\documentclass[aps,prd,twocolumn,superscriptaddress,10pt,nofootinbib,floatfix]{revtex4-1}

\usepackage[utf8]{inputenc} 	

\usepackage[TS1,T1]{fontenc} 	
\usepackage{microtype} 				

\usepackage[english]{babel}

\usepackage{amsmath} 
\usepackage{amssymb} 
\usepackage{amsthm} 
\usepackage{bm} 
\usepackage[breaklinks=true, colorlinks=false, pdfpagelabels, hidelinks]{hyperref} 
\usepackage{braket} 
\usepackage{mathtools} 
\usepackage[capitalize]{cleveref} 
\usepackage{dsfont} 
\usepackage{nicefrac} 
\usepackage{tensor} 
\usepackage{paralist} 
\usepackage{graphicx} 
\usepackage{color} 
\usepackage{simplewick} 
\usepackage{slashed} 
\usepackage[usenames,dvipsnames]{xcolor}
\usepackage[normalem]{ulem}

\newcommand{\source}[1]{Ref.~\cite{#1}} 
\newcommand{\sources}[1]{Refs.~\cite{#1}} 
\crefname{section}{Sec.}{Secs.}
\Crefname{section}{Sec.}{Secs.}
\crefname{lem}{lemma}{lemmas}
\Crefname{lem}{Lemma}{Lemmas}
\crefname{thm}{theorem}{theorems}
\Crefname{thm}{Theorem}{Theorems}

\newcommand{\lat}[1]{\textsl{#1}} 
\newcommand{\ii}{\mathrm{i}}
\newcommand{\dd}{\mathrm{d}}

\renewcommand{\vec}[1]{{\bm#1}} 
\newcommand{\op}[1]{#1} 
\newcommand{\Iop}[2][]{\op{#2}^{#1}} 
\newcommand{\conj}[1]{\overline{#1}}  

\newcommand{\define}{\equiv} 
\newcommand{\defineright}{\equiv} 
\DeclareMathOperator{\normalord}{:} 

\newcommand{\id}{\text{\upshape\sffamily \openone}} 

\newcommand{\laplace}{\Delta} 
\newcommand{\dalembert}{\Box} 

\DeclareMathOperator{\tr}{tr}
\DeclareMathOperator{\sgn}{sgn}
\DeclareMathOperator{\varRe}{Re}

\newcommand{\Z}{\mathds{Z}} 
\newcommand{\R}{\mathds{R}} 
\newcommand{\C}{\mathds{C}} 
\DeclareMathOperator{\Cliff}{Cliff} 

\newcommand{\Iket}[1]{\ket{#1}}

\newcommand{\Ibraket}[1]{\braket{#1}}

\newcommand{\inner}[2]{\braket{#1\,,\,#2}} 
\newcommand{\varinner}[2]{(#1\,,\,#2)}  

\theoremstyle{plain}

\newtheorem{thm}{Theorem}[section]
\newtheorem{lem}[thm]{Lemma}

\newtheorem*{WD}{Auxiliary diagrams for the detector}
\newtheorem*{RL}{Auxiliary diagrams for quantized real fields, linear coupling}
\newtheorem*{RQ}{Auxiliary diagrams for quantized real fields, quadratic coupling}
\newtheorem*{CQ}{Auxiliary diagrams for quantized complex fields, quadratic coupling}
\newtheorem*{SQ}{Auxiliary diagrams for quantized spinor fields, quadratic coupling}
\newtheorem*{GFR}{Generic Feynman rules}
\newtheorem*{EFD}{Elements in Feynman diagrams}
\newtheorem*{CFD}{Constructing a Feynman diagram}
\newtheorem*{EFDMF}{Elements in Feynman diagrams -- massive field}
\newtheorem*{EFDMLF}{Elements in Feynman diagrams -- massless field}

\renewcommand{\paragraph}[1]{\textit{#1}.---~}

\begin{document}
\selectlanguage{english} 


\title{Renormalized Unruh-DeWitt particle detector models for boson and fermion fields}
\author{Daniel Hümmer}
\date{\today}
\affiliation{Institut für Theoretische Physik, Ruprecht-Karls-Universität Heidelberg, Philosophenweg 16, 69120 Heidelberg, Germany}
\affiliation{Department of Applied Mathematics, University of Waterloo, 200 University
Avenue West, Waterloo, Ontario, N2L 3G1, Canada}
\author{Eduardo Mart\'{i}n-Mart\'{i}nez}
\affiliation{Department of Applied Mathematics, University of Waterloo, 200 University
Avenue West, Waterloo, Ontario, N2L 3G1, Canada}
\affiliation{Institute for Quantum Computing, University of Waterloo, Waterloo, Ontario, N2L 3G1, Canada}
\affiliation{Perimeter Institute for Theoretical Physics, 31 Caroline Street North, Waterloo, Ontario, N2L 2Y5, Canada}
\author{Achim Kempf}
\affiliation{Department of Applied Mathematics, University of Waterloo, 200 University
Avenue West, Waterloo, Ontario, N2L 3G1, Canada}
\affiliation{Institute for Quantum Computing, University of Waterloo, Waterloo, Ontario, N2L 3G1, Canada}
\affiliation{Perimeter Institute for Theoretical Physics, 31 Caroline Street North, Waterloo, Ontario, N2L 2Y5, Canada}

\begin{abstract}
Since quantum field theories do not possess proper position observables, Unruh-DeWitt detector models serve as a key theoretical tool for extracting localized spatiotemporal information from quantum fields. Most studies have been limited, however, to Unruh-DeWitt (UDW) detectors that are coupled linearly to a scalar bosonic field. Here, we investigate UDW detector models that probe fermionic as well as bosonic fields through both linear and quadratic couplings. In particular, we present a renormalization method that cures persistent divergencies of prior models. We then show how perturbative calculations with UDW detectors can be streamlined through the use of extended Feynman rules that include localized detector-field interactions. Our findings pave the way for the extension of previous studies of the Unruh and Hawking effects with UDW detectors, and provide new tools for studies in relativistic quantum information, for example,  regarding relativistic quantum communication and studies of the entanglement structure of the fermionic vacuum.
\end{abstract}
\maketitle


\section{Introduction}                                                

While the spatial and temporal meaning of wave functions in first quantization is straightforward, it is notoriously difficult to extract spatiotemporal information from the quantum fields of second quantization. This is in large part due to the fact that there are no position observables in quantum field theory (QFT) \cite{DannyLoc}. As a consequence, one often works with nonlocal field modes instead, for example, in order to calculate $S$-matrix elements in the case of particle physics or Bogolubov transforms in the case of QFT on curved space. 

Actual experiments, of course, probe quantum fields on finite patches of spacetime. A technique which does allow one to model the extraction of localized spatiotemporal information from quantum fields was pioneered by Unruh and DeWitt \cite{unruh_notes_1976,seligman_dewitt_quantum_1979}. The key idea is to probe quantum fields by coupling the quantum field in question to a localized first quantized system, called a detector. An excitation of the detector system is then interpreted 
as the absorption and therefore detection of a particle from the 
quantum field. 

For example, a hydrogen atom, first quantized, can serve as such a detector system for the photons of the second quantized electromagnetic field. The technique of probing quantum fields by coupling them (at least in gedanken experiments) to so-called Unruh-DeWitt (UDW) detectors has been used successfully, in particular, to analyze the Hawking and Unruh effects. One of the key insights gained from these studies has been the finding that and why the very notion of particle is observer dependent. 

More recently, UDW detectors have been used extensively in  studies on quantum communication via field quanta \cite{Jonsson2015,jonsson_quantum_2014} and, more generally, in studies on a host of effects related to the presence of spatially distributed entanglement in the quantum field theoretical vacuum \cite{Algebra1,Valentini1991,Reznik2003,VerSteeg2009,Farming}. In spite of these successes, however, most of these studies have been limited in scope to the model of a UDW detector that probes a quantized real scalar boson field through a linear coupling. 

Here, our aim is to widen the applicability of the general UDW detector approach. 
To this end, we investigate Unruh-DeWitt detector models that probe bosonic as well as fermionic, nonscalar fields through both linear and quadratic couplings. We show how certain persistent divergencies can be overcome through renormalization and we show how perturbative calculations with UDW detectors can be streamlined through the use of extended  Feynman rules that include the scattering of detector excitations.

\medskip{}

Historically, idealized particle detector models were first introduced by Unruh \cite{unruh_notes_1976}, in an attempt to resolve the well-known ambiguity in defining field states corresponding to physical particles on curved spacetimes or for noninertial observers \cite{fulling_nonuniqueness_1973,davies_scalar_1975,unruh_notes_1976}. This was the finding that different quantization procedures yield incompatible Fock spaces and thus lead to ambiguous definitions of particle number eigenstates \cite{grove_notes_1983,birrell_quantum_1984}. Unruh's groundbreaking idea is best summarized in his now famous dictum ``a particle is what a particle detector detects.'' He initially modeled a particle detector for a quantized scalar field as being a quantized scalar field itself, albeit one that is restricted to a cavity so that its excitation events come with spatiotemporal information. Unruh calculated that, when uniformly accelerated through the Minkowski vacuum, this detector will measure a flux of particles that are thermally distributed, which established the famous Unruh effect \cite{unruh_notes_1976,birrell_quantum_1984,crispino_unruh_2008}.

A few years later, DeWitt improved and simplified Unruh's original model by introducing approximations that effectively replace the detecting field with a nonrelativistic two-level quantum system \cite{seligman_dewitt_quantum_1979}, inventing what is now called the Unruh-DeWitt detector. The detector couples directly to the field $\Phi$ through its monopole moment $\op\mu$:
\begin{equation}\label{eqn: standard UDW detector}
  \op{H}_{\mathrm{int}} \propto \op\mu \op\Phi~.
\end{equation}
This model fixed minor problems with Unruh's initial suggestion \cite{grove_notes_1983}. More importantly, the UDW detector has the advantage of being much easier to work with in calculations, because the detector is described by quantum mechanics rather than quantum field theory. Other detector models for quantized scalar fields followed, using both two-level systems and full fields as detectors, and featuring different ways to couple the detector to the fundamental quantum field (see e.g.\,\cite{grove_notes_1983,hinton_particle_1983,hinton_particle_1984,sriramkumar_response_2001}). In the present paper, we will generally call any model that uses DeWitt's monopole as the detecting system a (UDW-type) detector, irrespective of the field it probes and the coupling it uses.

UDW particle detector models have proven to be very versatile and useful tools. They were first used to analyze the effects of accelerations, spacetime curvature and horizons on scalar quantum fields (see e.g.\,\cite{unruh_notes_1976,candelas_irreversible_1977}). Further, UDW detectors have been used in quantum field theory to quantify vacuum fluctuations and the structure of vacuum states \cite{takagi_vacuum_1986,reznik_violating_2005}. In relativistic quantum information, pairs of detectors are crucial for measuring and harvesting the entanglement of quantum fields \cite{cliche_vacuum_2011,martin-martinez_sustainable_2013,lin_entanglement_2010,olson_entanglement_2011}, and serve as senders and recipients of information in quantum signaling \cite{cliche_relativistic_2010,jonsson_quantum_2014}. Moreover, they can be good toy models for light-matter interaction in quantum optics \cite{martin-martinez_wavepacket_2013,Alhambra2013}.

Because UDW detector models for quantized scalar fields have proven so valuable, the question arises to what extent analogous techniques can be developed and applied to other types of fields, such as quantized spinor or vector fields, which in Nature are of course more common than quantized scalar fields. This question is of fundamental interest by itself and there are also concrete applications where detector models for quantized spinor fields could immediately be useful, such as for analyzing the fermionic Unruh effect \cite{soffel_dirac_1980}, and for investigating the related Hawking radiation of Dirac particles \cite{hawking_particle_1975}. Moreover, it has recently been pointed out that such more general UDW detector models could resolve ambiguities that arise in defining measures for the entanglement between disjoint patches of a fermionic field \cite{montero_fermionic_2011}.

A key aim of the present work is, therefore, to find a counterpart of the UDW detector for quantized spinor fields. More precisely, we are interested in a pair of particle detector models, one each for quantized scalar and spinor fields, that is comparable in the sense that differences in the detectors' reaction to scalar and to spinor fields are due to the probed field alone, not caused by any other feature of the respective detector models. We will focus on only half-integer spin fermions and integer-spin bosons, i.e., we will here not be concerned with nontrivial spin/statistic combinations that may arise in low dimensional systems. 

To the best of our knowledge, only two detector models for quantized spinor fields have been suggested so far: the first one by Iyer and Kumar \cite{iyer_detection_1980} uses a quantized scalar field in a cavity as detecting system, close in spirit to Unruh's original suggestion. The second one by Takagi \cite{takagi_response_1985,takagi_vacuum_1986} seeks to imitate DeWitt's simpler model by coupling an UDW-type two-level system to a quantized spinor field $\Psi$ through
\begin{equation}\label{eqn: Takagi's model}
  \op{H}_{\mathrm{int}} \propto \op\mu \op{\conj\Psi}\op\Psi~,
\end{equation}
making it a likely candidate for the closest spinor field equivalent to the scalar field UDW detector. While the above interaction may not be immediate from first principles, that is, from the Standard Model of particle physics, we may for example think of it as an attempt at a first-quantized, simplified version of a second-quantized cavity detector: in quantum electrodynamics, a spinor field (electrons) is coupled to a vector field $A_\mu$ (photons) through
\begin{equation}\label{eqn: interaction H QED}
  \op{H}_{\mathrm{QED}} \propto A_{\mu}\conj{\Psi}\gamma^{\mu}\Psi ~.
\end{equation}
In this logic, certain modes of the electromagnetic field restricted to a cavity with appropriate boundary conditions (realized, e.g., as superconducting mirrors) could serve as a detector for the electron field. The spatiotemporal profile of the detector is then determined by the extent of the cavity and the time over which the electromagnetic field is observed. It is conceivable that by a series of approximations---e.g., neglecting all but one mode in the cavity and restricting its occupation to at most one quantum---we could arrive at a model of the type of \cref{eqn: Takagi's model}. A two-level system coupled through \cref{eqn: Takagi's model} may then be thought of as an approximation to a cavity detector, much in the same way that the UDW detector is a simplification of Unruh's original cavity setup. In the following, however, we will keep with the original idea of such detectors as simple tools to obtain spatiotemporally resolved information about quantum fields, and set aside considerations of how exactly they could arise from first principles. This is justified by the fact that, as discussed above, UDW detectors provide a very useful tool for extracting spatiotemporal information from quantum fields, at least in gedanken experiments.

\medskip{}

We begin in \cref{sec: detector models} with an overview of the UDW-type detector models that we will study. In \cref{sec: VEP}, the excitation probability for the different detector models at rest in Minkowski vacuum is calculated for these models, also to test to what extent these models are comparable. We demonstrate that any UDW-type detector featuring an interaction Hamiltonian which is \emph{quadratic} in the field has divergent excitation probabilities. This is true for spinor field models like \cref{eqn: Takagi's model}, but also for quantized (real and complex) scalar field detector models with coupling
\begin{equation}\label{eqn: Hinton's model}
  \op{H}_{\mathrm{int}} \propto \op\mu \op\Phi^{\dagger}\Phi~.
\end{equation}
Unlike divergent probabilities reported before \cite{grove_notes_1983,takagi_vacuum_1986,sriramkumar_finite-time_1996,louko_transition_2008}, it is not possible to regularize these divergencies by considering a detector of finite size and switching it on and off adiabatically. However, the divergencies are readily understood from a field theoretic point of view: In \cref{sec: renormalizing} we show that they are closely related to tadpoles in quantum electrodynamics, and this analogy leads us to a renormalization scheme for quadratically coupled detectors. \Cref{sec: comparison} is dedicated to the comparison of the vacuum excitation probabilities of the detector models. We establish that it is indeed justified to compare the two models \cref{eqn: Takagi's model} (for quantized spinor fields) and \cref{eqn: Hinton's model} (for quantized complex fields), at least at leading order in the coupling constant. As an example, their response to massless quantum fields in $(1,1)$ dimensions is studied explicitly, using both a sudden switching of the detector, as well as a Gaussian switching function, in addition to a Gaussian spatial detector profile.

The second part of this work is dedicated to the derivation of fully fledged Feynman rules for all detector models discussed, valid on Minkowski spacetime. The Feynman rules are to facilitate investigating whether the discussed models are finite, or at least renormalizable, and whether they are comparable at higher orders as well. \Cref{sec: computation methods} summarizes the required calculation methods, such as the Feynman propagators and Wick's theorem, with the proofs moved to \cref{app: proofs wick} for better readability. 
After applying these methods in \cref{sec: VNRP} to the probability for a detector to remain in its ground state when interacting with a field in the vacuum state (the vacuum no-response probability), the Feynman rules then follow in \cref{sec: feynman rules}.

In \cref{app: quantum fields}, we summarize canonical results in classical and quantum field theory to introduce notation as well as for the convenience of the reader. \Cref{app: normalization} explains the origin of the different normalization conditions for scalar and spinor field mode functions, which plays a crucial role in the UV behavior of the theories.

\medskip{}

Throughout this work, natural units $c = 1 = \hbar$ are used. The metric of $(n+1)$-dimensional Minkowski spacetime has signature $(+,-,\cdots,-)$.


\section{Unruh-DeWitt-type particle detector models}\label{sec: detector models}
Unruh-DeWitt particle detectors are two-dimensional quantum systems. We choose the energy levels to be zero and some $\Omega > 0$ so that the detector's Hamiltonian, which generates translations in its proper time $\tau$, reads
\begin{equation}\label{eqn: H detector}
  \op{H}_{\mathrm{d}} = \Omega \ket{e}\bra{e} = \Omega \,\op\sigma^{+}   \op\sigma^{-}~.
\end{equation}
Here, the ladder operators
\begin{align}\label{eqn: detector ladder operators}
  \op\sigma^{+} &= \ket{e}\bra{g}~, & \op\sigma^{-} &= \ket{g}\bra{e}
\end{align}
are expressed in term of the ground and excited states $\ket{g}$ and $\ket{e}$. 
We parametrize the detector's trajectory $x(\tau) = (t(\tau), \vec{x}(\tau))$ by its proper time. 
Since the detector's trajectory will be prespecified rather than dynamical, spacetime translational symmetry is broken and overall energy and momentum are not conserved (unless one considers also the agent that keeps the detector on its trajectory).
The detector's effective spatial profile, which may be thought of as the distribution of its charge, centered around $\vec{x}$, will be described by a function
\begin{equation}
  p(\vec{x},\cdot): \R^{n}\to \R
\end{equation}
that is normalized to one unit of charge,
\begin{equation}\label{eqn: normalization spatial profile}
  \int_{\R^{n}}p(\vec{x},\vec{y})\,\dd \vec{y} = 1~,
\end{equation}
where $n$ is the number of spatial dimensions. For example, for a pointlike detector
\begin{equation}
  p(\vec{x},\vec{y}) = \delta(\vec{x}-\vec{y})~,
\end{equation}
while for a detector with Gaussian profile
\begin{equation}
  p(\vec{x},\vec{y}) = (2\pi\sigma^{2})^{-n/2}\,\,\exp\left(-\frac{(\vec{y}-\vec{x})^{2}}{2\sigma^{2}}\right)~.
\end{equation}
It will sometimes be convenient to combine the switching function $\chi$ and the spatial profile $p$ into a \emph{spacetime} profile $f$ by defining
\begin{equation}\label{eqn: spacetime profile}
 f(x) = f(t,\vec{x}) \define \chi(t)\,p(0,\vec{x})  
\end{equation}
in the detector's rest frame. 

An UDW detector is to detect field quanta by becoming excited. To this end, the field must couple to an operator, $\mu$, of the UDW detector that does not commute with the detector`s Hamiltonian $\op{H}_{\mathrm{d}}$. Without restricting generality, one can define this so-called \emph{monopole moment} operator of the detector to be
\begin{equation}\label{eqn: monopole operator}
  \op\mu = \op\sigma^{+} + \op\sigma^{-}.
\end{equation}
Let us now systematically consider how the monopole moment $\mu$ of an UDW detector can be coupled to the various types of fields.

\subsection{Linear coupling}\label{sec: linear coupling}
We begin by considering a pointlike UDW-type detector coupled to a quantized real field  $\op\Phi$ through the interaction Hamiltonian
\begin{equation}
  \op{H}_{\mathrm{int}} = \lambda\, \op{\mu}\, \op{\Phi}(\vec{x}) \label{12}~,
\end{equation}
where $\lambda \in \R$ is the coupling strength. This is the particle detector model originally proposed by DeWitt \cite{seligman_dewitt_quantum_1979}. It can be viewed as an idealization of the two lowest states of an atom interacting with the electromagnetic field. 

As is well known, UV divergencies in the excitation probabilities arise in this setup, in particular, if the detector is abruptly switched (i.e. coupled) on or off; see, e.g., \cite{grove_notes_1983,sriramkumar_finite-time_1996,louko_transition_2008}. The pointlike structure of this detector also leads to UV divergencies; see, e.g.,\,\cite{grove_notes_1983,takagi_vacuum_1986}. The divergencies can be regularized by switching the detector on and off gradually, i.e., by using a sufficiently smooth switching function $\chi(t)$, and by endowing the detector with a finite spatial profile $p$: 
\begin{equation}\label{eqn: interaction H real linear}
  \op{H}_{\mathrm{int}}(\tau) = \lambda\chi(\tau)\, \op{\mu} \int_{\R^{n}}p(\vec{x}(\tau),\vec{y})\,\op{\Phi}(\vec{y}) \,\dd \vec{y}~.
\end{equation}
Let us now consider the coupling of an UDW detector to a quantized complex field, $\op\Phi$.
In this case, one cannot couple them as in \cref{12} because $\op\Phi^{\dagger} \neq \op\Phi$ would imply that the interaction Hamiltonian is not self-adjoint. In principle, one may instead view $\op\Phi$ as composed of two quantized real fields $
  \op\Phi =  \frac{1}{\sqrt{2}}\left( \op\Phi_{1} + \ii  \op\Phi_{2}\right)
$ and couple any real linear combination of $\op\Phi_1$ and $\op\Phi_2$ to the detector. This approach, however, would single out a direction in the complex plane and would therefore make any $U(1)$ symmetry for $\op\Phi$ impossible \cite{takagi_vacuum_1986}. The same incompatibility arises with any unitary symmetry if $\op\Phi$ carries higher-dimensional complex group and/or spin representations. We conclude that the coupling of UDW detectors to quantized complex fields should not be linear.

\subsection{Quadratic coupling}\label{sec: quadratic coupling}
Any interaction Hamiltonian between an UDW detector and a field must be a self-adjoint scalar.
Since an UDW detector's monopole moment is a self-adjoint Lorentz scalar, UDW detectors always need to couple to an expression in the field which is also both scalar and self-adjoint. In the case of a quantized spinor field, the simplest self-adjoint Lorentz scalar is $\conj\Psi \Psi = \Psi^{\dagger}\gamma^{0}\Psi$. This suggests the following interaction Hamiltonian:
\begin{equation}\label{eqn: interaction H spinor quad}
  \op{H}_{\mathrm{int}}(\tau) = \lambda\chi(\tau)\, \op{\mu} \int_{\R^{n}}p(\vec{x}(\tau),\vec{y})\,\op{\conj\Psi}(\vec{y})\op{\Psi}(\vec{y}) \,\dd \vec{y}~.
\end{equation}
Note that the interaction is now quadratic in the field. This is plausible also because an UDW detector coupled linearly to a fermion field  would be able to violate fermion number conservation by creating and annihilating individual fermions. Coupling an UDW detector quadratically to the field ensures that the detector can only pair create or annihilate fermion antifermion pairs. An example of this type of coupling was first considered by Takagi \cite{takagi_response_1985,takagi_vacuum_1986} and has been used since; see \cite{diaz_radiative_2002,langlois_causal_2006,bessa_accelerated_2012,harikumar_uniformly_2013}.

It is straightforward to couple UDW detectors quadratically not only to quantized spinor fields but also to quantized (real and complex) scalar fields while preserving symmetries:
\begin{equation}\label{eqn: interaction H scalar quad}
  \op{H}_{\mathrm{int}}(\tau) = \lambda\chi(\tau)\, \op{\mu} \int_{\R^{n}}p(\vec{x}(\tau),\vec{y})\,\op{\Phi}^{\dagger}(\vec{y})\op{\Phi}(\vec{y}) \,\dd \vec{y}~.
\end{equation}
This type of coupling was first suggested by Hinton for quantized real fields \cite{hinton_particle_1984}. It may be possible to justify interpreting this model as an effective description of some fundamental interaction, for example when the scalar field describes composite particles at sufficiently low energy. However, in the present work, we  will use it as a theoretical tool for investigating what difference it makes to an UDW detector's behavior whether it is coupled to a bosonic (scalar) or to a fermionic (spinor) field. This is because to this end we can now compare model \cref{eqn: interaction H scalar quad} and model \cref{eqn: interaction H spinor quad} which are both quadratically coupled, i.e., they differ essentially only in the spinor structure but not in the type of coupling. We will separately investigate the impact of quadratic versus linear coupling, such as model \cref{eqn: interaction H scalar quad} versus model \cref{eqn: interaction H real linear}.

In \cref{tab: overview detector models}, the four different detector models we will discuss are summarized and numbered for easy reference: model 1 for quantized real fields couples linearly through \cref{eqn: interaction H real linear}; model 2 couples quadratically to quantized real fields; model 3 to quantized complex; and model 4 to quantized spinor fields with the Hamiltonians \cref{eqn: interaction H scalar quad} and \cref{eqn: interaction H spinor quad}, respectively.

\begin{table}
  \includegraphics{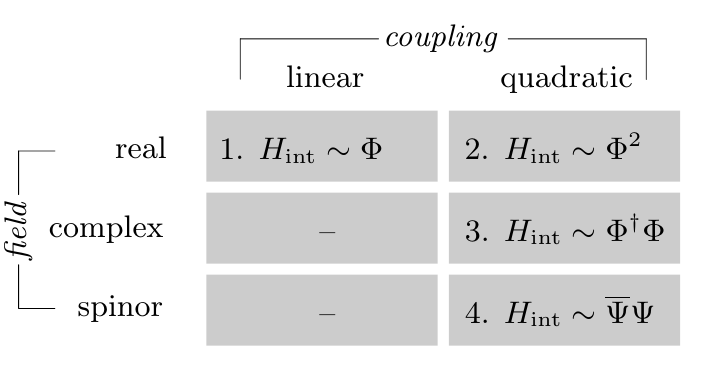}
  \caption{Overview of the particle detector models discussed in this paper -- listed is the power of the field to which the detector (i.e. $\mu$) is coupled.}
  \label{tab: overview detector models}
\end{table}

\section{First order: excitation probability in the vacuum}\label{sec: VEP}

For simplicity, only UDW-type detectors that are at rest in n+1 dimensional Minkowski spacetime will be considered in our investigation of the different detector models. In particular, we study vacuum excitation probabilities (VEP), i.e., probabilities $P_{0,g \to e}$ of the detector evolving from its ground state $\ket{g}$ to its excited state $\ket{e}$ when the field is initially in the vacuum state $\ket{0}$. We calculate the {vacuum excitation probability} or \emph{vacuum response} for the four different detector models under consideration, up to leading order in perturbation theory.  Calculation of the VEP serves a double purpose:\begin{inparaenum}[(a)]\item to demonstrate the need for renormalization of the models 2, 3 and 4 in \cref{tab: overview detector models} and, after successful renormalization, \item as a means to gauge whether these two models can be viewed as equivalent.\end{inparaenum}

For concrete calculations, we consider both bosonic and fermionic quantum fields $F$ quantized in the cylinder with toroidal spatial sections consisting of the following product of intervals with periodic boundary conditions:
\begin{equation}\label{eqn: cavity}
  B = [-\nicefrac{L}{2},\nicefrac{L}{2}]^{n} \subset \R^{n}~.
\end{equation}
Therefore,
\begin{equation}
  F(t,\vec{x} + \vec{\delta}) = F(t,\vec{x}) \quad\quad \forall\,\vec{\delta} \in L\cdot\Z^{n}.
\end{equation}

\paragraph{The zero mode}\label{sec: zero mode}
In the cylinder quantization, the field boundary conditions  allow for constant classical solutions, that is, solutions of vanishing momentum $\vec{k}$ and energy $\omega_{\vec{k}}$. These degrees of freedom mathematically behave like free particles, rather than harmonic oscillators. They therefore have to be quantized differently from the rest of the field modes. Since there is no canonical vacuum state for such a zero mode, further subtleties appear when we couple the zero mode to particle detectors \cite{martin-martinez_particle_2014}.

One could get rid of all the problems of the presence of a zero mode by using alternative boundary conditions which, unlike the periodic boundary conditions, do not permit solutions of zero momentum. A simple such choice is the use of Dirichlet boundary conditions, $F(t,\vec{x}) = 0$ for all $\vec{x} \in \partial B$, such that solutions of zero momentum vanish everywhere. These boundary conditions are straightforward to implement for scalar fields, but full Dirichlet boundary conditions are too restrictive for spinor fields in the Dirac representation on $\C^{4}$ because the only solution would be $\Psi = 0$. Instead, a weaker variant can be realized: imposing Dirichlet boundary conditions on the upper (or lower) two components of a four component spinor is possible; see for example \cite{alonso_boundary_1997}. However, it is not as clear how to construct a basis of the solution space that consists of eigenfunction of $\partial_{0}$ (corresponding to energy eigenfunctions in quantum theory) and $\sigma_{3}$ (corresponding to spin eigenfunctions), which makes quantization of the theory too involved. In this work we will therefore keep to periodic boundary conditions and simply exclude the zero mode from interactions, as described above.

Fortunately, in most cases, it is possible to choose ``safe'' states for the zero mode that would minimize its impact on the dynamics of the particle detectors \cite{martin-martinez_particle_2014}. Based on this fact, in this work  we will assume \lat{ad hoc} that the detector does not couple to the zero mode of the field, as it is often done in quantum optics. The question in what circumstances dropping the zero mode can be justified is not vital to the discussion here, and is still subject to active research.

\paragraph{Defining the vacuum response}
We assume that both field and detector are in their respective free ground states at time $t_{0}$. The state of the composite system is then characterized by the density operator
\begin{equation}
  \op{\rho}(t_{0}) = \ket{0,g;t_{0}}\bra{0,g;t_{0}} = \ket{0,g}\bra{0,g}~.
\end{equation}
The coupled system is allowed to evolve unitarily until time $t = t_{0}+T$; time-evolution is performed using the time-evolution operator $\op{U}$ generated by the total Hamiltonian
\begin{equation}        
\op{H} = \op{H}_{F} + \op{H}_{\mathrm{d}} + \op{H}_{\mathrm{int}}~,
\end{equation}
where $\op{H}_{F}$ is the Hamilton operator of the free field, $\op{H}_{\mathrm{d}}$ is given by \cref{eqn: H detector}, and $\op{H}_{\mathrm{int}}$ depends on the model under consideration. The system ends up in the state
\begin{equation}
  \op{\rho}(t) = \op{U}(t,t_{0})\, \op\rho \, \op{U}^{\dagger}(t,t_{0})
\end{equation}
and the probability for the detector to be excited is encoded in the corresponding component of the density operator, after tracing out the field:
\begin{equation}\label{eqn: VEP 1}
  P_{0,g\to e}(t,t_{0}) = \braket{e|\tr_{F} \op\rho(t)|e}~.
\end{equation}
At the end, we take the limit
\begin{equation}\label{eqn: VEP 1 t limit}
  P_{0,g\to e} = \lim_{\substack{%
    t_{0}\to -\infty\\%
    t\hphantom{_{0}}\to+\infty}%
  } P_{0,g\to e}(t,t_{0})
\end{equation}
such that the actual duration of the interaction is determined by the switching function $\chi(t)$  alone: at times outside the support of $\chi$, both systems evolve freely.

It is convenient to reformulate \cref{eqn: VEP 1} in terms of state vectors in the interaction picture:
\begin{equation}\label{eqn: VEP 2}
  P_{0,g\to e}(t,t_{0}) = \sum_{a} \left|\Ibraket{a,e;t| \Iop{U}(t,t_{0})|0, g;t_{0}}\right|^{2}~,
\end{equation}
where $\ket{a}$ is an orthonormal basis of the Hilbert space of states of the free quantum field $F$. Recall that the interaction picture is related to the Schrödinger picture through the partial inverse time evolution $ \op{U}_{0}(t) = e^{\ii  \op{H}_{0}t}$
which only takes into account the free Hamiltonian $\op{H}_{0} = \op{H}_{F} + \op{H}_\mathrm{d}$ \cite{sakurai_modern_2011}.

Finally, we set up the usual perturbation theory under the assumption that $\lambda$ is small. To this end, the time-evolution operator is expanded in the Dyson series
\begin{equation}\label{eqn: perturbative expansion}
  \Iop{U}(t,t_{0}) = \sum_{n=0}^{\infty} \frac{(-\ii  \lambda)^{n}}{n!}\, \op{U}^{(n)}(t,t_{0})~,
\end{equation}
where the operator at order $n$ is \cite{bjorken_relativistic_1965}
\begin{equation}\label{eqn: time-evolution order n}
  \op{U}^{(n)}(t,t_{0}) =
   \frac{1}{\lambda^{n}}T\prod_{i=1}^{n}\int_{t_{0}}^{t}\dd t_{i}  \Iop{H}_{\mathrm{int}}(t_{i})~.
\end{equation}

For $n=0$ we have $\op{U}^{(0)} = \id$ (no change over time at all), but since $\braket{e|g} = 0$, order zero does not contribute to $P_{0,g \to e}$ in \cref{eqn: VEP 2}. Therefore,
\begin{align}\label{eqn: VEP 3}
  P_{0,g\to e}(t,t_{0}) &\approx \sum_{a} \left|-\ii  \lambda\Ibraket{a,e;t| \op{U}^{(1)}(t,t_{0})|0, g;t_{0}}\right|^{2}\nonumber\\
  &\defineright \sum_{a} \left|\mathcal{A}_{a,e}^{(1)}\right|^{2}
\end{align}
to leading order in $\lambda$, where $\mathcal{A}_{a,e}^{(1)}$ is simply the probability amplitude for the system to evolve from $\ket{0,g}$ to $\ket{a,e}$ at leading order.

In the following, we calculate the VEP for quantized scalar fields in linear coupling, as well as for quantized scalar and spinor fields in quadratic coupling (all models in \cref{tab: overview detector models}).

\subsection{Linear coupling to quantized scalar fields}\label{sec: VEP real field}

\subsubsection{Vacuum excitation probability}

In preparation for the derivation of  new results concerning fermionic detector models, we begin by rederiving the standard case of a quantized real field $\Phi$ (see, among others, \cite{svaiter_inertial_1992,sriramkumar_finite-time_1996,louko_transition_2008}). In order to work as generally as possible, we will avoid choosing a particular spatial profile and switching function. We couple an UDW-type detector to the quantized scalar field using the interaction Hamiltonian \cref{eqn: interaction H real linear} (model 1 in \cref{tab: overview detector models}).  We are dealing with a detector at rest, so its trajectory is
\begin{align}\label{eqn: trajectory at rest}
  \vec{x}(\tau) &= \vec{x}_{0}~, & 
  t(\tau) &= \tau~.
\end{align}
The time-evolution operator at leading order reads
\begin{multline}
  \op{U}^{(1)}(t,t_{0}) =\\
   \int_{t_{0}}^{t}\dd t' \chi(t')\,\Iop{\mu}(t')\int_{\R^{n}} \dd \vec{y}\,p(\vec{x}_{0},\vec{y})\,\Iop{\Phi}(t',\vec{y})
\end{multline}
according to \cref{eqn: time-evolution order n}. Notice that for a sufficiently localized spatial profile $p$ (as for example a Gaussian profile), if the localization length $\sigma$ is much smaller than the compactification length $L$, we can extend the integration region to the full space to a good approximation.

Here, the monopole operator in the interaction picture is 
\begin{equation}\label{eqn: interaction p monopole operator}
  \Iop{\mu}(t) = e^{\ii\Omega t}\op{\sigma}^{+} + ^{-\ii  \Omega t}\op{\sigma}^{-}~.
\end{equation}
Similarly, the field operator
\begin{multline}\label{eqn: interaction p real field operator}
  \Iop\Phi(t,\vec{x})  = \sum_{\vec{k}}\frac{1}{\sqrt{2\omega_{\vec{k}}}} \Big[ \op{a}_{\vec{k}}\,e^{-\ii   \omega_{\vec{k}}t}\,\varphi_{\vec{k}}(\vec{x}) \\
  \hphantom{=~}+\op{a}^{\dagger}_{\vec{k}}\,e^{\ii \omega_{\vec{k}}t} \,\varphi_{\vec{k}}^{*}(\vec{x}) \Big]~,
\end{multline}
where $\omega_{\vec{k}}$ is the energy of a single mode, the ladder operators $\op{a}_{\vec{k}}$ satisfy the usual canonical commutation relations \cref{eqn: cr ladder op real field}, and the mode functions $\varphi_{\vec{k}}$ \cref{eqn: scalar mode functions 2} form an orthonormal basis of the solution space of the Klein-Gordon equation.  The field is restricted to a cylinder, so the momentum spectrum is discrete. 

Since the detector does not couple to the zero mode, we have
\begin{equation}
  \vec{k} \in \frac{2\pi}{L} \cdot \Z^{n}\setminus \{0\} ~.
\end{equation}
Note that we have not yet chosen a particular form for the  switching function or spatial profile.

From \cref{eqn: VEP 3}, the only nonvanishing amplitudes are $\mathcal{A}_{1_{\vec{k}},e}^{(1)}$. Therefore the time evolved states will be superpositions of the ground state $\ket{0,g}$ and states of the form $\ket{1_{\vec{k}},e}$, which feature a single quantum in the field. Summing up the squared moduli of the amplitudes and taking the time limits leads to the VEP:
\begin{multline}\label{eqn: VEP real linear 1}
  P_{0,g\to e} = \frac{\lambda^{2}}{2}\sum_{\vec{k}} \frac{1}{\omega_{\vec{k}}}\left| \int_{\R^{n}}p(\vec{x_{0},\vec{y}})\,\varphi_{\vec{k}}(\vec{y})\,\dd \vec{y}\right|^{2}\\
  \times \left| \int_{-\infty}^{\infty} \chi(t)\,e^{\ii (\Omega + \omega_{\vec{k}}) t}\, \dd t\right|^{2}~.
\end{multline}
Solving the Klein-Gordon equation in a cavity with periodic boundary conditions yields (see \cref{app: quantum fields}) for the spatial part of the mode function 
\begin{equation}\label{eqn: mode functions real field 1}
  \varphi_{\vec{k}}(\vec{x}) = \frac{1}{\sqrt{L^{n}}}~ e^{ \ii   \vec{k}\cdot \vec{x}}~,
\end{equation}
so the final expression for the VEP reads
\begin{multline}\label{eqn: VEP real linear 2}
  P_{0,g\to e} = \frac{\lambda^{2}}{2L^{n}}\sum_{\vec{k}} \frac{1}{\omega_{\vec{k}}}\left| \int_{\R^{n}}p(\vec{x_{0},\vec{y}})\,e^{-\ii  \vec{k}\vec{y}}\,\dd \vec{y}\right|^{2}\\
  \times \left| \int_{-\infty}^{\infty} \chi(t)\,e^{\ii (\Omega + \omega_{\vec{k}}) t}\, \dd t\right|^{2}~.
\end{multline}

\subsubsection[Dependence on Dimension]{Dependence on the dimension of spacetime}\label{sec: VEP discussion linear - dimension}

The momentum sum in \cref{eqn: VEP real linear 2} depends on the spatial dimension $n$, since $\vec{k} = (k^{1},k^{2},\dots,k^{n})$. Therefore the convergence behavior of $P_{0,g\to e}$ also strongly depends on the dimension. To demonstrate this, we assume the detector to be pointlike,
\begin{equation}\label{eqn: spatial profile pointlike detector}
  p(\vec{x},\vec{y}) = \delta(\vec{x}-\vec{y})~,
\end{equation}
and switch it on and off abruptly (sudden switching) by means of a window function:
\begin{equation}\label{eqn: sudden switching}
  \chi(t) = 
  \begin{cases}
    1 & \text{if } t \in [t_{0},t_{0}+T]\\
    0 & \text{else}
  \end{cases}~.
\end{equation}
This is the original UDW detector introduced by DeWitt \cite{seligman_dewitt_quantum_1979}. Plugging $\chi$ and $p(\vec{y},\cdot)$ in \cref{eqn: VEP real linear 2} and performing the time integrals, the VEP simplifies to
\begin{equation}\label{eqn: VEP real linear sudden switching pointlike detector}
   P_{0,g\to e} = \frac{2 \lambda^{2}}{L^{n}}\sum_{\vec{k}} \frac{1}{\omega_{\vec{k}}(\Omega+\omega_{\vec {k}})^{2}}  \sin^{2}\left( \frac{\Omega+\omega_{\vec{k}}}{2}T\right)~.
\end{equation}
In the massless case, the momentum sums are made explicit by writing out the dispersion relation
\begin{equation}\label{eqn: massless dispersion relation spelled out}
  \omega_{\vec{k}} = |\vec{k}| = \frac{2\pi}{L} \sqrt{l_{1}^{2}+l_{2}^{2}+\dots+ l_{n}^{2}}~,
\end{equation}
where $l_{i} \in \Z\setminus \{0\}$.

In one spatial dimension, we have
\begin{multline}\label{eqn: VEP real linear sudden switching pointlike 1 dim}
   P_{0,g\to e} = \frac{\lambda^{2}L^{2}}{2\pi^{3}}\sum_{l_{1} = 1}^{\infty}\\ \frac{1}{l_{1}\left(\frac{\Omega L}{2\pi}+l_{1}\right)^{2}}  \sin^{2}\left[ \frac{\pi}{L}\left(\frac{\Omega L}{2\pi}+l_{1}\right)T\right]~.
\end{multline}
This sum is bounded by
\begin{equation*}
   P_{0,g\to e} \leq \frac{\lambda^{2}L^{2}}{2\pi^{3}}\sum_{l_{1} = 1}^{\infty} \frac{1}{l_{1}^{3}}~,
\end{equation*}
which is convergent. 

For dimension $n=2$, the probability is
\begin{multline}\label{eqn: VEP real linear sudden switching pointlike 2 dim}
   P_{0,g\to e} = \frac{\lambda^{2}L}{\pi^{3}}\sum_{l_{1} = 1}^{\infty}\sum_{l_{2} = 1}^{\infty}\\ \left[\sqrt{l_{1}^{2}+l_{2}^{2}}\left(\frac{\Omega L}{2\pi}+\sqrt{l_{1}^{2}+l_{2}^{2}}\right)^{2}\right]^{-1}\\ \times  \sin^{2}\left[ \frac{\pi}{L}\left(\frac{\Omega L}{2\pi}+\sqrt{l_{1}^{2}+l_{2}^{2}}\right)T\right]~,
\end{multline}
which is similarly bounded by the convergent double sum
\begin{equation*}
   P_{0,g\to e} \leq \frac{\lambda^{2}L}{\pi^{3}}\sum_{l_{1} = 1}^{\infty}\sum_{l_{2} = 1}^{\infty} \frac{1}{\sqrt{l_{1}^{2}+l_{2}^{2}}^{3}}~.
\end{equation*}

In $n=3$ dimensions, however, the VEP diverges. Its formal expression is
\begin{multline}\label{eqn: VEP real linear sudden switching pointlike 3 dim}
   P_{0,g\to e} = \frac{2\lambda^{2}}{\pi^{3}}\sum_{l_{1} = 1}^{\infty}\sum_{l_{2} = 1}^{\infty}\sum_{l_{3} = 1}^{\infty}\\
   \left[\sqrt{\sum_{i=1}^{3}l_{i}^{2}}\left(\frac{\Omega L}{2\pi}+\sqrt{\sum_{i=1}^{3}l_{i}^{2}}\right)^{2}\right]^{-1}\\ \times  \sin^{2}\left[ \frac{\pi}{L}\left(\frac{\Omega L}{2\pi}+\sqrt{\sum_{i=1}^{3}l_{i}^{2}}\right)T\right]~.
\end{multline}
As $\Omega L(2\pi)^{-1} \geq 0$, we can estimate
\begin{multline}\label{eqn: first estimate VEP in 3 dim}
   \frac{\pi^{3}}{2\lambda^{2}}P_{0,g\to e} \geq \sum_{l_{1} = 1}^{\infty}\sum_{l_{2} = 1}^{\infty}\sum_{l_{3} = 1}^{\infty}\left[\frac{\Omega L}{2\pi}+\sqrt{\sum_{i=1}^{3}l_{i}^{2}}\right]^{-3}\\ \times  \sin^{2}\left[ \frac{\pi}{L}\left(\frac{\Omega L}{2\pi}+\sqrt{\sum_{i=1}^{3}l_{i}^{2}}\right)T\right]~.
\end{multline}
By exploiting $\cos(x) = \sin(x+\pi/2)$ and the Pythagorean trigonometric identity, it is straightforward to prove that the right-hand side converges for all choices of parameters if and only if
\begin{equation*}
  \sum_{l_{1} = 1}^{\infty}\sum_{l_{2} = 1}^{\infty}\sum_{l_{3} = 1}^{\infty} \left[A+\sqrt{\sum_{i=1}^{3}l_{i}^{2}}\right]^{-3}
\end{equation*}
does, where $A = \Omega L(2\pi)^{-1}$. This sum is in turn bounded by
\begin{equation}\label{eqn: last estimate VEP in 3 dim}
  \sum_{l_{1} = 1}^{\infty}\sum_{l_{2} = 1}^{\infty}\sum_{l_{3} = 1}^{\infty}\left[A+l_{1}+l_{2}+l_{3}\right]^{-3}~,
\end{equation}
which diverges logarithmically.

In higher dimensions convergence can only get worse, because even more sums will appear. For any dimension $(1,n)$ with $n\geq3$, the VEP is divergent \cite{louko_transition_2008}.

The above results show that the VEP strongly depends on the dimension of the underlying spacetime. In the above example, the vacuum response of a pointlike detector with sudden switching, coupled to a quantized real field is convergent in $(1,1)$ and $(1,2)$ dimensions, but divergent in all higher dimensions, necessitating regularization. 

\subsubsection[regularization]{Regularization through detector profile}\label{sec: VEP discussion linear - regularization}

Those divergences of the VEP can be understood as a result of the pointlike structure of the detector as well as the sudden switching \cite{louko_transition_2008}. In the linear coupling case, the divergences can be regularized by adiabatically switching the detector, or by ``smearing'' the interaction between detector and field over the spatial profile of the detector. Let us study both possibilities in two separate examples.

\paragraph{Gaussian switching and pointlike detector}
Consider the probability \cref{eqn: VEP real linear 2}, still with a pointlike detector \cref{eqn: spatial profile pointlike detector}, but this time employing a Gaussian switching function,
\begin{equation}\label{eqn: gaussian sw}
  \chi(t) = \exp\left(-\frac{t^{2}}{2T^{2}}\right)~,
\end{equation}
which is known to regularize the response of a resting detector in $(1,3)$-dimensional Minkowski spacetime \cite{sriramkumar_finite-time_1996}. We generalize this result to $(1,n)$ dimensions:
the time integral is readily solved, giving
\begin{equation}
  \int_{-\infty}^{\infty}e^{-t^{2}/(2T^{2})}\,e^{\ii (\Omega + \omega_{\vec{k}})t}\dd t = \sqrt{2\pi T^{2}}\,e^{- (\Omega+\omega_{\vec{k}})^{2}T^{2}/2}~.
\end{equation}
Substituting this in \cref{eqn: VEP real linear 2} yields
\begin{equation}\label{eqn: VEP real linear gaussian switching}
  P_{0,g\to e} = \frac{\pi\lambda^{2} T^{2}}{L^{n}}\sum_{\vec{k}} \frac{1}{\omega_{\vec{k}}} e^{-(\Omega + \omega_{\vec{k}})^{2}T^{2}}
\end{equation}
which is bounded by the convergent sum
\begin{equation}
P_{0,g\to e} \leq \frac{2^{n-1}\lambda^{2}T^{2}}{L^{n-1}} \prod_{i=1}^{n}\sum_{l_{i}=1}^{\infty} e^{-(2\pi T/L)^{2}\, l_{i}^{2}}~.
\end{equation}
This shows that a suitably smooth switching function is able to regularize the leading order excitation probability in arbitrary dimensions.

The excitation probability in this example is finite, but nonzero, even though the detector was in the ground state and the field in the vacuum. This should not be surprising considering that we have a time-dependent Hamiltonian: When the interaction is switched on, the  Hamiltonian of the composite system is changed, and the state $\ket{0,g}$ is no longer an energy eigenstate. In consequence, the state begins to evolve nontrivially, and there is a finite probability for the detector to be measured in its excited state $\ket{e}$. The more adiabatically  the interaction is switched on, the lower is the probability to excite the detector \cite{sriramkumar_finite-time_1996}. This is reflected in \cref{eqn: VEP real linear gaussian switching}: the larger $T$, that is, the slower the detector is switched, the smaller $P_{0,g\to e}$. In the limit $T \to \infty$, corresponding to adiabatic switching, $P_{0,g\to e} \to 0$ as expected.

\paragraph{Sudden switching and Gaussian detector profile}
Now consider a Gaussian detector profile,
\begin{equation}\label{eqn: gaussian detector profile}
  p(\vec{x}_{0},\vec{y}) = (2\pi\sigma^{2})^{-n/2}\,\exp\left(-\frac{(\vec{y}-\vec{x}_{0})^{2}}{2\sigma^{2}}\right)~,
\end{equation}
together with a sudden switching. In this case:
\begin{equation}
   P_{0,g\to e} = \frac{2 \lambda^{2}}{L^{n}}\sum_{\vec{k}} \frac{e^{-\vec{k}^{2}\sigma^{2}}}{\omega_{\vec{k}}(\Omega+\omega_{\vec {k}})^{2}}  \sin^{2}\left( \frac{\Omega+\omega_{\vec{k}}}{2}T\right)~,
\end{equation}
which is bounded by
\begin{equation}
   P_{0,g\to e} \leq \frac{2^{n-2} \lambda^{2}}{\pi^{3}L^{n-3}} \prod_{i=1}^{n}\sum_{l_i=1}^{\infty}e^{-(2\pi \sigma/L)^{2}\,l_{i}^{2}}
\end{equation}
and therefore again convergent in arbitrary dimensions---the spatial smearing of the interaction between detector and field is also able to regularize the vacuum response of the detector. Moreover, $P_{0,g\to e} \to 0$, if the detector is completely delocalized, that is for $\sigma \to \infty$.

\paragraph{General effect of spacetime profile}\label{sec: effect of spacetime profile}
The two preceding examples show that it is possible to regularize the VEP to leading order in perturbation theory. But which combinations of switching function $\chi$ and spatial smearing $p$ are able to regularize the VEP? To address this question, it is helpful to use the detector's spacetime profile $f$:
the VEP of the quantized real field in \cref{eqn: VEP real linear 2} can be reformulated as
\begin{equation}\label{eqn: VEP real linear spacetime profile 2}
  P_{0,g\to e} = \frac{\lambda^{2}}{2L^{n}}\sum_{\vec{k}} \frac{1}{\omega_{\vec{k}}} \left| \tilde{f}(\omega_\vec{k} + \Omega,\vec{k})\right|^{2}~,
\end{equation}
where $\tilde{f}$ is the $(1,n)$-dimensional Fourier transform of the spacetime profile $f$:
\begin{equation}\label{eqn: def Fourier transform}
  \tilde{f}(\omega_\vec{\xi},\vec{\xi}) = \int_{\R^{n}}\dd\vec{x} \int_{-\infty}^{\infty}\dd t\, f(\vec{x},t)\,e^{-\ii \omega_\vec{\xi} t } \,e^{\ii \vec{\xi} \vec{x}}~.
\end{equation}
Therefore, a given spacetime profile does regularize the VEP to leading order in $(1,n)$ dimensions if and only if the modulus squared of its $(1,n)$-dimensional Fourier transform $\tilde{f}$ decays fast enough in the UV such that the n-fold sum in the above expression is finite.

\subsection{Quadratic coupling to quantized spinor fields}\label{sec: VEP spinor field}

We now turn to the pivotal question: How can we construct a similar particle detector model for quantized spinor fields? To this end, let us begin by investigating and extending a model that has been commonly employed in the literature, namely, model 4 in \cref{tab: overview detector models} pioneered by Takagi in 1985 \cite{takagi_response_1985} and extensively used, e.g., in \cite{takagi_vacuum_1986,bessa_accelerated_2012,diaz_radiative_2002,harikumar_uniformly_2013,langlois_causal_2006}. We will assess the physicality of this class of models by considering the vacuum response of spatially smeared detectors.

Let $\Psi$ be a spinor field, that is, a field taking values in some representation space $\C^{m}$ of the complexified Clifford algebra $\Cliff(1,n)^{\C}$. In order to be able to do explicit calculations, we choose the irreducible Dirac representation of $\Cliff(1,3)$ on $\C^{4}$ in 4 spacetime dimensions. In 2 spacetime dimensions, a related but reducible representation of $\Cliff(1,1)$ on $\C^{4}$ is used, which yields similar spinor mode functions. Details on these conventions and standard results in classical and quantum field theory can be found in \cref{app: quantum fields}.

We will work again in the interaction picture, where the operator of the quantized field is
\begin{multline}\label{eqn: H pic ladd op mode expansion spinor field in text}
      \op{\Psi}(t,\vec{x}) =\\ \sum_{\vec{k},s} \op{a}_{\vec{k},s}\,e^{-\ii   \omega_{\vec{k}}t}\,\psi_{\vec{k},s,+}(\vec{x})
    + \op{b}_{\vec{k},s}^{\dagger}\,e^{\ii   \omega_{\vec{k}}t}\,\psi_{\vec{k},s,-}(\vec{x})~.
\end{multline}
The ladder operators $\op{a}_{\vec{k},s}$, $\op{b}_{\vec{k},s}$ satisfy the canonical \emph{anticommutation} relations \cref{eqn: acr ladder op spinor field}. The spinor-valued mode functions $\psi_{\vec{k},s,\epsilon}$ span a solution space of the Dirac equation. Depending on dimension and mass, the mode functions take slightly different forms, namely, (a) \cref{eqn: spinor mode function 4dim massive,eqn: spinor mode function 4dim massless} for massive and massless fields in $(1,3)$ dimensions; and (b) \cref{eqn: spinor mode function 4dim massive,eqn: spinor mode function 4dim massless} in $(1,1)$ dimensions. 

This quantized spinor field is coupled \emph{quadratically} to a resting UDW-type detector via the interaction Hamiltonian \cref{eqn: interaction H spinor quad}, which is model 4 in \cref{tab: overview detector models}. Using \cref{eqn: time-evolution order n} we then obtain the leading order contribution to the time evolution operator:
\begin{multline}\label{eqn: 1st order time evolution op spinor field}
  \op{U}^{(1)}(t,t_{0}) = \int_{t_{0}}^{t}\dd t' \chi(t')\,\op{\mu}(t')\\
  \times\int_{\R^{n}} \dd \vec{y}\,p(\vec{x}_{0},\vec{y})\,\op{\conj{\Psi}}(t',\vec{y})\op{\Psi}(t',\vec{y})~.
\end{multline}
The vacuum response \cref{eqn: VEP 3} yields
\begin{multline}\label{eqn: VEP spinor quad 1}
  P_{0,g \to e} = 
  \lambda^{2}  \left| \sum_{\vec{k},s} \int_{\R^{n}}p(\vec{x}_{0},\vec{y})\,\conj\psi_{\vec{k},s,-}(\vec{y})\,\psi_{\vec{k},s,-}(\vec{y})\,\dd \vec{y}\right|^{2}\\
   \times \left|\int_{-\infty}^{+\infty}\chi(t)\,e^{\ii \Omega t}\,\dd t\right|^{2}\\
  + \lambda^{2} \sum_{\vec{k},s}\sum_{\vec{p},r} \left| \int_{\R^{n}}p(\vec{x}_{0},\vec{y})\,\conj\psi_{\vec{k},s,+}(\vec{y})\,\psi_{\vec{p},r,-}(\vec{y})\,\dd \vec{y}\right|^{2}\\
  \times \left|\int_{-\infty}^{+\infty}\chi(t)\,e^{\ii (\Omega +\omega_{\vec{k}}+\omega_{\vec{p}})t}\,\dd t\right|^{2}~.
\end{multline}

There are obvious differences to the vacuum response of the usual UDW model (model 1 in \cref{tab: overview detector models}, i.e., a monopole detector linearly coupled to a quantized scalar field). Comparing \cref{eqn: VEP real linear 1} with \cref{eqn: 1st order time evolution op spinor field}, at leading order in $\lambda$, there are two kinds of terms instead of one.

The first term corresponds to excitation of the detector by creation and subsequent annihilation of a field quantum, with amplitude $\mathcal{A}^{(1)}_{0,e}$. More precisely, the detector is excited by creating and annihilating an antiparticle from the vacuum. However, the equivalent process featuring a particle does not occur.

The second term represents processes where the detector is excited by emission of a particle and an antiparticle, with amplitude $\mathcal{A}^{(1)}_{1_{(\vec{k},s)},\bar{1}_{(\vec{p},r)},e}$. These stood to be expected in the light of our discussion regarding fermion number conservation in \cref{sec: quadratic coupling}. Note that the momenta of the particle and antiparticle are not related since the detector is ``heavy'' and can absorb any amount of momentum.

We can make the mode functions $\psi_{\vec{k},s,\epsilon}$ explicit using \cref{eqn: spinor mode function 4dim massive,eqn: spinor mode function 4dim massless,eqn: spinor mode function 2dim massive,eqn: spinor mode function 2dim massless}. In all four cases considered here, massive and massless fields in $(1,3)$ and $(1,1)$ dimensions, the VEP when coupling quadratically to quantized spinor fields can be brought into the form
\begin{multline}\label{eqn: VEP spinor quad 2}
  P_{0,g\to e}
  = \frac{\lambda^{2}}{2L^{2n}}\Bigg[8m^{2} \left(\sum_{\vec{k}}\frac{1}{\omega_{\vec{k}}}\right)^{2} \left|\int_{-\infty}^{+\infty}\chi(t)\,e^{\ii \Omega t}\,\dd t\right|^{2}\\
  \hphantom{=~}+ \sum_{\vec{k},\vec{p}} \frac{(\omega_{\vec{k}}+m)(\omega_{\vec{p}}+m)}{\omega_{\vec{k}}\omega_{\vec{p}}}\left(\frac{\vec{k}}{\omega_{\vec{k}}+m}-\frac{\vec{p}}{\omega_{\vec{p}}+m}\right)^{2} \\
  \hphantom{=~=~}\times\left| \int_{\R^{n}}p(\vec{x}_{0},\vec{y})\,e^{-\ii  (\vec{k}+\vec{p})\vec{y}}\,\dd \vec{y}\right|^{2}\\
  \hphantom{=~=~}\times \left|\int_{-\infty}^{+\infty}\chi(t)\,e^{\ii (\Omega +\omega_{\vec{k}}+\omega_{\vec{p}})t}\,\dd t\right|^{2}\Bigg]~,
\end{multline}
where $n=1$ or $n=3$. The case of a massless field is correctly recovered by setting $m=0$, such that the first term vanishes. Notice that there is an additional problem with the first term in \cref{eqn: VEP spinor quad 2} for $m \neq 0$. Namely, the  term is proportional to
$\sum_{\vec{k}} \omega_{\vec{k}}^{-1}$. Even in the least divergent scenario, i.e. $n=1$, this sum diverges like $\sum_{a}a^{-1}$.

This divergency is fundamentally worse behaved than those encountered in the standard bosonic Unruh-DeWitt model, which arise from the pointlike structure of the detector \cite{grove_notes_1983,takagi_vacuum_1986}, or its overly abrupt switching \cite{sriramkumar_finite-time_1996,louko_transition_2008}. The first term in \cref{eqn: VEP spinor quad 2} \emph{cannot be regularized} by way of a smooth spacetime profile, and therefore the divergency in model 4 cannot be cured via this common method. This is because the spatial integral containing the (normalized) detector profile factors out and yields a factor of 1. The time integral here is not momentum dependent and merely contributes an overall factor which cannot regularize the divergent sum.

Comparing with the Unruh-DeWitt model, there are three conceivable origins of this new divergency: \begin{inparaenum}[(a)]\item it could be algebraic on the level of ladder operators, because we are now dealing with a field obeying Fermi statistics instead of Bose statistics, \item it could be analytic due to the field's spinorial structure and different inner product compared to the scalar fields, \item or the origin could be the coupling that is now quadratic in the field instead of linear\end{inparaenum}. The question is easily settled by computing the VEP of a \emph{quantized complex} field coupled quadratically to an UDW-type detector, keeping properties (a) and (b) the same as in the usual Unruh-DeWitt model.

\subsection{Quadratic coupling to quantized scalar fields}

Let $\Phi$ be a quantized real or complex field (models 2 and 3 in \cref{tab: overview detector models}, respectively), coupled quadratically to a resting UDW-type detector via the interaction Hamiltonian \cref{eqn: interaction H scalar quad}. Inserting the interaction Hamiltonian in \cref{eqn: time-evolution order n} yields
\begin{multline}
  \op{U}^{(1)}(t,t_{0}) =\\
   \int_{t_{0}}^{t}\dd t' \chi(t')\,\op{\mu}(t')\int_{\R^{n}} \dd \vec{y}\,p(\vec{x}_{0},\vec{y})\,\op{\Phi}^{\dagger}(t',\vec{y})\op{\Phi}(t',\vec{y})~,
\end{multline}
where the field operator (in the interaction picture) is given by \cref{eqn: interaction p real field operator} in case of a quantized real field. If the field is complex
\begin{multline}\label{eqn: interaction p complex field operator}
  \op\Phi(t,\vec{x})  = \sum_{\vec{k}}\frac{1}{\sqrt{2\omega_{\vec{k}}}} \Big[ \op{a}_{\vec{k}}\,e^{-\ii   \omega_{\vec{k}}t}\,\varphi_{\vec{k}}(\vec{x}) \\
  \hphantom{=~}+\op{b}^{\dagger}_{\vec{k}}\,e^{\ii  \omega_{\vec{k}}t} \,\varphi_{\vec{k}}^{*}(\vec{x}) \Big]~.
\end{multline}
Calculation of the VEP \cref{eqn: VEP 3} yields
\begin{align}\label{eqn: VEP scalar quad 1}
  P_{0,g \to e} 
  &= \frac{\lambda^{2}}{4} \Bigg[ \Bigg| \sum_{\vec{k}} \frac{1}{\omega_{\vec{k}}} \int_{\R^{n}}p(\vec{x}_{0},\vec{y})\left|\varphi_{\vec{k}}(\vec{y})\right|^{2} \,\dd \vec{y}\Bigg|^{2} \nonumber\\
  &\hphantom{=~=~}\times \left|\int_{-\infty}^{+\infty}\chi(t)\,e^{\ii \Omega t}\,\dd t\right|^{2}\nonumber \\
  &\hphantom{=~}+ \sum_{\vec{k},\vec{p}} \frac{1}{\omega_{\vec{k}}\omega_{\vec{p}}}\left| \int_{\R^{n}}p(\vec{x}_{0},\vec{y})  \,\varphi_{\vec{k}}^{*}(\vec{y})\,\varphi_{\vec{p}}^{*}(\vec{y})\,\dd \vec{y}\right|^{2} \nonumber \\
  &\hphantom{=~=~}\times \left|\int_{-\infty}^{+\infty}\chi(t)\,e^{\ii (\Omega +\omega_{\vec{k}}+\omega_{\vec{p}})t}\,\dd t\right|^{2}\Bigg]
\end{align}
in both cases at leading order. Notice that this is no longer true at higher orders. This expression has two terms like \cref{eqn: VEP spinor quad 1}, representing analogue processes: The first term again corresponds to excitation of the detector by emission and reabsorption of an antiparticle with amplitude $\mathcal{A}^{(1)}_{0,e}$, and the second term corresponds to processes where the detector is excited by emission of a particle-antiparticle pair with amplitude $\mathcal{A}^{(1)}_{1_{\vec{k}},\bar{1}_{\vec{p}},e}$. Explicitly inserting the mode functions \cref{eqn: mode functions real field 1} for periodic boundary conditions gives
\begin{align}\label{eqn: VEP scalar quad 2}
  P_{0,g\to e}
  &= \frac{\lambda^{2}}{4L^{2n}} \Bigg[ \left(\sum_{\vec{k}}\frac{1}{\omega_{\vec{k}}}\right)^{2} \left|\int_{-\infty}^{+\infty}\chi(t)\,e^{\ii \Omega t}\,\dd t\right|^{2}\nonumber\\
  &\hphantom{=~}+ \sum_{\vec{k},\vec{p}} \frac{1}{\omega_{\vec{k}}\omega_{\vec{p}}}\left| \int_{\R^{n}}p(\vec{x}_{0},\vec{y})\, e^{-\ii  (\vec{k}+\vec{p})\vec{y}}\,\dd \vec{y}\right|^{2}\nonumber\\
  &\hphantom{=~=~}\times \left|\int_{-\infty}^{+\infty}\chi(t)\,e^{\ii (\Omega +\omega_{\vec{k}}+\omega_{\vec{p}})t}\,\dd t\right|^{2}\Bigg]~.
\end{align}
The first term is proportional to $\sum_{\vec{k}}\omega_{\vec{k}}^{-1}$ and thus divergent in the same way as for spinor fields in \cref{eqn: VEP spinor quad 2}. Since this divergence appears in exactly the same form for both scalar and spinor fields, its origin is clearly the quadratic coupling. Also note that for quantized scalar fields the divergence arises even if $m=0$.

It is not surprising that these divergencies arise: Interactions containing second or higher powers of fields tend to lead to divergencies in observables, which require renormalization. In particular, the spinor field detector model by Iyler and Kumar \cite{iyer_detection_1980} faces the same problem.


\section{Renormalization of quadratic models}\label{sec: renormalizing}

Before investigating the renormalization of the quadratically coupled detector models \cref{eqn: interaction H spinor quad,eqn: interaction H scalar quad}, let us briefly review the literature.

\subsection{Literature review}\label{sec: literature review}

In \source{hinton_particle_1984}, Hinton introduces an UDW-type particle detector model coupled to a quantized real field $\Phi$ through the interaction Hamiltonian
\begin{equation}
  \op{H}_{\mathrm{int}} = \lambda \mu \op\Phi^{2}[\vec{x}(\tau)]~,
\end{equation}
which was the template for \cref{eqn: interaction H scalar quad}. Hinton calculates the VEP in \cite[eq. (14)]{hinton_particle_1984}, remarks that the expression is formally divergent and requires renormalization but no explicit renormalization scheme is given.

Takagi, who suggested coupling a quantized spinor field quadratically to DeWitt's two-level detector through \cref{eqn: Takagi's model}, followed a different approach. In the original publications, see \sources{takagi_response_1985,takagi_vacuum_1986}, Takagi points out an infinite term in the VEP at leading order but argues that this term can be dropped. Later publications employing Takagi's detector model or drawing on his results implicitly follow the same reasoning \cite{diaz_radiative_2002,langlois_causal_2006,bessa_accelerated_2012,harikumar_uniformly_2013}. A key point is that instead of the total excitation probability, excitation \emph{rates} are considered. As is well known, the calculations of rates involves fewer integrations and are therefore generally more regular. In addition, Takagi regularizes the detector response by using a smooth, namely exponential switching function. As a consequence, in the excitation rate
\begin{equation}\label{eqn: excitation rate}
  R = \lim_{\substack{%
    t_{0}\to -\infty\\%
    t\hphantom{_{0}}\to+\infty}%
  } \frac{P_{0,g\to e}(t,t_{0})}{t-t_{0}}
\end{equation}
the divergence in $P_{0,g\to e}$ is offset by a divergence in the denominator.

Takagi's expression for the rate \cite[eqs. (8.5.2) to (8.5.4)]{takagi_vacuum_1986} can easily be obtained in the notation developed earlier by using \cref{eqn: VEP 3,eqn: time-evolution order n,eqn: interaction H spinor quad} in \cref{eqn: excitation rate} and by using a pointlike detector with exponential switching function:
\begin{align}
  p(\vec{x},\vec{y}) &= \delta(\vec{x}-\vec{y})~, & \chi(t) &= e^{-s |t|}
\end{align}
The limit $s \to 0$ is taken at the end of the calculation in \source{takagi_vacuum_1986} so that ultimately the detector is switched on and off infinitely slowly. One obtains the rate
\begin{equation}
  R = \lambda^{2} |\braket{e|\mu|g}|^{2} F(\Omega)~.
\end{equation}
Here,
\begin{multline}
  F(\Omega) =  \lim_{s\to 0}\lim_{\substack{t_{0}\to -\infty\\t\hphantom{_{0}}\to+\infty}} \int_{t_{0}}^{t}\dd t_{1}\int_{t_{0}}^{t}\dd t_{2} \frac{1}{t-t_{0}}\\
  \times e^{-s(|t_{1}|+|t_{2}|)}e^{-\ii  \Omega(t_{1}-t_{2})} S(t_{1}-t_{2})~,
\end{multline}
which is usually called the \emph{detector response function} \cite{birrell_quantum_1984}, and
\begin{equation}
  S(t_{1}-t_{2}) = \braket{0|\op{\conj{\Psi}}(t_{1},\vec{x})\op{\Psi}(t_{1},\vec{x})\op{\conj{\Psi}}(t_{2},\vec{x})\op{\Psi}(t_{2},\vec{x})|0}
\end{equation}
is a 4-point correlation function. This step is not made explicit in \source{takagi_vacuum_1986}; rather, Takagi directly argues that the exact form of the switching function is not relevant, since it has only been introduced as a way to regularize the detector response. Therefore, he argues, one is allowed to temporarily replace the exponential switching function with a Gaussian one, and to exploit that only the limits $t\to\infty$, $t_{0}\to -\infty$ and $s\to 0$ are of interest---see end of section 3.3 in \source{takagi_vacuum_1986}. Finally, the Gaussian is replaced \lat{ad hoc} with an exponential function again. If we accept Takagi's reasoning, the above response function can be rewritten as
\begin{equation}
  F(\Omega) = \lim_{s\to 0}\int_{-\infty}^{\infty}\dd t\, e^{-\ii  \Omega(t)} e^{-s|t|} S(t)~,
\end{equation}
which indeed corresponds to \cite[eqs. (8.5.2) to (8.5.4)]{takagi_vacuum_1986} (with two differences: in the original paper, the factor $e^{-s|t|}$ due to the switching function is suppressed, and $\lambda = 1$).
In this notation, the divergence we encountered earlier is now hidden in the correlator $S$: by Wick's theorem it can be expanded as
\begin{equation}\label{eqn: correlator ordered by Wick Theorem}
  S(t_{1}-t_{2}) = [\tr S^{-}(0)]^{2} + \tr[ S^{-}(t_{1}-t_{2}) S^{+}(t_{1}-t_{2})]
\end{equation}
where $S^{-}$, $S^{+}$ are the Wightman functions
\begin{equation}\begin{aligned}
  S^{-}(t_{1},t_{2})^{A}_{B} &= \braket{0|\conj{\op{\Psi}}_{B}(t_{2},\vec{x})\op{\Psi}^{A}(t_{1},\vec{x})|0} \\   S^{+}(t_{1},t_{2})^{A}_{B} &= \braket{0|\op{\Psi}^{A}(t_{1},\vec{x})\conj{\op\Psi}_{B}(t_{2},\vec{x})|0}
\end{aligned}\end{equation}
and the first term is infinite, as expected:
\begin{equation}
  \tr S^{-}(0) \to \infty~.
\end{equation}
Takagi treats this infinite term as formally constant  in \cref{eqn: correlator ordered by Wick Theorem}, which allows one to obtain a delta distribution from the time integral:
\begin{multline}
   F(\Omega) =  \delta(\Omega)\, [\tr S^{-}(0)]^{2} 
 + \lim_{s\to 0}\int_{-\infty}^{\infty}\dd t' e^{-\ii  \Omega(t')} e^{-s|t'|}\\ \times \tr S^{-}(t_{2},t_{1}) S^{+}(t_{1},t_{2})~.
\end{multline}
He concludes, therefore, that the term $\tr S^{-}(0)$ can be dropped, as long as the detector gap is finite, $\Omega>0$, thereby obtaining a finite excitation rate $R$.

This argumentation is not entirely satisfactory. In particular, it merely yields a finite excitation rate. This is not enough to fully characterize the response of a particle detector. For that, one would need to be able to compute its full density matrix, which requires computing probability amplitudes, not rates. Takagi's techniques, however, do not avoid divergences in probabilities. The challenge remains to fully renormalize particle detector models with an interaction Hamiltonian that is quadratic in the field.

\subsection{Renormalization at leading order}\label{sec: renormalization at leading order}

As we discussed above, the first term in \cref{eqn: VEP spinor quad 2,eqn: VEP scalar quad 2} contains a divergence that cannot be regularized by means of a spacetime profile, unlike in the usual linearly coupled UDW model. As we will show now, the occurrence of such divergencies is not surprising
from the point of view of a quantum field theoretical renormalization theory, and can be handled straightforwardly. To this end, let us compare the two problematic detector models to quantum electrodynamics (QED). The interaction Hamiltonian of QED, $\op{H}_{\mathrm{QED}} \propto A_{\mu}\conj{\Psi}\gamma^{\mu}\Psi$,
is also quadratic in the spinor field (electron field), the only difference being that the coupling is not to a scalar detector but to the vectorial photon field. Indeed, QED has a similar divergence, namely the well-known divergent tadppole subdiagram. In QED, it is renormalized straightforardly  by normal-ordering the interaction Hamiltonian (see, e.g.\,\cite{greiner_quantum_2008}). In the following two sections we demonstrate that and how the divergencies in the detector models are related to tadpole diagrams, and we show that normal-ordering of the interaction Hamiltonian renormalizes the detector models as well.

\subsubsection{Structure of the divergence}  As mentioned before, there are two types of perturbative processes that excite the detector at leading order, thereby contributing to the VEP for quantized spinor fields in \cref{eqn: VEP spinor quad 1}, or quantized scalar fields in \cref{eqn: VEP scalar quad 1}. To see this, recall that at leading order, time evolution amounts to applying the interaction Hamiltonian to the initial state $\ket{0,g}$ and integrating over the entire duration of the interaction. In case of a spinor field, the final state is $\ket{e,\psi}$. Given  the quadratic coupling in the interaction Hamiltonian $\sim \op{\conj\Psi}\op\Psi$ for the fermionic case and  $\sim \op{\Phi}^{\dagger}\op\Phi$ for the bosonic case, the algebraic structure of the leading order time evolved field state is of the form:
\begin{equation}\label{eqn: algebraic structure quadratic VEP}
  \ket{\psi},\,\ket{\phi} = \sum_{\vec{k}} E_{\vec{k}}\ket{0} + \sum_{\vec{k},\vec{p}} F_{\vec{k},\vec{p}} \ket{1_{\vec{k}},\bar{1}_{\vec{p}}}~.
\end{equation}
The spinorial degrees of freedom are suppressed and $E_{\vec{k}}$, $F_{\vec{k},\vec{p}}$ are coefficients depending on the momenta (and for spinor fields additionally on spins). In other words: When interacting with the field in the vacuum state, the detector can be excited (to leading order in perturbation theory) either by emission and subsequent annihilation of an antiparticle [divergent first term in \cref{eqn: algebraic structure quadratic VEP,eqn: VEP spinor quad 1,eqn: VEP scalar quad 1}], or by emitting a particle-antiparticle pair (second term in the same equations). Since we are only interested in the state of the detector after the interaction, both types of processes contribute to the vacuum excitation probability. They are visualized as diagrams in \cref{fig: VEP processes 1st order}. Note that these simplified diagrams (used for illustration) are not strictly Feynman diagrams since the detector is not yet second quantized, as we will see in more detail in \cref{sec: feynman rules} when we build the detector-field interaction Feynman rules.
\begin{figure}
  \includegraphics{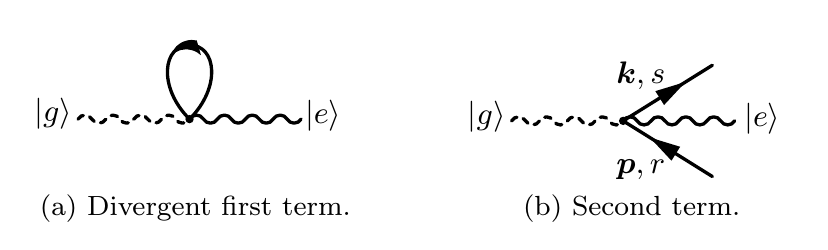}
  \caption{Schematic representation of the processes leading to the state \cref{eqn: algebraic structure quadratic VEP}. The wiggly line stands for the detector (dashed when in ground state, solid if excited), the solid straight line for the quantized scalar or spinor field. The time direction is from left to right.}
  \label{fig: VEP processes 1st order}
\end{figure}

The following discussion is limited to the case of a quantized spinor field, but runs analogously for quantized scalar fields. In comparing the spinor field detector with QED, the spinor field corresponds to the electron field in QED, and the detector to the photon field. Thus, the detector ground state translates to the absence of photons, while the excited state corresponds to the presence of one photon. In this sense, the processes in \cref{fig: VEP processes 1st order} are analogous to the QED diagrams drawn in \cref{fig: QED processes 1st order}. 
\begin{figure}
  \includegraphics{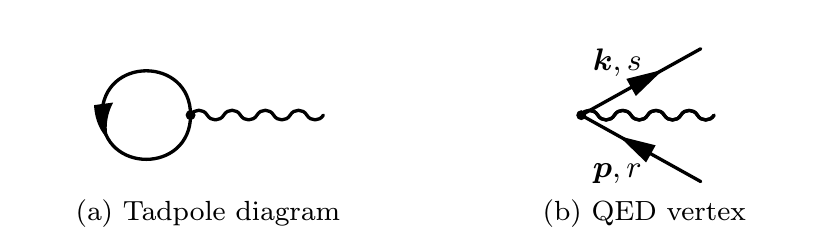}
  \caption{QED Feynman diagrams corresponding to the processes in \cref{fig: VEP processes 1st order}.}
  \label{fig: QED processes 1st order}
\end{figure}
In particular, the problematic first term in \cref{fig: VEP processes 1st order} corresponds to the tadpole diagram in \cref{fig: QED processes 1st order}. The divergence of its amplitude $\mathcal{A}^{(1)}_{0,e}$ is mirrored by QED: the amplitude of the tadpole diagram is divergent as well; see \cite{greiner_quantum_2008}.

While the ``heavy'' detector guarantees energy and momentum conservation in \cref{fig: VEP processes 1st order}, both QED processes in \cref{fig: QED processes 1st order} are dynamically forbidden. They may, however, appear as subdiagrams in larger processes, and it is possible to calculate their contribution to the total amplitudes. Indeed, the analogy between the left-hand diagrams in \cref{fig: VEP processes 1st order,fig: QED processes 1st order} extends to their respective amplitudes: according to \cref{eqn: VEP 3,eqn: 1st order time evolution op spinor field}, the amplitude of the left diagram in \cref{fig: VEP processes 1st order} is
\begin{multline}\label{eqn: 1st term amplitude}
  \mathcal{A}^{(1)}_{0,e} = -\ii  \lambda e^{-\ii  \Omega t}\\
  \times\bra{0,e}\int_{t_{0}}^{t} \chi(t')\,\op{\mu}(t') \op{\conj{\Psi}}(t',\vec{x})\,\op{\Psi}(t',\vec{x})\,\dd t'\ket{0,g}\\
  \sim \braket{0|\,\op{\conj{\Psi}}(t',\vec{x})\,\op{\Psi}(t',\vec{x})|0} 
\end{multline}
when suppressing the time integration. The detector was assumed to be pointlike for simplicity, since we know already from \cref{eqn: VEP spinor quad 1} that the spatial profile has no regularizing effect. Using the fermionic time-ordering operator, this can be rewritten as
\begin{multline}
  \braket{0|\,\op{\conj{\Psi}}_{A}(t,\vec{x})\,\op{\Psi}^{A}(t,\vec{x})|0} \\
  = - \braket{0|T\,\op{\Psi}^{A}(t,\vec{x})\,\op{\conj{\Psi}}_{A}(t,\vec{x})|0}~,
\end{multline}
which is essentially the Feynman propagator:
\begin{equation}
  S_{F}(x-y)^{A}_{B} = \braket{0|T \op{\Psi}^{A}(x)\op{\conj{\Psi}}_{B}(y)|0}~.
\end{equation}
We find
\begin{equation}\label{eqn: 1st term as Feynman propagator}
  \mathcal{A}^{(1)}_{0,e} \sim \tr S_{F}(0) \to \infty~,
\end{equation}
where the trace runs over the spinor indices.

By applying position space Feynman rules to the tadpole diagram in \cref{fig: QED processes 1st order}, on the other hand, one obtains the amplitude
\begin{equation}\label{eqn: tadpole as Feynman propagator}
  \tr [\gamma^{\mu} S_{F}(0)] \to \infty~,
\end{equation}
(see e.g.\,\cite{greiner_quantum_2008}). The divergence of the two amplitudes \cref{eqn: 1st term as Feynman propagator,eqn: tadpole as Feynman propagator} has a common structure: they diverge because the Feynman propagator of the quantized spinor field is ill defined when closed in itself. The gamma matrices in the second amplitude merely appear because electrodynamics has a vector coupling \cref{eqn: interaction H QED}, as opposed to the scalar coupling \cref{eqn: interaction H spinor quad} for the detector model.

In conclusion, it is justified to interpret the divergence found in the quadratically coupling detector models as the analogue of tadpole diagrams known from quantum field theories such as QED.

\subsubsection{Renormalization by normal-ordering}
Tadpoles are renormalized by normal-ordering of the interaction Hamiltonian. This amounts to setting the amplitude of the tadpole to zero, or---equivalently---ignoring any diagram containing tadpoles as subdiagrams (see, for instance, \cite{greiner_quantum_2008}). The same procedure works for both the detector models as well: normal-ordering the interaction Hamiltonians \cref{eqn: interaction H spinor quad} leads to
\begin{multline}
  \mathcal{A}^{(1)}_{0,e} = -\ii  \lambda e^{-\ii  \Omega t}\\
  \times\bra{0,e}\int_{t_{0}}^{t} \chi(t')\,\op{\mu}(t') \normalord{}\op{\conj{\Psi}}(t',\vec{x})\,\op{\Psi}(t',\vec{x})\normalord{}\,\dd t'\ket{0,g} \\ = 0
\end{multline}
instead of \cref{eqn: 1st term amplitude,eqn: 1st term as Feynman propagator}, such that the divergence is renormalized to zero. Equivalently, for detector models coupled quadratically to quantized real or complex fields via the Hamiltonian \cref{eqn: interaction H scalar quad},
\begin{align}
  \mathcal{A}^{(1)}_{0,e} \sim \braket{0|\normalord{}\op{\Phi}^{\dagger}\,\op{\Phi}\normalord{}|0} = 0~.
\end{align}

The above renormalization can be interpreted as shifting the expectation value of the interaction energy from infinity to zero: If detector and field do not interact, $\ket{0,g}$ is the ground state, with energy eigenvalue zero:
\begin{equation}
  \bra{0,g}H\ket{0,g} = 0~.
\end{equation}
Here, $H = H_{\mathrm{d}} + H_{\Phi}$ for quantized scalar fields, and $H = H_{\mathrm{d}} + H_{\Psi}$ for quantized spinor fields as given in \cref{eqn: H detector,eqn: hamilton operator spinor field,eqn: hamilton operator real field,eqn: hamilton operator complex field}. However, as soon as the coupling is switched on, $\lambda >0$, we would naively have
\begin{equation}
  \bra{0,g}H + H_{\mathrm{int}}\ket{0,g} \to \infty
\end{equation}
for the interactions quadratic in the field \cref{eqn: interaction H scalar quad,eqn: interaction H spinor quad}. This would be unphysical and normal-ordering the interaction Hamiltonian is therefore appropriate to ensure that
\begin{equation}
  \bra{0,g}H + \normalord{}H_{\mathrm{int}}\normalord{}\ket{0,g} = 0~.
\end{equation}
In comparison, for linear coupling \cref{eqn: interaction H real linear} the energy expectation value vanishes automatically.

\subsection{The renormalized models and their VEP}\label{sec: renormalized models}

In summary, renormalization of the spinor field detector, model 4 in \cref{tab: overview detector models}, is achieved by coupling through the Hamiltonian
\begin{equation}\label{eqn: interaction H spinor quad renorm}
  \op{H}_{\mathrm{int}}(t) = \lambda\chi(t)\, \op{\mu} \int_{\R^{n}}p(\vec{x}(t),\vec{y})\,\normalord{}\op{\conj\Psi}(\vec{y})\op{\Psi}(\vec{y})\normalord{} \,\dd \vec{y}~.
\end{equation}
Detector models 2 and 3 for real or complex fields are similarly renormalized by coupling through
\begin{equation}\label{eqn: interaction H scalar quad renorm}
  \op{H}_{\mathrm{int}}(t) = \lambda\chi(t)\, \op{\mu} \int_{\R^{n}}p(\vec{x}(t),\vec{y})\,\normalord{}\op{\Phi}^{\dagger}(\vec{y})\op{\Phi}(\vec{y})\normalord{} \,\dd \vec{y}~.
\end{equation}

The vacuum response of the renormalized spinor field detector is obtained by dropping the first term in \cref{eqn: VEP spinor quad 1}:
\begin{multline}\label{eqn: VEP spinor quad 1 renorm}
  P_{0,g\to e} = \\
   \lambda^{2} \sum_{\vec{k},s}\sum_{\vec{p},r} \left| \int_{\R^{n}}p(\vec{x}_{0},\vec{y})\,\conj\psi_{\vec{k},s,+}(\vec{y})\,\psi_{\vec{p},r,-}(\vec{y})\,\dd \vec{y}\right|^{2}\\
  \times \left|\int_{-\infty}^{+\infty}\chi(t)\,e^{\ii (\Omega +\omega_{\vec{k}}+\omega_{\vec{p}})t}\,\dd t\right|^{2}~.
\end{multline}
Explicitly, for periodic boundary conditions:
\begin{multline}\label{eqn: VEP spinor quad 2 renorm}
  P_{0,g\to e}  = \\
  \frac{\lambda^{2}}{2L^{2n}} \sum_{\vec{k},\vec{p}} \frac{(\omega_{\vec{k}}+m)(\omega_{\vec{p}}+m)}{\omega_{\vec{k}}\omega_{\vec{p}}}\left(\frac{\vec{k}}{\omega_{\vec{k}}+m}-\frac{\vec{p}}{\omega_{\vec{p}}+m}\right)^{2} \\
  \times\left| \int_{\R^{n}}p(\vec{x}_{0},\vec{y})\,e^{-\ii  (\vec{k}+\vec{p})\vec{y}}\,\dd \vec{y}\right|^{2} \\
  \times \left|\int_{-\infty}^{+\infty}\chi(t)\,e^{\ii (\Omega +\omega_{\vec{k}}+\omega_{\vec{p}})t}\,\dd t\right|^{2}~.
\end{multline}
Similarly, the leading-order vacuum response of the renormalized real and complex field detector becomes
\begin{multline}\label{eqn: VEP scalar quad 1 renorm}
  P_{0,g\to e}  =\\
   \frac{\lambda^{2}}{4}  \sum_{\vec{k},\vec{p}} \frac{1}{\omega_{\vec{k}}\omega_{\vec{p}}}\left| \int_{\R^{n}}p(\vec{x}_{0},\vec{y}) \, \varphi_{\vec{k}}^{*}(\vec{y})\,\varphi_{\vec{p}}^{*}(\vec{y})\,\dd \vec{y}\right|^{2}\\
  \times \left|\int_{-\infty}^{+\infty}\chi(t)\,e^{\ii (\Omega +\omega_{\vec{k}}+\omega_{\vec{p}})t}\,\dd t\right|^{2}~,
\end{multline}
or explicitly
\begin{multline}\label{eqn: VEP scalar quad 2 renorm}
  P_{0,g \to e}  = \\
  \frac{\lambda^{2}}{4L^{2n}} \sum_{\vec{k},\vec{p}} \frac{1}{\omega_{\vec{k}}\omega_{\vec{p}}}\left| \int_{\R^{n}}p(\vec{x}_{0},\vec{y})\,e^{-\ii  (\vec{k}+\vec{p})\vec{y}}\,\dd \vec{y}\right|^{2}\\
  \times \left|\int_{-\infty}^{+\infty}\chi(t)\,e^{\ii (\Omega +\omega_{\vec{k}}+\omega_{\vec{p}})t}\,\dd t\right|^{2}
\end{multline}
by dropping the first term from \cref{eqn: VEP scalar quad 1,eqn: VEP scalar quad 2}.

As will be demonstrated in the next section, the VEPs \cref{eqn: VEP scalar quad 2 renorm,eqn: VEP spinor quad 2 renorm} may still be divergent even after this renormalization and further regularization of the VEP is necessary. Recall that the VEP of detector model 1 coupling to a quantized scalar field could be rewritten in terms of the Fourier transform of its spacetime profile $f$ in \cref{eqn: VEP real linear spacetime profile 2}. The same is possible for the quadratically coupling detectors: using the Fourier transformation \cref{eqn: def Fourier transform}, the VEP \cref{eqn: VEP scalar quad 2 renorm} for quantized scalar fields can be rewritten as
\begin{equation}\label{eqn: VEP complex quad 2 renorm 2}
  P_{0,g \to e}  = \frac{\lambda^{2}}{4L^{2n}} \sum_{\vec{k},\vec{p}} \frac{1}{\omega_{\vec{k}}\omega_{\vec{p}}}\left| \tilde{f}(-d-k-p)\right|^{2}~.
\end{equation}
The probability in the case of \cref{eqn: VEP spinor quad 2 renorm} is
\begin{multline}\label{eqn: VEP spinor quad 2 renorm 2}
  P_{0,g\to e}
  = \frac{\lambda^{2}}{2L^{2n}} \sum_{\vec{k},\vec{p}} \frac{(\omega_{\vec{k}}+m)(\omega_{\vec{p}}+m)}{\omega_{\vec{k}}\omega_{\vec{p}}}\\
  \times\left(\frac{\vec{k}}{\omega_{\vec{k}}+m}-\frac{\vec{p}}{\omega_{\vec{p}}+m}\right)^{2} \left| \tilde{f}(-d-k-p)\right|^{2}~.
\end{multline}
Thus, the UV behavior of its Fourier transform again decides which spacetime profile $f$ is suitable for regularization. Note that the prefactor in front of the $\tilde{f}$ decays faster in the UV for quantized scalar fields than for quantized spinor fields. This means that spinor field detectors  require a ``stronger'' regularization.


\section{Comparison of particle detector models}\label{sec: comparison}

So far, we have discussed how to obtain a finite vacuum response through renormalization and regularization for four different detector models. We will now investigate whether any of the three scalar field detector models (1,2 and 3 in \cref{tab: overview detector models})  is comparable to the spinor field detector, in the sense that a comparison of their respective responses offers insight into the properties of the field they probe.

For the different models considered, the relevant features that can have a significant impact on the response of UDW-type detectors for quantized scalar and spinor fields are
\begin{compactenum}
  \item The nature of the coupling: linear versus quadratic.
  \item Internal degrees of freedom of the field: for example, the $U(1)$ charge.
  \item The field statistics: bosonic or fermionic.
  \item The analytic structure of the field: scalar or spinorial.
\end{compactenum}
From the results obtained in the previous sections, we discuss below each of these points individually.

\subsection{Linear vs quadratic coupling}

The coupling has a profound influence on the response of an UDW-type detector:  The vacuum response of model 1---i.e. coupling linearly to a quantized real field---given in \cref{eqn: VEP real linear 1},  is fundamentally different from the response [\cref{eqn: VEP scalar quad 1 renorm}] of model 2 to the same field when using the (renormalized) quadratic coupling of \cref{eqn: interaction H scalar quad renorm}. In order to assess the differences between the response of detectors to fermionic (spinor) fields and bosonic (scalar) fields, we will now compare models 2 to 4 that couple to the same power of the fields.

\subsection{Charged vs uncharged field}

Quadratic coupling to quantized scalar fields gives the same VEP \cref{eqn: VEP scalar quad 1 renorm} for uncharged (i.e. real) and charged (i.e. complex) quantized scalar fields to leading order in perturbation theory. However, as we shall see in \cref{sec: feynman rules}, the detector response is not identical in general; the field charge does have an influence. Since spinor fields are generally charged, one  should always use a charged scalar field when comparing quantized scalar fields with quantized spinor fields in order to single out effects coming exclusively from the field statistics or analytic structure. In other words, only model 3 and model 4 can be rightfully compared.

\subsection{Boson vs fermion statistics}

The statistics of the probed field influence the reaction of an UDW-type detector in two ways. First of all, the Pauli exclusion principle will prevent the creation of more than one quantum of a given charge, momentum and spin in a fermionic field. This restricts the set of possible processes that, for example, excite the detector. This would not be the case for bosonic fields. Secondly, the anticommutation of fermionic operators implies that two distinct processes leading to the same final state may have a relative minus sign in their amplitudes. Fermionic fields, unlike bosonic fields, can therefore have cancellations between amplitudes.

Nevertheless, the field statistics do not influence the VEP of the four detector models discussed here at leading order: merely a single quantum (linear coupling), or at most a particle-antiparticle pair (quadratic coupling) is created from the vacuum of the field, so the Pauli exclusion principle has no effect. And since there is only one possible process contributing to the VEP, no cancellation between different processes occurs. This implies that the leading order VEP remains unaffected by the field statistics.

As soon as the field is initially not in its vacuum state, however, the Pauli exclusion principle does affect the response of a particle detector to a fermionic field. And even for the VEP, the exclusion principle and different sign amplitude interference will be relevant at higher orders in perturbation theory. Interestingly, in dimensions larger than two, the spin-statistics theorem ties together field statistics and the spin of a field (which is reflected in its analytic structure). This makes it generally hard to distinguish the effect of the analytic structure and the effect of field  statistics. In the leading order vacuum response, however, due to the uniqueness of the leading order VEP described above, it is possible to independently study the effect of the spinor structure versus a scalar structure.

\subsection{Scalar vs spinor field}

Comparing \cref{eqn: VEP scalar quad 1 renorm,eqn: VEP spinor quad 1 renorm}, one finds that the VEP of model 3 has summands proportional to
\begin{multline}
  \propto \frac{1}{4\omega_{\vec{k}}\omega_{\vec{p}}}\left| \int_{\R^{n}}p(\vec{x}_{0},\vec{y}) \, \varphi_{\vec{k}}^{*}(\vec{y})\,\varphi_{\vec{p}}^{*}(\vec{y})\,\dd \vec{y}\right|^{2}\\
  = \left| \int_{\R^{n}}p(\vec{x}_{0},\vec{y}) \, \tilde\varphi_{\vec{k}}^{*}(t,\vec{y})\,\tilde\varphi_{\vec{p}}^{*}(t,\vec{y})\,\dd \vec{y}\right|^{2}~,
\end{multline}
while the VEP of the spinor field detector, model 4, has summands proportional to
\begin{multline}
  \propto \sum_{s,r}\left| \int_{\R^{n}}p(\vec{x}_{0},\vec{y})\,\conj\psi_{\vec{k},s,+}(\vec{y})\,\psi_{\vec{p},r,-}(\vec{y})\,\dd \vec{y}\right|^{2}\\
  = \sum_{s,r}\left| \int_{\R^{n}}p(\vec{x}_{0},\vec{y})\,\tilde\psi_{\vec{k},s,+}(t,\vec{y})\,\gamma^{0}\,\tilde\psi_{\vec{p},r,-}(t,\vec{y})\,\dd \vec{y}\right|^{2}~.
\end{multline}
The sum over all spins in the second expression is, of course, not present in the first one for quantized scalar fields. The factor $(4\omega_{\vec{k}}\omega_{\vec{p}})^{-1}$ in the first expression is simply part of the normalization factor of the scalar mode functions \cref{eqn: scalar mode functions 1}. This can be made explicit by rewriting both of the above expressions in terms of the full, time-dependent mode functions \cref{eqn: scalar mode functions 1,eqn: spinor mode functions} (at arbitrary time $t$). In this form, it is obvious that the \emph{only} difference between the VEP for scalar and spinor fields is due to the different mode functions.

So does it make any difference at all whether the detector is coupled to a quantized scalar or a spinor field? The answer is yes: Comparing the VEPs \cref{eqn: VEP complex quad 2 renorm 2} and \cref{eqn: VEP spinor quad 2 renorm 2} at the end of the previous chapter, we found that quantized scalar fields are better behaved for large frequencies since the UV behavior of the mode functions is different; the scalar mode functions generally have a factor $\sim (\omega_{\vec{k}})^{-1/2}$ which the spinor mode functions do not: As a simple example, compare the mode functions of massless fields in $(1,1)$ dimensions: For the quantized (real or complex) scalar field, the particle and antiparticle mode functions are
\begin{equation}
  \tilde\varphi_{\vec{k},\epsilon}^{*}(t,x) = \frac{1}{\sqrt{2|k|L}} \,e^{+\ii   |k|t}\,e^{-\ii  kx}
\end{equation}
[see \cref{app: quantum fields},  \cref{eqn: scalar mode functions 2,eqn: mode expansion complex field,eqn: mode expansion real field}], while the corresponding spinor mode functions for particles and antiparticles (in this order) are
\begin{equation}\begin{aligned}
  \tilde\psi_{k,s,+}^{\dagger}(t,x) &= \frac{1}{\sqrt{2L}}
  \begin{bmatrix}
    \xi_{s} \\
    \sgn (k) \,\sigma^{3}\xi_{s}
  \end{bmatrix}^{\dagger}e^{+\ii  |k|t}\,e^{-\ii  kx}\\
  \tilde\psi_{k,s,-}(t,x) &= \frac{1}{\sqrt{2L}}
  \begin{bmatrix}
    \sgn (k) \,\sigma^{3}\xi_{s}\\
    \xi_{s}
  \end{bmatrix}\hphantom{^{\dagger}}e^{+\ii  |k|t}\,e^{- ikx}
\end{aligned}\end{equation}
[see \cref{app: quantum fields},  \cref{eqn: spinor mode function 2dim massive,eqn: spinor mode function 2dim massless,eqn: spinor mode functions,eqn: mode expansion spinor field}].
The spinorial parts do no scale with $k$, while the scalar field has an additional factor $|k|^{-1}= \omega_{\vec{k}}^{-1}$.

This difference has notable implications: it causes the quantized scalar fields to exhibit better convergence properties when summing over all momenta---the VEP for detectors coupling quadratically to quantized scalar fields---\cref{eqn: VEP complex quad 2 renorm 2}---will in general converge better than the VEP for detectors coupling quadratically to quantized spinor fields---\cref{eqn: VEP spinor quad 2 renorm 2}. This is a dramatic difference. For instance, there are switching functions and spatial profiles that give a finite VEP for a detector quadratically coupled to a quantized scalar field  but yield a divergent VEP for quantized spinor fields. In \cref{sec: scalar vs spinor examples}, examples where this situation actually occurs are discussed.

Superficially, the different normalization conditions, stemming from a mathematically different inner product, are the reason for this difference: The scalar mode functions are orthonormal with respect to
\begin{multline}\label{eqn: normalization 1}
  \varinner{\tilde\varphi_{\vec{k}}}{\tilde\varphi_{\vec{p}}} \equiv\\
    -\ii   \int_{B}  \tilde\varphi_{\vec{k}}(x) \big[\partial_{0}\tilde\varphi_{\vec{p}}^{*}(x)\big] 
   - \big[\partial_{0}\tilde\varphi_{\vec{k}}(x)\big]\tilde\varphi_{\vec{p}}^{*}(x)\, \dd \vec{x}\\ = \delta_{\vec{k},\vec{p}}~,
\end{multline}
while the spinor mode functions are normalized according to
\begin{multline}\label{eqn: normalization 2}
  \inner{\tilde\psi_{\vec{k},s,\epsilon}}{\tilde\psi_{\vec{p},r,\delta}} \equiv \\
  \int_{B} \tilde\psi_{\vec{k},s,\epsilon}^{\dagger}(x)\,\tilde\psi_{\vec{p},r,\delta}(x)\,\dd \vec{x} = \delta_{\vec{k},\vec{p}}\,\delta_{s,r}\,\delta_{\epsilon,\delta}~.
\end{multline}
The additional time derivative in the product for scalar functions generates a factor $\omega_{\vec{k}}$. Thus, the normalization factor of the scalar fields contains an additional factor $(\omega_{\vec{k}})^{-1/2}$ compared to the normalization of the spinor fields.

Ultimately, the reason is that the equation of motion of scalar fields, the Klein-Gordon equation, is of second order and features a double time derivative, while the equation of motion of the spinor field, the Dirac equation, is of first order, containing only a single time derivative (see \cref{app: normalization}).

\medskip{}
In order to demonstrate the influence this has on the response of the different detector models, we compare their VEPs in three simple examples on $(1,1)$ dimensional Minkowski spacetime for massless fields: (1) sudden switching with a pointlike detector, (2) Gaussian switching with a pointlike detector, and (3) sudden switching with a Gaussian detector profile.

\subsection{Examples in (1,1) dimensions}\label{sec: scalar vs spinor examples}

Let us start by recalling the general expressions for the massless (1,1) dimension scenario. For model 3 in \cref{tab: overview detector models} (quadratic coupling quantized scalar field detector), the VEP \cref{eqn: VEP scalar quad 2 renorm} simplifies in $(1,1)$ dimensions and for massless fields to
\begin{multline}\label{eqn: VEP complex quad 2 renorm 2dim}
  P_{0,g\to e}
  = \frac{\lambda^{2}}{4L^{2}} \sum_{k,p\neq0} \frac{1}{|kp|}\left| \int_{-\infty}^{\infty}p(x_{0},y) e^{-\ii  (k+p)y}\,\dd y\right|^{2}\\
  \times \left|\int_{-\infty}^{+\infty}\chi(t)\,e^{\ii (\Omega +|k|+|p|)t}\,\dd t\right|^{2}.
\end{multline}
On the other hand, for the spinor field detector (model 4 in \cref{tab: overview detector models}), one similarly obtains 
\begin{align}\label{eqn: VEP spinor quad 2 renorm 2dim}
  &P_{0,g\to e}
  = \frac{2\lambda^{2}}{L^{2}} \sum_{k,p>0}  \left[\left| \int_{-\infty}^{\infty}p(x_{0},y)e^{-\ii  (k-p)y}\dd y\right|^{2}\right.\\
 &\!\!\!\!\! + \left.\left| \int_{-\infty}^{\infty}\!\!p(x_{0},y)e^{-\ii  (p-k)y}\dd y\right|^{2}  \right]
  \left|\int_{-\infty}^{+\infty}\!\!\!\chi(t)\,e^{\ii (\Omega +k+p)t}\,\dd t\right|^{2},\nonumber
\end{align}
directly from \cref{eqn: VEP spinor quad 1 renorm} by plugging in the mode functions and spinors \cref{eqn: spinor mode function 2dim massive,eqn: spinor mode function 2dim massless}.

\subsubsection{Sudden switching and pointlike detector}

In the most basic scenario, the detector is pointlike [delta spatial profile as given in \cref{eqn: spatial profile pointlike detector}], and it is suddenly coupled to the field as given in \cref{eqn: sudden switching}. This is the same situation that was considered in \cref{sec: VEP discussion linear - dimension} for the case of a quantized real field with linear coupling (model 1).

\paragraph{Quantized scalar field}
Plugging \cref{eqn: spatial profile pointlike detector,eqn: sudden switching} in \cref{eqn: VEP complex quad 2 renorm 2dim}, the resulting VEP for the quantized scalar field is
\begin{multline}\label{eqn: VEP complex quad 2 renorm 2dim sudden sw pointlike}
  P_{0,g\to e}
  =\frac{\lambda^{2}L^{2}}{4 \pi^{4}} \sum_{l_{1}=1}^{\infty} \sum_{l_{2}=1}^{\infty} \frac{1}{l_{1}l_{2}\left(\frac{\Omega L}{2\pi}+l_{1}+l_{2}\right)^{2}}\\ \times\sin^{2}\left[\frac{\pi}{L}\left( \frac{\Omega L}{2\pi}+l_{1}+l_{2}\right)T\right]
\end{multline}
This double sum is convergent: we can estimate
\begin{equation*}
  \frac{4 \pi^{4}}{\lambda^{2}L^{2}} P_{0,g\to e} \leq  \sum_{l_{1}=1}^{\infty} \sum_{l_{2}=1}^{\infty} \frac{1}{l_{1}l_{2}\left(l_{1}+l_{2}\right)^{2}}~.
\end{equation*}
The sums on the right-hand side converge if and only if
\begin{equation*}
 S= \sum_{l_{1}=2}^{\infty} \sum_{l_{2}=2}^{\infty} \frac{1}{l_{1}l_{2}\left(l_{1}+l_{2}\right)^{2}}
\end{equation*}
does, where the summation starts at $2$ instead of $1$. We can now exploit that $l_{1}l_{2}\geq l_{1}+l_{2}$ for all $l_{i}\geq 2$ to find the upper bound
\begin{equation*}
 S \leq \sum_{l_{1}=2}^{\infty} \sum_{l_{2}=2}^{\infty} \frac{1}{\left(l_{1}+l_{2}\right)^{3}}~,
\end{equation*}
where the sum on the right-hand side is convergent.

\paragraph{Quantized spinor field}
The formal expression for the VEP in case of a quantized spinor field is obtained by plugging \cref{eqn: spatial profile pointlike detector,eqn: sudden switching} in \cref{eqn: VEP spinor quad 2 renorm 2dim}:
\begin{multline}\label{eqn: VEP spinor quad 2 renorm 2dim sudden sw pointlike}
  P_{0,g\to e}
  =\frac{4\lambda^{2}}{\pi^{2}} \sum_{l_{1}=1}^{\infty} \sum_{l_{2}=1}^{\infty} \frac{1}{\left(\frac{\Omega L}{2\pi}+l_{1}+l_{2}\right)^{2}}\\ \times\sin^{2}\left[\frac{\pi}{L}\left( \frac{\Omega L}{2\pi}+l_{1}+l_{2}\right)T\right]~.
\end{multline}
Similar to \cref{eqn: first estimate VEP in 3 dim}, it is easy to prove that this sum converges if and only if
\begin{equation*}
  \sum_{l_{1}=1}^\infty\sum_{l_{2}=1}^{\infty} \frac{1}{\left(\frac{\Omega L}{2\pi}+l_{1}+l_{2}\right)^{2}}
\end{equation*}
does. However, this sum diverges, since the corresponding integral is ill defined:
\begin{multline*}
  \int_{1}^{\infty}\dd l_{1} \int_{1}^{\infty}\dd l_{2} \frac{1}{(A+l_{1}+l_{2})^{2}} =\\
   \lim_{a \to \infty} \ln(a) -\ln(2+A) \to \infty~.
\end{multline*}
Thus, the VEP \cref{eqn: VEP spinor quad 2 renorm 2dim sudden sw pointlike} is divergent.

This result nicely illustrates that the detector response to quantized scalar fields is better behaved in the UV: On one hand, all the scalar fields detectors (models 1-3) have a finite VEP in $(1,1)$ dimensions, even if the detector is simply pointlike and uses sudden switching [see \cref{eqn: VEP complex quad 2 renorm 2dim sudden sw pointlike,eqn: VEP real linear sudden switching pointlike 1 dim}]. On the other hand, in the case of a quantized spinor field the VEP is already divergent in $(1,1)$ dimensions for this spacetime profile.

\subsubsection{Gaussian switching and pointlike detector}

It was demonstrated in \cref{sec: VEP discussion linear - regularization} that the original UDW detector (model 1), which has a divergent VEP in dimensions $(1,n)$ for $n\geq 3$, can be regularized by introducing a Gaussian switching function \cref{eqn: gaussian sw} of width $T$. We can prove that Gaussian switching also regularizes the spinor field detector, model 4, in $(1,1)$ dimensions.

\paragraph{Quantized spinor field}
Inserting \cref{eqn: gaussian sw} in \cref{eqn: VEP spinor quad 2 renorm 2dim}, the VEP for a spinor field reads
\begin{equation}\label{eqn: VEP spinor quad 2 renorm 2dim gaussian sw pointlike}
  P_{0,g\to e}
  =\frac{8 \pi T^{2}\lambda^{2}}{L^{2}} \sum_{l_{1}=1}^{\infty} \sum_{l_{2}=1}^{\infty} e^{-\frac{4\pi^{2}T^{2}}{L^{2}}\left(\frac{\Omega L}{2\pi}+l_{1}+l_{2}\right)^{2}}~.
\end{equation}
This excitation probability is now clearly finite: the sum converges if and only if 
\begin{equation*}
  \int_{1}^{\infty}\dd l_{1}\int_{1}^{\infty}\dd l_{2}\,e^{-A(B+l_{1}+l_{2})^{2}}
\end{equation*}
does, where $A = 4\pi^{2}T^{2}L^{-2}$ and $B=\Omega L(2\pi)^{-1}$, and the latter expression is bounded from above by
\begin{equation*}
  \int_{1}^{\infty}\dd l_{1}\,e^{-Al_{1}^{2}}\int_{1}^{\infty}\dd l_{2}\,e^{-Al_{2}^{2}} = \frac{\pi}{A}~.
\end{equation*}

\paragraph{Quantized scalar field} 
For model 3, the same switching function yields
\begin{equation}\label{eqn: VEP complex quad 2 renorm 2dim gaussian sw pointlike}
    P_{0,g\to e} =\frac{T^{2}\lambda^{2}}{2\pi} \sum_{l_{1}=1}^{\infty} \sum_{l_{2}=1}^{\infty} \frac{1}{l_{1}l_{2}}e^{-\frac{4\pi^{2}T^{2}}{L^{2}}\left(\frac{\Omega L}{2\pi}+l_{1}+l_{2}\right)^{2}}
\end{equation}
which is of course also finite since the summand decays faster than in \cref{eqn: VEP spinor quad 2 renorm 2dim gaussian sw pointlike} [and the VEP had been convergent in \cref{eqn: VEP complex quad 2 renorm 2dim sudden sw pointlike} even without regularization].

Note that for both field types, $P_{0,g\to e} \to 0$ in the adiabatic limit $T \to \infty$, just like for the linear coupling in \cref{eqn: VEP real linear gaussian switching}.

\subsubsection{Sudden switching and Gaussian detector profile}

The third example discussed in \cref{sec: VEP discussion linear - regularization} for model 1 demonstrated regularization through a Gaussian detector profile \cref{eqn: gaussian detector profile} with variance $\sigma>0$ and sudden switching. We repeat the example in the present case.

\paragraph{Quantized spinor field}
In case of the quantized spinor field, the probability is
\begin{multline}\label{eqn: VEP spinor quad 2 renorm 2dim sudden sw gaussian profile}
  P_{0,g\to e}
  = \frac{4\lambda^{2}}{\pi^{2}} \sum_{l_{1}=1}^{\infty} \sum_{l_{2}=1}^{\infty}
  \frac{e^{-\frac{4\pi^{2}\sigma^{2}}{L^{2}}(l_{1}-l_{2})^{2}}}{\left(\frac{\Omega L}{2\pi}+l_{1}+l_{2}\right)^{2}}\\ \times  \sin^{2}\left[\frac{\pi}{L}\left( \frac{\Omega L}{2\pi}+l_{1}+l_{2}\right)T\right]
\end{multline}
and convergence is harder to assess than for the Gaussian switching function in \cref{eqn: VEP spinor quad 2 renorm 2dim gaussian sw pointlike}: the sine squared is positive, bounded from above by 1 and can thus be ignored. For $l_1$ = $l_2$, the exponential factor is always $1$, but then the denominator decays as $l_1^{-2}$; as it happens, this is just fast enough. As a first step, we can estimate
\begin{equation*}
   \frac{\pi^{2}}{4\lambda^{2}}  P_{0,g\to e} \leq \sum_{l_{1}=1}^{\infty} \sum_{l_{2}=1}^{\infty} \frac{e^{-A(l_{1}-l_{2})^{2}}}{\left(l_{1}+l_{2}\right)^{2}}
\end{equation*}
where $A = 4\pi^{2}\sigma^{2}L^{-2} > 0$. This sum converges if the corresponding integral
\begin{equation*}
 I= \int_{l_{1}=1}^{\infty}\dd l_{1} \int_{l_{2}=1}^{\infty}\dd l_{2}\, \frac{e^{-A(l_{1}-l_{2})^{2}}}{\left(l_{1}+l_{2}\right)^{2}}
\end{equation*}
does. By substituting $l = l_{2}-l_{1}$,
\begin{equation*}
 I = \int_{l_{1}=1}^{\infty}\dd l_{1} \int_{l=1-l_{1}}^{\infty}\dd l\, \frac{e^{-Al^{2}}}{\left(l+2l_{1}\right)^{2}}
\end{equation*}
and since for $l\geq 1-l_{1}$ we have $(l+2l_{1})^{-2} < l_{1}^{-2}$:
\begin{multline*}
  I  < \int_{l_{1}=1}^{\infty}\dd l_{1} \,\frac{1}{l_{1}^{2}} \int_{l=1-l_{1}}^{\infty}\dd l\, e^{-Al^{2}} \\
  \leq \int_{l_{1}=1}^{\infty}\dd l_{1}\, \frac{1}{l_{1}^{2}} \int_{l=-\infty}^{\infty}\dd l\, e^{-Al^{2}} = \sqrt{\frac{\pi}{A}} < \infty~.
\end{multline*}
The VEP \cref{eqn: VEP spinor quad 2 renorm 2dim sudden sw gaussian profile} is finite.

\paragraph{Quantized scalar field}
For the quantized scalar field, the probability
\begin{multline}\label{eqn: VEP complex quad 2 renorm 2dim sudden sw gaussian profile}
  P_{0,g\to e}
  = \frac{\lambda^{2}L^{2}}{8\pi^{4}} \\
  \times \sum_{l_{1}}^{\infty}\sum_{l_{2}=1}^{\infty} \frac{e^{-\frac{4\pi^{2}\sigma^{2}}{L^{2}}(l_{1}+l_{2})^{2}} + e^{-\frac{4\pi^{2}\sigma^{2}}{L^{2}}(l_{1}-l_{2})^{2}}}{l_{1}l_{2}\left(\frac{\Omega L}{2\pi}+l_{1}+l_{2}\right)^{2}}\\ \times\sin^{2}\left[\frac{\pi}{L}\left( \frac{\Omega L}{2\pi}+l_{1}+l_{2}\right)T\right]
\end{multline}
is once again convergent because \cref{eqn: VEP complex quad 2 renorm 2dim sudden sw pointlike} already was even without regularization.

Again, $P_{0,g\to e} \to 0$ for both fields if the detector is maximally delocalized by $\sigma \to \infty$.

\subsection{Comparability of the models}\label{sec: conclusion comparison}

From the analysis above we conclude that  the differences in the response of models 3 and 4 are only due to the analytic structure of the fields and the field statistics.

The field statistics do not enter the leading order VEP, but will have a profound influence on the detector response at higher orders, or also at leading order if the field is initially not in the vacuum state. On the level of the analytic structure, model 4 is sensitive to the additional spin degree of freedom, and, more importantly, requires stronger regularization since it displays a worse-behaved UV response. Interestingly, it is not enough to regularize with a smooth switching function as opposed to the scalar case. Rather, the spinor model requires a spatial smearing in order to yield finite VEP.

Note that only a qualitative comparison is feasible. A naïve quantitative comparison fails because the coupling strength $\lambda$ has different dimensions in model 3 and 4: Scalar fields on $(1,n)$-dimensional spacetime have mass dimension $(n-1)/2$, while spinor fields have mass dimension $n/2$. To obtain a Hamiltonian of mass dimension $1$, the respective coupling constants therefore need to be of different dimensions. A direct comparison is therefore only possible with regards to the functional dependence on model parameters like the detector trajectory $x(\tau)$, the detector gap $\Omega$, and more concretely the interaction time $T$ and the detector size $\sigma$. In this sense, the leading-order vacuum response of models 3 and 4 is equivalent in all (convergent) examples discussed in \cref{sec: scalar vs spinor examples}.

In conclusion, for finite size detectors after renormalization, the pair of models, model 3 (for quantized complex fields) and model 4 (for quantized spinor fields) can be used to reliably study the differences between fermionic and bosonic fields via particle detectors, at least at leading order in perturbation theory.

\section{Computation methods at arbitrary order}\label{sec: computation methods}

Up to now, the discussion was limited to the leading order in the coupling strength $\lambda$. However, potential applications of particle detectors will undoubtedly require calculations to higher orders in perturbation theory: some phenomena that have been investigated for quantized real fields using the original UDW detector, for example, rely on higher-order effects order (see  \sources{martin-martinez_sustainable_2013,martin-martinez_Unruh2013} among others). In preparation of investigations at higher orders, it will be convenient to derive a set of Feynman rules for each detector model. We will adhere to the following roadmap: In this section we develop the computation methods necessary to obtain Feynman rules. In \cref{sec: VNRP}, we will apply these methods to calculations up to second order in the transition amplitudes by way of example, and in \cref{sec: feynman rules} we will generalize the emerging pattern to Feynman rules.

It is to be expected that new divergencies appear at each order in perturbation theory. As in any quantum field theory, there are three sources of divergencies in transition probabilities of UDW-type detector models:
\begin{compactenum}
  \item The amplitude of a single process can be infinite, as it was the case for the tadpolelike diagram in \cref{fig: VEP processes 1st order}.
  \item Finite amplitudes can still amount to a divergent transition probability of the detector when the field is traced out: the sum over all the ingoing and outgoing field configurations can diverge.
  \item The perturbative series itself can diverge when taking into account all orders.
\end{compactenum}
Notice, therefore,  it is important to ascertain that those divergences can be regularized and renormalized away in order to be able to fully trust insights gained from using a particle detector model to describe the physics of a quantum field.

\paragraph{Example: vacuum no-response probability}
In order to show which quantities we are going to need to evaluate, let us first consider a transition probability which involves the second order term in the Dyson expansion \cref{eqn: perturbative expansion}. Every application of the monopole operator $\op{\mu}$ \cref{eqn: monopole operator} to one of the two energy eigenstates of the detector switches it back and forth between ground state $\ket{g}$ and the excited state $\ket{e}$. In consequence, if the detector starts out in the ground state, only even orders in perturbation theory contribute to the probability for it to remain in the ground state, and only odd orders contribute to the probability for the detector to be excited. There are thus no  corrections to the VEP coming from $U^{(2)}$ in \cref{eqn: perturbative expansion,eqn: time-evolution order n}. Instead, we calculate the probability for the detector to \emph{remain} in the ground state when interacting with a quantum field in the vacuum state, and call this probability the \emph{vacuum no-response probability} (VNRP) $P_{0,g\to g}$.

Notice, however,  that it would not be strictly necessary to go through the $U^{(2)}$ calculation in order to compute the VNRP: Since the same order perturbative corrections to the density matrix are traceless \cite{jonsson_quantum_2014}, it follows that $P_{0,g\to g}=1-P_{0,g\to e}$. We will nevertheless compute it using $U^{(2)}$ to illustrate higher order techniques.

Let $F$ be the probed quantum field, and $\op{\rho}(t_{0}) = \ket{0,g}\bra{0,g}$ the density matrix of the coupled system at starting time $t_{0}$. The VNRP is the $(g,g)$ component after tracing out the field:
\begin{multline}\label{eqn: def VNRP}
  P_{0,g\to g}(t,t_{0}) \define\\ \braket{g|\tr_{F} \op\rho(t)|g} = \sum_{a} \left|\braket{a,g;t| \op{U}(t,t_{0})|0, g;t_{0}}\right|^{2}
\end{multline}
at time $t>t_{0}$. In the limit $t_0 \to -\infty$, $t \to \infty$, the actual duration of the interaction between detector and field is determined by the switching function $\chi$.
Using the perturbative expansion \cref{eqn: perturbative expansion,eqn: time-evolution order n}, the VNRP can be written as
\begin{equation}\label{eqn: VNRP t-o operators 1}
    P_{0,g\to g}(t,t_{0}) = 1 + 2 \sum_{a} \left[ \varRe{\left(\mathcal{A}^{(0)}_{a,g}\mathcal{A}^{(2)}_{a,g}\right)} \right] + \mathcal{O}(\lambda^4)~,
\end{equation}
where the transition amplitudes at order $n$ are
\begin{equation}\label{eqn: def transition amplitudes VNRP}
  \mathcal{A}^{(n)}_{a,g} \define (-\ii  \lambda)^{n}\Ibraket{a,g;t|\op{U}^{(n)}(t,t_{0})|0,g;t_{0}}~,
\end{equation}
The contribution at order zero is only nonzero for the final state $\ket{0,g}$ since $\op{U}^{0} = \id$:
\begin{equation}\label{eqn: VNRP amplitude at order 0}
  \mathcal{A}^{(0)}_{a,g} =
  \begin{cases}
    1 & \text{for } a = 0\\
    0 & \text{else}
  \end{cases}~.
\end{equation}
As we have argued earlier, order $n=1$ does not contribute at all, $\mathcal{A}^{(1)}_{a,g} = 0$.
Thus, only the amplitude at order $n=2$,
\begin{multline}\label{eqn: VNRP amplitude at order 2}
  \mathcal{A}^{(2)}_{a,g} = (-\ii  )^{2}\int_{t_{0}}^{t}\dd t_{1}\int_{t_{0}}^{t}\dd t_{2}\\
  \Ibraket{a,g;t|T\Iop{H}_{\mathrm{int}}(t_{1}) \Iop{H}_{\mathrm{int}}(t_{2})|0,g;t_{0}}~,
\end{multline}
remains to be evaluated to obtain the VNRP.

As an example, consider the renormalized detector model 3 for quantized complex fields with interaction Hamiltonian \cref{eqn: interaction H scalar quad renorm}. The second order correction is
\begin{multline}\label{eqn: VNRP scalar amplitudes order 2 t-o op}
  \mathcal{A}^{(2)}_{\phi,g} =  (-\ii  \lambda)^{2}
  e^{-\ii  \omega_{\phi}t} \\
  \times \int_{t_{0}}^{t}\dd t_{1}\int_{t_{0}}^{t}\dd t_{2}\,\chi(t_{1})\chi(t_{2})\,\braket{g|T\op{\mu}(t_{1})\op{\mu}(t_{2})|g}\\
  \times \int_{\R^{n}}\dd \vec{y}_{1}\int_{\R^{n}}\dd \vec{y}_{2}\,p(\vec{x}_{0},\vec{y}_{1})p(\vec{x}_{0},\vec{y}_{2})\\
  \times \braket{\phi|T\normalord \op{\Phi}^{\dagger}(y_{1}) \op{\Phi}(y_{1}) \normalord \normalord \op{\Phi}^{\dagger}(y_{2}) \op{\Phi}(y_{2}) \normalord|0}~.
\end{multline}
The main difficulty now rests in evaluating the time-ordered expectation values
\begin{equation}
  \braket{g|T\op{\mu}(t_{1})\op{\mu}(t_{2})|g}
\end{equation}
and
\begin{equation}\label{eqn: generic expectation value at second order for complex VNRP}
  \braket{\phi|T\normalord \op{\Phi}^{\dagger}(y_{1}) \op{\Phi}(y_{1}) \normalord \normalord \op{\Phi}^{\dagger}(y_{2}) \op{\Phi}(y_{2}) \normalord|0}~.
\end{equation}

\paragraph{Higher orders}
Generalizing the above example to order $n$ and arbitrary initial and final states, we need to evaluate expressions of the form
\begin{equation}\label{eqn: general correlator detector}
  \braket{g|[\sigma^{-}]^{a}[\,T\op{\mu}(t_{1})\cdots\op{\mu}(t_{n})\,]\,[\sigma^{+}]^{b}|g}
\end{equation}
for the detector, where $a,b\in\{0,1\}$ to allow for initial and final states $\ket{g}$ and $\ket{e}$. Similarly, for model 1 and quantized real fields we need to calculate
\begin{equation}\label{eqn: general correlator real field}
  \braket{0|(\op{a}_{\vec{p}})_{\vec{p}}\,[\,T \op{\Phi}(y_{1})  \cdots  \op{\Phi}(y_{n}) \,]\,(\op{a}^{\dagger}_{\vec{k}})_{\vec{k}}|0}~,
\end{equation}
where $(\op{a}^{\dagger}_{\vec{k}})_{\vec{k}} \equiv \op{a}^{\dagger}_{\vec{k}_{1}} \dots \op{a}^{\dagger}_{\vec{k}_{N}}$ is an arbitrary product of creation operators [such as, e.g., $(\op{a}^{\dagger}_{\vec{k}})_{\vec{k}} = a_{\vec{k}_1}^\dagger a_{\vec{k}_2}^\dagger  a_{\vec{k}_3}^\dagger a_{\vec{k}_1}^\dagger$ or $(\op{a}^{\dagger}_{\vec{k}})_{\vec{k}} = \id$] allowing for arbitrary number eigenstates content in the initial state, and $(\op{a}_{\vec{p}})_{\vec{p}}$ are annihilation operators for arbitrary final states. In case of model 3 (for quantized complex fields), we need
\begin{multline}\label{eqn: general correlator complex field}
  \bra{0}(\op{a}_{\vec{p}})_{\vec{p}}\,(\op{b}_{\bar{\vec{p}}})_{\bar{\vec{p}}}\,[\,T\normalord \op{\Phi}^{\dagger}(y_{1}) \op{\Phi}(y_{1}) \normalord \cdots\\ \cdots \normalord \op{\Phi}^{\dagger}(y_{n}) \op{\Phi}(y_{n}) \normalord\,]\,(\op{a}^{\dagger}_{\vec{k}})_{\vec{k}}\,(\op{b}^{\dagger}_{\bar{\vec{k}}})_{\bar{\vec{k}}}\ket{0}
\end{multline}
and in the case of model 2 (for quantized real fields), the previous expression simplifies because the field is self-adjoint, $\Phi^{\dagger} = \Phi$, and there is only one type of ladder operators. Finally, for model 4 coupling to quantized spinor fields expressions of the form
\begin{multline}\label{eqn: general correlator spinor field}
  \bra{0} (\op{a}_{\vec{p},r})_{\vec{p},r}\,(\op{b}_{\bar{\vec{p}},\bar{r}})_{\bar{\vec{p}},\bar{r}}\,[\,T\normalord \op{\conj\Psi}(y_{1}) \op{\Psi}(y_{1}) \normalord \cdots\\
   \cdots \normalord \op{\conj\Psi}(y_{n}) \op{\Psi}(y_{n}) \normalord\,]\,(\op{a}^{\dagger}_{\vec{k},s})_{\vec{k},s}\,(\op{b}^{\dagger}_{\bar{\vec{k}},\bar{s}})_{\bar{\vec{k}},\bar{s}}\ket{0}
\end{multline}
appear.

In short, we need to evaluate the vacuum expectation value of time-ordered products of field operators, which may contain normal-ordered subproducts, with a prepended product of annihilation operators and an appended product of creation operators. Since the field operators and the monopole operator are essentially sums of ladder operators as well, a good strategy is to normal-order the entire product 
by successively commuting (or anticommuting) ladder operators. This procedure will furnish us with variants of \emph{Wick's theorem}. Moreover, the presence of the time-ordering symbol leads to the appearance of \emph{Feynman propagators} instead of ordinary commutators. This approach is closely related to standard techniques in elementary particle physics, where elements of the scattering matrix $S$ are computed to obtain probabilistic predictions about the outcome of scattering experiments. There, the Lehmann-Symanzik-Zimmermann reduction formula similarly allows one to express elements of $S$ in terms of vacuum expectation values of products of time-ordered fields \cite{peskin_introduction_1995,bjorken_relativistic_1965}.

\paragraph{Time- and normal-ordering}
Let us briefly revise the definition of time-ordering and normal-ordering for bosonic and fermionic fields. In the bosonic case, the time-ordering symbol is defined as
\begin{multline}
  T\Phi(x)\Phi(y) \\= \Theta(x^{0}-y^{0}) \Phi(x)\Phi(y) +\Theta(y^{0}-x^{0}) \Phi(y)\Phi(x)~.
\end{multline}
In particular, exchanging two operators inside a time-ordered expression does not engender any change:
\begin{equation}
  T\Phi(x)\Phi(y) = T\Phi(y)\Phi(x) ~.
\end{equation}
Normal-ordering is achieved by moving all annihilation operators to the right:
\begin{align}
  \normalord a_{\vec{k}} a^{\dagger}_{\vec{p}} \normalord &=a^{\dagger}_{\vec{p}}  a_{\vec{k}}~, & \normalord a^{\dagger}_{\vec{p}} a_{\vec{k}} \normalord &=a^{\dagger}_{\vec{p}}  a_{\vec{k}}~.
\end{align}
For fermionic fields, however,
\begin{multline}
  T\Psi(x)\Psi(y) \\= \Theta(x^{0}-y^{0}) \Psi(x)\Psi(y) -\Theta(y^{0}-x^{0}) \Psi(y)\Psi(x)
\end{multline}
so exchanging two fields introduces an additional sign
\begin{equation}
  T\Psi(x)\Psi(y) = -T\Psi(y)\Psi(x) ~.
\end{equation}
Similarly, a sign is introduced when exchanging ladder operators in normal-ordering:
\begin{align}
  \normalord a_{\vec{k},s} a^{\dagger}_{\vec{p},r} \normalord &= -a^{\dagger}_{\vec{p},r}  a_{\vec{k},s}~, &   \normalord a^{\dagger}_{\vec{p},r}a_{\vec{k},s}  \normalord &= a^{\dagger}_{\vec{p},r}  a_{\vec{k},s} ~.
\end{align}

\subsection{The monopole moment as a field}

A two-level particle detector is, in a way, a fermionic system: the ladder operators \cref{eqn: detector ladder operators} of the monopole operator satisfy the anticommutation relations
\begin{align}\label{eqn: acr ladd op monopole}
  \{\sigma^{+},\sigma^{+}\} &= 0~, & \{\sigma^{+},\sigma^{-}\} &= \id~, & \{\sigma^{-},\sigma^{-}\} &= 0~.
\end{align}
By comparison with the ladder operators of the quantized spinor field it is easy to see that, formally, $\sigma^{-}$, $\sigma^{+}$ correspond to the ladder operators of a fermionic field which 
\begin{inparaenum}[(1)]
  \item has only one mode,
  \item does not have spin degrees of freedom, and
  \item has particle modes which are their own antiparticles.
\end{inparaenum}
The analogy extends to the monopole operator: in the Heisenberg picture,
\begin{equation}\label{eqn: monopole moment H picture}
  \mu(t) = \sigma^{+}e^{+\ii  \Omega t} + \sigma^{-}e^{-\ii  \Omega t}~,
\end{equation}
which can be seen as a simplification of the mode expansion of the quantized spinor field \cref{eqn: mode expansion spinor field}. In the following, we will therefore regard the monopole operator as a fermionic quantum field $\mu(x) = \mu (t,\vec{x})\define \mu(t)$
on $(1,n)$-dimensional Minkowski spacetime which is constant in space, rather than as an observable of a system in first quantization.
The time-evolution of this field is described by a Klein-Gordon equation: from \cref{eqn: monopole moment H picture} immediately follows $(\partial_{t}^{2}+\Omega^{2}) \mu(t) = 0$, which can be rewritten as
\begin{equation}
  (\dalembert+\Omega^{2}) \mu(x) = 0
\end{equation}
since $\laplace \mu(x) = 0$.
While this field is defined on the entire spacetime, it interacts with the quantum field only for a limited time and in a certain place, according to the detector's spacetime profile $f$.

In terms of particle number eigenstates, this field is very simple: since there is only one mode, no spin quantum numbers and no distinct antiparticles, there can at most be one quantum; excitation of more quanta is forbidden by the Pauli exclusion principle. The two particle number eigenstates are simply the two energy eigenstates of the detector: the ground state $\ket{g}$ for no quantum, the excited state $\ket{e}$ for one quantum. The detector behaves as a sort of completely delocalized Grassman scalar \cite{montero_fermionic_2011,Bradler2012,Montero2012b}.

\subsection{Feynman propagators}\label{sec: propagators}

\begin{figure}
  \includegraphics{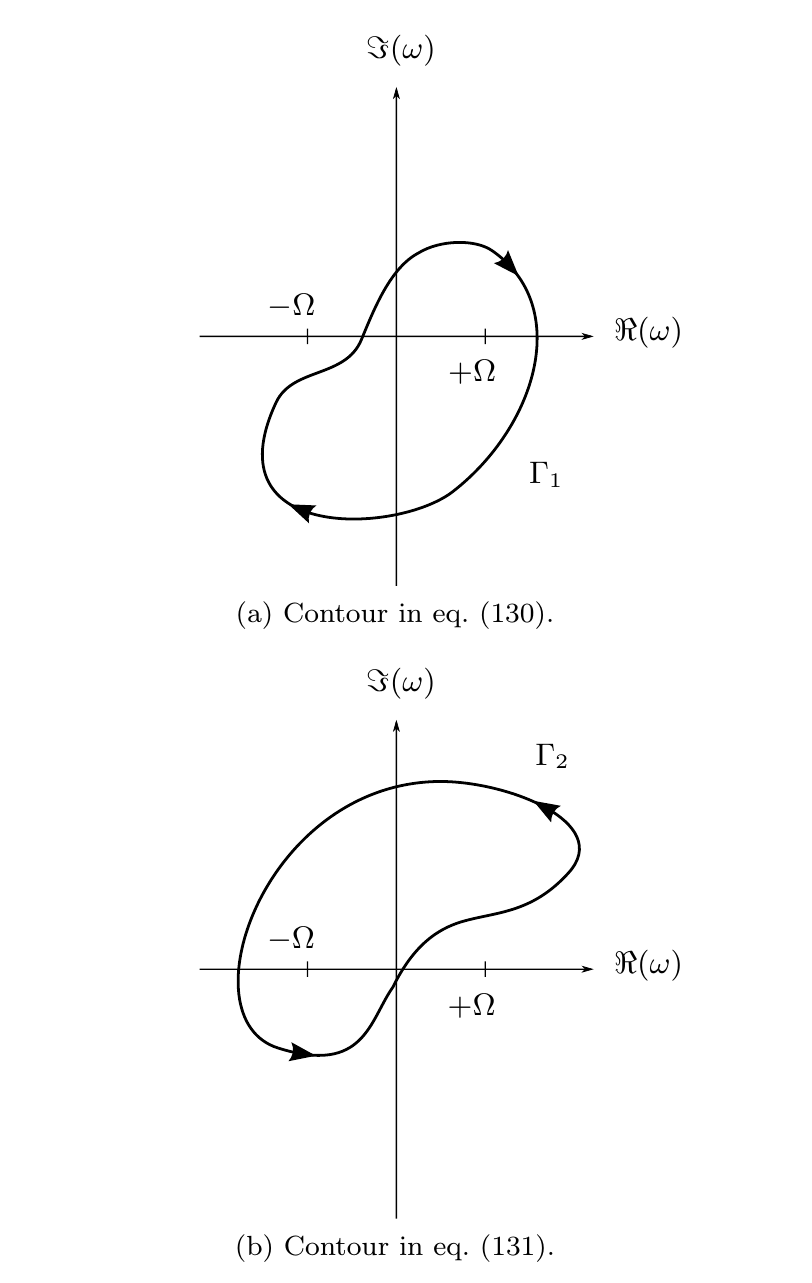}
  \caption{Curves for the complex contour integrations in \cref{eqn: intro contour 1,eqn: intro contour 2}.}
  \label{fig: contours 1}
\end{figure}

Deriving the Feynman propagator $D_{F}$ of the UDW-type detector (or, more precisely, its monopole moment) runs analogously to the derivation of the Feynman propagator $S_{F}$ of the quantized spinor field. We will treat the detector first, and generalize to the well-known Feynman propagators of the other fields for the sake of completeness.

\subsubsection{UDW-type detector}

The Feynman propagator of the UDW-type detector is defined as the time-ordered vacuum expectation value
\begin{equation}\label{eqn: def Feynman P detector}
  D_{F}(t,s) \define \braket{g|T\op{\mu}(t) \op{\mu}(s)|g}~,
\end{equation}
the resulting two 2-point correlators are
\begin{equation}\label{eqn: Feynman P detector 2}
  D_{F}(t,s)\equiv D_{F}(t-s) =
  \begin{cases}
    \hphantom{+}e^{-\ii  \Omega(t-s)} & \text{for }t > s\\
    -e^{+\ii  \Omega(t-s)} & \text{for }t < s
  \end{cases}~.
\end{equation}
Note that $D_{F}(0)$ is ill defined, since $D_{F}\to \pm1$, depending on whether $t<s$ or $t>s$.
For $t \neq s$, we may simply express the Feynman propagator using the step function $\Theta$ as
\begin{equation}\label{eqn: Feynman P detector 1}
    D_{F}(t,s) =\\
     \Theta(t-s) e^{-\ii  \Omega(t-s)} -  \Theta(s-t) e^{+\ii  \Omega(t-s)}~.
\end{equation}
This form of the Feynman propagator has the disadvantage that the exponential factor switches signs depending on which time is larger---in calculations of transition probabilities, one would therefore not end up with the Fourier transform of the spacetime profile $f$ when integrating over all times in the Feynman rules later on. It is therefore useful to rewrite the expression using a complex contour integral (as it is usually for the Feynman propagator of a quantum field \cite{peskin_introduction_1995}). The two cases of time-ordering are then implemented as two different contours.

By the residue theorem, we can rewrite
\begin{equation}\label{eqn: intro contour 1}
  e^{-\ii  \Omega(t-s)} = 2\Omega \oint_{\Gamma_{1}}\frac{d\omega}{2\pi}\, \frac{\ii}{(\omega-\Omega)(\omega+\Omega)}e^{-\ii  \omega(t-s)}~,
\end{equation}
where the curve $\Gamma_{1}$ is given in the first part of \cref{fig: contours 1} and encloses the pole at $\omega = +\Omega$. Similarly,
\begin{equation}\label{eqn: intro contour 2}
  e^{+\ii  \Omega(t-s)} = 2\Omega \oint_{\Gamma_{2}}\frac{d\omega}{2\pi}\, \frac{\ii}{(\omega-\Omega)(\omega+\Omega)}e^{-\ii  \omega(t-s)}
\end{equation}
for $\Gamma_{2}$ as shown in the second part of \cref{fig: contours 1}, enclosing the pole at $\omega = -\Omega$. In order to be able to integrate along the real axis, we shift the poles off the axis by $\pm \ii   \epsilon$ and straighten out the curves as shown in \cref{fig: contours 2}. After the evaluation of the integrals, we will take the limit $\epsilon \to 0$. Since $\Gamma_{1}' = \Gamma(r) \cup \Gamma_{1}(r)$ and $\Gamma_{2}' = \Gamma(r) \cup \Gamma_{2}(r)$ as drawn in \cref{fig: contours 2}, the integral over each of the entire closed curves can be split into two integrals. In the limit $r \to \infty$, the first integral simply turns into an integral over the entire real line. The second one, on the other hand, evaluates to zero in both cases:
\begin{multline}\label{eqn: complex contour 1}
  \Theta(t-s)\,e^{-\ii  \Omega(t-s)} = \Theta(t-s)\,\lim_{\epsilon \to 0} 2\Omega\\ \int_{-\infty}^{\infty}\frac{d\omega}{2\pi}\,\frac{\ii}{(\omega-\Omega+\ii  \epsilon)(\omega+\Omega-\ii  \epsilon)}e^{-\ii  \omega(t-s)}
\end{multline}
and
\begin{multline}\label{eqn: complex contour 2}
  \Theta(s-t)\,e^{+\ii  \Omega(t-s)} = \Theta(s-t)\,\lim_{\epsilon \to 0} 2\Omega\\ \int_{-\infty}^{\infty}\frac{d\omega}{2\pi}\, \frac{\ii}{(\omega-\Omega+\ii  \epsilon)(\omega+\Omega-\ii  \epsilon)}e^{-\ii  \omega(t-s)}~.
\end{multline}
By introducing time derivatives
\begin{align}
  e^{-\ii  \Omega(t-s)} &= \frac{i\partial_{t}}{\Omega}e^{-\ii  \Omega(t-s)}~, &  e^{+\ii  \Omega(t-s)} &= -\frac{i\partial_{t}}{\Omega}e^{+\ii  \Omega(t-s)}
\end{align}
and evaluating them inside the integral, we obtain
\begin{multline}\label{eqn: Feynman P detector 3}
  D_{F}(t-s) =\\ \lim_{\epsilon \to 0} 2 \int_{-\infty}^{\infty}\frac{d\omega}{2\pi}\,\frac{i\omega}{(\omega-\Omega+\ii  \epsilon)(\omega+\Omega-\ii  \epsilon)}e^{-\ii  \omega(t-s)}~.
\end{multline}
Since in the limit
\begin{equation}\label{eqn: pole shift approximation}
  \lim_{\epsilon \to 0} (\omega-\Omega+\ii  \epsilon)(\omega+\Omega-\ii  \epsilon) = \lim_{\delta \to 0}\omega^{2}-\Omega^{2}+\ii  \delta~,
\end{equation}
the above expression is abridged in the usual way to
\begin{equation}\label{eqn: Feynman P detector}
  D_{F}(t-s) = \lim_{\epsilon \to 0} 2 \int_{-\infty}^{\infty}\frac{d\omega}{2\pi}\,\frac{i\omega}{\omega^{2}-\Omega^{2}+\ii  \epsilon}e^{-\ii  \omega(t-s)}~.
\end{equation}

If the arguments of the propagator are exchanged, a sign is picked up,
\begin{equation}
  D_{F}(s-t) = - D_{F}(t-s)~,
\end{equation}
since $T\mu(t)\mu(s) = -T\mu(s)\mu(t)$.

\begin{figure}
  \includegraphics{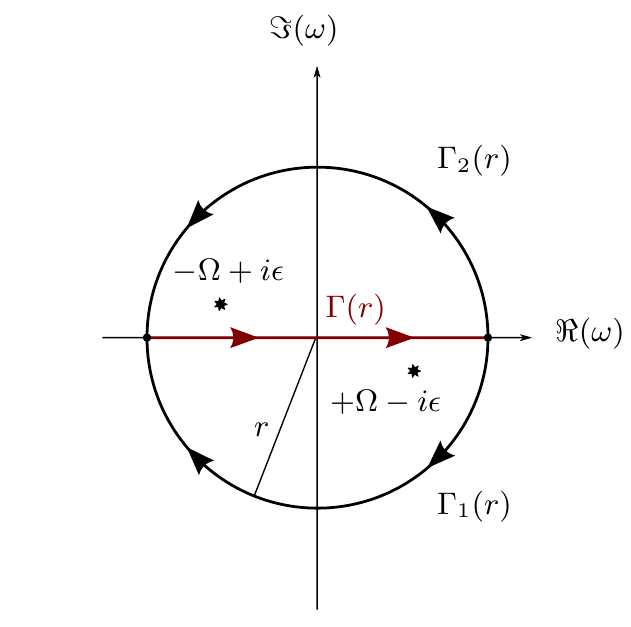}
  \caption{Curves for the complex contour integration in \cref{eqn: complex contour 1,eqn: complex contour 2}.}
  \label{fig: contours 2}
\end{figure}

\subsubsection{Quantized scalar field}

Although well known in the literature, we repeat the definition of the Feynman propagator for quantized scalar fields:
\begin{equation}\label{eqn: def Feynman P complex}
  G_{F}(x-y) \define \braket{0|T\op{\Phi}(x) \op{\Phi}^{\dagger}(y)|0}~.
\end{equation}
It can be expressed as
\begin{multline}\label{eqn: Feynman P scalar field}
  G_{F}(x-y) =\\ \lim_{\epsilon \to 0} \frac{1}{L^{n}}\sum_{\vec{k}} \int_{-\infty}^{\infty}\frac{dk^{0}}{2\pi}\,\frac{\ii}{k^{2}-m^{2}+\ii  \epsilon}e^{-\ii  k(x-y)}~,
\end{multline}
see, for example, \source{peskin_introduction_1995}. 
Note that for $x = y$,
\begin{equation}
  G_{F}(0) = \frac{1}{2L^{n}}\sum_{\vec{k}}\frac{1}{\omega_{\vec{k}}} \to \infty~.
\end{equation}
Unlike for the detector, no sign is picked up under exchange of the arguments:
\begin{equation}
  G_{F}(x-y) = G_{F}(y-x)~.
\end{equation}

\subsubsection{Quantized spinor field}

The Feynman propagator of the quantized spinor field is, by definition,
\begin{equation}\label{eqn: def Feynman P spinor field}
  {S_{F}}(x-y)^{A}_{B} \define \braket{0|T \op{\Psi}^{A}(x)\op{\conj\Psi}_{B}(y)|0}~,
\end{equation}
and can be expressed as \cite{peskin_introduction_1995}
\begin{multline}\label{eqn: Feynman P spinor field}
  {S_{F}}(x-y)^{A}_{B} =\\
  \lim_{\epsilon \to 0}\frac{1}{L^{n}}\sum_{\vec{k}}\int_{-\infty}^{\infty}\frac{dk^{0}}{2\pi}\frac{i\left[k\cdot \gamma+ m\right]^{A}_{B}}{k^{2}-m^{2}+\ii  \epsilon}\,e^{-\ii  k(x-y)}~.
\end{multline}
The above expression is valid in $(1,n)$ dimensions, for both $n=1$ and $n=3$; the massless case is correctly recovered by setting $m=0$.

It is straightforward to check that the Feynman propagator diverges in the coincidence limit
\begin{equation}
  S_{F}(0) \to \infty~.
\end{equation}
Unlike for the detector, where the exchange of the propagator's argument only produced a sign, there is no simple relation between $S_{F}(x-y)$ and $S_{F}(y-x)$:
\begin{multline}
  S_{F}(y-x)^{A}_{B} \\= - \lim_{\epsilon \to 0}\frac{1}{L^{n}}\sum_{\vec{k}}\int_{-\infty}^{\infty}\frac{dk^{0}}{2\pi}\frac{i\left[k\cdot \gamma- m\right]^{A}_{B}}{k^{2}-m^{2}+\ii  \epsilon}\,e^{-\ii  k(x-y)}~.
\end{multline}

\subsection{Wick's theorem}\label{sec: wick}

In this section we will first address Wick's theorem for quantized scalar fields. Then, we will formulate the fermionic version of the theorem for spinor fields, and subsequently simplify it in order to arrive to the case of the detector (monopole field). Notice that full proofs of the theorems are given in \cref{app: proofs wick}.

\subsubsection{Quantized complex field}\label{sec: Wick complex}

Ultimately, we want to calculate expectation values like \cref{eqn: general correlator complex field}, or more generally:
\begin{multline}
  \bra{0}(\op{a}_{\vec{p}})_{\vec{p}}\,(\op{b}_{\bar{\vec{p}}})_{\bar{\vec{p}}}\,[\,T\op{\Phi}^{\epsilon_{1}}(x_{1}) \cdots\\\cdots \op{\Phi}^{\epsilon_{n}}(x_{n})  \,]\,(\op{a}^{\dagger}_{\vec{k}})_{\vec{k}}\,(\op{b}^{\dagger}_{\bar{\vec{k}}})_{\bar{\vec{k}}}\ket{0}~,
\end{multline}
where the index $\epsilon_{i}$ indicates that the fields may be conjugated or not. Since the vacuum expectation value of normal-ordered products of operators vanishes, it is useful to rewrite the original sequence of operators as a sum of normal-ordered expressions. To achieve normal-ordering, the fields need to be expanded in terms of ladder operators, and the ladder operators commuted with each other until normal-ordering is reached. The well-known Wick theorem allows to do so systematically. We will proceed in two steps in evaluating the above expression. In the first step, we present Wick theorem for time-ordered products of field operators. In the second step, we then extend this theorem and allow for more ladder operators to be prepended and appended.

It is useful to decompose the field operator into a part containing the creation operators and a second one containing the annihilation operators:
\begin{equation}\label{eqn: field separation complex}
  \Phi(x) = \Phi_{a}(x) + \Phi^{\dagger}_{b}(x)~,
\end{equation}
where
\begin{align}\label{eqn: scalar field parts}
  \Phi_{a}(x)&\define \sum_{\vec{k}} \op{a}_{\vec{k}}\,\varphi_{\vec{k}}(x)~, &
  \Phi_{b}(x)&\define \sum_{\vec{k}} \op{b}_{\vec{k}}\,\varphi_{\vec{k}}(x)
\end{align}
and $\varphi_{\vec{k}}$ are the mode functions of the quantised scalar field.

\paragraph{Time-ordered product of quantized complex fields}
The above separation facilitates normal-ordering series of fields. If there are only a few field operators, this can still rapidly be done by hand. For example,
\begin{equation}\label{eqn: Wick example complex 1}
  \Phi(x)\Phi^{\dagger}(y) = \normalord \Phi(x) \Phi^{\dagger}(y) \normalord + [\Phi_{a}(x),\Phi_{a}^{\dagger}(y)]~.
\end{equation}
It is straightforward to check that the commutator is identical to a nontrivial two-point function of the quantized complex field:
\begin{equation}\begin{aligned}\label{eqn: commutators complex field parts}
  [\Phi_{a}(x),\Phi_{a}^{\dagger}(y)] &= \id\braket{0|\Phi(x)\Phi^{\dagger}(y)|0}~.
\end{aligned}\end{equation}
Thus, time-ordering of the fields introduces the Feynman propagator, such that for example:
\begin{equation}
  T\Phi(x)\Phi^{\dagger}(y) = \normalord \Phi(x)\Phi^{\dagger}(y) \normalord + \id\,G_{F}(x-y)~.
\end{equation}

These results are readily generalized if more fields are included. To keep the notation light, Feynman propagators are customarily written as \emph{contractions}
\begin{align}
  \contraction{}{\Phi}{(x)}{\Phi}%
  \Phi(x) \Phi^{\dagger}(y)
   &\define \id\,G_{F}(x-y)~, &
   \contraction{}{\Phi}{(x)}{\Phi}%
  \Phi(x) \Phi(y)
   &\define 0~,\nonumber\\
  \contraction{}{\Phi}{^{\dagger}(x)}{\Phi}%
  \Phi^{\dagger}(x) \Phi(y)
   &\define \id\,G_{F}(x-y)~, & 
  \contraction{}{\Phi}{^{\dagger}(x)}{\Phi}%
  \Phi^{\dagger}(x) \Phi^{\dagger}(y)
   & \define 0~.
\end{align}

\begin{thm}[Wick's theorem for time-ordered quantized complex fields]\label{thm: Wick complex t-o}
  A product of $n$ time-ordered field operators of the quantized complex field $\Phi$, which may contain subproducts that are normal-ordered, can be rewritten as follows:
  \begin{align}\label{eqn: Wick complex t-o}
    &T\op{\Phi}^{\epsilon_{1}}(x_{1})\op{\Phi}^{\epsilon_{2}}(x_{2}) \cdots \op{\Phi}^{\epsilon_{n}}(x_{n}) \nonumber\\
    &= \normalord \op{\Phi}^{\epsilon_{1}}(x_{1})\op{\Phi}^{\epsilon_{2}}(x_{2}) \cdots \op{\Phi}^{\epsilon_{n}}(x_{n}) \normalord\nonumber\\
    &\hphantom{={}}+ \sum_{\mathclap{\substack{\text{single}\\\text{contractions}}}} \normalord \contraction{\op{\Phi}^{\epsilon_{1}}(x_{1})\cdots}{\op{\Phi}}{^{\epsilon_{i}}(x_{i}) \cdots }{\op{\Phi}} \op{\Phi}^{\epsilon_{1}}(x_{1})\cdots \op{\Phi}^{\epsilon_{i}}(x_{i}) \cdots \op{\Phi}^{\epsilon_{j}}(x_{j}) \cdots \op{\Phi}^{\epsilon_{n}}(x_{n}) \normalord\nonumber\\
    &\hphantom{={}}+ \sum_{\mathclap{\substack{\text{double}\\\text{contractions}}}}\normalord \contraction{\op{\Phi}^{\epsilon_{1}}(x_{1})\cdots}{\op{\Phi}}{^{\epsilon_{i}}(x_{i}) \cdots }{\op{\Phi}} \op{\Phi}^{\epsilon_{1}}(x_{1})\cdots \op{\Phi}^{\epsilon_{i}}(x_{i})\cdots \op{\Phi}^{\epsilon_{j}}(x_{j}) \cdots\nonumber\\
    &\hphantom{={}+}\cdots  \contraction{}{\op{\Phi}}{^{\epsilon_{k}}(x_{k}) \cdots }{\op{\Phi}}  \op{\Phi}^{\epsilon_{k}}(x_{k})\cdots \op{\Phi}^{\epsilon_{l}}(x_{l})  \normalord+ \cdots~.
  \end{align}
  Two contracted fields belonging to different normal-ordered subproducts are to be replaced as follows:
  \begin{equation}\begin{aligned}\label{eqn: complex field contractions t-o}
    \contraction{}{\Phi}{^{\dagger}(x_{i}) \cdots }{\Phi} \Phi^{\dagger}(x_{i}) \cdots \Phi(x_{j}) &= \id\,G_{F}(x_{i}-x_{j})\\\
    \contraction{}{\Phi}{(x_{i}) \cdots }{\Phi} \Phi(x_{i}) \cdots \Phi^{\dagger}(x_{j}) &= \id\,G_{F}(x_{i}-x_{j}) \\
    \contraction{}{\Phi}{(x_{i}) \cdots }{\Phi} \Phi(x_{i}) \cdots \Phi(x_{j}) &= 0\\
    \contraction{}{\Phi}{^{\dagger}(x_{i}) \cdots }{\Phi} \Phi^{\dagger}(x_{i}) \cdots \Phi^{\dagger}(x_{j}) &= 0~.
  \end{aligned}\end{equation}
  Contraction of any pair of fields belonging to the same normal-ordered subproduct vanish.
\end{thm}

\noindent{}For three fields, for example,
\begin{multline}
  T\Phi(x)\normalord\Phi^{\dagger}(y)\Phi(z)\normalord \\ 
  =\normalord \Phi(x)\Phi^{\dagger}(y)\Phi(z) \normalord %
  + \contraction{\normalord }{\Phi}{(x)}{\Phi}%
  \normalord \Phi(x)\Phi^{\dagger}(y)\Phi(z) \normalord\\
  = \normalord \Phi(x)\Phi^{\dagger}(y)\Phi(z) \normalord + \id\,G_{F}(x-y)\Phi(z)~.
\end{multline}
This theorem is a canonical result in quantum field theory \cite{peskin_introduction_1995,bjorken_relativistic_1965}.

\begin{thm}[Wick's theorem for full expectation values of the quantized complex field]\label{thm: Wick complex t-o full vev}
  Consider a time-ordered product of quantized complex fields $T\Phi^{\epsilon_{1}}(x_{1}) \cdots  \Phi^{\epsilon_{n}}(x_{n})$ that may contain normal-ordered subproducts. The indices $\epsilon_{i} \in \{ \circ,\dagger\}$ indicate whether the field operator is conjugated ($\epsilon_{i} = \dagger$, such that $\Phi^{\epsilon_{i}} = \Phi^{\dagger}$) or not ($\epsilon_{i} = \circ$, such that $\Phi^{\epsilon_{i}} = \Phi$). Then prepend products of annihilation operators
  \begin{align}
  (\op{a}_{\vec{p}})_{\vec{p}} &\equiv \op{a}_{\vec{p}_{1}} \dots \op{a}_{\vec{p}_{M}}~,  & (\op{b}_{\bar{\vec{p}}})_{\bar{\vec{p}}} &\equiv \op{b}_{\bar{\vec{p}}_{1}} \dots \op{b}_{\bar{\vec{p}}_{\bar{M}}}
  \end{align}
  and append products of creation operators
  \begin{align}
  (\op{a}^{\dagger}_{\vec{k}})_{\vec{k}} &\equiv \op{a}^{\dagger}_{\vec{k}_{1}} \dots \op{a}^{\dagger}_{\vec{k}_{N}}~,  & (\op{b}^{\dagger}_{\bar{\vec{k}}})_{\bar{\vec{k}}} &\equiv \op{b}^{\dagger}_{\bar{\vec{k}}_{1}} \dots\op{b}^{\dagger}_{\bar{\vec{k}}_{\bar{N}}}~.
  \end{align}
  The vacuum expectation value of the resulting product of operators vanishes if the total number of operators is odd. If it is even, it can be expressed as the sum over all full contractions:
\begin{multline}\label{eqn: Wick complex t-o full vev}
  \braket{0|(\op{a}_{\vec{p}})_{\vec{p}}\,(\op{b}_{\bar{\vec{p}}})_{\bar{\vec{p}}}\,[\,T\Phi^{\epsilon_{1}}(x_{1}) \cdots  \Phi^{\epsilon_{n}}(x_{n}) \,]\,(\op{a}^{\dagger}_{\vec{k}})_{\vec{k}}\,(\op{b}^{\dagger}_{\bar{\vec{k}}})_{\bar{\vec{k}}}|0}\\
  = \sum_{\mathclap{\substack{\text{all full}\\\text{contractions}}}}%
  \contraction{}{\vphantom{\Phi}a}{_{\vec{p}_{1}} \cdots \Phi^{\epsilon_{i}}(x_{i}) \cdots}{\Phi}%
  \bcontraction{a_{\vec{p}_{1}} \cdots }{\Phi}{^{\epsilon_{i}}(x_{i}) \cdots \Phi^{\epsilon_{j}}(x_{j}) \cdots }{\Phi}%
  a_{\vec{p}_{1}} \cdots \Phi^{\epsilon_{i}}(x_{i}) \cdots \Phi^{\epsilon_{j}}(x_{j}) \cdots \Phi^{\epsilon_{k}}(x_{k})\cdots\\
   \cdots \contraction{}{\Phi}{^{\epsilon_{l}}(x_{l}) \cdots}{b} \Phi^{\epsilon_{l}}(x_{l}) \cdots b^{\dagger}_{\bar{\vec{k}}_{\bar{N}}}~.
\end{multline}
The contractions between fields that do not belong to the same normal-ordered subproducts are
  \begin{equation}\label{eqn: complex field full vev contractions 1}
    \contraction{}{\Phi}{^{\dagger}(x_{i}) \cdots }{\Phi} \Phi^{\dagger}(x_{i}) \cdots \Phi(x_{j}) = \id\,G_{F}(x_{i}-x_{j}) = \contraction{}{\Phi}{(x_{i}) \cdots }{\Phi} \Phi(x_{i}) \cdots \Phi^{\dagger}(x_{j})~,
  \end{equation}
  the contractions between ladder operators
  \begin{equation}\label{eqn: complex field full vev contractions 2}
    \contraction{}{\op{a}}{_{\vec{k}}\cdots }{a} \op{a}_{\vec{k}} \cdots a_{\vec{p}}^{\dagger} = \id\, \delta_{\vec{k},\vec{p}} =
    \contraction{}{\op{b}}{_{\vec{k}}\cdots }{b} \op{b}_{\vec{k}} \cdots b_{\vec{p}}^{\dagger}~,
  \end{equation}
  and, finally, the mixed contractions
  \begin{equation}\begin{aligned}\label{eqn: complex field full vev contractions 3}
    \contraction{}{\vphantom{\Phi}\op{a}}{_{\vec{k}}\cdots }{\Phi} \op{a}_{\vec{k}}\cdots  \Phi^{\dagger}(x) &= \id\,\varphi^{*}_{\vec{k}}(x)~, &\quad\quad{}
    \contraction{}{\Phi}{(x)\cdots }{\op{a}}  \Phi(x) \cdots \op{a}^{\dagger}_{\vec{k}}&= \id\,\varphi_{\vec{k}}(x)\\
    \contraction{}{\vphantom{\Phi}\op{b}}{_{\vec{k}}\cdots }{\Phi} \op{b}_{\vec{k}} \cdots \Phi(x) &= \id\,\varphi^{*}_{\vec{k}}(x)~, &
    \contraction{}{\Phi}{^{\dagger}(x)\cdots }{\op{b}}  \Phi^{\dagger}(x)\cdots  \op{b}^{\dagger}_{\vec{k}}&= \id\,\varphi_{\vec{k}}(x)~.
  \end{aligned}\end{equation}
  All other contractions vanish, in particular contractions between fields belonging to the same normal-ordered subproducts.
\end{thm}
\noindent{}A proof of this theorem is given in \cref{sec: proofs wick scalar}.

\subsubsection{Quantized real field}\label{sec: Wick real}

Note that the field operator can be separated as
\begin{equation}\label{eqn: field separation real}
  \Phi(x) = \Phi_{a}(x) + \Phi^{\dagger}_{a}(x)~.
\end{equation}
Unlike for the quantized complex field, the commutator
\begin{equation}\label{eqn: commutators real field part}
  [\Phi_{a}(x),\Phi(y)] = \braket{0|\Phi(x)\Phi(y)|0}
\end{equation}
does not vanish. This changes the contractions in the equivalent of \cref{thm: Wick complex t-o} for the quantized real field: contractions between $\Phi(x)$ and $\Phi(y)$ do not vanish for the quantized real field. Similarly, the contraction between $\op{a}_{\vec{k}}$ and $\Phi(x)$ does not vanish in the equivalent of \cref{thm: Wick complex t-o full vev}, since
\begin{equation}\begin{aligned}
    [a_{\vec{k}},\Phi(x)] &=  \id\,\varphi_{\vec{k}}^{*}(x)\\
    [\Phi(x),a_{\vec{k}}^{\dagger}] &=  \id\,\varphi_{\vec{k}}(x)
\end{aligned}\end{equation}
are now nonzero. Keeping this in mind, we can immediately state:

\begin{thm}[Wick's theorem for full expectation values of the quantized real field]\label{thm: Wick real t-o full vev}
  Consider a time-ordered product of quantized real fields $T\Phi(x_{1}) \cdots  \Phi(x_{n})$ that may contain normal-ordered subproducts. Then prepend products of annihilation operators
  \begin{equation}
    (\op{a}_{\vec{p}})_{\vec{p}} \equiv \op{a}_{\vec{p}_{1}} \dots \op{a}_{\vec{p}_{M}}
  \end{equation}
  and append products of creation operators
  \begin{equation}
    (\op{a}^{\dagger}_{\vec{k}})_{\vec{k}} \equiv \op{a}^{\dagger}_{\vec{k}_{1}}\dots\op{a}^{\dagger}_{\vec{k}_{N}}~.
  \end{equation}
  The vacuum expectation value of the resulting product of operators vanishes if the total number of operators is odd. If it is even, it can be expressed as the sum over all full contractions:
\begin{multline}\label{eqn: Wick real t-o full vev}
  \braket{0|(\op{a}_{\vec{p}})_{\vec{p}}\,[\,T\Phi(x_{1}) \cdots  \Phi(x_{n}) \,]\,(\op{a}^{\dagger}_{\vec{k}})_{\vec{k}}|0}\\
  = \sum_{\mathclap{\substack{\text{all full}\\\text{contractions}}}}%
  \contraction{}{\vphantom{\Phi}a}{_{\vec{p}_{1}} \cdots \Phi(x_{i}) \cdots}{\Phi}%
  \bcontraction{a_{\vec{p}_{1}} \cdots }{\Phi}{(x_{i}) \cdots \Phi(x_{j}) \cdots }{\Phi}%
  a_{\vec{p}_{1}} \cdots \Phi(x_{i}) \cdots \Phi(x_{j}) \cdots \Phi(x_{k}) \cdots \contraction{}{\Phi}{(x_{l}) \cdots}{a} \Phi(x_{l}) \cdots a^{\dagger}_{\vec{k}_{N}}~.
\end{multline}
The contractions between fields that do not belong to the same normal-ordered subproduct are
  \begin{equation}\label{eqn: real field full vev contractions 1}
    \contraction{}{\Phi}{(x_{i}) \cdots }{\Phi} \Phi(x_{i}) \cdots \Phi(x_{j}) = \id\,G_{F}(x_{i}-x_{j})~,
  \end{equation}
  the contractions between ladder operators
  \begin{equation}\label{eqn: real field full vev contractions 2}
    \contraction{}{\op{a}}{_{\vec{k}}\cdots }{a} \op{a}_{\vec{k}} \cdots a_{\vec{p}}^{\dagger} = \id\,\delta_{\vec{k},\vec{p}}~,
  \end{equation}
  and the mixed contractions
  \begin{align}\label{eqn: real field full vev contractions 3}
    \contraction{}{\vphantom{\Phi}\op{a}}{_{\vec{k}}\cdots }{\Phi} \vphantom{\Phi}\op{a}_{\vec{k}} \cdots \Phi(x) &= \id\,\varphi^{*}_{\vec{k}}(x) &
    \contraction{}{\Phi}{(x)\cdots }{\op{a}}  \Phi(x)\cdots  \op{a}^{\dagger}_{\vec{k}}&= \id\,\varphi_{\vec{k}}(x)~.
  \end{align}
  All other contractions vanish, in particular contractions between fields belonging to the same normal-ordered subproducts.
\end{thm}

\subsubsection{Quantized spinor field}\label{sec: Wick spinor}

For spinor fields, we need to work out how to evaluate vacuum expectation values of the form
\begin{multline}
  \bra{0} (\op{a}_{\vec{p},r})_{\vec{p},r}\,(\op{b}_{\bar{\vec{p}},\bar{r}})_{\bar{\vec{p}},\bar{r}}\,[\,T \op{\Psi}^{\epsilon_{1}}(x_{1})\cdots\\ \cdots \op{\Psi}^{\epsilon_{n}}(x_{n}) \,]\,(\op{a}^{\dagger}_{\vec{k},s})_{\vec{k},s}\,(\op{b}^{\dagger}_{\bar{\vec{k}},\bar{s}})_{\bar{\vec{k}},\bar{s}}\ket{0}
\end{multline}
where once more the index $\epsilon_{i}$ indicates whether the field is conjugated ($\conj\Psi$), or not ($\Psi$). Let us again split the field in two parts, according to the ladder operators, by defining
\begin{equation}\begin{aligned}\label{eqn: spinor field parts}
  \Psi_{a}(x) &\define \sum_{\vec{k},s}a_{\vec{k},s}\,\psi_{\vec{k},s,+}(x) \\
  \Psi_{b}(x) &\define \sum_{\vec{k},s}b^{\dagger}_{\vec{k},s}\,\psi_{\vec{k},s,-}(x)~,
\end{aligned}\end{equation}
such that
\begin{equation}\begin{aligned}\label{eqn: field separation spinor}
  \Psi(x) &= \Psi_{a}(x) + \Psi_{b}(x) \\
  \conj\Psi(x) &= \conj\Psi_{a}(x) + \conj\Psi_{b}(x)~.
\end{aligned}\end{equation}
Note that this convention varies slightly from \cref{eqn: scalar field parts,eqn: field separation complex} used for quantized scalar fields.

\paragraph{Time-ordered products of quantized spinor fields}
Two spinor fields are easily normal-ordered by hand, for example:
\begin{equation}\label{eqn: Wick example spinor 1}
  \Psi(x)\conj\Psi(y) = \normalord \Psi(x) \conj\Psi(y) \normalord + \{\Psi_{a}(x),\conj\Psi_{a}(y)\}~.
\end{equation}
The anticommutators can be evaluated using \cref{eqn: inverse spinor-bar-prod 3dim massive} for massive fields in $(1,3)$ dimensions [and the analogous equations for massless fields, and in $(1,1)$ dimensions respectively]. It is straightforward to see that
\begin{equation}
  \{\Psi_{a}^{A}(x),\conj\Psi_{a,B}(y)\} =  \id \braket{0|\Psi^{A}(x)\conj\Psi_{B}(y)|0}~.
\end{equation}
The Feynman propagator appears again if the fields are time-ordered. For example,
\begin{equation}
  T\Psi(x)\conj\Psi(y) = \normalord \Psi(x) \conj\Psi(y)  \normalord + \id\,S_{F}(x-y)~,
\end{equation}
where we have used that $\normalord \conj\Psi(y) \Psi(x)  \normalord = -\normalord \Psi(x) \conj\Psi(y) \normalord$.

The generalization of the above is given by the fermionic version of Wick's theorem for spinor fields:

\begin{thm}[Wick's theorem for time-ordered quantized spinor fields]\label{thm: Wick spinor t-o}
  Consider a time-ordered product of spinor fields $T\Psi^{\epsilon_{1}}(x_{1}) \cdots  \Psi^{\epsilon_{n}}(x_{n})$ which may contain normal-ordered subproducts. The indices $\epsilon_{i} \in \{ \circ,\bar{~}\}$ indicate whether the field operator is Dirac conjugated ($\epsilon_{i} = \bar{~}$, such that $\Psi^{\epsilon_{i}} = \conj\Psi$) or not ($\epsilon_{i} = \circ$, such that $\Psi^{\epsilon_{i}} = \Psi$). It can be normal-ordered as follows:
  \begin{align}\label{eqn: Wick spinor t-o}
    &T\op{\Psi}^{\epsilon_{1}}(x_{1})\op{\Psi}^{\epsilon_{2}}(x_{2}) \cdots \op{\Psi}^{\epsilon_{n}}(x_{n}) \nonumber\\
    &= \normalord  \op{\Psi}^{\epsilon_{1}}(x_{1})\op{\Psi}^{\epsilon_{2}}(x_{2}) \cdots \op{\Psi}^{\epsilon_{n}}(x_{n}) \normalord\nonumber\\
    &\hphantom{={}}+ \sum_{\mathclap{\substack{\text{single}\\\text{contractions}}}} \normalord \contraction{\op{\Psi}^{\epsilon_{1}}(x_{1})\cdots}{\op{\Psi}}{^{\epsilon_{i}}(x_{i}) \cdots }{\op{\Psi}} \op{\Psi}^{\epsilon_{1}}(x_{1})\cdots \op{\Psi}^{\epsilon_{i}}(x_{i}) \cdots \op{\Psi}^{\epsilon_{j}}(x_{j}) \cdots \op{\Psi}^{\epsilon_{n}}(x_{n}) \normalord\nonumber\\
    &\hphantom{={}}+ \sum_{\mathclap{\substack{\text{double}\\\text{contractions}}}}\normalord \contraction{\op{\Psi}^{\epsilon_{1}}(x_{1})\cdots}{\op{\Psi}}{^{\epsilon_{i}}(x_{i}) \cdots }{\op{\Psi}} \op{\Psi}^{\epsilon_{1}}(x_{1})\cdots \op{\Psi}^{\epsilon_{i}}(x_{i})\cdots \op{\Psi}^{\epsilon_{j}}(x_{j}) \cdots \nonumber\\
    & \hphantom{={}+\sum_{\mathclap{\substack{\text{double}\\\text{contractions}}}}\normalord} \cdots \contraction{}{\op{\Psi}}{^{\epsilon_{k}}(x_{k}) \cdots }{\op{\Psi}}  \op{\Psi}^{\epsilon_{k}}(x_{k})\cdots \op{\Psi}^{\epsilon_{l}}(x_{l})  \normalord + \cdots~.
  \end{align}
  Contractions of any pair of fields belonging to the same normal-ordered subproduct vanish. Two contracted fields belonging to different normal-ordered subproducts are to be replaced as follows:
  \begin{equation}\begin{aligned}\label{eqn: spinor field contractions t-o}
    \contraction{}{\conj\Psi}{_{A_{i}}(x_{i})  }{\Psi} \conj\Psi_{A_{i}}(x_{i})  \Psi^{A_{j}}(x_{j}) &= -\id\,S_{F}(x_{j}-x_{i})^{A_{j}}_{A_{i}}\\
    \contraction{}{\vphantom{\conj\Psi}\Psi}{^{A_{i}}(x_{i})  }{\conj\Psi} \Psi^{A_{i}}(x_{i})  \conj\Psi_{A_{j}}(x_{j}) &= +\id\,S_{F}(x_{i}-x_{j})^{A_{i}}_{A_{j}}\\
    \contraction{}{\Psi}{(x_{i})  }{\Psi} \Psi(x_{i})  \Psi(x_{j}) &= 0\\
    \contraction{}{\conj\Psi}{(x_{i})  }{\conj\Psi} \conj\Psi(x_{i})  \conj\Psi(x_{j}) &= 0~.
  \end{aligned}\end{equation}
  If contracted operators are not adjacent, the term has to be reordered until they are, with every exchange of two operators introducing an overall factor of $(-1)$ to the term.
\end{thm}
\noindent{}The last requirement that contracted operators need to be adjacent means for example that
\begin{multline}
  \contraction{\normalord}{\conj\Psi}{(x_{1})\Psi(x_{2})\conj\Psi(x_{3})}{\Psi}%
  \bcontraction{\normalord \conj\Psi(x_{1})}{\Psi}{(x_{2})}{\conj\Psi}%
  \normalord \conj\Psi(x_{1})\Psi(x_{2})\conj\Psi(x_{3})\Psi(x_{4})\normalord \\
  = (-1)^{3}\contraction{\normalord}{\vphantom{\conj\Psi}\Psi}{(x_{4})}{\conj\Psi}%
  \bcontraction{\normalord \Psi(x_{4})\conj\Psi(x_{1})}{\Psi}{(x_{2})}{\conj\Psi}%
  \normalord \Psi(x_{4})\conj\Psi(x_{1})\Psi(x_{2})\conj\Psi(x_{3})\normalord   
\end{multline}
Since this is again a canonical result \cite{peskin_introduction_1995,bjorken_relativistic_1965}, we abstain from giving the proof here.

\begin{thm}[Wick's theorem for full expectation values of the quantized spinor field]\label{thm: Wick spinor t-o full vev}
  Consider a time-ordered product of spinor fields $T\Psi^{\epsilon_{1}}(x_{1}) \cdots  \Psi^{\epsilon_{n}}(x_{n})$ like in \cref{thm: Wick spinor t-o}. Then prepend products of annihilation operators   \begin{equation}\begin{aligned}
  (\op{a}_{\vec{p},r})_{\vec{p},r} &\equiv \op{a}_{(\vec{p},r)_{1}} \dots \op{a}_{(\vec{p},r)_{M}} \\
  (\op{b}_{\bar{\vec{p}},\bar{r}})_{\bar{\vec{p}},\bar{r}} &\equiv \op{b}_{(\bar{\vec{p}},\bar{r})_{1}} \dots \op{b}_{(\bar{\vec{p}},\bar{r})_{\bar{M}}}~,
\end{aligned}\end{equation}
and append products of creation operators
  \begin{equation}\begin{aligned}
  (\op{a}^{\dagger}_{\vec{k},s})_{\vec{k},s} &\equiv \op{a}^{\dagger}_{(\vec{k},s)_{1}} \dots \op{a}^{\dagger}_{(\vec{k},s)_{N}} \\
  (\op{b}^{\dagger}_{\bar{\vec{k}},\bar{s}})_{\bar{\vec{k}},\bar{s}} &\equiv \op{b}^{\dagger}_{(\bar{\vec{k}},\bar{s})_{1}} \dots \op{b}^{\dagger}_{(\bar{\vec{k}},\bar{s})_{\bar{N}}} ~.
\end{aligned}\end{equation}
 The vacuum expectation value of the resulting product of operators vanishes if the total number of operators is odd. If it is even, it can be expressed as the sum over all full contractions:
\begin{multline}\label{eqn: Wick spinor t-o full vev}
  \bra{0}(\op{a}_{\vec{p},r})_{\vec{p},r}\,(\op{b}_{\bar{\vec{p}},\bar{r}})_{\bar{\vec{p}},\bar{r}}\,[\,T\Psi^{\epsilon_{1}}(x_{1}) \cdots \\ \cdots \Psi^{\epsilon_{n}}(x_{n}) \,]\,(\op{a}^{\dagger}_{\vec{k},s})_{\vec{k},s}\,(\op{b}^{\dagger}_{\bar{\vec{k}},\bar{s}})_{\bar{\vec{k}},\bar{s}}\ket{0}\\
  = \sum_{\mathclap{\substack{\text{all full}\\\text{contract.}}}}%
  \contraction{}{\vphantom{\Psi}a}{_{(\vec{p},r)_{1}} \cdots \Psi^{\epsilon_{i}}(x_{i}) \cdots}{\Psi}%
  \bcontraction{a_{(\vec{p},r)_{1}} \cdots }{\Psi}{^{\epsilon_{i}}(x_{i}) \cdots \Psi^{\epsilon_{j}}(x_{j}) \cdots }{\Psi}%
  a_{(\vec{p},r)_{1}} \cdots \Psi^{\epsilon_{i}}(x_{i}) \cdots \Psi^{\epsilon_{j}}(x_{j}) \cdots \Psi^{\epsilon_{k}}(x_{k}) \cdots \\ \cdots \contraction{}{\Psi}{^{\epsilon_{l}}(x_{l}) \cdots}{b} \Psi^{\epsilon_{l}}(x_{l}) \cdots b^{\dagger}_{(\bar{\vec{k}},\bar{s})_{\bar{N}}}~.
\end{multline}
The contractions between fields that do not belong to the same normal-ordered subproduct are
  \begin{equation}\begin{aligned}\label{eqn: spinor field full vev contractions 1}
    \contraction{}{\conj\Psi}{_{A_{i}}(x_{i})  }{\Psi} \conj\Psi_{A_{i}}(x_{i})  \Psi^{A_{j}}(x_{j}) &= -\id\,S_{F}(x_{j}-x_{i})^{A_{j}}_{A_{i}}\\
    \contraction{}{\vphantom{\conj\Psi}\Psi}{^{A_{i}}(x_{i})  }{\conj\Psi} \Psi^{A_{i}}(x_{i})  \conj\Psi_{A_{j}}(x_{j}) &= +\id\,S_{F}(x_{i}-x_{j})^{A_{i}}_{A_{j}}~,
  \end{aligned}\end{equation}
  the contractions between ladder operators
  \begin{equation}\label{eqn: spinor field full vev contractions 2}
    \contraction{}{\op{a}}{_{\vec{k},s}}{a} \op{a}_{\vec{k},s}  a_{\vec{p},r}^{\dagger} = \id\,\delta_{\vec{k},\vec{p}}\,\delta_{s,r} =
    \contraction{}{\op{b}}{_{\vec{k},s}}{b} \op{b}_{\vec{k},s}  b_{\vec{p},r}^{\dagger}~,
  \end{equation}
  and the mixed contractions
  \begin{equation}\begin{aligned}\label{eqn: spinor field full vev contractions 3}
    \contraction{}{\vphantom{\conj\Psi}\op{a}}{_{\vec{k},s}}{\Psi} \op{a}_{\vec{k},s} \conj\Psi_{A}(x) &= \id\,\conj{\psi}_{\vec{k},s,+,A}(x)\\
    \contraction{}{\Psi}{^{A}(x)}{\op{a}}  \Psi^{A}(x) \op{a}^{\dagger}_{\vec{k},s}&= \id\,\psi_{\vec{k},s,+}^{A}(x)\\
    \contraction{}{\vphantom{\Psi}\op{b}}{_{\vec{k},s}}{\Psi} \op{b}_{\vec{k},s} \Psi^{A}(x) &= \id\,\psi_{\vec{k},s,-}^{A}(x) \\
    \contraction{}{\conj\Psi}{_{A}(x)}{\op{b}}  \conj\Psi_{A}(x) \op{b}^{\dagger}_{\vec{k},s}&= \id\,\conj{\psi}_{\vec{k},s.-,A}(x)~.
  \end{aligned}\end{equation}
  If contracted operators are not adjacent, the term has to be reordered until they are, with every exchange of two operators introducing an overall factor of $(-1)$ to the term. All other contractions vanish, in particular contractions between fields belonging to the same normal-ordered subproduct.
\end{thm}
\noindent{}A proof is given in \cref{sec: proofs wick spinor}.

\subsubsection{UDW-type detector}\label{sec: Wick detector}

Last but not least, we would like to compute vacuum expectation values of the from given in \cref{eqn: general correlator detector},
\begin{equation}\label{eqn: general correlator detector 2}
  \braket{g|[\sigma^{-}]^{a}[\,T\op{\mu}(t_{1})\cdots\op{\mu}(t_{n})\,]\,[\sigma^{+}]^{b}|g}
\end{equation}
As before, we decompose the monopole operator into a creation and an annihilation part,
\begin{align}
  \mu_{-}(t) &\define e^{-\ii  \Omega t}\,\sigma^{-}~, & \mu_{+}(t) &\define e^{+\ii  \Omega t} \,\sigma^{+}~,
\end{align}
such that
\begin{equation}
  \mu(t) = \mu_{-}(t) + \mu_{+}(t)~.
\end{equation}

\paragraph{Time-ordered products of monopole operators}
For two monopole operators, normal-ordering is swiftly accomplished and yields
\begin{equation}
  \mu(t_{1})\mu(t_{2}) = \normalord \mu(t_{1}) \mu(t_{2}) \normalord + \,\{\mu_{-}(t_{1}),\mu_{+}(t_{2})\}~,
\end{equation}
where
\begin{equation}
  \{\mu_{-}(t_{1}),\mu_{+}(t_{2})\} = \id\,e^{-\ii  \Omega(t_{1}-t_{2})}~.
\end{equation}
Now consider what happens if the product of monopole operators is time-ordered. For example,
\begin{equation}
  T\mu(t_{1})\mu(t_{2}) =  \normalord \mu(t_{1}) \mu(t_{2}) \normalord + \id\,D_{F}(t_{1}-t_{2})~,
\end{equation}
using the definition of the Feynman propagator \cref{eqn: Feynman P detector 1}.
Since the monopole operator $\mu(x)$ (as a field) has basically the same structure as the quantized spinor field, we can directly conclude that the generalization to $n$ operators is given by

\begin{thm}[Wick's theorem for time-ordered monopole operators]\label{thm: Wick detector t-o}
  A product of $n$ time-ordered monopole operators $\mu$ can be rewritten as a sum of normal-ordered products as follows:
  \begin{align}\label{eqn: Wick detector t-o}
    &T\mu(t_{1}) \cdots \mu(t_{n}) \nonumber\\
    &= \normalord \mu(t_{1}) \cdots \mu(t_{n})  \normalord  \nonumber\\
    &\hphantom{={}}+ \sum_{\mathclap{\substack{\text{single}\\\text{contractions}}}}
    \contraction{\normalord \mu(t_{1}) \cdots}{\mu}{(t_{i}) \cdots }{\mu}%
    \normalord \mu(t_{1}) \cdots \mu(t_{i}) \cdots \mu(t_{j}) \cdots \mu(t_{n}) \normalord \nonumber\\
    &\hphantom{={}}+ \sum_{\mathclap{\substack{\text{double}\\\text{contractions}}}}
    \contraction{\normalord \mu(t_{1}) \cdots }{\mu}{(t_{i}) \cdots \mu(t_{k})\cdots }{\mu}
    \bcontraction{\normalord \mu(t_{1}) \cdots \mu(t_{i}) \cdots }{\mu}{(t_{k})\cdots \mu(t_{j})\cdots }{\mu}
    \normalord \mu(t_{1}) \cdots \mu(t_{i}) \cdots \mu(t_{k})\cdots \mu(t_{j})\cdots \mu(t_{l}) \cdots \nonumber\\
    &\hphantom{={}+}\cdots \mu(t_{n})\normalord + \cdots~.
  \end{align}
  The contraction of adjacent operators is defined as
  \begin{equation}
    \contraction{}{\mu}{(t_{i})}{\mu} \mu(t_{i})\mu(t_{j}) \define \id\,D_{F}(t_{i}-t_{j})~.
  \end{equation}
  If contracted operators are not adjacent, the term has to be reordered until they are, with every exchange of two operators introducing an overall factor of $(-1)$ to the term.
\end{thm}

Similar to the spinor field, this last remark implies, for example,
\begin{multline}
  \contraction{\normalord }{\mu}{(t_{1})\mu(t_{2})}{\mu}%
  \bcontraction{\normalord \mu(t_{1})}{\mu}{(t_{2})\mu(t_{3})}{\mu}%
  \normalord \mu(t_{1})\mu(t_{2})\mu(t_{3})\mu(t_{4})\normalord = 
  -\contraction{\normalord }{\mu}{(t_{1})}{\mu}%
  \bcontraction{\normalord \mu(t_{1})\mu(t_{3})}{\mu}{(t_{2})}{\mu}%
  \normalord \mu(t_{1})\mu(t_{3})\mu(t_{2})\mu(t_{4})\normalord \\
  = -\id\,e^{-\ii  \Omega(t_{1}+t_{2}-t_{3}-t_{4})}~.
\end{multline}

\subsubsection{Full expectation value}

When prepending zero or one annihilation operator $\sigma^{-}$ and appending zero or one creation operator $\sigma^{+}$ and taking the vacuum expectation value, more contractions appear:

\begin{thm}[Wick's theorem for the full expectation value of the detector]\label{thm: Wick detector t-o full vev}
  Consider a product of $n$ time-ordered monopole operators $\mu$ with $a\in\{0,1\}$ annihilation operators and $b\in\{0,1\}$ creation operators. Its vacuum expectation value is
  \begin{multline}\label{eqn: Wick detector t-o full vev}
    \braket{g|[\sigma^{-}]^{a}[\,T\op{\mu}(t_{1})\cdots\op{\mu}(t_{n})\,]\,[\sigma^{+}]^{b}|g} \\
    = \sum_{\mathclap{\substack{\text{all full}\\\text{contractions}}}}%
    \contraction{}{\sigma}{^{-} \cdots }{\mu}%
    \contraction{\sigma^{-} \cdots \mu(t_{i}) \cdots }{\mu}{(t_{k}) \cdots \mu(t_{j}) \cdots }{\mu}%
    \bcontraction{\sigma^{-} \cdots \mu(t_{i}) \cdots \mu(t_{k}) \cdots }{\mu}{(t_{j}) \cdots \mu(t_{l})  \cdots }{\sigma}%
    \sigma^{-} \cdots \mu(t_{i}) \cdots \mu(t_{k}) \cdots \mu(t_{j}) \cdots \mu(t_{l})  \cdots \sigma^{+}~.
  \end{multline}
  The contraction between two adjacent monopole fields is defined as
  \begin{equation}\label{eqn: detector full vev contractions 1}
    \contraction{}{\mu}{(t_{i})}{\mu} \mu(t_{i})\mu(t_{j}) \define \id\,D_{F}(t_{i}-t_{j})~,
  \end{equation}
  the contraction between ladder operator and monopole field
  \begin{align}\label{eqn: detector full vev contractions 2}
    \contraction{}{\sigma}{^{-}}{\mu}%
    \sigma^{-} \mu(t) &= \id\,e^{+\ii  \Omega t}~, &
    \contraction{}{\mu}{(t) }{\sigma}%
    \mu(t) \sigma^{+}  &= \id\,e^{-\ii  \Omega t},
  \end{align}
  and finally the contraction between ladder operators
  \begin{equation}\label{eqn: detector full vev contractions 3}
    \contraction{}{\sigma}{^{-} }{\sigma} \sigma^{-} \sigma^{+} = \id\,~.
  \end{equation}
  If contracted operators are not adjacent, the term has to be reordered until they are, with every exchange of two operators introducing an overall factor of $(-1)$ to the term.
\end{thm}
\noindent{}The proof is presented in \cref{sec: proofs wick detector}.

\subsection{Diagrammatic computation}\label{sec: diagrammatic}

A usual way to visualize the contractions between field operators in the above theorems, as well as the next step towards Feynman diagrams, is the use of line diagrams. Operators in a given product are represented by points (vertices), and contractions are visualized as lines between the vertices. The vacuum expectation value of a given product of operators can be obtained in a structured manner by drawing such a diagram for every nontrivial full contraction (meaning that there are no wholly unconnected vertices), assigning the corresponding contraction as a factor to each line, and then summing up the terms obtained from each diagram. 

To distinguish them from the fully fledged Feynman diagrams that we will develop in \cref{sec: feynman rules}, we will in the following sometimes refer to these diagrams as \emph{auxiliary diagrams}.

\subsubsection{UDW-type detector}\label{sec: auxiliary diagrammatic detector}

The following set of rules can be applied as a help in evaluating expectation values of the form \cref{eqn: general correlator detector 2} using \cref{thm: Wick detector t-o full vev}:

\begin{WD}
\setdefaultenum{1.}{}{}{} 
\hfill\begin{compactenum}
  \item If a creation operator is present, draw a vertex on the left end of the diagram. If an annihilation operator is present, draw a vertex on the right end of the diagram.  These vertices will be referred to as \emph{external vertices} (\emph{ingoing} on the left  and \emph{outgoing} on the right).
  \item In the middle, draw a horizontal series of vertices, one for each $t_{i}$, starting with $t_{n}$ on the left, and label them accordingly. These vertices will be referred to as \emph{internal vertices}.
  \item To account for the nontrivial contractions \cref{eqn: detector full vev contractions 1,eqn: detector full vev contractions 2,eqn: detector full vev contractions 3}, connect all vertices by lines such that each external and internal vertex has one line attached to it.
  \item Draw one diagram for every way to connect all vertices using the rules in 3, such that no vertex is left completely unconnected. If this is impossible, the expectation value is zero.
  \item Associate the following factors to each line in a diagram:
  \begin{compactitem}
    \item $1$ to a line connecting two external vertices
    \item $e^{-\ii  \Omega t_{i}}$ to a line connecting an ingoing vertex with an internal vertex $t_{i}$
    \item $e^{+\ii  \Omega t_{i}}$ to a line connecting an outgoing vertex with an internal vertex $t_{i}$
    \item $D_{F}(t_{i}-t_{j})$ to a line connecting two internal vertices $t_{i}$, $t_{j}$; since $D_{F}(t_{i}-t_{j}) = -D_{F}(t_{j}-t_{i})$, and the overall sign is determined in other ways, the order is irrelevant.
  \end{compactitem}
  \item To obtain the correct overall sign of a diagram, fall back to the interaction picture and reorder the operators such that contracted pairs are adjacent, as described in \cref{thm: Wick detector t-o full vev}.
  \item The vacuum expectation value is obtained by summing up the values associated with each diagram.
\end{compactenum}
\end{WD}

\paragraph{Example}
Assume the detector is excited both at the beginning and at the end. Including second order corrections to the amplitude, one encounters the expectation value
\begin{equation}\label{eqn: detector example 2}
  E = \braket{0|\sigma^{-}\,[\,T \mu(t_{1}) \mu(t_{2})\,]\,\sigma^{+}|0}~.
\end{equation}
There are three full contractions which can be written as follows:
\begin{multline}
  E = \contraction{\bra{0}}{\sigma}{^{-}\,[\,T }{\mu}%
  \contraction{\bra{0}\sigma^{-}\,[\,T \mu(t_{1}) }{\mu}{(t_{2})\,]\,}{\sigma}%
  \bra{0}\sigma^{-}\,[\,T \mu(t_{1}) \mu(t_{2})\,]\,\sigma^{+}\ket{0}\\ 
  + \contraction{\bra{0}}{\sigma}{^{-}\,[\,T \mu(t_{1}) }{\mu}%
  \bcontraction{\bra{0}\sigma^{-}\,[\,T }{\mu}{(t_{1}) \mu(t_{2})\,]\,}{\sigma}%
  \bra{0}\sigma^{-}\,[\,T \mu(t_{1}) \mu(t_{2})\,]\,\sigma^{+}\ket{0}\\
  + \contraction{\bra{0}}{\sigma}{^{-}\,[\,T \mu(t_{1}) \mu(t_{2})\,]\,}{\sigma}%
  \bcontraction{\bra{0}\sigma^{-}\,[\,T }{\mu}{(t_{1}) }{\mu}
  \bra{0}\sigma^{-}\,[\,T \mu(t_{1}) \mu(t_{2})\,]\,\sigma^{+}\ket{0}\\
  \contraction{}{\sigma}{^{-}}{\mu}%
  \contraction{\sigma^{-}\mu(t_{1}) }{\mu}{(t_{2})}{\sigma}%
  \sigma^{-}\mu(t_{1}) \mu(t_{2})\sigma^{+}
  - \contraction{}{\sigma}{^{-}}{\mu}%
  \contraction{\sigma^{-} \mu(t_{2}) }{\mu}{(t_{2})}{\sigma}%
  \sigma^{-} \mu(t_{2}) \mu(t_{1}) \sigma^{+}
  \\+ \contraction{}{\sigma}{^{-}}{\sigma^{+}}%
  \contraction{\sigma^{-}\sigma^{+} }{\mu}{(t_{1})}{\mu}%
  \sigma^{-}\sigma^{+} \mu(t_{1}) \mu(t_{2})\\
  = e^{+\ii  \Omega (t_{1}-t_{2})}-e^{-\ii  \Omega(t_{1}-t_{2})}+D_{F}(t_{1}-t_{2})~,
\end{multline}
according to \cref{thm: Wick detector t-o full vev}.
Up to the signs, the same is obtained from the three corresponding diagrams shown in \cref{fig: Wick detector example 2}.

\begin{figure}
  \includegraphics{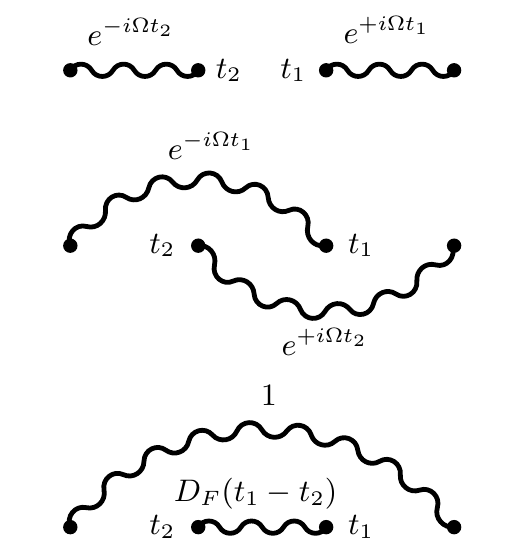}
  \caption{The three diagrams belonging to the example for the monopole operator.}
  \label{fig: Wick detector example 2}
\end{figure}

\subsubsection{Linear coupling to quantized real fields}\label{sec: Wick real linear}

To simplify the use of \cref{thm: Wick real t-o full vev} to evaluate the vacuum expectation values of the form 
\begin{equation}
  \braket{0|(\op{a}_{\vec{p}})_{\vec{p}}\,[\,T\op{\Phi}(y_{1}) \cdots \op{\Phi}(y_{n})\,]\,(\op{a}^{\dagger}_{\vec{k}})_{\vec{k}}|0}~,
\end{equation}
the following equivalent diagrammatic method can be used:

\begin{RL}
  \hfill
  \setdefaultenum{1.}{}{}{} 
  \begin{compactenum}
    \item On the left end of the diagram, draw vertices for every creation operator, and label them by their momenta. Similarly, draw vertices for every annihilation operator on the right-hand side and label them by their momenta. As before, these vertices will be referred to as \emph{external vertices} (\emph{ingoing} on the left  and \emph{outgoing} on the right).
    \item In the middle, draw a horizontal series of vertices, one for each $y_{i}$, starting with $y_{n}$ on the left, and label them accordingly. Once again, these are the \emph{internal vertices}.
    \item To account for the nontrivial contractions \cref{eqn: real field full vev contractions 1,eqn: real field full vev contractions 2,eqn: real field full vev contractions 3}, connect all vertices by lines in the following way:
    \begin{compactitem}
      \item both external and internal vertices have one line attached
      \item direct connections between ingoing vertices are forbidden; direct connections between outgoing vertices are forbidden
    \end{compactitem}
    \item Draw one diagram for every way to connect all vertices using the rules in 3. If this is not possible, the expectation value is zero.
    \item Associate the following factors to each line in a diagram:
    \begin{compactitem}
      \item $\delta_{\vec{k},\vec{p}}$ to a line connecting two external vertices $\vec{k}$ and $\vec{p}$
      \item $\varphi_{\vec{k}}(y_{i})$ to a line connecting an ingoing vertex $\vec{k}$ with an internal vertex $y_{i}$
      \item $\varphi_{\vec{k}}^{*}(y_{i})$ to a line connecting an outgoing vertex $\vec{k}$ with an internal vertex $y_{i}$
      \item $G_{F}(y_{i}-y_{j})$ to a line connecting two internal vertices $y_{i}$, $y_{j}$; since $G_{F}(x_{i}-x_{j}) = G_{F}(x_{j}-x_{i})$, the order is irrelevant.
    \end{compactitem}
    \item The vacuum expectation value is obtained by summing up the values associated with each diagram.
  \end{compactenum}
  \setdefaultenum{(a)}{}{}{} 
\end{RL}

\paragraph{Example}
Add two annihilation operators in front of the two time-ordered fields. The calculations can still be performed by hand in a reasonable time, yielding:
\begin{multline}\label{eqn: real linear example 2}
  \braket{0|a_{\vec{k}}\op{a}_{\vec{p}}\,[\,T \op{\Phi}(y_{1})  \op{\Phi}(y_{2}) ]|0} \\= \varphi^{*}_{\vec{k}}(y_{2})\,\varphi^{*}_{\vec{p}}(y_{1}) + \varphi^{*}_{\vec{k}}(y_{1})\,\varphi^{*}_{\vec{p}}(y_{2})~.
\end{multline}
Indeed, applying the above diagrammatic rules results in the two diagrams depicted in \cref{fig: Wick real linear example 2} and gives the same result.

\begin{figure}
  \includegraphics{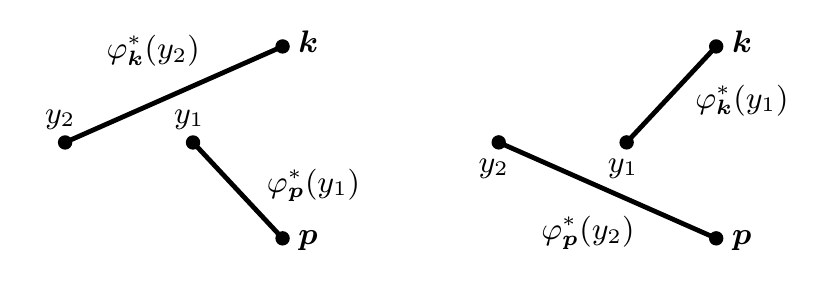}
  \caption{The two diagrams belonging to the example for linear coupling to quantized real fields.}
  \label{fig: Wick real linear example 2}
\end{figure}

\subsubsection{Quadratic coupling to quantized real fields}\label{sec: auxiliary diagrammatic real quad}

If the detector model 2 coupling quadratically to a quantized real field is used instead, expectation values of the following type are encountered:
\begin{multline}
  \bra{0}(\op{a}_{\vec{p}})_{\vec{p}}\,[\,T\normalord \op{\Phi}(y_{1}) \op{\Phi}(y_{1}) \normalord \cdots \\ \cdots \normalord \op{\Phi}(y_{n}) \op{\Phi}(y_{n}) \normalord\,]\,(\op{a}^{\dagger}_{\vec{k}})_{\vec{k}}\ket{0}~.
\end{multline}
Since there are now normal-ordered subproducts of fields with the same argument, the above rules have to be adapted.

\begin{RQ}
  \hfill\\
  The rules for quantized real fields and linear coupling still apply, with one exception in point 3: internal vertices have two lines attached; these are \emph{not} interchangeable.
\end{RQ}

\paragraph{Example}
We evaluate an expectation value similar to the example for linear coupling. There are four ways to contract two fields $\Phi(y_{1})$, $\Phi(y_{2})$, and then two ways to contract the remaining fields with the ladder operators. In consequence there are eight terms in total, yielding:
\begin{multline}\label{eqn: Wick real quad example 1}
  \braket{0|a_{\vec{k}}\op{a}_{\vec{p}}\,[\,T\normalord \op{\Phi}(y_{1}) \op{\Phi}(y_{1}) \normalord \normalord \op{\Phi}(y_{2}) \op{\Phi}(y_{2}) \normalord]|0} \\
  = 4\,G_{F}(y_{1}-y_{2})\left[\varphi^{*}_{\vec{k}}(y_{2})\,\varphi^{*}_{\vec{p}}(y_{1}) + \varphi^{*}_{\vec{k}}(y_{1})\,\varphi^{*}_{\vec{p}}(y_{2}) \right]~.
\end{multline}
The corresponding diagrams are drawn in \cref{fig: Wick real quad example 1}. The same result can be obtained more economically by dropping the distinction between the two lines attached to each internal vertex, leading to merely two distinct diagrams in \cref{fig: Wick real quad example 1 version 2}. To compensate, we assign an additional symmetry factor of four to the line between the two vertices, which counts the different ways the fields $\Phi(y_{1})$ and $\Phi(y_{2})$ can be contracted while now resulting in the same diagram. In accordance with the convention in particle physics, we will in the following use symmetry factors rather than distinguishing different legs of the same vertex.

\begin{figure}
  \includegraphics{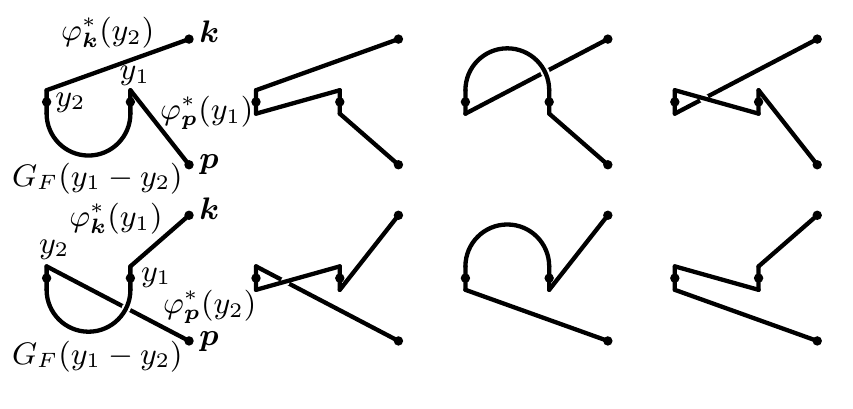}
  \caption{The eight diagrams belonging to the example for quadratic coupling to quantized real fields. In each of the two lines, all four ways to connect the two internal vertices are realized. In the upper line, one field $\Phi(y_{2})$ is always contracted with $a_{\vec{k}}$, and one field $\Phi(y_{1})$ with $a_{\vec{p}}$. In the lower line, the inverse case is accounted for.}
  \label{fig: Wick real quad example 1}
\end{figure}

\begin{figure}
  \includegraphics{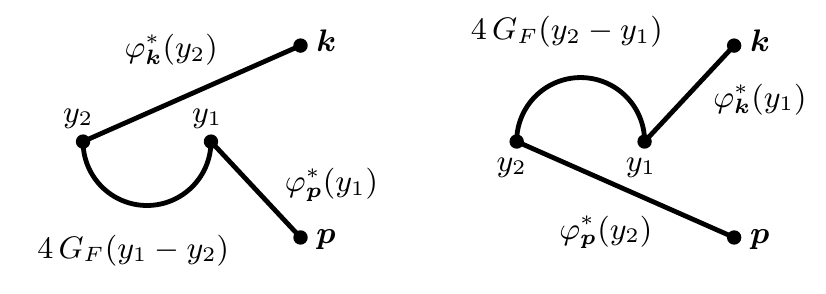}
  \caption{Diagrams belonging to the example for quadratic coupling to quantized real fields, reduced to two distinct diagrams by dropping the distinction between upwards and downwards lines of internal vertices. To compensate, an additional factor four is assigned to their contraction.}
  \label{fig: Wick real quad example 1 version 2}
\end{figure}

\subsubsection{Quadratic coupling to quantized complex fields}\label{sec: auxiliary diagrammatic complex quad}

\Cref{thm: Wick complex t-o full vev} is useful to evaluate the vacuum expectation values of the form
\begin{multline}
  \bra{0}(\op{a}_{\vec{p}})_{\vec{p}}\,(\op{b}_{\bar{\vec{p}}})_{\bar{\vec{p}}}\,[\,T\normalord \op{\Phi}^{\dagger}(y_{1}) \op{\Phi}(y_{1}) \normalord \cdots \\ \cdots \normalord \op{\Phi}^{\dagger}(y_{n}) \op{\Phi}(y_{n}) \normalord\,]\,(\op{a}^{\dagger}_{\vec{k}})_{\vec{k}}\,(\op{b}^{\dagger}_{\bar{\vec{k}}})_{\bar{\vec{k}}}\ket{0}~,
\end{multline}
which arise when calculating the response of the quadratically coupling detector model 3. Here, the existence of two independent types of ladder operators (and accordingly the fact that the field is no longer self-adjoint) has to be taken into account. This is most conveniently achieved by giving the connection lines between vertices a direction, that is, by turning them into arrows. The following rules can be used to apply the theorem:

\begin{CQ}
  \hfill
  \setdefaultenum{1.}{}{}{} 
  \begin{compactenum}
    \item On the left end of the diagram, draw vertices for every creation operator $a^{\dagger}_{\vec{k}_{1}}$ and $b^{\dagger}_{\vec{k}_{2}}$, and label them by their momenta with $\vec{k}_{1}$, $\overline{\vec{k}_{2}}$, where the bar is used to mark antiparticle operators. Similarly, draw vertices for every annihilation operator $a_{\vec{p}_{1}}$ and $b_{\vec{p}_{2}}$ on the right-hand side and label them by their momenta $\vec{p}_{1}$ and $\overline{\vec{p}_{2}}$.
    \item In the middle, draw a horizontal series of vertices for each $y_{i}$, starting with $y_{n}$ on the left, and label them accordingly.
    \item To account for the nontrivial contractions \cref{eqn: complex field full vev contractions 1,eqn: complex field full vev contractions 2,eqn: complex field full vev contractions 3}, connect all vertices by arrows in the following way:
    \begin{compactitem}
      \item vertices representing $\op{a}^{\dagger}$ and $\op{b}$ are origins of one arrow each
      \item vertices representing $\op{a}$ and $\op{b}^{\dagger}$ are endpoints of one arrow each
      \item vertices representing $y_{i}$ are origin of one arrow, and endpoint of one arrow each
      \item direct connections between ingoing vertices are forbidden; direct connections between outgoing vertices are forbidden
      \item connecting an internal vertex with itself is forbidden
      \item arrow directions must not collide
    \end{compactitem}
    \item Draw one diagram for every way to connect all vertices using the rules in 3. If this is not possible, the expectation value is zero.
    \item Associate the following factors to each arrow in a diagram:
    \begin{compactitem}
      \item $\delta_{\vec{k},\vec{p}}$ to an arrow connecting two external vertices $\vec{k}$ and $\vec{p}$
      \item $\varphi_{\vec{k}}(y_{i})$ to an arrow connecting an ingoing vertex $\vec{k}$ with an internal vertex $y_{i}$
      \item $\varphi_{\vec{k}}^{*}(y_{i})$ to an arrow connecting an outgoing vertex $\vec{k}$ with an internal vertex $y_{i}$
      \item $G_{F}(y_{i}-y_{j})$ to an arrow connecting two internal vertices $y_{i}$, $y_{j}$
    \end{compactitem}
    \item The vacuum expectation value is obtained by summing up the values associated with each diagram.
  \end{compactenum}
  \setdefaultenum{(a)}{}{}{} 
\end{CQ}

\paragraph{Example 1}
Consider the term
\begin{equation}\label{eqn: complex quad example 1}
  \braket{0|\,T\normalord \op{\Phi}^{\dagger}(y_{1}) \op{\Phi}(y_{1}) \normalord \normalord \op{\Phi}^{\dagger}(y_{2}) \op{\Phi}(y_{2}) \normalord|0}~.
\end{equation}
There is only one contributing diagram, drawn in \cref{fig: Wick complex quad example 1}, and we can immediately read off that
\begin{multline}
  \braket{0|\,T\normalord \op{\Phi}^{\dagger}(y_{1}) \op{\Phi}(y_{1}) \normalord \normalord \op{\Phi}^{\dagger}(y_{2}) \op{\Phi}(y_{2}) \normalord|0}  \\ = [G_{F}(y_{1}-y_{2})]^{2}~,
\end{multline}
where we have used that $G_{F}(y_{2}-y_{1}) = G_{F}(y_{1}-y_{2})$.
The same result is obtained when doing the ladder operator algebra directly by hand.

\begin{figure}
  \includegraphics{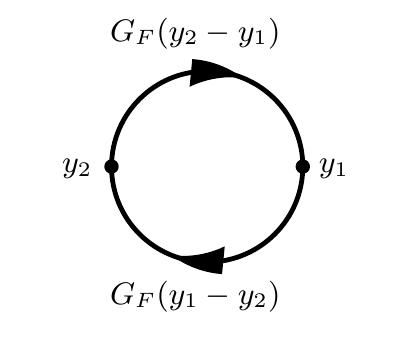}
  \caption{Diagram belonging to Example 1 for quantized complex fields.}
  \label{fig: Wick complex quad example 1}
\end{figure}

\paragraph{Example 2}
Now consider the analogue of the expectation value from the example for quadratic coupling to quantized real fields. Both annihilation operators have to be contracted, and according to \cref{thm: Wick complex t-o full vev}, they can only (nontrivially) be contracted with conjugated fields $\Phi^{\dagger}(y_{i})$. This only leaves $\Phi(y_{1})$, $\Phi(y_{2})$, whose contraction vanishes -- unlike for the quantized real field, this expectation value vanishes for a quantized complex field:
\begin{equation}
  \braket{0|a_{\vec{k}}\op{a}_{\vec{p}}\,[\,T\normalord \op{\Phi}^{\dagger}(y_{1}) \op{\Phi}(y_{1}) \normalord \normalord \op{\Phi}^{\dagger}(y_{2}) \op{\Phi}(y_{2}) \normalord]|0} = 0~.
\end{equation}
This is reflected in the diagrammatic rules through the orientation of the lines, and is the manifestation of $U(1)$ charge conservation.

\paragraph{Example 3}
We replace one outgoing particle state with an antiparticle state. The two corresponding diagrams are shown in \cref{fig: Wick complex quad example 3}, and give
\begin{multline}\label{eqn: complex quad example 3}
  \braket{0|a_{\vec{k}}\op{b}_{\vec{p}}\,[\,T\normalord \op{\Phi}^{\dagger}(y_{1}) \op{\Phi}(y_{1}) \normalord \normalord \op{\Phi}^{\dagger}(y_{2}) \op{\Phi}(y_{2}) \normalord]|0} \\
  = G_{F}(y_{1}-y_{2})\left[\varphi^{*}_{\vec{k}}(y_{2})\,\varphi^{*}_{\vec{p}}(y_{1}) + \varphi^{*}_{\vec{k}}(y_{1})\,\varphi^{*}_{\vec{p}}(y_{2}) \right]~.
\end{multline}
Unlike in the example for quadratic coupling to quantized real fields, there is only one possibility to draw a diagram if $y_{2}$ is connected to $\vec{k}$ instead of four, and a second possibility if $y_{2}$ is connected to $\vec{p}$; in total, there are thus two instead of eight diagrams.

\begin{figure}
  \includegraphics{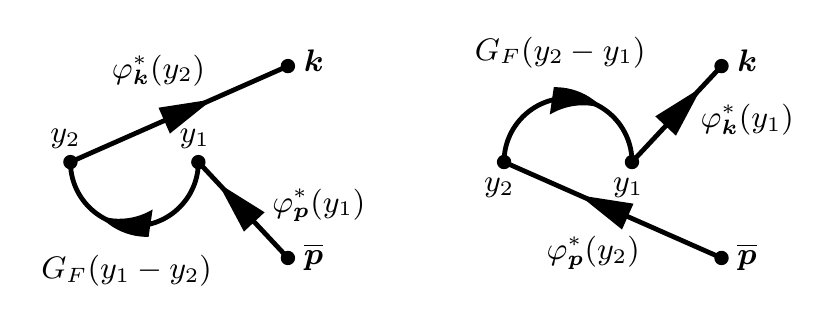}
  \caption{The two diagrams belonging to Example 3 for quantized complex fields.}
  \label{fig: Wick complex quad example 3}
\end{figure}

\subsubsection{Quadratic coupling to quantized spinor fields}\label{sec: auxiliary diagrammatic spinor quad}

For spinor fields, the rules need to be extended once more to account for the fermionic nature of the field. To evaluate expressions of the type
\begin{multline}
  \bra{0} (\op{a}_{\vec{p},r})_{\vec{p},r}\,(\op{b}_{\bar{\vec{p}},\bar{r}})_{\bar{\vec{p}},\bar{r}}\,[\,T\normalord \op{\conj\Psi}_{A_{1}}(y_{1}) \op{\Psi}^{A_{1}}(y_{1}) \normalord \cdots\\
   \cdots \normalord \op{\conj\Psi}_{A_{n}}(y_{n}) \op{\Psi}^{A_{n}}(y_{n}) \normalord\,]\,(\op{a}^{\dagger}_{\vec{k},s})_{\vec{k},s}\,(\op{b}^{\dagger}_{\bar{\vec{k}},\bar{s}})_{\bar{\vec{k}},\bar{s}}\ket{0}
\end{multline}
using \cref{thm: Wick spinor t-o full vev}, the following set of rules may be employed:

\begin{SQ}
  \hfill
  \setdefaultenum{1.}{}{}{} 
  \begin{compactenum}
    \item On the left end of the diagram, draw vertices for every creation operator $a^{\dagger}_{\vec{k}_{1},s_{1}}$ and $b^{\dagger}_{\vec{k}_{2},s_{2}}$, and label them by their momenta and spins $(\vec{k}_{1},s)$, $\overline{(\vec{k}_{2},s_{2})}$, where the bar is used to mark antiparticle operators. Similarly, draw vertices for every annihilation operator $a_{\vec{p}_{1},r_{1}}$ and $b_{\vec{p}_{2},r_{2}}$ on the right-hand side and label them by their momenta and spins $(\vec{p}_{1},r_{1})$ and $\overline{(\vec{p}_{2},r_{2})}$. Include a horizontal setoff between all vertices, such that the leftmost annihilation operator ends up at the very right, and the rightmost creation operator on the very left. The ordering of all vertices has to be inverse to the ordering of the corresponding ladder operators.
    \item In the middle, draw a horizontal series of vertices for each $y_{i}$, starting with $y_{n}$ on the left, and label them accordingly. Also, indicate the appropriate spinor index $A_{i}$.
    \item To account for the nontrivial contractions \cref{eqn: spinor field full vev contractions 1,eqn: spinor field full vev contractions 2,eqn: spinor field full vev contractions 3}, connect all vertices by arrows in the following way:
    \begin{compactitem}
      \item external vertices representing $\op{a}^{\dagger}$ and $\op{b}$ are origins of one arrow each
      \item external vertices representing $\op{a}$ and $\op{b}^{\dagger}$ are endpoints of one arrow each
      \item internal vertices representing $y_{i}$ are origin of one arrow, and endpoint of one arrow each
      \item direct connections between ingoing vertices are forbidden; direct connections between outgoing vertices are forbidden
      \item connecting an internal vertex with itself is forbidden
      \item arrow directions must not collide
    \end{compactitem}
    \item Draw one diagram for every way to connect all vertices using the rules in 3. If this is not possible, the expectation value is zero.
    \item Associate the following factors to each arrow in a diagram:
    \begin{compactitem}
      \item $\delta_{\vec{k},\vec{p}}\,\delta_{s,r}$ to an arrow connecting two external vertices $(\vec{k},s)$, $(\vec{p},r)$ (or their antiparticle pendants)
      \item $\psi_{\vec{k},s,+}^{A}(x)$ to an arrow connecting an ingoing particle vertex $(\vec{k},s)$ with an internal vertex $(x,A)$
      \item $\conj\psi_{\vec{k},s,-,A}(x)$ to an arrow connecting an ingoing antiparticle vertex $\overline{(\vec{k},s)}$ with an internal vertex $(x,A)$
      \item $\conj\psi_{\vec{p},r,+,A}(x)$ to an arrow connecting an outgoing particle vertex $(\vec{p},r)$ with an internal vertex $(x,A)$
      \item $\psi_{\vec{p},r,-}^{A}(x)$ to an arrow connecting an outgoing antiparticle vertex $\overline{(\vec{p},r)}$ with an internal vertex $(x,A)$
      \item $S_{F}(x-y)^{A}_{B}$ to an arrow going from an internal vertex $(y,B)$ to an internal vertex $(x,A)$
    \end{compactitem}
    \item To determine the overall sign of a diagram, fall back to the interaction picture calculations and order the operators according to the contractions, as described in \cref{thm: Wick spinor t-o full vev}.
    \item The vacuum expectation value is obtained by summing up the values associated with each diagram.
  \end{compactenum}
  \setdefaultenum{(a)}{}{}{} 
\end{SQ}

\paragraph{Example 1}
Consider the following term, which allows only one nontrivial contraction according to \cref{thm: Wick spinor t-o full vev}:
\begin{multline}\label{eqn: spinor quad example 1}
  \bra{0}T\normalord \op{\conj\Psi}(y_{1}) \op{\Psi}(y_{1}) \normalord  \normalord \op{\conj\Psi}(y_{2}) \op{\Psi}(y_{2}) \ket{0}\\
  = \contraction{\bra{0}T\normalord }{\op{\conj\Psi}}{_{A}(y_{1}) \op{\Psi}^{A}(y_{1}) \normalord  \normalord \op{\conj\Psi}_{B}(y_{2})}{\op{\Psi}}%
  \bcontraction{\bra{0}T\normalord \op{\conj\Psi}_{A}(y_{1}) }{\op{\Psi}}{^{A}(y_{1}) \normalord  \normalord }{\op{\conj\Psi}}
  \bra{0}T\normalord \op{\conj\Psi}_{A}(y_{1}) \op{\Psi}^{A}(y_{1}) \normalord  \normalord \op{\conj\Psi}_{B}(y_{2}) \op{\Psi}^{B}(y_{2}) \ket{0}\\
  = -\tr S_{F}(y_{1}-y_{2})\,S_{F}(y_{2}-y_{1})~,
\end{multline}
where the trace is over spinor indices. The same factors are obtained diagrammatically from \cref{fig: Wick spinor quad example 1}.

\begin{figure}
  \includegraphics{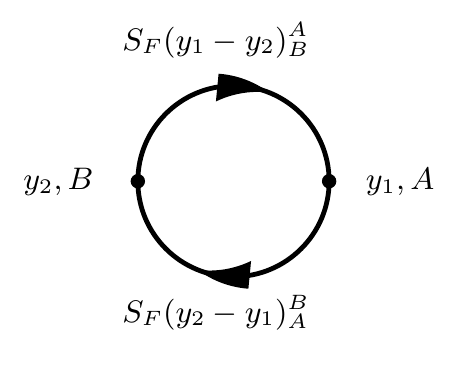}
  \caption{Diagram belonging to Example 1 for quantized spinor fields.}
  \label{fig: Wick spinor quad example 1}
\end{figure}

\begin{figure}
  \includegraphics{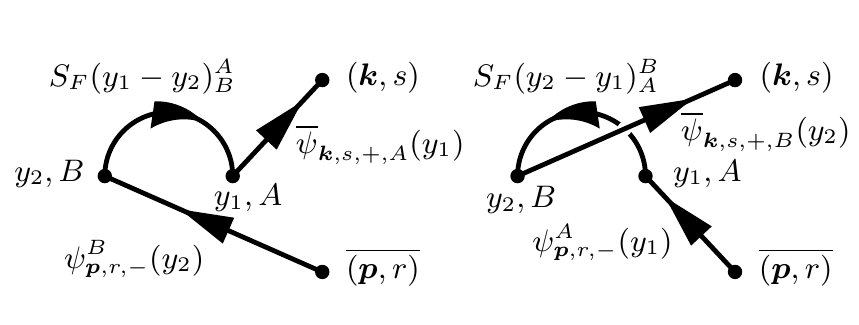}
  \caption{The two diagrams belonging to Example 2 for quantized spinor fields.}
  \label{fig: Wick spinor quad example 2}
\end{figure}

\paragraph{Example 2}
Consider the term
\begin{equation}
  E = \bra{0} a_{\vec{k},s}b_{\vec{p},r}\,[\,T\normalord \op{\conj\Psi}(y_{1}) \op{\Psi}(y_{1}) \normalord  \normalord \op{\conj\Psi}(y_{2}) \op{\Psi}(y_{2}) \normalord\,]\ket{0}~.
\end{equation}
According to \cref{thm: Wick spinor t-o full vev}, there are two nontrivial ways to contract. By reordering and writing out the contractions it can be cast into the form:
\begin{multline}\label{eqn: spinor quad example 2}
 E = \contraction{(-1)^{3}}{\vphantom{\conj\Psi}a}{_{\vec{k},s} }{\op{\conj\Psi}}%
  \contraction{(-1)^{3}a_{\vec{k},s} \op{\conj\Psi}_{A}(y_{1}) }{\op{\Psi}}{^{A}(y_{1}) }{\op{\conj\Psi}}%
  \contraction{(-1)^{3}a_{\vec{k},s} \op{\conj\Psi}_{A}(y_{1}) \op{\Psi}^{A}(y_{1}) \op{\conj\Psi}_{B}(y_{2}) }{\vphantom{\Psi}b}{_{\vec{p},r}}{\op{\Psi}}
  (-1)^{3}a_{\vec{k},s} \op{\conj\Psi}_{A}(y_{1}) \op{\Psi}^{A}(y_{1}) \op{\conj\Psi}_{B}(y_{2}) b_{\vec{p},r}\op{\Psi}^{B}(y_{2})\\
  + \contraction{(-1)^{5}}{\vphantom{\conj\Psi}a}{_{\vec{k},s} }{\op{\conj\Psi}}%
  \contraction{(-1)^{5}a_{\vec{k},s} \op{\conj\Psi}_{B}(y_{2}) }{\op{\Psi}}{^{B}(y_{2}) }{\op{\conj\Psi}}%
  \contraction{(-1)^{5}a_{\vec{k},s} \op{\conj\Psi}_{B}(y_{2}) \op{\Psi}^{B}(y_{2}) \op{\conj\Psi}_{A}(y_{1}) }{b}{_{\vec{p},r}}{\op{\Psi}}
   (-1)^{5}a_{\vec{k},s} \op{\conj\Psi}_{B}(y_{2}) \op{\Psi}^{B}(y_{2}) \op{\conj\Psi}_{A}(y_{1}) b_{\vec{p},r}\op{\Psi}^{A}(y_{1})\\
  = -\tr \conj\psi_{\vec{k},s,+}(y_{1})\, S_{F}(y_{1}-y_{2}) \,\psi_{\vec{p},r,-}(y_{2}) \\
  -\tr \conj\psi_{\vec{k},s,+}(y_{2}) \,S_{F}(y_{2}-y_{1})\, \psi_{\vec{p},r,-}(y_{1})~.
\end{multline}
The two corresponding diagrams are shown in \cref{fig: Wick spinor quad example 2}.


\section{Second order: vacuum no-response probability}\label{sec: VNRP}

We now turn to calculating the VNRP \cref{eqn: VNRP t-o operators 1} for all four detector models up to second order as an example. The emerging patterns can then easily be generalized to full Feynman rules. We have already determined the transition amplitudes $\mathcal{A}^{(0)}_{a,g}$ in \cref{eqn: VNRP amplitude at order 0}, and $\mathcal{A}^{(1)}_{a,g}$ vanishes. Only the second order amplitude $\mathcal{A}^{(2)}_{a,g}$ given in \cref{eqn: VNRP amplitude at order 2} remains to be evaluated.

\subsection{Linear coupling to quantized real fields}

To calculate the VNRP for the original UDW detector, model 1 in \cref{tab: overview detector models}, we insert the interaction Hamiltonian \cref{eqn: interaction H real linear} in the amplitude \cref{eqn: VNRP amplitude at order 2}. Combining the switching function and the spatial profile into the spacetime profile $f$ gives the amplitude
\begin{multline}\label{eqn: VNRP amplitude order 2 real linear}
  \mathcal{A}^{(2)}_{\phi,g} = (-\ii  \lambda)^{2}e^{-\ii  \omega_{\phi}}\int dy_{1} f(y_{1})\int \dd y_{2} f(y_{2})\\
   \times \braket{g|T\mu(y_{1})\mu(y_{2})|g}  \braket{\phi|T\Phi(y_{1})\Phi(y_{2})|0}~,
\end{multline}
where each integration runs over $[t_{0},t]\times \R^{n}$, and $\ket{\phi}$ are energy eigenstates  with energy $\omega_{\phi}$. \Cref{thm: Wick real t-o full vev} and \cref{thm: Wick detector t-o full vev}, and their diagrammatic representations can now be applied to evaluate these amplitudes for any given final state $\phi$ of the quantized real field.
If the final state is the vacuum, $\ket{\phi} = \ket{0}$, there is only one way to contract all fields:
\begin{multline}
  \mathcal{A}^{(2)}_{0,g} = (-\ii  \lambda)^{2}\int \dd y_{1} f(y_{1})\int \dd y_{2} f(y_{2})\\
  \contraction{\times \bra{g}T}{\mu}{(y_{1})}{\mu}%
  \contraction{\times \bra{g}T\mu(y_{1})\mu(y_{2})\ket{g}  \bra{0}T}{\Phi}{(y_{1})}{\Phi}%
  \times \bra{g}T\mu(y_{1})\mu(y_{2})\ket{g}  \bra{0}T\Phi(y_{1})\Phi(y_{2})\ket{0}~.
\end{multline}
The corresponding auxiliary diagrams for detector and field share the same vertices, and can be combined as shown in \cref{fig: VNRP linear 1a}. Writing the propagators out introduces factors of the form $e^{\pm id\cdot y_{i}}$---where $d = (\omega,\vec{0})$---in case of the detector, and $e^{\pm ik \cdot y_{i}}$ in case of the field. This allows to perform the integration over the $y_{i}$; for example:
\begin{equation}\label{eqn: integration over spacetime profile}
  \int \dd y_{1} f(y_{1})e^{-\ii  (d+k)y_{1}} \defineright \tilde{f}(d+k)~,
\end{equation}
where the domain of integration is $[t_{0},t]\times \R^{n}$. In the limit $t_{0} \to -\infty$, $t\to \infty$, this is just the Fourier transform of $f$ as defined in \cref{eqn: def Fourier transform}. For simplicity, the same symbol $\tilde{f}$ will be used both for finite times and in the limit $\pm\infty$. The amplitude then reads
\begin{multline}\label{eqn: VNRP scalar linear t-o op amplitude 2}
  \mathcal{A}^{(2)}_{0,g} = \lim_{\epsilon\to 0}\lim_{\delta \to 0} \int_{-\infty}^{\infty}\frac{d\omega}{2\pi} \,\frac{1}{L^{n}}\sum_{\vec{k}}\int_{-\infty}^{\infty}\frac{dk^{0}}{2\pi}\, (-\ii  \lambda)\tilde{f}(d+k)\\
  \times \frac{2\ii\omega}{\omega^{2}-\Omega^{2}+\ii  \epsilon} \,\frac{\ii}{k^{2}-m^{2}+\ii  \delta} 
   \,(-\ii  \lambda)\tilde{f}(-d-k)~.
\end{multline}

\begin{figure}
  \includegraphics{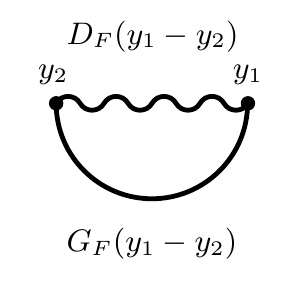}
  \caption{Linear coupling to quantized real field: auxiliary diagram for detector and field if final state $\ket{\phi} = \ket{0}$.}
  \label{fig: VNRP linear 1a}
\end{figure}

The factors and integrals that now additionally appear in the amplitude can be incorporated directly into the diagrammatic rules to obtain Feynman rules: For example, the last amplitude as a whole can be represented by the Feynman diagram shown in \cref{fig: VNRP linear 1b}. To regain the amplitude from the diagram, we simply have to integrate (sum) the product of all the factors over the energies and momenta $d$, $k$ running in the loop (with the appropriate measure).

\begin{figure}
  \includegraphics{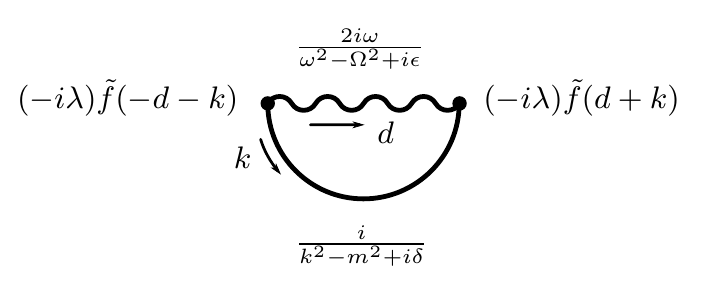}
  \caption{Linear coupling to quantized real field: Feynman diagram at second order for the VNRP amplitude, final state $\ket{\phi} = \ket{0}$.}
  \label{fig: VNRP linear 1b}
\end{figure}

If the final state contains one particle, $\ket{\phi} = \ket{1_{\vec{k}}}$, the amplitude vanishes,
\begin{equation}
  \mathcal{A}^{(2)}_{1_{\vec{k}},g} \sim \bra{0}a_{\vec{k}}T\Phi(y_{1})\Phi(y_{2})\ket{0} = 0~,
\end{equation}
since the presence of three operators does not allow for a full contraction.

If the final state is $\ket{\phi} = \ket{1_{\vec{k}},1_{\vec{p}}}$, there are two ways to contract the operators in the field part:
\begin{multline}
  \mathcal{A}^{(2)}_{1_{\vec{k}},1_{\vec{p}},g} =  (-\ii  \lambda)^{2}e^{-\ii  (\omega_{\vec{k}}+\omega_{\vec{p}})t}\\
  \times \int \dd y_{1} f(y_{1})\int \dd y_{2} f(y_{2})
  \contraction{\bra{g}T}{\mu}{(y_{1})}{\mu}%
  \bra{g}T\mu(y_{1})\mu(y_{2})\ket{g} \\ \times \Big[
  \contraction{\bra{0}}{\vphantom{\Phi}a}{_{\vec{k}}a_{\vec{p}}T}{\Phi}%
  \bcontraction{\bra{0}a_{\vec{k}}}{a}{_{\vec{p}}T\Phi(y_{1})}{\Phi}%
  \bra{0}a_{\vec{k}}a_{\vec{p}}T\Phi(y_{1})\Phi(y_{2})\ket{0} +%
  \contraction{\bra{0}}{\vphantom{\Phi}a}{_{\vec{k}}a_{\vec{p}}T\Phi(y_{1})}{\Phi}%
  \bcontraction{\bra{0}a_{\vec{k}}}{a}{_{\vec{p}}T}{\Phi}%
  \bra{0}a_{\vec{k}}a_{\vec{p}}T\Phi(y_{1})\Phi(y_{2})\ket{0} \Big]~.
\end{multline}
We have already calculated these contractions in \cref{eqn: real linear example 2}; the corresponding diagrams were drawn in \cref{fig: Wick real linear example 2}. The resulting amplitude is
\begin{multline}\label{eqn: VNRP scalar linear t-o op amplitude 3}
  \mathcal{A}^{(2)}_{1_{\vec{k}},1_{\vec{p}},g} = \frac{1}{\sqrt{2\omega_{\vec{k}}L^{n}}}e^{-\ii  \omega_{\vec{k}}t}\,\frac{1}{\sqrt{2\omega_{\vec{p}}L^{n}}}e^{-\ii  \omega_{\vec{p}}t}\\
  \times \int_{-\infty}^{\infty}\frac{d\omega}{2\pi} \frac{2\ii\omega}{\omega^{2}-\Omega^{2}+\ii  \epsilon}
  \times \Big[ (-\ii  \lambda)\tilde{f}(d-p)\, (-\ii  \lambda)\tilde{f}(-d-k) \\
  + (-\ii  \lambda)\tilde{f}(d-k)\, (-\ii  \lambda)\tilde{f}(-d-p) \Big]~,
\end{multline}
where we have suppressed the limit $\epsilon \to 0$, the two corresponding Feynman diagrams are drawn in \cref{fig: VNRP linear 2}. Note that both diagrams essentially have the same amplitude since the two outgoing momenta $\vec{k}$, $\vec{p}$ appear symmetrically.

\begin{figure}
  \includegraphics{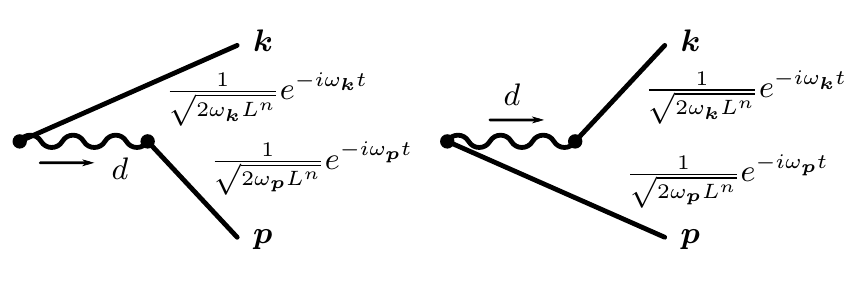}
  \caption{Linear coupling to quantized real field: Feynman diagrams at second order for the VNRP amplitude, final state $\ket{\phi} = \ket{1_{\vec{k}},1_{\vec{p}}}$. The vertex factors and the detector propagator are not shown.}
  \label{fig: VNRP linear 2}
\end{figure}

Final states $\ket{\phi}$ with more outgoing field quanta cannot contribute, because there are only two fields available to contract annihilation operators (outgoing states) with. Equivalently, on the level of Feynman diagrams, there are no further diagrams at second order. The vacuum no-response probability is now obtained by inserting all amplitudes in \cref{eqn: VNRP t-o operators 1}; we refrain from writing it out here.

\subsection{Quadratic coupling to quantized scalar fields}

\subsubsection{Quantized real fields}

For model 2 with the interaction Hamiltonian \cref{eqn: interaction H scalar quad renorm}, there are three nonzero amplitudes of the form
\begin{multline}
  \mathcal{A}^{(2)}_{\phi,g} = (-\ii  \lambda)^{2}e^{-\ii  \omega_{\phi}}\int \dd y_{1} f(y_{1})\int \dd y_{2} f(y_{2})\\
   \times \braket{g|T\mu(y_{1})\mu(y_{2})|g} \\
   \times \braket{\phi|T\normalord \Phi(y_{1})\Phi(y_{1}) \normalord \normalord \Phi(y_{2})\Phi(y_{2})\normalord|0}~.
\end{multline}
For final state $\ket{0}$, there are two ways to contract all field operators. There is a single corresponding Feynman diagram, drawn in \cref{fig: VNRP real quad 1}, its amplitude therefore contains a symmetry factor of two:
\begin{multline}\label{eqn: VNRP real quad t-o op amplitude 1}
  \mathcal{A}^{(2)}_{0,g} = 2\int_{-\infty}^{\infty}\frac{d\omega}{2\pi}\,\frac{1}{L^{n}}\sum_{\vec{k}}\int_{-\infty}^{\infty}\frac{dk^{0}}{2\pi}\,\frac{1}{L^{n}}\sum_{\vec{p}}\int_{-\infty}^{\infty}\frac{dp^{0}}{2\pi} \\
  \times (-\ii  \lambda)\tilde{f}(d+k+p) \frac{2\ii\omega}{\omega^{2}-\Omega^{2}+\ii  \epsilon} \,\frac{\ii}{k^{2}-m^{2}+\ii  \delta}\\
  \times  \,\frac{\ii}{p^{2}-m^{2}+\ii  \nu} \,(-\ii  \lambda)\tilde{f}(-d-k-p)~.
\end{multline}

\begin{figure}
  \includegraphics{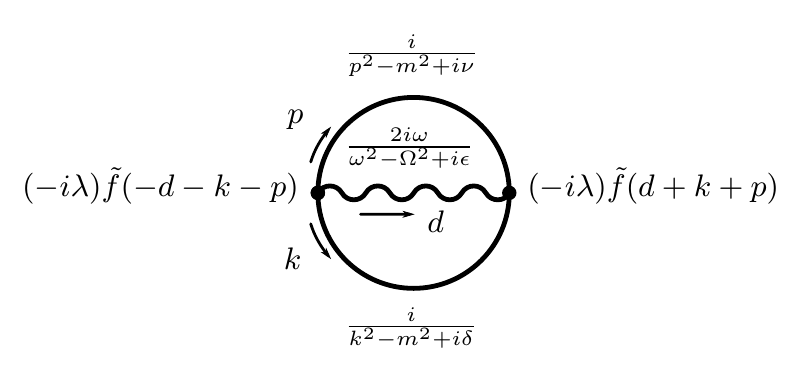}
  \caption{Quadratic coupling to quantized real field: Feynman diagram contributing to $\mathcal{A}^{(2)}_{0,g}$.}
  \label{fig: VNRP real quad 1}
\end{figure}

If the final state contains two field quanta, the expectation value in the field part corresponds to \cref{eqn: Wick real quad example 1}. Recall that there are eight ways to fully contract this expression; the two corresponding diagrams are drawn in \cref{fig: Wick real quad example 1 version 2}. The complete Feynman diagrams are drawn in \cref{fig: VNRP real quad 2} and have symmetry factor four. The sum of the corresponding amplitudes is
  \begin{multline}\label{eqn: VNRP real quad t-o op amplitude 2}
    \mathcal{A}^{(2)}_{1_{\vec{k}},1_{\vec{p}},g} =\\
      4 \frac{e^{-\ii  \omega_{\vec{k}}t}}{\sqrt{2\omega_{\vec{k}}L^{n}}}\,\frac{e^{-\ii  \omega_{\vec{p}}t}}{\sqrt{2\omega_{\vec{p}}L^{n}}} \int_{-\infty}^{\infty}\frac{d\omega}{2\pi} \frac{1}{L^{n}}\sum_{\vec{q}}\int_{-\infty}^{\infty}\frac{dq^{0}}{2\pi}\\
     \times \Big[ (-\ii  \lambda)\tilde{f}(d-p+q)\frac{2\ii\omega}{\omega^{2}-\Omega^{2}+\ii  \epsilon}\\
     \times \frac{\ii}{q^{2}-m^{2}+\ii  \delta} (-\ii  \lambda)\tilde{f}(-d-k-q) \\
     + (-\ii  \lambda)\tilde{f}(d-k+q)\frac{2\ii\omega}{\omega^{2}-\Omega^{2}+\ii  \epsilon}\\
     \times \frac{\ii}{q^{2}-m^{2}+\ii  \delta} (-\ii  \lambda)\tilde{f}(-d-p-q) \Big]~.
  \end{multline}

\begin{figure}
  \includegraphics{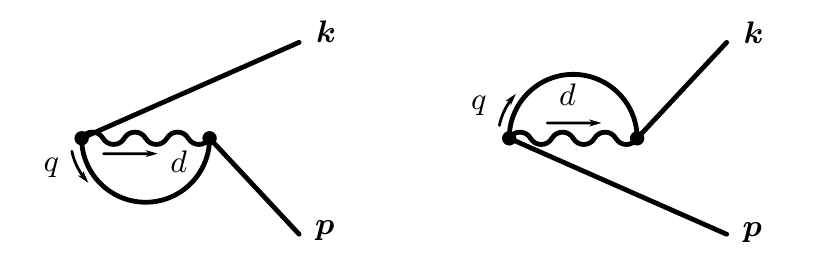}
  \caption{Quadratic coupling to quantized real field: two out of eight Feynman diagrams contributing to $\mathcal{A}^{(2)}_{1_{\vec{k}},1_{\vec{p}},g}$.}
  \label{fig: VNRP real quad 2}
\end{figure}

Finally, if the final state has four particles, there are four ladder operators for the outgoing states that must be contracted with four field operators, which gives $4! = 24$ possible combinations. \Cref{fig: VNRP real quad 3} shows the six corresponding distinct Feynman diagrams. The symmetry factor for each diagram is four, and their amplitude totals to
  \begin{multline}\label{eqn: VNRP real quad t-o op amplitude 3}
    \mathcal{A}^{(2)}_{1_{\vec{k}},1_{\vec{p}},1_{\vec{q}},1_{\vec{l}}} =  4 \frac{e^{-\ii  \omega_{\vec{k}}t}}{\sqrt{2\omega_{\vec{k}}L^{n}}}\,\frac{e^{-\ii  \omega_{\vec{p}}t}}{\sqrt{2\omega_{\vec{p}}L^{n}}}\,\frac{e^{-\ii  \omega_{\vec{q}}t}}{\sqrt{2\omega_{\vec{q}}L^{n}}}\,\frac{e^{-\ii  \omega_{\vec{l}}t} }{\sqrt{2\omega_{\vec{l}}L^{n}}}\\
    \times \int_{-\infty}^{\infty}\frac{d\omega}{2\pi} (-\ii  \lambda)^{2}\frac{2\ii\omega}{\omega^{2}-\Omega^{2}+\ii  \epsilon}\\
    \Big[ \tilde{f}(d-p-q)\,\tilde{f}(-d-k-l) + \tilde{f}(d-p-l)\,\tilde{f}(-d-k-q)    \\
    + \tilde{f}(d-q-l)\,\tilde{f}(-d-k-p) +\tilde{f}(d-k-q)\,\tilde{f}(-d-p-l) \\
    +\tilde{f}(d-k-p)\,\tilde{f}(-d-q-l) +\tilde{f}(d-k-l)\,\tilde{f}(-d-p-q)\Big]~.
  \end{multline}

\begin{figure}
  \includegraphics{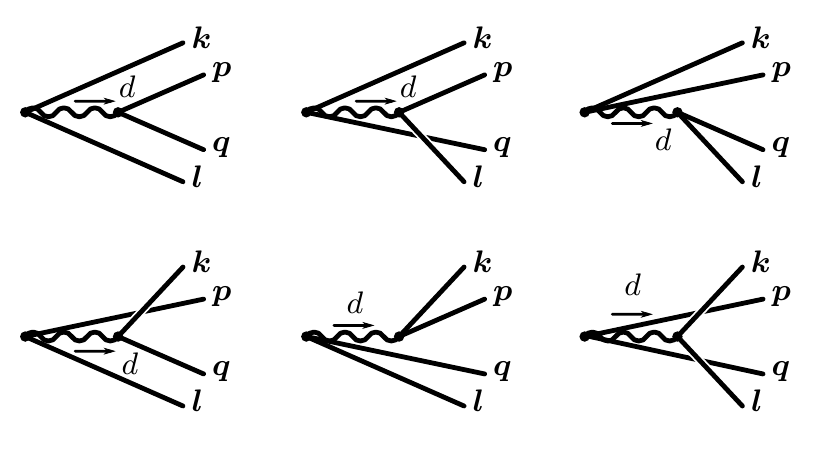}
  \caption{Quadratic coupling to quantized real field: six out 24 Feynman diagrams contributing to $\mathcal{A}^{(2)}_{1_{\vec{k}},1_{\vec{p}},1_{\vec{q}},1_{\vec{l}},g}$.}
  \label{fig: VNRP real quad 3}
\end{figure}

\subsubsection{Quantized complex fields}

When considering model 3 for the quantized complex field, there are again three nonzero amplitudes of the form
\begin{multline}
  \mathcal{A}^{(2)}_{\phi,g} = (-\ii  \lambda)^{2}e^{-\ii  \omega_{\phi}}\int \dd y_{1} f(y_{1})\int \dd y_{2} f(y_{2})\\
   \times \braket{g|T\mu(y_{1})\mu(y_{2})|g}\\
   \times   \braket{\phi|T\normalord \Phi^{\dagger}(y_{1})\Phi(y_{1}) \normalord \normalord \Phi^{\dagger}(y_{2})\Phi(y_{2})\normalord|0}~.
\end{multline}
If the final state is the vacuum, there is only one way to fully contract, unlike for the quantized real field. The expression in the field part was already evaluated in \cref{eqn: complex quad example 1}, the corresponding auxiliary diagram is given in \cref{fig: Wick complex quad example 1}. In combination with the diagram of the detector's Feynman propagator, it becomes the Feynman diagram \cref{fig: VNRP complex quad 1}. Its amplitude is simply half of the amplitude in the case of a quantized real field as given in \cref{eqn: VNRP real quad t-o op amplitude 1}.

\begin{figure}
  \includegraphics{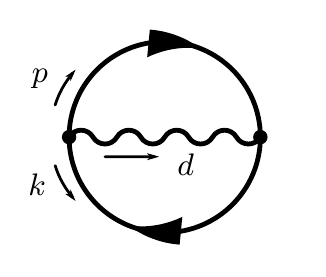}
  \caption{Quadratic coupling to quantized complex field: Feynman diagram corresponding to $\mathcal{A}^{(2)}_{0,g}$.}
  \label{fig: VNRP complex quad 1}
\end{figure}

If the final state of the quantized complex field contains one particle and one antiparticle, the amplitude admits two different ways to fully contract. The expression in the field has already been calculated  in \cref{eqn: complex quad example 3}; the two corresponding Feynman diagrams based on \cref{fig: Wick complex quad example 3} are given in \cref{fig: VNRP complex quad 2}. The resulting amplitude is
\begin{equation}\label{eqn: VNRP complex quad t-o op amplitude 2}
  \mathcal{A}^{(2)}_{1_{\vec{k}},\bar{1}_{\vec{p}},g} = \frac{1}{4}\mathcal{A}^{(2)}_{1_{\vec{k}},1_{\vec{p}},g}~,
\end{equation}
differing from $\mathcal{A}^{(2)}_{1_{\vec{k}},1_{\vec{p}},g}$ in \cref{eqn: VNRP real quad t-o op amplitude 2} only by lack of the symmetry factor.

\begin{figure}
  \includegraphics{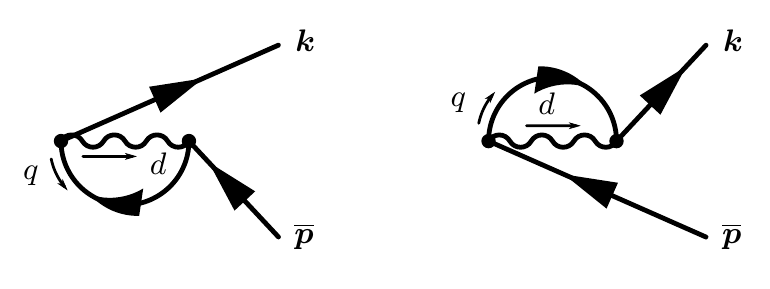}
  \caption{Quadratic coupling to quantized complex field: the two Feynman diagrams contributing to  $\mathcal{A}^{(2)}_{1_{\vec{k}},\bar{1}_{\vec{p}},g}$.}
  \label{fig: VNRP complex quad 2}
\end{figure}

If the final state has two particles and two antiparticles, there are only four full contractions possible, as opposed to the quantized real field with 24 contractions. They correspond to the diagrams in \cref{fig: VNRP complex quad 3}, and the resulting amplitude is
\begin{multline}\label{eqn: VNRP complex quad t-o op amplitude 3}
\mathcal{A}^{(2)}_{1_{\vec{k}},\bar{1}_{\vec{p}},1_{\vec{q}},\bar{1}_{\vec{l}}} =  \frac{e^{-\ii  \omega_{\vec{k}}t}}{\sqrt{2\omega_{\vec{k}}L^{n}}}\,\frac{e^{-\ii  \omega_{\vec{p}}t}}{\sqrt{2\omega_{\vec{p}}L^{n}}}\,\frac{e^{-\ii  \omega_{\vec{q}}t}}{\sqrt{2\omega_{\vec{q}}L^{n}}}\,\frac{e^{-\ii  \omega_{\vec{l}}t} }{\sqrt{2\omega_{\vec{l}}L^{n}}}\\
\times \int_{-\infty}^{\infty}\frac{d\omega}{2\pi} (-\ii  \lambda)^{2}\frac{2\ii\omega}{\omega^{2}-\Omega^{2}+\ii  \epsilon} \Big[ \tilde{f}(d-p-q)\,\tilde{f}(-d-k-l) \\
+ \tilde{f}(d-q-l)\,\tilde{f}(-d-k-p) +\tilde{f}(d-k-p)\,\tilde{f}(-d-q-l) \\
+\tilde{f}(d-k-l)\,\tilde{f}(-d-p-q)\Big]~.
\end{multline}

\begin{figure}
  \includegraphics{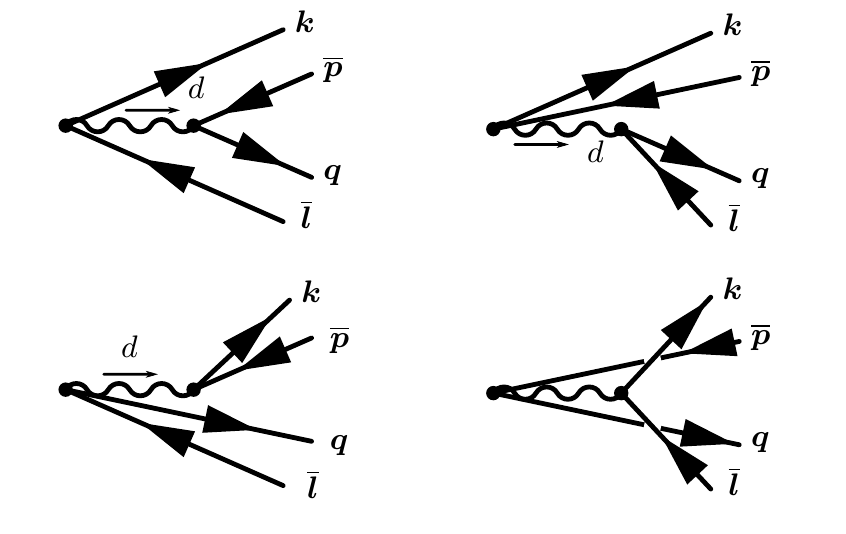}
  \caption{Quadratic coupling to quantized complex field: the four Feynman diagrams contributing to $\mathcal{A}^{(2)}_{1_{\vec{k}},\bar{1}_{\vec{p}},1_{\vec{q}},\bar{1}_{\vec{l}},g}$.}
  \label{fig: VNRP complex quad 3}
\end{figure}

\subsection{Quadratic coupling to quantized spinor fields}

For the spinor field detector, model 4, the missing amplitudes are of the form
\begin{multline}\label{eqn: VNRP amplitude order 2 spinor quad}
  \mathcal{A}^{(2)}_{\psi,g} = (-\ii  \lambda)^{2}e^{-\ii  \omega_{\psi}t}\int \dd y_{1} f(y_{1})\int \dd y_{2} f(y_{2})\\
   \times \braket{g|T\mu(y_{1})\mu(y_{2})|g}\\
   \times  \braket{\psi|T\normalord \conj\Psi(y_{1}) \Psi(y_{1}) \normalord\normalord \conj\Psi(y_{2}) \Psi(y_{2}) \normalord|0}~.
\end{multline}
If the final state is $\ket{0}$, there is only one full contraction. The expression involving the field operators has already been evaluated in \cref{eqn: spinor quad example 1} and yields twice the Feynman propagator of the spinor field. Writing it out according to \cref{eqn: Feynman P spinor field}, the amplitude can be expressed as
\begin{multline}
  \mathcal{A}^{(2)}_{0,g} =
  -\int_{-\infty}^{\infty}\frac{d\omega}{2\pi}\,\frac{1}{L^{n}}\sum_{\vec{k}} \int_{-\infty}^{\infty}\frac{dk^{0}}{2\pi}\,\frac{1}{L^{n}} \sum_{\vec{p}} \int_{-\infty}^{\infty}\frac{dp^{0}}{2\pi}\\
  \times(-\ii  \lambda)\tilde{f}(d-k+p)\times \frac{2\ii\omega}{\omega^{2}-\Omega^{2}+\ii  \epsilon}\frac{\ii[\slashed{k}+m]^{B}_{A}}{k^{2}-m^{2}+\ii  \delta}\\
  \times\frac{\ii[\slashed{p}+m]^{A}_{B}}{p^{2}-m^{2}+\ii  \nu}(-\ii  \lambda)\tilde{f}(-d+k-p)~,
\end{multline}
where we have used Feynman slash notation $\slashed{k} = k \cdot \gamma$. Combining the two auxiliary diagrams corresponding to this contraction yields the Feynman diagram \cref{fig: VNRP spinor quad 1}. Note the overall sign $-1$ of the amplitude that arises from the reordering of the fermionic field operators.

\begin{figure}
  \includegraphics{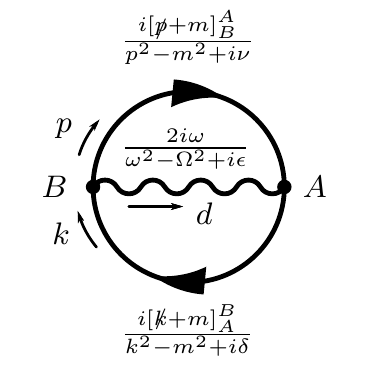}
  \caption{Quadratic coupling to quantized spinor field: Feynman diagram with factors corresponding to $\mathcal{A}^{(2)}_{0,g}$. Note that the sign of the momentum in the spinor field propagator is such that it flows into the direction of the field line (fermion flow).}
  \label{fig: VNRP spinor quad 1}
\end{figure}

If there are two field quanta in the final state (particle and antiparticle), there are two ways to fully contract, analogous to the quantized complex field. The expression in the field part is similar to the one given in \cref{eqn: spinor quad example 2}, but the annihilation operators are exchanged. This introduces an overall sign $-1$, such that
\begin{multline}\label{eqn: vnrp amplitude spinor field detector}
  \mathcal{A}^{(2)}_{1_{\vec{k},s},\bar{1}_{\vec{p},r},g} =\\
   \sqrt{\frac{m}{\omega_{\vec{k}}L^{n}}}e^{-\ii  \omega_{\vec{k}}t}\conj{u}_{\vec{k},s,B}\,\sqrt{\frac{m}{\omega_{\vec{p}}L^{n}}}e^{-\ii  \omega_{\vec{p}}t}v_{\vec{p},r}^{A}\int_{-\infty}^{\infty}\frac{d\omega}{2\pi}\\
  \times \frac{1}{L^{n}}\sum_{\vec{q}} \int_{-\infty}^{\infty}\frac{dq^{0}}{2\pi} (-\ii  \lambda)\tilde{f}(d-p-q)\frac{2\ii\omega}{\omega^{2}-\Omega^{2}+\ii  \epsilon}\\
  \times \frac{\ii[\slashed{q}+m]^{B}_{A}}{q^{2}-m^{2}+\ii  \delta} (-\ii  \lambda)\tilde{f}(-d-k-q)\\
  + \sqrt{\frac{m}{\omega_{\vec{k}}L^{n}}}e^{-\ii  \omega_{\vec{k}}t}\conj{u}_{\vec{k},s,A}\,\sqrt{\frac{m}{\omega_{\vec{p}}L^{n}}}e^{-\ii  \omega_{\vec{p}}t}v_{\vec{p},r}^{B}\int_{-\infty}^{\infty}\frac{d\omega}{2\pi}\\
  \times \frac{1}{L^{n}}\sum_{\vec{q}} \int_{-\infty}^{\infty}\frac{dq^{0}}{2\pi} (-\ii  \lambda)\tilde{f}(d-k+q)\frac{2\ii\omega}{\omega^{2}-\Omega^{2}+\ii  \epsilon}\\
  \times \frac{\ii[\slashed{q}+m]^{B}_{A}}{q^{2}-m^{2}+\ii  \delta} (-\ii  \lambda)\tilde{f}(-d-p-q)~.
\end{multline}
The expression for massless fields is recovered by $m\,\omega_{\vec{k}}^{-1}L^{-n} \to L^{-n}$. The corresponding Feynman diagrams are drawn in \cref{fig: VNRP spinor quad 2}.

\begin{figure}
  \includegraphics{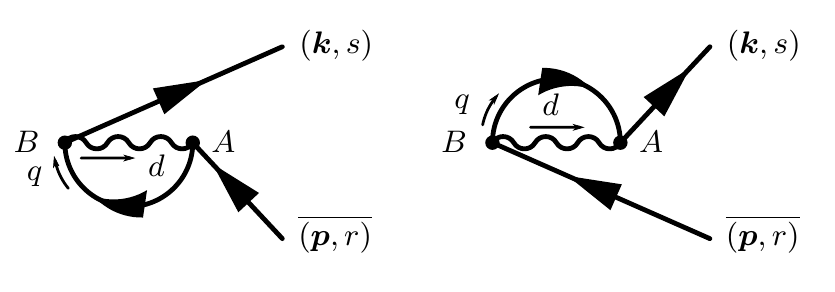}
  \caption{Quadratic coupling to quantized spinor field: the two Feynman diagrams contributing to  $\mathcal{A}^{(2)}_{1_{\vec{k},s},\bar{1}_{\vec{p},r},g}$.}
  \label{fig: VNRP spinor quad 2}
\end{figure}

The only other final state with nonzero probability amplitude contains two particle states and two antiparticle states. Note that the fermionic nature of the field requires $(\vec{k},s)\neq(\vec{q},u)$ and $(\vec{p},r)\neq(\vec{l},v)$ due to the Pauli exclusion principle. As for the quantized complex field in the comparable situation, there are four possible contractions---see \cref{fig: VNRP spinor quad 3} for the corresponding Feynman diagrams. For massive fields, the amplitude is
\begin{multline}
  \mathcal{A}^{(2)}_{1_{\vec{k},s},\bar{1}_{\vec{p},r},1_{\vec{q},u},\bar{1}_{\vec{l},v},g} =\\
  (-\ii  \lambda)^{2}\sqrt{\frac{m}{\omega_{\vec{k}}L^{n}}}e^{-\ii  \omega_{\vec{k}}t}\,\sqrt{\frac{m}{\omega_{\vec{p}}L^{n}}}e^{-\ii  \omega_{\vec{p}}t}\\
  \times \sqrt{\frac{m}{\omega_{\vec{q}}L^{n}}}e^{-\ii  \omega_{\vec{q}}t}\sqrt{\frac{m}{\omega_{\vec{l}}L^{n}}}e^{-\ii  \omega_{\vec{l}}t}\int_{-\infty}^{\infty}\frac{d\omega}{2\pi} \frac{2\ii\omega}{\omega^{2}-\Omega^{2}+\ii  \epsilon}\\
  \times \Big[\tilde{f}(d-p-q)\,\tilde{f}(-d-k-l)(\conj{u}_{\vec{q},u,A}v_{\vec{p},r}^{A})(\conj{u}_{\vec{k},s,B}v_{\vec{l},v}^{B}) \\
  +\tilde{f}(d-q-l)\,\tilde{f}(-d-k-p)(\conj{u}_{\vec{q},u,A}v_{\vec{l},v}^{A})(\conj{u}_{\vec{k},s,B}v_{\vec{p},r}^{B}) \\
  +\tilde{f}(d-k-p)\,\tilde{f}(-d-q-l)(\conj{u}_{\vec{k},s,A}v_{\vec{p},r}^{A})(\conj{u}_{\vec{q},u,B}v_{\vec{l},v}^{B}) \\
  +\tilde{f}(d-k-l)\,\tilde{f}(-d-p-q)(\conj{u}_{\vec{k},s,A}v_{\vec{l},v}^{A})(\conj{u}_{\vec{q},u,B}v_{\vec{p},r}^{B}) \Big]~,
\end{multline}
and the case of a massless field is recovered like for \cref{eqn: vnrp amplitude spinor field detector}.

\begin{figure}
  \includegraphics{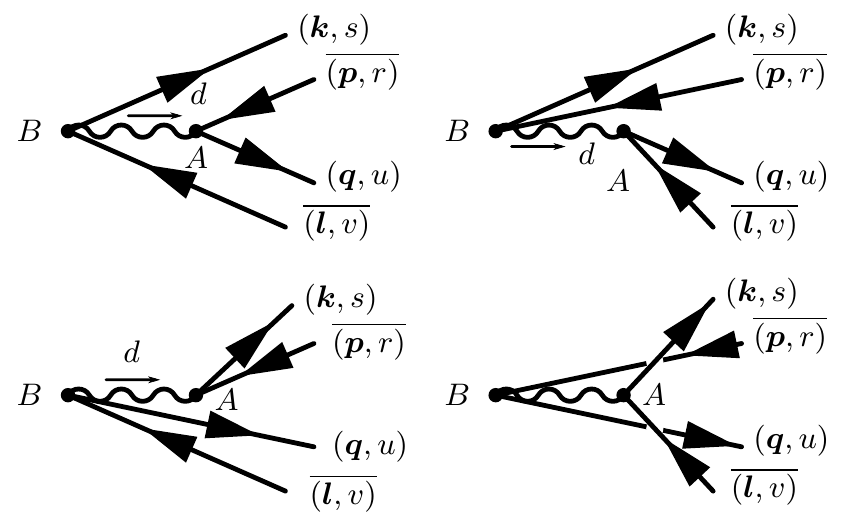}
  \caption{Quadratic coupling to quantized spinor field: the four Feynman diagrams contributing to $\mathcal{A}^{(2)}_{1_{\vec{k},s},\bar{1}_{\vec{p},r},1_{\vec{q},u},\bar{1}_{\vec{l},v},g}$.}
  \label{fig: VNRP spinor quad 3}
\end{figure}


\section{Feynman rules for detector models}\label{sec: feynman rules}

Feynman diagrams and Feynman rules can be understood as an extension of the diagrammatic representation of Wick's theorem.In order to obtain a transition amplitude, there are additional factors and summations (integrations) to be performed. Since their occurrence also follows simple patterns, it is straightforward to include them into the diagrammatic rules: Combining the auxiliary diagrams of the different field types into one diagram containing all field types gives us a Feynman diagram, and incorporating all additional operations into the rules for evaluating auxiliary diagrams gives Feynman rules.

Note that the formalism presented here does not make use of asymptotic free states, in contrast to the usual Feynman rules of particle physics. Rather, ingoing and outgoing states are tensor products of Fock states of the noninteracting systems.

\subsection{Towards Feynman rules}\label{sec: towards Feynman rules}

\paragraph{Feynman diagrams} Consider, for example, the expression
\begin{multline}
  \bra{g}\sigma^{-}T\mu(y_{1})\mu(y_{2})\sigma^{+}\ket{g}\\
  \times \bra{0}b_{\vec{p},r}a_{\vec{k},s}T\normalord \conj\Psi(y_{1}) \Psi(y_{1}) \normalord\normalord \conj\Psi(y_{2}) \Psi(y_{2}) \normalord\ket{0}~.
\end{multline}
The detector part allows for three different contractions, giving rise to three terms and the three corresponding diagrams drawn in \cref{fig: Wick detector example 2}. The field part has two possible contractions, with diagrams shown in \cref{fig: Wick spinor quad example 2}. Since these two parts are multiplied, one arrives at six summands in total. Both the diagrams for the detector and the field always have the same amount of vertices (equal to the order in perturbation theory), and can be combined into six Feynman diagrams which simultaneously show contractions for the detector and the field. 

\paragraph{Feynman rules} From the VNRP calculations in \cref{sec: VNRP}, we can directly deduce that outgoing quanta contribute the factors
\begin{equation}\begin{aligned}
  e^{-\ii  \Omega t_{0}}~, &\quad\quad \frac{1}{\sqrt{2\omega_{\vec{k}}L^{n}}}e^{-\ii  \omega_{\vec{k}}t_{0}}~,\\  \frac{m}{\sqrt{\omega_{\vec{k}}L^{n}}}e^{-\ii  \omega_{\vec{k}}t_{0}}\,u_{\vec{k},s}^{A}~, &\quad\quad \frac{m}{\sqrt{\omega_{\vec{k}}L^{n}}}e^{-\ii  \omega_{\vec{k}}t_{0}}\,\conj{v}_{\vec{k},s,A}
\end{aligned}\end{equation}
to the amplitude [for the detector, quantized scalar fields, quantized spinor field particles and (massive) spinor field antiparticles from left to right and top to bottom]. Internal lines between vertices give (Fourier transformed) Feynman propagators as usual; for the detector, scalar fields, and (massive) spinor fields from left to right:
\begin{align}
 \frac{2\ii\omega}{\omega^{2}-\Omega^{2}+\ii  \epsilon}~, && \frac{\ii}{k^{2}-m^{2}+\ii  \epsilon}~, && \frac{\ii[\slashed{k}+m]^{A}_{B}}{k^{2}-m^{2}+\ii  \epsilon}~.
\end{align}
Each vertex contributes
$
  \tilde{f}(k)
$
where $k$ is the energy-momentum flowing into the vertex. Finally, all internal momenta (the ones in the propagators) need to be integrated or summed according to
  \begin{align}
    &\int_{-\infty}^{\infty}\frac{d\omega}{2\pi}~, & \frac{1}{L^{n}}\sum_{\vec{k}}\int_{-\infty}^{\infty}\frac{dk^{0}}{2\pi}
  \end{align}
for the detector (left) and the fields (right).

Let us now consider the factors associated with incoming quanta on one hand, and cases where two ladder operators are contracted, i.e., diagrammatically, when an incoming line is directly connected to an outgoing line. The contractions related to the first case are $\contraction{}{\sigma}{^{-}}{\mu}\sigma^{-}\mu(t_{i}) = \id\, e^{+\ii  \Omega t_{i}}$ for the detector, $\contraction{}{\vphantom{\Phi}\op{a}}{_{\vec{k}}\cdots }{\Phi} \op{a}_{\vec{k}}\cdots  \Phi^{\dagger}(x) =\id\,  \varphi^{*}_{\vec{k}}(x)$ for a quantized real field, and so on. The exponential functions simply contribute to the Fourier transform of $f$, such that the incoming momenta appear in the momentum balance of the vertex they are attached to. The normalization factors and spinors are left over and assigned to the incoming diagram leg. Moreover, since we are working in the interaction picture, there is a time-dependent phase in addition to the ladder operator creating the state, for example,
\begin{align}
  \Iket{e;t_{0}} &= e^{+\ii  \Omega t_{0}}\ket{e}, & \Iket{1_{\vec{k}};t_{0}} &= e^{+\ii  \omega_{\vec{k}}t_{0}}\ket{1_{\vec{k}}}~.
\end{align}
Altogether, the incoming legs of the diagram contribute factors
\begin{equation}\begin{aligned}
  e^{+\ii  \Omega t}~, &\quad\quad \frac{1}{\sqrt{2\omega_{\vec{k}}L^{n}}}e^{+\ii  \omega_{\vec{k}}t}~,\\
  \frac{m}{\sqrt{\omega_{\vec{k}}L^{n}}}e^{+\ii  \omega_{\vec{k}}t}\conj{u}_{\vec{k},s,A}~, &\quad\quad \frac{m}{\sqrt{\omega_{\vec{k}}L^{n}}}e^{+\ii  \omega_{\vec{k}}t}\,v_{\vec{k},s}^{A}
\end{aligned}\end{equation}
in the same order as before.

If two ladder operators are contracted, they merely contribute Kronecker deltas ensuring that both are of the same type (momentum, spin) according to \cref{eqn: real field full vev contractions 2,eqn: complex field full vev contractions 2,eqn: spinor field full vev contractions 2,eqn: detector full vev contractions 3}. In the full amplitude, they additionally contribute phases from the interaction picture. For example, in case of a quantized real field:
\begin{multline}
  \Ibraket{1_{\vec{k}},g;t|U^{(0)}(t,t_{0})|1_{\vec{p}},g;t_{0}}
  = e^{-\ii  \omega_{\vec{k}}(t-t_{0})} \,\delta_{\vec{k},\vec{p}}~.
\end{multline}
Similar results hold for the monopole moment, and the remaining fields. Disconnected lines from the left to the right end of the diagram will thus contribute
\begin{align}
  e^{-\ii  \Omega (t-t_{0})}~, && e^{-\ii  \omega_{\vec{k}}(t-t_{0})}\,\delta_{\vec{k},\vec{p}}~, && e^{-\ii  \omega_{\vec{k}}(t-t_{0})}\,\delta_{\vec{k},\vec{p}}\,\delta_{s,r}
\end{align}
for the detector, quantized scalar fields, and quantized spinor fields, in that order.

The Feynman diagrams presented in this chapter, as well as the auxiliary diagrams in \cref{sec: diagrammatic}, do not contain lines for the detector in its ground state---in this they differ from the representations we used in \cref{sec: VEP,sec: renormalizing} and it differs from the diagrams used in \cite{reznik_violating_2005}. The reason for this is that Feynman diagrams generally describe the energy-momentum flows involved in the interaction processes. The detectors (as macroscopic physical objects) do carry energy and momentum, but in the UDW-type detector models, only the excitations of its internal degree of freedom couple to the field. This is why lines are drawn not for the detector itself but instead for the excitation, i.e., for the quantum of energy, that the detector can carry.

\subsection{Feynman rules for single detectors}

We summarize the Feynman rules for all four UDW-type detector models considered in this work. The rules are applicable for detectors on Minkowski spacetime , the trajectory of the detector is determined by its spacetime profile $f$: if the spatial dependence varies with time in the restframe of the field, the detector is moving. Assume the field is initially in the number eigenstate $\ket{a_{i}}$ at time $t_{0}$, and the detector in the state $\ket{\beta_{i}}$ (which may either be the ground state $\ket{g}$ or the excited state $\ket{e}$). The aim is to calculate the probability amplitude
\begin{equation}\label{eqn: generic probability amplitude}
  \mathcal{A} = \braket{a_{f},\beta_{f};t|U(t,t_{0})|a_{i},\beta_{i};t_{0}}
\end{equation}
for the detector to end up in the state $\ket{\beta_{f}}$ and the field in $\ket{a_{f}}$ at time $t>t_{0}$ to order $\lambda^{n}$ in perturbation theory. In general, this is accomplished by combining the perturbative expansion \cref{eqn: perturbative expansion,eqn: time-evolution order n} of the time-evolution operator with the appropriate Wick theorem. 

The calculations can be organized using a set of Feynman rules; the procedure itself is independent of the detector model and the type of field:

\begin{GFR}
  In order to calculate a probability amplitude of the type \cref{eqn: generic probability amplitude} to order $\lambda^{n}$ in perturbation theory, proceed as follows:
  \setdefaultenum{1.}{}{}{} 
  \begin{compactenum}
    \item At a given order $k\leq n$, draw all possible Feynman diagrams using the rules appropriate for the detector model and field in question, as described in the following sections.
    \item Evaluate the amplitude $\mathcal{A}^{(k)}_{i}$ of each diagram by first multiplying the factors associated with the individual parts of that diagram, and then summing (integrating) over all undetermined energies and momenta.
    \item Sum up the amplitudes at order $k$: their contribution to the total amplitude $\mathcal{A}$ is
    \begin{equation}
      \mathcal{A}^{(k)} = \sum_{i}\mathcal{A}^{(k)}_{i}
    \end{equation}
    \item Repeat the above steps for every order $0 \leq k \leq n$, and add the resulting amplitudes to obtain $\mathcal{A}$:
    \begin{equation}
      \mathcal{A} = \sum_{k=0}^{n}\frac{1}{k!}\mathcal{A}^{(k)}
    \end{equation}
  \end{compactenum}
  \setdefaultenum{(a)}{}{}{} 
\end{GFR}

The four sections below specify how to draw and evaluate the Feynman diagrams at a given order $k$ for the different detector models and fields.

\subsubsection{Linear coupling to quantized real fields}\label{sec: feynman rules real linear}

The Feynman diagrams for the detector model 1 coupling a single UDW-type detector to a quantized real field consist of the seven different types of elements listed below. Each element is associated with a factor in the amplitude $\mathcal{A}^{(k)}_{i}$ of an individual diagram. The corresponding factor is printed to the right of each element. Straight lines represent excitations of the quantized scalar field, wiggly lines excitations of the UDW-type detector.
  \newlength{\diagramheight}
  \setlength{\diagramheight}{1cm}
  
\begin{EFD}
  \hfill\\
  \includegraphics{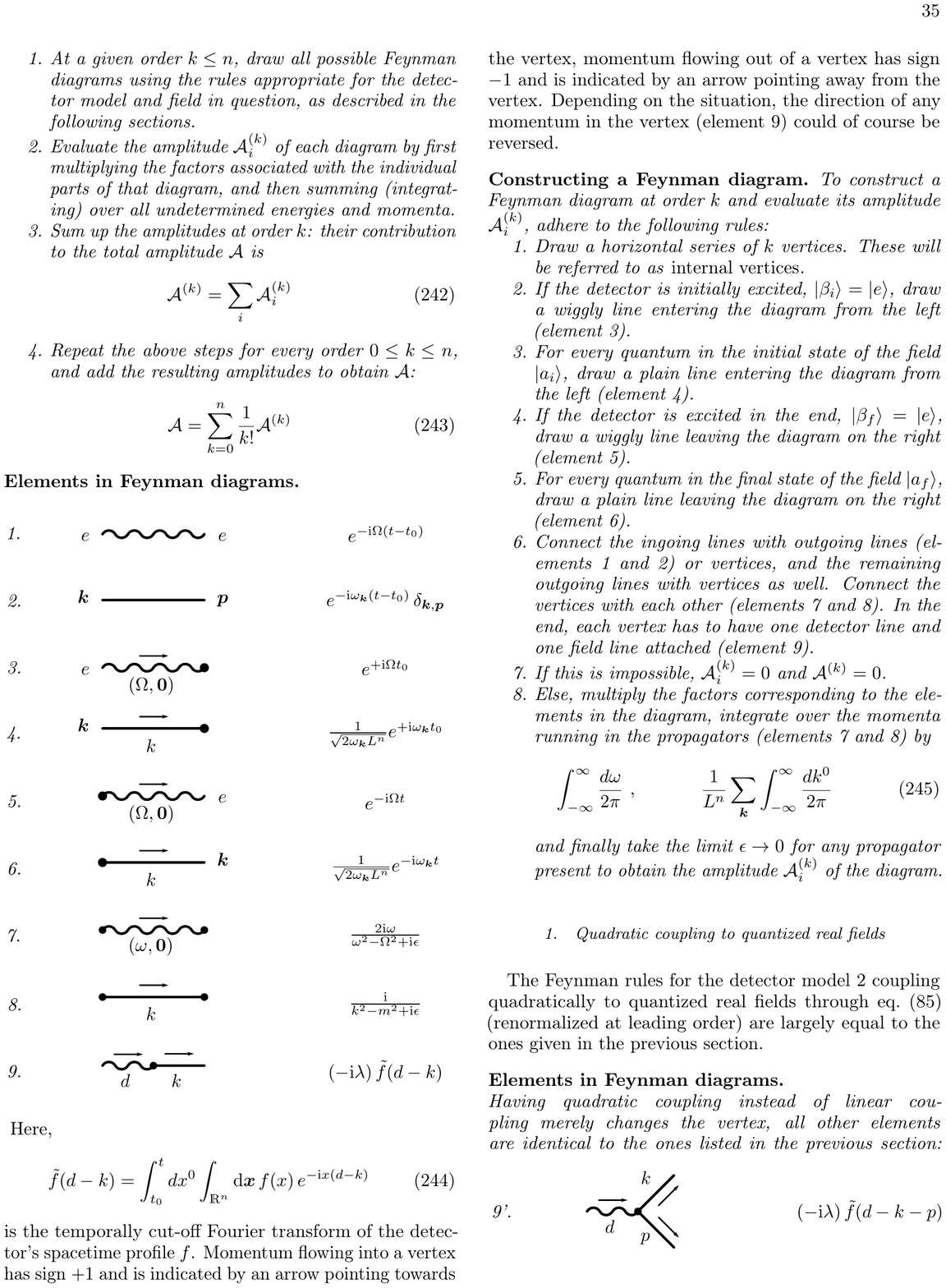}
\end{EFD}

\noindent{}
Here,
\begin{equation}
  \tilde{f}(d-k) = \int_{t_{0}}^{t}dx^{0}\int_{\R^{n}}\dd \vec{x}\,f(x)\,e^{-\ii  x(d-k)}
\end{equation}
is the temporally cutoff Fourier transform of the detector's spacetime profile $f$. Momentum flowing into a vertex has sign $+1$ and is indicated by an arrow pointing towards the vertex, momentum flowing out of a vertex has sign $-1$ and is indicated by an arrow pointing away from the vertex. Depending on the situation, the direction of any momentum in the vertex (element 9) could of course be reversed.

\begin{CFD}
  To construct a Feynman diagram at order $k$ and evaluate its amplitude $\mathcal{A}^{(k)}_{i}$, adhere to the following rules:
  \setdefaultenum{1.}{}{}{} 
  \begin{compactenum}
    \item Draw a horizontal series of $k$ vertices. These will be referred to as \emph{internal vertices}.
    \item If the detector is initially excited, $\ket{\beta_{i}} = \ket{e}$, draw a wiggly line entering the diagram from the left (element 3).
    \item For every quantum in the initial state of the field $\ket{a_{i}}$, draw a plain line entering the diagram from the left (element 4).
    \item If the detector is excited in the end, $\ket{\beta_{f}} = \ket{e}$, draw a wiggly line leaving the diagram on the right (element 5).
    \item For every quantum in the final state of the field $\ket{a_{f}}$, draw a plain line leaving the diagram on the right  (element 6).
    \item Connect the ingoing lines with outgoing lines (elements 1 and 2) or vertices, and the remaining outgoing lines with vertices as well. Connect the vertices with each other (elements 7 and 8). In the end, each vertex has to have one detector line and one field line attached (element 9).
    \item If this is impossible, $\mathcal{A}^{(k)}_{i} = 0$ and $\mathcal{A}^{(k)} = 0$.
    \item Else, multiply the factors corresponding to the elements in the diagram, integrate over the momenta running in the propagators (elements 7 and 8) by
    \begin{align}
      &\int_{-\infty}^{\infty}\frac{d\omega}{2\pi}~, & \frac{1}{L^{n}}\sum_{\vec{k}}\int_{-\infty}^{\infty}\frac{dk^{0}}{2\pi}
    \end{align}
    and finally take the limit $\epsilon \to 0$ for any propagator present to obtain  the amplitude $\mathcal{A}^{(k)}_{i}$ of the diagram.
  \end{compactenum}
  \setdefaultenum{(a)}{}{}{} 
\end{CFD}

\subsubsection{Quadratic coupling to quantized real fields}\label{sec: feynman rules real quad}

The Feynman rules for the detector model 2 coupling quadratically to quantized real fields through \cref{eqn: interaction H scalar quad renorm} (renormalized at leading order) are largely equal to the ones given in the previous section.

\begin{EFD}
  \hfill{}\\
  Having quadratic coupling instead of linear coupling merely changes the vertex, all other elements are identical to the ones listed in the previous section:
  \includegraphics{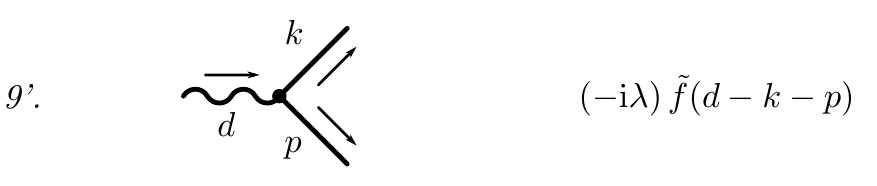}
\end{EFD}

\medskip{}

\begin{CFD}
  Construction of a Feynman diagram at a given order $k$ proceeds as before, with the exception of
  \begin{compactenum}
    \item[6'.] [\dots] In the end, each vertex has to have one detector line and \emph{two} field lines attached (element 9').
  \end{compactenum}
and the appearance of symmetry factors because each vertex corresponds to two identical field operators $\Phi(y_{i})\Phi(y_{i})$:
  \begin{compactenum}
    \item[9.] Multiply the amplitude with the symmetry factor of the diagram (compare \cref{sec: auxiliary diagrammatic real quad}).
  \end{compactenum}
\end{CFD}

\subsubsection{Quadratic coupling to quantized complex fields}\label{sec: feynman rules complex quad}
In case of the (leading order renormalized) model 3 and a quantized complex field, one additionally has to distinguish between particle and antiparticle quanta. Diagrammatically, this is achieved by turning field lines into arrows.

\begin{EFD}
  \hfill\\
  \includegraphics{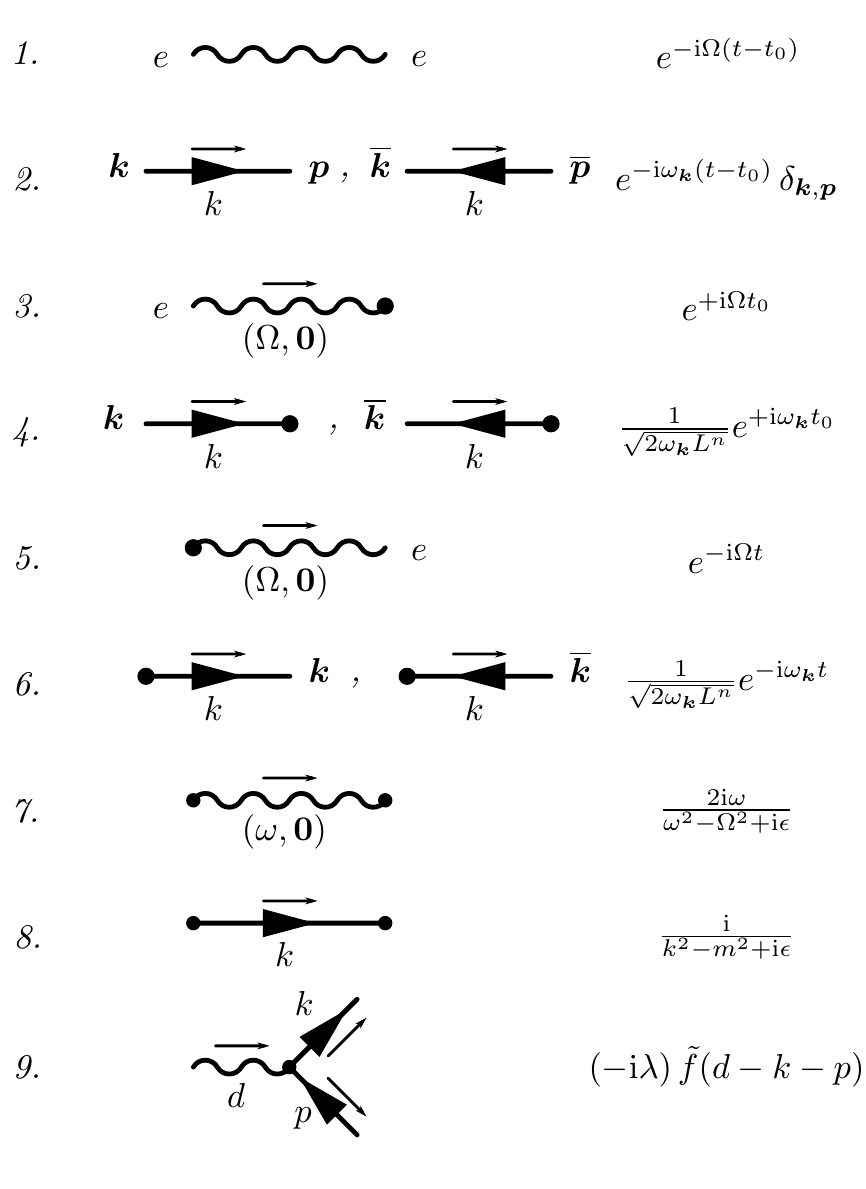}
\end{EFD}

\medskip{}

\begin{CFD}
  For quantized complex fields, step 6.\,from \cref{sec: feynman rules real linear} has to be amended once more:
  \begin{compactenum}
    \item[6'.] [\dots] In the end, each vertex has to have one detector line and two field lines attached (element 9). \emph{Arrow directions must not collide.}
  \end{compactenum}
\end{CFD}
\noindent{}Unlike for the quantized real field, there is no need for introducing symmetry factors, since the two field operators underlying a vertex, $\Phi^{\dagger}(y_{i})\Phi(y_{i})$, are not identical (see \cref{sec: auxiliary diagrammatic complex quad}).

\subsubsection{Quadratic coupling to quantized spinor fields}\label{sec: feynman rules spinor quad}

For a quantized spinor field coupled to an UDW-type detector through \cref{eqn: interaction H spinor quad renorm}, the same elements appear, but are assigned different factors:

\begin{EFDMF}
  \hfill\\
   \includegraphics{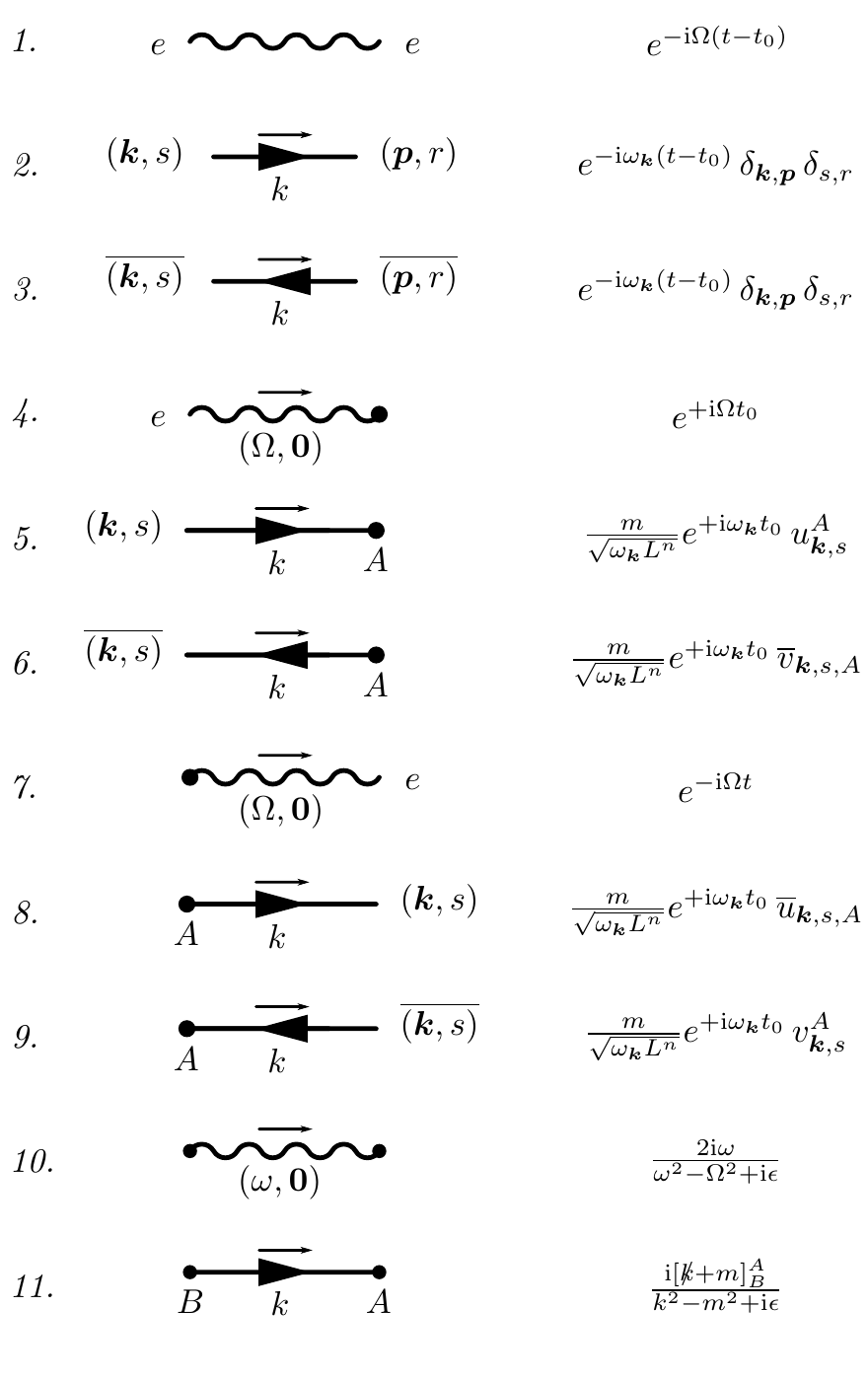}
\end{EFDMF}

When dealing with massless fields, the redefinition of the normalization factors in the mode functions (see \cref{app: spinor fields}) requires redefinition of part of the  Feynman rules just stated:

\vfill{}

\begin{EFDMLF}
  \hfill\\
  \includegraphics{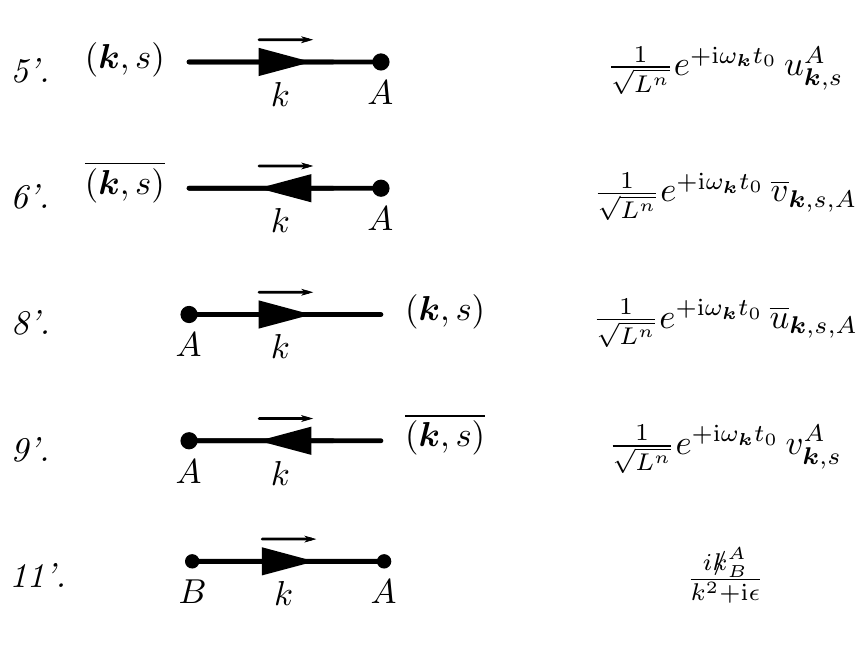}
\end{EFDMLF}

\medskip{}

\begin{CFD}
  For spinor fields, the steps from \cref{sec: feynman rules complex quad} are still applicable, but it is convenient keep track of spinor indices in addition:
  \begin{compactenum}
    \item[1'.] Draw a horizontal series of $k$ vertices, \emph{and label each one with a spinor index} $A_{i}$. These will be referred to as \emph{internal vertices}.
  \end{compactenum}
\end{CFD}
\noindent{}Note that these rules do \emph{not} provide the overall sign of the amplitude $\mathcal{A}^{(k)}_{i}$ of an individual diagram at order $k$. As for the auxiliary diagrams in \cref{sec: diagrammatic}, the sign is most easily determined by writing out the underlying contractions and reordering the operators explicitly.

\subsection{Feynman rules for multiple detectors}

For questions in relativistic quantum information theory, often two or more detectors coupled to the same quantum field $F$ are considered. Each detector $i$ has its own Hilbert space, and a monopole operator $\mu^{(i)}$ acting on it. For $N$ detectors, the  Hamiltonian of the entire system is
\begin{equation}
H = H_{F} + \sum_{i=1}^{N} H_{\mathrm{d}}^{(i)} + H_{\mathrm{int}}^{(i)}
\end{equation}
where each $H_{\mathrm{d}}^{(i)}$ is the Hamiltonian of a free detector as given in \cref{eqn: H detector}, and each detector by itself interacts with the field as before. The above Feynman rules readily generalize to such situations: each additional detector is equivalent to an additional monopole field $\mu^{(i)}(x)$. Because each of these has its own Hilbert space, quanta (and the corresponding lines in Feynman diagrams) of different detectors always need to be distinguished. Since there are only two fields (for linear coupling) or three fields (for quadratic coupling) present in each interaction Hamiltonian, there are still only two- or three-legged vertices, respectively. While direct interactions of the detectors do not occur in this model, they can of course interact indirectly through the quantum field they probe---\cref{fig: fr example two detectors} demonstrates such a process.

\begin{figure}
  \includegraphics{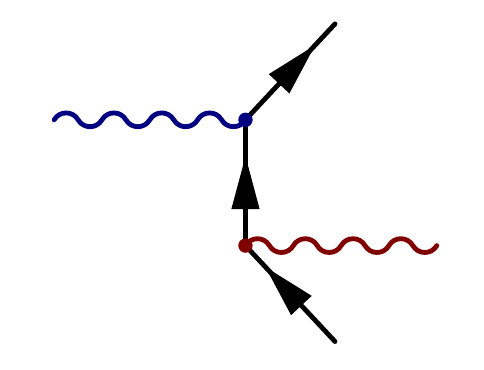}
  \caption{Example of a leading-order process involving two different detectors, blue and red (upper and lower wavy lines), quadratically interacting with a quantized complex or spinor field. Initially, the blue (upper) detector is excited. By exchange of a field quantum, the excitation is swapped to the red (lower) detector under emission of a particle-antiparticle pair.}
  \label{fig: fr example two detectors}
\end{figure}


\section{Conclusions and Outlook}

A key feature of Unruh-DeWitt detectors, i.e., of first quantized systems that are coupled to quantum fields, is that they can be localized in space and time. UDW detectors therefore provide an excellent tool for extracting spatiotemporal information from quantum fields, at least in gedanken experiments. So far, most studies have used the simplest kind of Unruh-DeWitt detectors, namely detectors that couple linearly to a bosonic real scalar field. The use of those Unruh-DeWitt detector models has yielded important results, for example about the Unruh and Hawking effects \cite{unruh_notes_1976,crispino_unruh_2008,Jormacircular,Keith2014,JormaJuarez}, and, more recently,  about the extraction of entanglement from the vacuum and communication via quantum fields \cite{jonsson_quantum_2014,Jonsson2015}. 

These successes  motivate the development of Unruh-DeWitt detector models also for the detection of quanta other than real scalar bosons, such as fermionic and complex quanta.
There are pioneering studies in this direction \cite{takagi_response_1985,takagi_vacuum_1986,hinton_particle_1984,iyer_detection_1980}, but so far they have not gone further than the calculation of transition rates. Indeed, we found that a direct application of these models beyond the computation of transition rates leads to divergencies that cannot be regularized in the usual fashion, i.e., through spatial smearing and smooth switching functions.    

Here, we  introduced improved Unruh-DeWitt type detector models that cure these divergencies by means of renormalization. The new UDW detector models therefore possess the full versatility of the usual bosonic Unruh-DeWitt detectors. For example, they now allow for the investigation of the entanglement structure of the fermionic vacuum.

To this end, we started by analyzing the four different particle detector models summarized in \cref{tab: overview detector models}. 
We then showed that, without further treatment, the vacuum excitation probability is already divergent at leading order in perturbation theory for the models 2-4, which are coupling quadratically in the field. These divergences cannot  be regularized by the introduction of a spacetime profile (switching function and spatial smearing), unlike divergencies appearing in the case of the traditional linearly coupled UDW detector \cite{grove_notes_1983,sriramkumar_finite-time_1996,louko_transition_2008,takagi_vacuum_1986}. 

We found the origin of these divergences to be the fact that the quadratic coupling gives rise to loop diagrams (see left side of \cref{fig: VEP processes 1st order})  which are the exact equivalent of tadpole diagrams in quantum electrodynamics (QED) \cite{greiner_quantum_2008}, with the detector playing the role of the electromagnetic field. Motivated by the analogy to QED we renormalized models 2-4 at leading order in perturbation theory by normal-ordering the interaction Hamiltonian. Physically, this corresponds to setting the (otherwise infinite) interaction energy to zero.

Finally, we derived Feynman rules for all four detector models in \cref{sec: feynman rules}. The Feynman rules allow for the systematic, graph-based calculation and renormalization of, for example, transition probabilities to arbitrary order, in the  way familiar from particle physics.

With the Feynman rules for the various detector models at hand, one can now straightforwardly study the behaviors not only of bosonic but also of fermionic fields, for example, concerning the communication between two first quantized quantum systems (UDW detectors) via a quantum field, or the harvesting or farming of entanglement from the vacuum. In particular, the latter may now also allow one to develop quantitative tools to measure entanglement in the fermionic vacuum. 

There are various further promising possibilities, in particular for applying the new renormalized detector model for spinor fields. Among others, this detector model can now be used to study the fermionic Unruh effect beyond the detector transition rates. This should allow one to study the universality and the thermalization of a detector in the fermionic Unruh effect \cite{martin-martinez_Unruh2013}  and to study possible finite time effects (see for instance \cite{Anti_Unruh2013}). 

We also addressed the question of which bosonic detector model allows for a fair comparison with a fermionic model. We found that the leading order behavior of the renormalized models 3 and 4 can be compared on equal footing. 

Finally, while UDW detectors are of great value for gedanken experiments, let us recall that the UDW detector for bosonic fields can be viewed as an idealization of a realistic system, such as the lowest two energy levels of an atom or molecule \cite{martin-martinez_wavepacket_2013,Alhambra2013}. The atom detects photons by internal excitation that can later be measured via standard experimental procedures. It should be interesting to explore to what extent the UDW detector for fermionic fields, model 4 in \cref{tab: overview detector models}, can also be viewed as a simple idealization of a realistic system, such as the optical cavity envisioned in the introduction, which would detect fermions through electromagnetic mode excitations. More generally, it should be interesting to further investigate how the proposed detector models along with their mode selection and switching functions could find experimental realizations.

\appendix


\section{Notes on quantum fields}\label{app: quantum fields}

In this section, we summarize results of the classical and quantum field theory of scalar fields on Minkowski spacetime, as well as spinor fields on $\R^{1,1}$ and $\R^{1,3}$. All fields are spatially restricted to the cavity
\begin{equation}
  B = [-\nicefrac{L}{2},\nicefrac{L}{2}]^{n} \subset \R^{n}
\end{equation}
with periodic boundary conditions
\begin{equation}
  \Phi(t,\vec{x} + \vec{\delta}) = \Phi(t,\vec{x}) \quad\quad \forall\,\vec{\delta} \in L\cdot\Z^{n}
\end{equation}
The objective is to establish notation and provide a reference basis for the main body of this paper. These results are obtained in close analogy to the standard treatment of unrestricted fields, see e.g.\,\cite{bjorken_relativistic_1965,peskin_introduction_1995}.

\subsection{Scalar fields}

\subsubsection{Results in classical theory}

Consider a real or complex field $\Phi$ on Minkowski spacetime. The dynamics of this field are determined by the Lagrangian densities
\begin{equation}\begin{aligned}\label{eqn: lagrange density real}
  \mathcal{L}_{\Phi} &= \frac{1}{2} (\partial_{\mu}\Phi)(\partial^{\mu}\Phi) - \frac{1}{2}m^{2}\Phi^{2}\\
  \mathcal{L}_{\Phi}^{c} &= (\partial_{\mu}\Phi^{*})(\partial^{\mu}\Phi) - m^{2}\Phi^{*}\Phi
\end{aligned}\end{equation}
(for the real and the complex field, respectively), which both yield Klein-Gordon equation
\begin{equation}\label{eqn: Klein-Gordon equation}
  (\partial_{\mu}\partial^{\mu}+m^{2})\Phi(t,\vec{x}) = 0
\end{equation}
as equation of motion. A solution space to the Klein-Gordon equation is spanned by the family of mode functions
\begin{equation}\label{eqn: scalar mode functions 1}
  \tilde\varphi_{\vec{k}}(t,\vec{x}) = \frac{1}{\sqrt{2\omega_{\vec{k}}L^{n}}} \,e^{-\ii  k\cdot x}~,
\end{equation}
where the momentum takes the following values:
\begin{align}\label{eqn: allowed momenta}
  \vec{k} &\in 2\pi/L \cdot \Z^{n}~.
\end{align}
The energy $\omega_{\vec{k}}$ is determined by the dispersion relation
\begin{equation}\label{eqn: dispersion relation}
  \omega_{\vec{k}} = \sqrt{|\vec{k}|^{2}+m^{2}}~.
\end{equation}
These mode functions have been normalized according to the condition
\begin{equation}\label{eqn: normalization scalar mode functions}
  \varinner{\tilde\varphi_{\vec{k}}}{\tilde\varphi_{\vec{p}}} =  \delta_{\vec{k},\vec{p}}
\end{equation}
with respect to the Klein-Gordon product
\begin{multline}\label{eqn: Klein-Gordon product}
  \varinner{\Phi_{1}}{\Phi_{2}}\define  -\ii   \int_{B}  \Phi_{1}(t,\vec{x}) \left[\partial_{0}\Phi_{2}^{*}(t,\vec{x})\right]\\
   - \left[\partial_{0}\Phi_{1}(t,\vec{x})\right]\Phi_{2}^{*}(t,\vec{x}) \dd \vec{x}~.
\end{multline}

It is useful to define purely spatial mode functions
\begin{equation}\label{eqn: scalar mode functions 2}
  \varphi_{\vec{k}}(\vec{x}) = \frac{1}{\sqrt{L^{n}}}~ e^{ \ii   \vec{k}\cdot \vec{x}}~.
\end{equation}
These are orthonormal with respect to the $L^{2}$-scalar product
\begin{equation}
  \inner{\varphi_{\vec{k}}}{\varphi_{\vec{p}}} = \int_{B} \varphi_{\vec{k}}^{*}(\vec{x}) \varphi_{\vec{p}}(\vec{x}) \dd \vec{x} = \delta_{\vec{k},\vec{p}}~.
\end{equation}
To avoid cluttered notation, we will sometimes drop the tilde distinguishing time-dependent from time-independent mode functions.

\subsubsection{Elements of canonical quantization}

Canonical quantization  of the free real field leads to the mode expansion
\begin{equation}\label{eqn: mode expansion real field}
  \op\Phi(t,\vec{x}) = \sum_{\vec{k}}\op{a}_{\vec{k}}\,\tilde\varphi_{\vec{k}}(t,\vec{x})
  + \op{a}^{\dagger}_{\vec{k}}\,\tilde\varphi^{*}_{\vec{k}}(t,\vec{x})
\end{equation}
for the Heisenberg picture field operator in terms of the annihilators $\op{a}_{\vec{k}}$ and creators $\op{a}^{\dagger}_{\vec{k}}$. In the Schrödinger picture,
\begin{equation}\label{eqn: SP mode expansion real field}
  \op\Phi(\vec{x}) = \sum_{\vec{k}} \frac{1}{\sqrt{2\omega_{\vec{k}}}} \Big[ \op{a}_{\vec{k}}\,\varphi_{\vec{k}}(\vec{x})
  + \op{a}^{\dagger}_{\vec{k}}\,\varphi^{*}_{\vec{k}}(\vec{x})\Big]~.
\end{equation}
The field operator is self-adjoined, $\op{\Phi}^{\dagger} = \op{\Phi}$ , and the ladder operators satisfy the canonical commutation relations
\begin{equation}\begin{aligned} \label{eqn: cr ladder op real field}
  [\op{a}_{\vec{k}},\op{a}^{\dagger}_{\vec{p}}] &= \id\,\delta_{\vec{k},\vec{p}}\\
  [\op{a}_{\vec{k}},\op{a}_{\vec{p}}] &= 0 = [\op{a}_{\vec{k}}^{\dagger},\op{a}^{\dagger}_{\vec{p}}]~.
\end{aligned}\end{equation}
After renormalizing the vacuum energy by normal-ordering, the dynamics of the system are determined by the Hamilton operator
\begin{align}\label{eqn: hamilton operator real field}
  \op{H}_{\Phi} = \sum_{\vec{k}} \omega_{\vec{k}} \op{a}^{\dagger}_{\vec{k}}\op{a}_{\vec{k}}~.
\end{align}
Particle number eigenstates are created by acting on the vacuum state with creation operators 
\begin{equation}\label{eqn: number eigenstates real}
  \ket{(n_{\vec{k}})_\vec{k}} =  \prod_{\vec{k}}\frac{1}{\sqrt{n_{\vec{k}}!}}(\op{a}_{\vec{k}}^{\dagger})^{n_{\vec{k}}}|0\rangle~.
\end{equation}
Creation operators commute, so their ordering is not relevant.

Quantization of the free complex field leads to the mode expansion
\begin{multline}\label{eqn: mode expansion complex field}
  \op\Phi(t,\vec{x}) = \sum_{\vec{k}} \frac{1}{\sqrt{2\omega_{\vec{k}}}} \Big[ \op{a}_{\vec{k}}\,e^{-\ii  \omega_{\vec{k}}t}\,\varphi_{\vec{k}}(\vec{x})  \\
   + \op{b}_{\vec{k}}^{\dagger}\,e^{+\ii  \omega_{\vec{k}}t}\,\varphi_{\vec{k}}^{*}(\vec{x})\Big]
\end{multline}
in the Heisenberg picture. The field operator is \emph{not} self-adjoint, the ladder operators satisfy the canonical commutation relations 
\begin{align}\label{eqn: cr ladder op complex field}
  [\op{a}_{\vec{k}},\op{a}^{\dagger}_{\vec{p}}] &= \id\, \delta_{\vec{k},\vec{p}}~, & [\op{b}_{\vec{k}},\op{b}^{\dagger}_{\vec{p}}] &= \id\,\delta_{\vec{k},\vec{p}}~,
\end{align}
and all other commutators not yet fixed by these relations vanish. The Hamilton operator of the free complex field is
\begin{equation}\label{eqn: hamilton operator complex field}
  \op{H} = \sum_{\vec{k}} \omega_{\vec{k}} \left( \op{a}_{\vec{k}}^{\dagger}\op{a}_{\vec{k}} + \op{b}_{\vec{k}}^{\dagger} \op{b}_{\vec{k}}\right)~,
\end{equation}
and particle number eigenstates are again generated by applying the creation operators on the vacuum. The complex field has $U(1)$ charge, so we need to distinguish between particle states created by $\op{a}$ and antiparticle states created by $\op{b}$:
\begin{multline}\label{eqn: number eigenstates complex}
  \ket{(n_{\vec{k}})_\vec{k},(\bar{n}_{\vec{p}})_\vec{p}} =\\
    \prod_{\vec{k}}\frac{1}{\sqrt{\vphantom{\bar{n}_{\vec{k}}}n_{\vec{k}}!}}(\op{a}_{\vec{k}}^{\dagger})^{n_{\vec{k}}}    \prod_{\vec{p}}\frac{1}{\sqrt{\bar{n}_{\vec{p}}!}}(\op{b}_{\vec{p}}^{\dagger})^{\bar{n}_{\vec{p}}}|0\rangle~.
\end{multline}

\subsection{Spinor fields}\label{app: spinor fields}

\subsubsection{Results in classical theory}

Let $\Psi$ be a spinor field on Minkowski spacetime $\R^{1,1}$ or $\R^{1,3}$, taking values in $\C^{4}$ as a representation space of the Clifford algebras $\Cliff(1,1)$ or $\Cliff(1,3)$, respectively. In $(1,3)$ dimensions we choose the irreducible Dirac representation, generated by the Dirac matrices
\begin{align}
  \gamma^{0} &= \begin{bmatrix} \id_{2} & 0 \\ 0 & -\id_{2} \end{bmatrix}~, & \gamma^{i} &= \begin{bmatrix} 0 & \sigma_{i} \\ -\sigma_{i} & 0 \end{bmatrix}~,
\end{align}
where $\id_{2}$ is the two-dimensional identity matrix, and
\begin{align}
  \sigma_{1} &= \begin{pmatrix} 0 & 1 \\ 1 & 0 \end{pmatrix}~, & \sigma_{2} &= \begin{pmatrix} 0 & -\ii   \\ \ii   & 0 \end{pmatrix}~, & \sigma_{3} &= \begin{pmatrix} 1 & 0 \\ 0 & -1 \end{pmatrix}
\end{align}
are the Pauli matrices. The Dirac matrices satisfy the anticommutation relations
\begin{equation*}
  \{\gamma^{\mu},\gamma^{\nu}\} = 2 \eta^{\mu\nu}\id_{2}
\end{equation*}
for $\mu, \nu \in \{0,1,2,3\}$, and therefore generate a representation of the Clifford algebra $\Cliff(1,3)^{\C}$ on $\C^{4}$. In $(1,1)$ dimensions, we use $\gamma^{0}$ and $\gamma^{3}$ to generate a (reducible) representation of $\Cliff(1,1)$ on $\C^{4}$ very similar to the Dirac representation.

The dynamics of the field are governed by the Lagrange density
\begin{equation}\label{eqn: lagrange density spinor}
  \mathcal{L}_{\Psi} = \conj\Psi (i\gamma^{\mu}\partial_{\mu}-m)\Psi{}~,
\end{equation}
leading to the Dirac equation:
\begin{equation}\label{eqn: Dirac equation}
  (i\gamma^{\mu}\partial_{\mu}-m)\Psi(t,\vec{x}) = 0~.
\end{equation}
Here $\gamma^{\mu}$ are the Dirac matrices. Solutions to the Dirac equation are linear combinations of the mode functions
\begin{equation}\label{eqn: spinor mode functions}
  \tilde\psi_{\vec{k},s,\epsilon}(t,\vec{x}) = \psi_{\vec{k},s,\epsilon}(\vec{x})\,e^{-\epsilon i\omega_{\vec{k}}t}~,
\end{equation}
where $s \in \{-\nicefrac{1}{2},\nicefrac{1}{2}\}$ describes the spin degree of freedom, and $\vec{k}$ and $\epsilon$ can take the same values as in \cref{eqn: allowed momenta}. These mode functions are taken to be normalized according to the condition
\begin{equation}\label{eqn: normalization spinor mode functions}
  \inner{\tilde\psi_{\vec{k},s,\epsilon{}}}{\tilde\psi_{\vec{p},r,\delta}} = \delta_{\vec{k},\vec{p}}\,\delta_{s,r}\,\delta_{\epsilon,\delta} = \inner{\psi_{\vec{k},s,\epsilon{}}}{\psi_{\vec{p},r,\delta}}~,
\end{equation}
where we use the (generalized) scalar product
\begin{equation}\label{eqn: spinor inner product}
  \inner{\Psi_{1}}{\Psi_{2}} \define \int_{B} \Psi_{1}^{\dagger}(t,\vec{x})\Psi_{2}(t,\vec{x})\dd \vec{x}~.
\end{equation}
Again, to lighten notation, we will sometimes drop the tilde if safely possible.

The solution of the Dirac equation in the cavity $B$ can be accomplished in close analogy to the case of a free spinor field, a detailed treatment of the latter can e.g.\,be found in \cite{bjorken_relativistic_1965}. The actual form and the properties of the mode functions depend on the spatial dimension as well as on the mass parameter $m$.

\paragraph{(1,3) dimensions}
In case of a finite mass $m \neq 0$, the mode functions in three spatial dimensions can be decomposed as
\begin{equation}\label{eqn: spinor mode function 4dim massive}
  \psi_{\vec{k},s,\epsilon}(\vec{x}) = \sqrt{\frac{m}{\omega_{\vec{k}}L^{3}}}\,u_{\vec{k},s,\epsilon}\,e^{+\epsilon \ii   \vec{k}\vec{x}}~,
\end{equation}
where the spatial dependence has been separated from the spinorial part. For the latter, we will use the well-established notation $u_{\vec{k},s,\epsilon=+1} 
= u_{\vec{k},s}$ and $u_{\vec{k},s,\epsilon=-1} = v_{\vec{k},s}$. Explicitly, they are
\begin{equation}\begin{aligned}\label{eqn: spinors 4dim massive}
  u_{\vec{k},s} = \sqrt{\frac{\omega_{\vec{k}}+m}{2m}}
  \begin{bmatrix}
    \xi_{s}\\
    \frac1{\omega_{\vec{k}}+m} \left(\sum_{i=1}^{3}\sigma^{i}k^{i}\right)\xi_{s}
  \end{bmatrix}\\
  v_{\vec{k},s} = \sqrt{\frac{\omega_{\vec{k}}+m}{2m}}
  \begin{bmatrix}
    \frac1{\omega_{\vec{k}}+m} \left(\sum_{i=1}^{3}\sigma^{i}k^{i}\right)\xi_{s}\\
    \xi_{s}
  \end{bmatrix}~,
\end{aligned}\end{equation}
where $\sigma^{i}$ are the Pauli matrices, and
\begin{align}
  \xi_{\nicefrac{1}{2}} &= \begin{pmatrix}1 \\ 0\end{pmatrix}~, & \xi_{-\nicefrac{1}{2}} &= \begin{pmatrix}0 \\ 1\end{pmatrix}
\end{align}
are eigenvectors of $\sigma^{3}$. Multiplying two spinors yields
\begin{multline}
  \conj{u}_{\vec{k},s}\,u_{\vec{p},r} = \frac{1}{2m} \Bigg[\sqrt{(\omega_{\vec{k}}+m)(\omega_{\vec{p}}+m)} \\
  -\frac{ 2si(k^{1}p^{2}-k^{2}p^{1})+\vec{k}\cdot \vec{p}}{\sqrt{(\omega_{\vec{k}}+m)(\omega_{\vec{p}}+m)}}\Bigg]\delta_{s,r}\\
  + \frac{1}{2m}\frac{2s (k^{1}p^{3}-k^{3}p^{1}) -\ii  (k^{2}p^{3}-k^{3}p^{2})}{\sqrt{(\omega_{\vec{k}}+m)(\omega_{\vec{p}}+m)}}\delta_{s,-r}
\end{multline}
\begin{multline}\label{eqn: spinor prod 4dim massive}
  \conj{u}_{\vec{k},s}\,v_{\vec{p},r} = \frac{\sqrt{(\omega_{\vec{k}}+m)(\omega_{\vec{p}}+m)}}{2m}\\
  \times \Bigg[\left( \frac{2sp^{3}}{\omega_{\vec{p}}+m}-\frac{2sk^{3}}{\omega_{\vec{k}}+m}\right)\delta_{s,r}\\
  +\left( \frac{p^{1}-2sip^{2}}{\omega_{\vec{p}}+m}-\frac{k^{1}-2sik^{2}}{\omega_{\vec{k}}+m}\right)\delta_{s,-r}\Bigg]
\end{multline}
and
\begin{equation}
  \conj{v}_{\vec{k},s}\,v_{\vec{p},r} = - \conj{u}_{\vec{k},s}\,u_{\vec{p},r}
\end{equation}
\begin{equation}
  \conj{v}_{\vec{k},s}\,u_{\vec{p},r} =  - \left(\conj{u}_{\vec{k},s}\,v_{\vec{p},r}\right)^{*} ~.
\end{equation}
The bar indicates the Dirac conjugate $\bar{u} = u^{\dagger}\gamma^{0}$ as usual.
In particular, the normalization is chosen such that for $\vec{k} = \vec{p}$ all products simplify to
\begin{equation}\label{eqn: spinor prod 4dim massive 2}
  \conj{u}_{\vec{k},s,\epsilon}\,u_{\vec{k},r,\delta} = \epsilon\,\delta_{s,r}\,\delta_{\epsilon,\delta}~.
\end{equation}
Moreover, products with the inverse order of spinors (such that they multiply to a $4\times4$ matrix) can be brought into a compact form as well:
\begin{equation}\begin{aligned}\label{eqn: inverse spinor-bar-prod 3dim massive}
  \sum_{s}u_{\vec{k},s}^{A}\conj{u}_{\vec{k},s,B} = \frac{1}{2m} \left[k \cdot \gamma +\id_{4}m \right]^{A}_{B}\\
  \sum_{s}v_{\vec{k},s}^{A}\conj{v}_{\vec{k},s,B} = \frac{1}{2m} \left[k \cdot \gamma -\id_{4}m \right]^{A}_{B}~.
\end{aligned}\end{equation}

For vanishing mass $m = 0$, the mode functions
\begin{equation}\label{eqn: spinor mode function 4dim massless}
  \psi_{\vec{k},s,\epsilon}(\vec{x}) = \frac{1}{\sqrt{L^{3}}}\,u_{\vec{k},s,\epsilon}\,e^{+\epsilon i\vec{k}\cdot\vec{x}}
\end{equation}
are simply the limit $m \to 0$ of \cref{eqn: spinor mode function 4dim massive}, but the normalization of the spinor parts is changed:
\begin{equation}\begin{aligned}\label{spinors 3dim massless periodic-bc}
  u_{\vec{k},s} = \frac{1}{\sqrt{2}}
  \begin{bmatrix}
    \xi_{s}\\
    \frac{1}{|\vec{k}|} \left(\sum_{i=1}^{3}\sigma^{i}k^{i}\right)\xi_{s}
  \end{bmatrix}\\
  v_{\vec{k},s} =\frac{1}{\sqrt{2}}
  \begin{bmatrix}
    \frac{1}{|\vec{k}|} \left(\sum_{i=1}^{3}\sigma^{i}k^{i}\right)\xi_{s}\\
    \xi_{s}
  \end{bmatrix}  ~.
\end{aligned}\end{equation}
This ensures that products of spinors are still finite. Explicitly,
\begin{multline}
  \conj{u}_{\vec{k},s}\,u_{\vec{p},r} =\\
  \frac{1}{2} \left[1-\frac{1}{|\vec{k}||\vec{p}|} \left( 2si(k^{1}p^{2}-k^{2}p^{1})+\vec{k}\cdot\vec{p}\right)\right]\delta_{s,r}\\
  + \frac{1}{2}\frac{1}{|\vec{k}||\vec{p}|} \left[ 2s (k^{1}p^{3}-k^{3}p^{1}) -\ii  (k^{2}p^{3}-k^{3}p^{2})\right]\delta_{s,-r}
\end{multline}
\begin{multline}\label{eqn: spinor prod 4dim massless}
  \conj{u}_{\vec{k},s}\,v_{\vec{p},r} =
   \frac{1}{2} \left( \frac{2sp^{3}}{|\vec{p}|}-\frac{2sk^{3}}{|\vec{k}|}\right)\delta_{s,r}\\
   \frac{1}{2} \left( \frac{p^{1}-2sip^{2}}{|\vec{p}|}-\frac{k^{1}-2sik^{2}}{|\vec{k}|}\right)\delta_{s,-r}
\end{multline}
\begin{equation}
  \conj{v}_{\vec{k},s}\,v_{\vec{p},r} = - \conj{u}_{\vec{k},s}\,u_{\vec{p},r}
\end{equation}
\begin{equation}
  \conj{v}_{\vec{k},s}\,u_{\vec{p},r} = - \left(\conj{u}_{\vec{k},s}\,v_{\vec{p},r}\right)^{*}~,
\end{equation}
and for $\vec{k} = \vec{p}$ all products vanish:
\begin{align}\label{eqn: spinor prod 4dim massless 2}
  \conj{u}_{\vec{k},s,\epsilon}\,u_{\vec{k},r,\delta} = 0~.
\end{align}
Note that these products are \emph{not} simply the limits $m \to 0$ of \cref{eqn: spinor prod 4dim massive,eqn: spinor prod 4dim massive 2}. Finally,
\begin{equation}\begin{aligned}\label{eqn: inverse spinor-bar-prod 3dim massless}
  \sum_{s}u_{\vec{k},s}^{A}\conj{u}_{\vec{k},s,B} = \frac{1}{2|\vec{k}|} \left[k \cdot \gamma \right]^{A}_{B}\\
  \sum_{s}v_{\vec{k},s}^{A}\conj{v}_{\vec{k},s,B} = \frac{1}{2|\vec{k}|} \left[k \cdot \gamma \right]^{A}_{B}~.
\end{aligned}\end{equation}

\paragraph{(1,1) dimension}
For finite mass $m \neq 0$, the mode functions can be decomposed as
\begin{equation}\label{eqn: spinor mode function 2dim massive}
  \psi_{k,s,\epsilon}(x) = \sqrt{\frac{m}{\omega_{k}L}}\,u_{k,s,\epsilon}\,e^{+\epsilon \ii   kx}~,
\end{equation}
where the spinorial parts are
\begin{equation}\begin{aligned}
  u_{k,s} = \sqrt{\frac{\omega_{k}+m}{2m}}
  \begin{bmatrix}
    \xi_{s}\\
    \frac{k}{\omega_{k}+m} \sigma^{3}\xi_{s}
  \end{bmatrix}\\
  v_{k,s} = \sqrt{\frac{\omega_{k}+m}{2m}}
  \begin{bmatrix}
    \frac{k}{\omega_{k}+m} \sigma^{3}\xi_{s}\\
    \xi_{s}
  \end{bmatrix}~.
\end{aligned}\end{equation}
Here, $\sigma^{3}$ is the third Pauli matrix. The products of spinors are
\begin{multline}
  \conj{u}_{k,s}\,u_{p,r} = \frac{\sqrt{(\omega_{k}+m)(\omega_{p}+m)}}{2m}\\
  \times \left[1-\frac{kp}{(\omega_{k}+m)(\omega_{p}+m)}\right]\delta_{s,r}
\end{multline}
\begin{multline}\label{eqn: spinor prod 2dim massive}
  \conj{u}_{k,s}\,v_{p,r} = \frac{\sqrt{(\omega_{k}+m)(\omega_{p}+m)}}{2m}\\
  \times \left( \frac{2sp}{\omega_{p}+m}-\frac{2sk}{\omega_{k}+m}\right)\delta_{s,r}
\end{multline}
\begin{equation}
  \conj{v}_{k,s}\,v_{p,r} = - \conj{u}_{k,s}\,u_{p,r}
\end{equation}
\begin{equation}
  \conj{v}_{k,s}\,u_{p,r} = -\conj{u}_{k,s}\,v_{p,r}~.
\end{equation}
For $k = p$, all relations simplify to
\begin{equation}\label{eqn: spinor prod 2dim massive 2}
  \conj{u}_{k,s,\epsilon}\,u_{k,r,\delta} = \epsilon\,\delta_{s,r}\,\delta_{\epsilon,\delta}
\end{equation}
as before. Finally,
\begin{equation}\begin{aligned}\label{eqn: inverse spinor-bar-prod 1dim massive}
  \sum_{s}u_{\vec{k},s}^{A}\conj{u}_{\vec{k},s,B} = \frac{1}{2m} \left[k_{\mu} \gamma^{\mu} +\id_{4}m \right]^{A}_{B}\\
  \sum_{s}v_{\vec{k},s}^{A}\conj{v}_{\vec{k},s,B} = \frac{1}{2m} \left[k_{\mu} \gamma^{\mu} -\id_{4}m \right]^{A}_{B}~,
\end{aligned}\end{equation}
where here $k_{\mu} \gamma^{\mu} = k^{0}\gamma^{0}-k^{3}\gamma^{3} \equiv \omega_{\vec{k}}\gamma^{0}-k\gamma^{3}$.

In case of vanishing mass $m = 0$, the mode functions can be decomposed as
\begin{equation}\label{eqn: spinor mode function 2dim massless}
  \psi_{k,s,\epsilon}(x) = \frac{1}{\sqrt{L}}\,u_{k,s,\epsilon}\,e^{+\epsilon ikx}
\end{equation}
with spinors
\begin{equation}\begin{aligned}\label{eqn: spinors 2dim massless}
   u_{k,s} = \frac{1}{\sqrt{2}}
  \begin{bmatrix}
    \xi_{s}\\
    \sgn k \,\sigma^{3}\xi_{s}
  \end{bmatrix}\\
  v_{k,s} = \frac{1}{\sqrt{2}}
  \begin{bmatrix}
    \sgn k \,\sigma^{3}\xi_{s}\\
    \xi_{s}
  \end{bmatrix}~.
\end{aligned}\end{equation}
The products then read
\begin{equation}\begin{aligned}\label{eqn: spinor prod 2dim massless}
  \conj{u}_{k,s}\,u_{p,r} &= \frac{1}{2}[1-\sgn(kp)]\,\delta_{s,r}\\
  \conj{u}_{k,s}\,v_{p,r} &= s [\sgn p-\sgn k]\,\delta_{s,r}\\
  \conj{v}_{k,s}\,v_{p,r} &= - \conj{u}_{k,s}\,u_{p,r}\\
  \conj{v}_{k,s}\,u_{p,r} &= -\conj{u}_{k,s}\,v_{p,r}~,
\end{aligned}\end{equation}
and for $k = p$ all products vanish like in four dimensions
\begin{equation}\label{eqn: spinor prod 2dim massless 2}
  \conj{u}_{k,s,\epsilon}\,u_{k,r,\delta} = 0~.
\end{equation}
Lastly,
\begin{equation}\begin{aligned}\label{eqn: inverse spinor-bar-prod 1dim massless}
  \sum_{s}u_{\vec{k},s}^{A}\conj{u}_{\vec{k},s,B} = \frac{1}{2|k|} \left[k_{\mu} \gamma^{\mu} \right]^{A}_{B}\\
  \sum_{s}v_{\vec{k},s}^{A}\conj{v}_{\vec{k},s,B} = \frac{1}{2|k|} \left[k_{\mu} \gamma^{\mu} \right]^{A}_{B}~.
\end{aligned}\end{equation}

\subsubsection{Elements of canonical quantization}

Canonical quantization of the free spinor field results in the mode expansion
\begin{equation}\label{eqn: mode expansion spinor field}
  \op\Psi(t,\vec{x}) =\sum_{\vec{k},s} \op{a}_{\vec{k},s} \tilde\psi_{\vec{k},s,+}(t,\vec{x}) + \op{b}^{\dagger}_{\vec{k},s} \tilde\psi_{\vec{k},s,-}(t,\vec{x})
\end{equation}
in the Heisenberg picture, or
\begin{equation}\label{eqn: SP mode expansion spinor field}
  \op\Psi(\vec{x}) =\sum_{\vec{k},s} \op{a}_{\vec{k},s} \psi_{\vec{k},s,+}(\vec{x}) + \op{b}^{\dagger}_{\vec{k},s} \psi_{\vec{k},s,-}(\vec{x})
\end{equation}
in the Schrödinger picture.
Here, $\op{a}_{\vec{k},s}$ and $\op{b}_{\vec{k},s}$ are particle and antiparticle ladder operators. They satisfy the canonical anticommutation relations
\begin{align}\label{eqn: acr ladder op spinor field}
  \{\op{a}_{\vec{k},s},\op{a}^{\dagger}_{\vec{p},r}\} &= \delta_{\vec{k},\vec{p}}\delta_{s,r}~, & 
  \{\op{b}_{\vec{k},s},\op{b}^{\dagger}_{\vec{p},r}\} &= \delta_{\vec{k},\vec{p}}\delta_{s,r}~,
\end{align}
and all yet undetermined relations vanish. Again after renormalizing the vacuum energy, the Hamiltonian of the free spinor field is
\begin{equation}\label{eqn: hamilton operator spinor field}
  \op{H}_{\Psi} = \sum_{\vec{k},s} \omega_{\vec{k}} \,\left(\op{a}^{\dagger}_{\vec{k},s}\op{a}_{\vec{k},s} + \op{b}_{\vec{k},s}^{\dagger}\op{b}_{\vec{k},s} \right) ~.
\end{equation}
As before, particle number eigenstates are generated from the vacuum by application of creation operators. However, the order of the creation operators is now relevant, because they anticommute.


\section{On the normalization of field modes}\label{app: normalization}

We briefly sketch why the mode functions of the scalar field are normalized with respect to \cref{eqn: normalization 1}, and those of the spinor field according to \cref{eqn: normalization 2}.

The underlying reason is that the Lagrange density of the free scalar field $\Phi$ given in \cref{eqn: lagrange density real} contains a double time derivative of the field, which leads to the conjugate momentum $\Pi =  \partial_{0}\Phi$. The Lagrange density \cref{eqn: lagrange density spinor} for spinor fields $\Psi$, on the other hand, only contains a single time derivative of the field, and leads to the conjugate momentum $\Pi = \ii   \Psi^{\dagger}$. This in turn makes said normalization conditions necessary to ensure that the canonical (anti)commutation relations of the ladder operators are met. To make the connection evident, assume the modified relations
\begin{align}\label{eqn: ccr ladder op mod}
  [\op{a}_{\vec{k}},\op{a}^{\dagger}_{\vec{p}}] &= \alpha\,\delta_{\vec{k},\vec{p}}~, & \{\op{a}_{\vec{k},s},\op{a}_{\vec{p},r}^{\dagger}\} &= \alpha\,\delta_{\vec{k},\vec{p}}\,\delta_{s,r}~,
\end{align}
and so on.

\subsection{Scalar field}

To see how the normalization arises naturally, quantize the scalar field directly in the Heisenberg picture. The task is to find a linear operator $\op\Phi(x)$ which
\begin{inparaenum}[(a)]
  \item solves the Klein-Gordon equation
  \begin{equation}
    (\dalembert +m^{2})\op\Phi(x) = 0
  \end{equation}
  as an operator equation,
  \item satisfies the usual commutation relations with its canonical conjugate field $\Pi = \partial_{0}\Phi$, namely
  \begin{equation}
    [\Phi(x),\Pi(x)]_{|_{x_{0}=y_{0}}} = \ii   \delta(\vec{x}-\vec{y})~,
  \end{equation}
  and \item is self-adjoint:
  \begin{align}
    \op\Phi^{\dagger}(x) &= \op\Phi(x)~, & \op\Pi^{\dagger}(x) &= \op\Pi(x)~.
  \end{align}
\end{inparaenum}

In the first step, assume that $\op\Phi$ can be expanded in a Fourier series using the mode functions \cref{eqn: scalar mode functions 2}:
\begin{align}
  \op\Phi(x) &= \sum_{\vec{k}} \op\Phi_{\vec{k}}(t)\,\varphi_{\vec{k}}(\vec{x})~, & \op\Pi(x) &= \sum_{\vec{k}} \op\Pi_{\vec{k}}(t)\,\varphi_{\vec{k}}(\vec{x})~,
\end{align}
where $\Pi_{\vec{k}}(t) = \partial_{0}\Phi_{\vec{k}}(t)$. In the Fourier domain, the Klein-Gordon equation becomes
\begin{equation}
  (\partial_{0}^{2}+\omega_{\vec{k}}^{2})\op\Phi_{\vec{k}}(t)
\end{equation}
and the canonical commutation relations imply
\begin{equation}\label{eqn: Fourier space ccr}
  [\op\Phi_{\vec{k}}(t),\op\Pi_{\vec{p}}^{\dagger}(t)] = \ii   \delta_{\vec{k},\vec{p}}~,
\end{equation}
where we have already exploited that
\begin{align}
  \op\Phi_{\vec{k}}^{\dagger}(t) &= \op\Phi_{-\vec{k}}(t)~, &     \op\Pi_{\vec{k}}^{\dagger}(t) &= \op\Pi_{-\vec{k}}(t)~.
\end{align}
Solving the problem now amounts to finding $\op\Phi_{\vec{k}}(t)$.

The next step is to make an ansatz for $\op\Phi_{\vec{k}}$:
\begin{equation}
  \op\Phi_{\vec{k}}(t) = u_{\vec{k}}(t)\,\op{a}_{\vec{k}} + u_{-\vec{k}}^{*}(t) \,\op{a}_{-\vec{k}}^{\dagger}~.
\end{equation}
The $u_{\vec{k}}(t)$ are simply functions, while $\op{a}_{\vec{k}}$ are linear operators. The Klein-Gordon equation reduces to 
\begin{equation}
  (\partial_{0}^{2}+\omega_{\vec{k}}^{2})\,u_{\vec{k}}(t) = 0~,
\end{equation}
which means that $u_{\vec{k}}$ is a linear combination of $e^{\pm \ii   \omega_{\vec{k}}t}$. Requiring the commutation relations \cref{eqn: ccr ladder op mod} for the ladder operators, the commutation relations \cref{eqn: Fourier space ccr} for the field lead to the condition
\begin{equation}\label{eqn: second part of scalar normalization condition}
   -\ii  \alpha\left[u_{\vec{k}}(t)\partial_{0}u_{\vec{k}}^{*}(t) - u_{\vec{k}}^{*}(t)\partial_{0}u_{\vec{k}}(t)\right] = 1~.
\end{equation}
In combination with $\inner{\varphi_{\vec{p}}}{\varphi_{\vec{k}}} = \delta_{\vec{k},\vec{p}}$, we obtain
\begin{equation}\label{eqn: scalar mode function normalization 1 mod}
  \varinner{\tilde\varphi_{\vec{k}}}{\tilde\varphi_{\vec{p}}} = \frac{1}{\alpha} \,\delta_{\vec{k},\vec{p}}~,
\end{equation}
which, in case of the canonical commutation relations of the ladder operators $\alpha = 1$, is indeed the normalization condition \cref{eqn: normalization 1}.

\subsection{Spinor field}

A very similar analysis for spinor fields leads to the normalization \cref{eqn: normalization 2} for spinor mode functions. This time, the task is to find a linear operators $\Psi(x)$ which
\begin{inparaenum}[(a)]
  \item solves the Dirac equation
  \begin{equation}
    (i\gamma^{\mu}\partial_{\mu}-m)\op\Psi(x) = 0
  \end{equation}
  as an operator equation, and
  \item satisfies the canonical anticommutation relations with its conjugate field $\Pi = \ii   \Psi^{\dagger}$:
  \begin{equation}
    \{\op\Psi^{A}(x),\op\Pi_{B}(y)\}_{|_{x^{0}=y^{0}}} = \ii   \delta(\vec{x}-\vec{y})\,\delta^{A}_{B}~.
  \end{equation}
\end{inparaenum}
The Hermitian conjugate of the field is automatically $\op\Psi_{B}^{\dagger}(x) = -\ii   \op\Pi_{B}(x)$ by definition of $\op\Pi$.

Now assume once again that the operator can be expanded in a Fourier series
\begin{equation}\begin{aligned}
  \op\Psi^{A}(x) &= \sum_{\vec{k}} \op\Psi_{\vec{k}}^{A}(t)\,\varphi_{\vec{k}}(\vec{x}) \\  \op\Pi_{B}(x) &= \sum_{\vec{k}} \op\Pi_{\vec{k},B}(t)\,\varphi_{\vec{k}}(\vec{x})~,
\end{aligned}\end{equation}
where $\op\Pi_{\vec{k},B}(t) = \ii   \op\Psi^{\dagger}_{-\vec{k},B}(t)$.
Translated to the Fourier domain, the problem becomes the following: find a linear operator $\op\Psi^{A}_{\vec{k}}(t)$ that solves the differential equation
\begin{equation}
  (i\gamma^{0}\partial_{0}+\gamma^{i}k_{i}-m)\op\Psi_{\vec{k}}^{A}(t) = 0
\end{equation}
and satisfies 
\begin{equation}\label{eqn: Fourier domain ACR}
  \{\op\Psi_{\vec{k}}^{A}(t),\op\Pi_{-\vec{p},B}(t)\} = \ii   \delta_{\vec{k},\vec{p}}\,\delta^{A}_{B}~.
\end{equation}

In the next step, make an ansatz that again separates the functional part from the operator part:
\begin{equation}\begin{aligned}
  \op\Psi^{A}_{\vec{k}}(t) &= \sum_{s} u^{A}_{\vec{k},s}(t)\,\op{a}_{\vec{k},s} + v^{A}_{-\vec{k},s}(t)\,\op{b}_{-\vec{k},s}^{\dagger}\\
  \op\Pi_{\vec{k},B}(t) &= \ii   \sum_{s} v^{\dagger}_{\vec{k},s,B}(t)\,\op{b}_{\vec{k},s} + u_{-\vec{k},s,B}^{\dagger}(t)\,\op{a}_{-\vec{k},s}^{\dagger}~.
\end{aligned}\end{equation}
The Dirac equation then reduces to
\begin{equation}\begin{aligned}
  (i\gamma^{0}\partial_{0}+\gamma^{i}k_{i}-m)\,u_{\vec{k},s}^{A}(t) &= 0\\
  (i\gamma^{0}\partial_{0}-\gamma^{i}k_{i}-m)\,v_{\vec{k},s}^{A}(t) &= 0~.
\end{aligned}\end{equation}
The spinors \cref{eqn: spinors 4dim massive} that solve the classical Dirac equation satisfy
\begin{equation}\begin{aligned}\label{eqn: dirac eq for spinors alone}
  (\gamma^{\mu}k_{\mu}-m)\,u_{\vec{k},s} &= 0\\
  (\gamma^{\mu}k_{\mu}+m)\,v_{\vec{k},s} &= 0~,
\end{aligned}\end{equation}
so it is easy to see that the above equations are solved by
\begin{equation}\begin{aligned}
  u_{\vec{k},s}^{A}(t) = C_{s}\,u_{\vec{k},s}\,e^{-\ii  \omega_{\vec{k}}t}\\
  v_{\vec{k},s}^{A}(t) = D_{s}\,v_{\vec{k},s}\,e^{+\ii  \omega_{\vec{k}}t}~,
\end{aligned}\end{equation}
where $C_{s}$, $D_{s} \in \C$. Our ansatz thus reads
\begin{multline}
  \op\Psi^{A}_{\vec{k}}(t) =\\
   \sum_{s}C_{s}\,u_{\vec{k},s}^{A}\,e^{-\ii  \omega_{\vec{k}}t}\,\op{a}_{\vec{k},s} +D_{s}v_{-\vec{k},s}^{A}\,e^{+\ii  \omega_{\vec{k}}t}\,\op{b}_{-\vec{k},s}^{\dagger}~.
\end{multline}
The anticommutation relations \cref{eqn: Fourier domain ACR}, together with those for the ladder operators \cref{eqn: ccr ladder op mod}, then lead to the condition
\begin{equation}
  \alpha \sum_{s} C_{s}^{2}u_{\vec{k},s}^{A}u_{\vec{k},r,B}^{\dagger} + D_{s}^{2}v_{-\vec{k},s}^{A}v_{-\vec{k},r,B}^{\dagger} = \delta^{A}_{B}~.
\end{equation}
Direct computation shows that this condition is met for $C_{s} = D_{s} = \sqrt{m\,(\alpha\omega_{\vec{k}})^{-1}}$. The field operator thus enjoys the mode expansion
\begin{equation}\label{eqn: spinor field operator full mode expansion}
  \op\Psi^{A}(x) = \sum_{\vec{k},s}\tilde\psi_{\vec{k},s,+}(x)\,\op{a}_{\vec{k},s} + \tilde\psi_{\vec{k},s,-}(x)\,\op{b}_{\vec{k},s}^{\dagger}
\end{equation}
with mode functions
\begin{equation}\begin{aligned}
  \tilde\psi_{\vec{k},s,+}(x) &= \sqrt{\frac{m}{\alpha\omega_{\vec{k}}L^{n}}}u_{\vec{k},s}\,e^{-\ii  k\cdot x}\\
  \tilde\psi_{\vec{k},s,-}(x) &= \sqrt{\frac{m}{\alpha\omega_{\vec{k}}L^{n}}}v_{\vec{k},s}\,e^{+\ii  k\cdot x}~.
\end{aligned}\end{equation}
These satisfy
\begin{equation}
  \inner{\tilde\psi_{\vec{k},s,\epsilon}}{\tilde\psi_{\vec{p},r,\delta}} = \frac{1}{\alpha} \delta_{\vec{k},\vec{p}}\,\delta_{s,r}\,\delta_{\epsilon,\delta}~,
\end{equation}
and, choosing $\alpha = 1$, we finally obtain \cref{eqn: normalization 2}.


\section{Proofs of Wick normal-ordering}\label{app: proofs wick}

\subsection{Quantized scalar field}\label{sec: proofs wick scalar}

In this section, \cref{thm: Wick complex t-o full vev} is proven on the basis of \cref{thm: Wick complex t-o}. We already know from \cref{thm: Wick complex t-o} how to rewrite the time-ordered product of fields as a sum of normal-ordered products with more and more contractions. To normal-order the entire product and obtain the expression stated in the theorem, we therefore merely need to commute all creation operators from right to the very left, and all annihilation operators from the left to the very right. In doing so, we will pick up commutators, between ladder operators and fields as well as between the different ladder operators themselves.
By \cref{thm: Wick complex t-o}, the full product
\begin{equation*}
(\op{a}_{\vec{p}})_{\vec{p}}\,(\op{b}_{\bar{\vec{p}}})_{\bar{\vec{p}}}\,[\,T\Phi^{\epsilon_{1}}(x_{1}) \cdots  \Phi^{\epsilon_{n}}(x_{n}) \,]\,(\op{a}^{\dagger}_{\vec{k}})_{\vec{k}}\,(\op{b}^{\dagger}_{\bar{\vec{k}}})_{\bar{\vec{k}}}
\end{equation*}
can be reduced to a sum of products of the form
\begin{equation*}
  (\op{a}_{\vec{p}})_{\vec{p}}\,(\op{b}_{\bar{\vec{p}}})_{\bar{\vec{p}}}\,\normalord \Phi^{\epsilon_{1}}(x_{1}) \cdots  \Phi^{\epsilon_{n}}(x_{n}) \normalord\,(\op{a}^{\dagger}_{\vec{k}})_{\vec{k}}\,(\op{b}^{\dagger}_{\bar{\vec{k}}})_{\bar{\vec{k}}}~.
\end{equation*}
For these holds

\begin{lem}\label{lem: towards Wick complex t-o full vev}
  Using the notation of \cref{thm: Wick complex t-o full vev}, the following identity holds:
  \begin{align}\label{eqn: towards Wick complex t-o full vev}
    &(\op{a}_{\vec{p}})_{\vec{p}}\,(\op{b}_{\bar{\vec{p}}})_{\bar{\vec{p}}}\,\normalord \Phi^{\epsilon_{1}}(x_{1}) \cdots  \Phi^{\epsilon_{n}}(x_{n}) \normalord\,(\op{a}^{\dagger}_{\vec{k}})_{\vec{k}}\,(\op{b}^{\dagger}_{\bar{\vec{k}}})_{\bar{\vec{k}}} \nonumber\\
    &= \normalord(\op{a}_{\vec{p}})_{\vec{p}}\,(\op{b}_{\bar{\vec{p}}})_{\bar{\vec{p}}}\,\Phi^{\epsilon_{1}}(x_{1}) \cdots  \Phi^{\epsilon_{n}}(x_{n}) \,(\op{a}^{\dagger}_{\vec{k}})_{\vec{k}}\,(\op{b}^{\dagger}_{\bar{\vec{k}}})_{\bar{\vec{k}}}\normalord\nonumber\\
    &\hphantom{={}}+ \sum_{\mathclap{\substack{\text{single}\\\text{contractions}}}} \normalord
    \contraction{a_{\vec{p}_{1}} \cdots}{\vphantom{\Phi}a}{_{\vec{p}_{i}} \cdots }{\Phi}%
    a_{\vec{p}_{1}} \cdots a_{\vec{p}_{i}} \cdots \Phi^{\epsilon_{j}}(x_{j}) \cdots b^{\dagger}_{\bar{\vec{k}}_{\bar{N}}}\normalord\nonumber\\
    &\hphantom{={}}+ \sum_{\mathclap{\substack{\text{double}\\\text{contractions}}}} \cdots~.
  \end{align}
  where contractions are only allowed between two ladder operators, and between a ladder operator and a field (not between two fields).
\end{lem}

\begin{proof}[Proof of \Cref{lem: towards Wick complex t-o full vev}]
  To achieve normal-ordering, the annihilation operators have to be moved past the field operators as well as the creation operators to the very right of the expression. Then, the creation operators have to be commuted past the field operators to the very left. Recall the commutation relations of the ladder operators of the quantized complex field \cref{eqn: cr ladder op real field}, with all other pairings vanishing. In the first step, we need the resulting commutation relations between annihilation operators and the field parts. The  four possible commutators are
  \begin{equation}\begin{aligned}\label{eqn: cr w fields towards Wick complex t-o full vev 1}
    [a_{\vec{k}},\Phi(x)] &= 0~, & [a_{\vec{k}},\Phi^{\dagger}(x)] &=  \id\, \varphi^{*}_{\vec{k}}(x)~, \\
    [b_{\vec{k}},\Phi(x)] &= \id\,  \varphi^{*}_{\vec{k}}(x)~, & [b_{\vec{k}},\Phi^{\dagger}(x)] &= 0~.
  \end{aligned}\end{equation}
  
  Consider the simplest case, when there is only one annihilation operator $a_{\vec{k}}$ on the left, and no annihilation operators on the right. To move $a_{\vec{k}}$ to the right, it has to be commuted with each and every field $\Psi^{\epsilon_{i}}(x_{i})$ operator sooner or later---or more precisely with their constituent ladder operators. Each field operator decomposes into two part, one containing the annihilation operators, and the other the creation operators:
  \begin{multline}
    a_{\vec{k}}\,\normalord \Phi^{\epsilon_{1}}(x_{1}) \cdots  \Phi^{\epsilon_{n}}(x_{n}) \normalord \\
    = a_{\vec{k}}\,\normalord \Phi^{\epsilon_{1}}(x_{1}) \cdots  [\Phi_{a}^{\epsilon_{i}}(x_{i}) + \Phi_{b}^{\dagger\epsilon_{i}}(x_{i})] \cdots \Phi^{\epsilon_{n}}(x_{n}) \normalord~.
  \end{multline}
  This decomposition can be done for all the other $n-1$ field operators as well. The result are $2^{n-1}$ summands containing $\Phi_{a}^{\epsilon_{i}}$ as a factor (and all possible combinations of the parts of the other fields), and another set of $2^{n-1}$ summands containing $\Phi_{b}^{\dagger\epsilon_{i}}$ multiplied with the same set of combinations of parts of the other fields. Let us pick one such set and call it $S$ for brevity. For example, $S$ might be
  \begin{equation*}
    S = \Phi_{a}^{\epsilon_{1}}(x_{1}) \cdots \Phi_{a}^{\epsilon_{i-1}}(x_{i-1}) \Phi_{a}^{\epsilon_{i+1}}(x_{i+1}) \cdots \Phi_{a}^{\epsilon_{n}}(x_{n}) ~.
  \end{equation*}
  Note that
  \begin{equation}\label{eqn: first S sum}
    \sum_{S} S = \normalord \Phi^{\epsilon_{1}}(x_{1}) \cdots \widehat{\Phi^{\epsilon_{i}}(x_{i})} \cdots \Phi^{\epsilon_{n}}(x_{n}) \normalord~,
  \end{equation}
  where the term with hat is omitted. Since the entire product is normal-ordered at the end, $\Phi_{a}^{\epsilon_{i}}$ and $\Phi_{b}^{\dagger\epsilon_{i}}$ will end up at different positions inside $S$:
  \begin{multline} \label{eqn: second S sum}
    \normalord \Phi^{\epsilon_{1}}(x_{1}) \cdots \Phi^{\epsilon_{n}}(x_{n}) \normalord \\
    = \sum_{S} \left[S_{1}\Psi^{\epsilon_{i}}_{a}(x_{i})S_{2} + S^{*}_{1}\Psi^{\dagger\epsilon_{i}}_{b}(x_{i})S^{*}_{2}  \right]~,
  \end{multline}
  where $S_{i}$, $S^{*}_{i}$ are two differen partitions of $S$: $S_{1}S_{2} = S$ and $S^{*}_{1}S^{*}_{2} = S$. Thus,
  \begin{multline}
    a_{\vec{k}}\,\normalord \Phi^{\epsilon_{1}}(x_{1}) \cdots  \Phi^{\epsilon_{n}}(x_{n}) \normalord \\
    = \cdots + a_{\vec{k}}S_{1}\Phi^{\epsilon_{i}}(x_{i})S_{2} + a_{\vec{k}}S^{*}_{1}\Phi_{b}^{\dagger\epsilon_{i}}(x_{i})S^{*}_{2} + \cdots~.
  \end{multline}
   We now start moving $a_{\vec{k}}$ to the right, and the contribution from $\Phi(x_{i})$ will be the following: $a_{\vec{k}}$ is commuted through $S_{1}$ (creating an additional term containing the commutator with other field parts in each step), and will eventually end up next to $\Phi^{\epsilon_{i}}(x_{i})$. Commuting $a_{\vec{k}}$ and $\Phi^{\epsilon_{i}}(x_{i})$ gives two terms
  \begin{equation}
    S_{1}a_{\vec{k}}\Phi^{\epsilon_{i}}(x_{i})S_{2} = [a_{\vec{k}},\Phi^{\epsilon_{i}}(x_{i})]S + S_{1}a_{\vec{k}}\Phi^{\epsilon_{i}}(x_{i})S_{2}~,
  \end{equation}
  where we have used that the commutator is proportional to the identity operator $\id$ in any case. The same is true for the second term containing $\Phi_{b}^{\dagger\epsilon_{i}}$:
  \begin{multline}
    S^{*}_{1}a_{\vec{k}}\Phi_{b}^{\dagger\epsilon_{i}}(x_{i})S^{*}_{2} \\
    = [a_{\vec{k}},\Phi_{b}^{\dagger\epsilon_{i}}(x_{i})]S + S^{*}_{1}a_{\vec{k}}\Phi_{b}^{\dagger\epsilon_{i}}(x_{i})S^{*}_{2}~.
  \end{multline}
   After combining the two commutators one obtains
  \begin{multline}
    a_{\vec{k}}\,\normalord \Phi^{\epsilon_{1}}(x_{1}) \cdots  \Phi^{\epsilon_{n}}(x_{n}) \normalord \\
    = \cdots + [a_{\vec{k}},\Phi^{\epsilon_{i}}(x_{i})]S + S_{1}\Phi^{\epsilon_{i}}(x_{i})a_{\vec{k}}S_{2} \\
    + S^{*}_{1}\Phi_{b}^{\dagger\epsilon_{i}}a_{\vec{k}}S^{*}_{2} + \cdots ~.
  \end{multline}
  All the terms containing commutators with the other fields have not been made explicit. Since this works for any combination $S$ of the ladder operators of the other fields, the sum of all these will just give
  \begin{equation}
    [a_{\vec{k}},\Phi^{\epsilon_{i}}(x_{i})]\normalord \Phi^{\epsilon_{1}}(x_{1}) \cdots \widehat{\Phi^{\epsilon_{i}}(x_{i})} \cdots \Phi^{\epsilon_{n}}(x_{n}) \normalord
  \end{equation}
  from the first term according to \cref{eqn: first S sum}. The other two terms need to be processed further, giving rise to commutators with the other fields. In the end, after commuting $a_{\vec{k}}$ all the way through $S$ there will be a commutator term like the above for every field, and a term which yields
  \begin{multline}
    \sum_{S} \left[S_{1}\Psi^{\epsilon_{i}}_{a}(x_{i})S_{2} + S^{*}_{1}\Psi^{\dagger\epsilon_{i}}_{b}(x_{i})S^{*}_{2}  \right]a_{\vec{k}} \\
    = \normalord \Phi^{\epsilon_{1}}(x_{1}) \cdots \Phi^{\epsilon_{n}}(x_{n}) \normalord a_{\vec{k}} \\
    = \normalord a_{\vec{k}} \Phi^{\epsilon_{1}}(x_{1}) \cdots \Phi^{\epsilon_{n}}(x_{n})  \normalord 
  \end{multline}
  according to \cref{eqn: second S sum}. By virtue of the commutators \cref{eqn: cr w fields towards Wick complex t-o full vev 1} and the definition \cref{eqn: complex field full vev contractions 3} of the contractions, we finally obtain
  \begin{multline}\label{eqn: towards Wick complex t-o full vev intermediate result 1}
    a_{\vec{k}}\,\normalord \Phi^{\epsilon_{1}}(x_{1}) \cdots  \Phi^{\epsilon_{n}}(x_{n}) \normalord = \normalord a_{\vec{k}}\Phi^{\epsilon_{1}}(x_{1}) \cdots  \Phi^{\epsilon_{n}}(x_{n}) \normalord \\
    + \sum_{i=1}^{n}%
    \contraction{\normalord}{\vphantom{\Phi}a}{_{\vec{k}}\Phi^{\epsilon_{1}}(x_{1}) \cdots}{\Phi}%
    \normalord a_{\vec{k}}\Phi^{\epsilon_{1}}(x_{1}) \cdots  \Phi^{\epsilon_{i}}(x_{i}) \cdots \Phi^{\epsilon_{n}}(x_{n}) \normalord~.
  \end{multline}
  Since it was nowhere used that $a$ is a particle annihilation operator, the same relation holds for antiparticle annihilation operators $b$.
  
  Now, if there is a single creation operator on the right instead, the relevant commutation relations are:
  \begin{equation}\begin{aligned}\label{eqn: cr w fields towards Wick complex t-o full vev 2}
    [\Phi(x),a^{\dagger}_{\vec{k}}] &= \id\, \varphi_{\vec{k}}(x)~, & [\Phi^{\dagger}(x),a^{\dagger}_{\vec{k}}] &=  0~,\\
    [\Phi(x),b^{\dagger}_{\vec{k}}] &= 0~, & [\Phi^{\dagger}(x),b^{\dagger}_{\vec{k}}] &= \id\, \varphi_{\vec{k}}(x) ~,
  \end{aligned}\end{equation}
  and a similar reasoning eventually gives
  \begin{multline}\label{eqn: towards Wick complex t-o full vev intermediate result 2}
    \,\normalord \Phi^{\epsilon_{1}}(x_{1}) \cdots  \Phi^{\epsilon_{n}}(x_{n}) \normalord \,a^{\dagger}_{\vec{p}} = \normalord \Phi^{\epsilon_{1}}(x_{1}) \cdots  \Phi^{\epsilon_{n}}(x_{n}) a^{\dagger}_{\vec{p}}\normalord \\
    + \sum_{i=1}^{n}%
    \contraction{\normalord \Phi^{\epsilon_{1}}(x_{1}) \cdots}{\Phi}{^{\epsilon_{i}}(x_{i}) \cdots \Phi^{\epsilon_{n}}(x_{n}) }{a}%
    \normalord \Phi^{\epsilon_{1}}(x_{1}) \cdots  \Phi^{\epsilon_{i}}(x_{i}) \cdots \Phi^{\epsilon_{n}}(x_{n}) a^{\dagger}_{\vec{p}}\normalord
  \end{multline}
  in terms of the contractions \cref{eqn: complex field full vev contractions 3}.
  
  If both types are present, \cref{eqn: towards Wick complex t-o full vev intermediate result 1} immediately allows to write:
  \begin{multline}\label{eqn: towards Wick complex t-o full vev intermediate 1}
    a_{\vec{k}}\,\normalord \Phi^{\epsilon_{1}}(x_{1}) \cdots  \Phi^{\epsilon_{n}}(x_{n}) \normalord\,a^{\dagger}_{\vec{p}} \\
    = \normalord \Phi^{\epsilon_{1}}(x_{1}) \cdots  \Phi^{\epsilon_{n}}(x_{n}) \normalord \,a_{\vec{k}}a^{\dagger}_{\vec{p}}  \\
    + \sum_{i=1}^{n}%
    \contraction{\normalord}{\vphantom{\Phi}a}{_{\vec{k}}\Phi^{\epsilon_{1}}(x_{1}) \cdots}{\Phi}%
    \normalord a_{\vec{k}}\Phi^{\epsilon_{1}}(x_{1}) \cdots  \Phi^{\epsilon_{i}}(x_{i}) \cdots \Phi^{\epsilon_{n}}(x_{n}) \normalord a^{\dagger}_{\vec{p}}~.
  \end{multline}
  The first term on the right-hand side evaluates to
  \begin{multline}\label{eqn: towards Wick complex t-o full vev intermediate 2}
  \normalord  \Phi^{\epsilon_{1}}(x_{1}) \cdots  \Phi^{\epsilon_{n}}(x_{n})   \normalord \,a^{\dagger}_{\vec{p}}a_{\vec{k}}   \normalord \\
  + %
    \contraction{ \normalord}{a}{_{\vec{k}} \Phi^{\epsilon_{1}}(x_{1}) \cdots  \Phi^{\epsilon_{n}}(x_{n}) }{a}%
    \normalord a_{\vec{k}} \Phi^{\epsilon_{1}}(x_{1}) \cdots  \Phi^{\epsilon_{n}}(x_{n}) a^{\dagger}_{\vec{p}}  \normalord
  \end{multline}
  using the contractions \cref{eqn: complex field full vev contractions 2}. The second term in \cref{eqn: towards Wick complex t-o full vev intermediate 1}, as well as the first term in \cref{eqn: towards Wick complex t-o full vev intermediate 2} can then be evaluated using \cref{eqn: towards Wick complex t-o full vev intermediate result 2}. In total, one obtains
  \begin{multline}
    a_{\vec{k}}  \,\normalord \Phi^{\epsilon_{1}}(x_{1}) \cdots  \Phi^{\epsilon_{n}}(x_{n}) \normalord \,a^{\dagger}_{\vec{p}} =
    \normalord a_{\vec{k}}  \Phi^{\epsilon_{1}}(x_{1}) \cdots  \Phi^{\epsilon_{n}}(x_{n})  a^{\dagger}_{\vec{p}} \normalord \\
    +\sum_{i=1}^{n}%
    \contraction{\normalord}{\vphantom{\Phi}a}{_{\vec{k}}\Phi^{\epsilon_{1}}(x_{1}) \cdots }{\Phi}%
    \normalord a_{\vec{k}}\Phi^{\epsilon_{1}}(x_{1}) \cdots \Phi^{\epsilon_{i}}(x_{i}) \cdots \Phi^{\epsilon_{n}}(x_{n}) a^{\dagger}_{\vec{p}}\normalord \\
    + \contraction{\normalord\, }{a}{_{\vec{k}}\Phi^{\epsilon_{1}}(x_{1})  \cdots \Phi^{\epsilon_{n}}(x_{n})}{a}%
    \normalord\, a_{\vec{k}}\Phi^{\epsilon_{1}}(x_{1})  \cdots \Phi^{\epsilon_{n}}(x_{n})  a^{\dagger}_{\vec{p}}\normalord\\
    +\sum_{i=1}^{n}%
    \contraction{\normalord a_{\vec{k}}\Phi^{\epsilon_{1}}(x_{1}) \cdots }{\Phi}{^{\epsilon_{i}}(x_{i}) \cdots  \Phi^{\epsilon_{n}}(x_{n})}{a}%
    \normalord a_{\vec{k}}\Phi^{\epsilon_{1}}(x_{1}) \cdots \Phi^{\epsilon_{i}}(x_{i})  \Phi^{\epsilon_{n}}(x_{n}) \cdots a^{\dagger}_{\vec{p}}\normalord \\
    +\sum_{i=1}^{n}\sum_{\substack{j=1\\j\neq i}}^{n}%
    \contraction{\normalord}{\vphantom{\Phi}a}{_{\vec{k}}\Phi^{\epsilon_{1}}(x_{1}) \cdots \Phi^{\epsilon_{j}}(x_{j}) \cdots}{\Phi}%
    \bcontraction{\normalord a_{\vec{k}}\Phi^{\epsilon_{1}}(x_{1}) \cdots}{\Phi}{^{\epsilon_{j}}(x_{j}) \cdots \Phi^{\epsilon_{i}}(x_{i}) \cdots \Phi^{\epsilon_{n}}(x_{n})}{a}
    \normalord a_{\vec{k}}\Phi^{\epsilon_{1}}(x_{1}) \cdots \Phi^{\epsilon_{j}}(x_{j}) \cdots \Phi^{\epsilon_{i}}(x_{i}) \cdots \Phi^{\epsilon_{n}}(x_{n}) a^{\dagger}_{\vec{p}}\normalord~.
  \end{multline}
  Repeated application of \cref{eqn: towards Wick complex t-o full vev intermediate result 1,eqn: towards Wick complex t-o full vev intermediate result 2} finally yields \cref{eqn: towards Wick complex t-o full vev} in the general case.
  
\end{proof}

\begin{proof}[Proof of \Cref{thm: Wick complex t-o full vev}]
  Combining \cref{thm: Wick complex t-o} and \cref{lem: towards Wick complex t-o full vev} directly allows one to express
  \begin{equation*}
    (\op{a}_{\vec{p}})_{\vec{p}}\,(\op{b}_{\bar{\vec{p}}})_{\bar{\vec{p}}}\,[\,T\Phi^{\epsilon_{1}}(x_{1}) \cdots  \Phi^{\epsilon_{n}}(x_{n}) \,]\,(\op{a}^{\dagger}_{\vec{k}})_{\vec{k}}\,(\op{b}^{\dagger}_{\bar{\vec{k}}})_{\bar{\vec{k}}}
  \end{equation*}
  as normal-ordered, plus sums of the same normal-ordered products with all the contractions (single, double,\dots fully contracted), defined in \cref{eqn: complex field contractions t-o,eqn: complex field full vev contractions 2,eqn: complex field full vev contractions 3}. If the vacuum expectation value is taken, all terms which still contain uncontracted normal-ordered operators vanish. If the total number of operators (ladder operators and field operators) is odd, there will be at least one normal-ordered operator remaining and all terms vanish. Only if the number of operators is even, the terms where all operators are contracted survives since all the contractions are proportional to the identity operator. Note that some or all of these surviving terms can still be zero, since many of the  contractions also vanish.
\end{proof}

\subsection{Quantized spinor field}\label{sec: proofs wick spinor}

We now prove \cref{thm: Wick spinor t-o full vev}. Analogous to the proof of \cref{thm: Wick complex t-o full vev}, we first prove the following
\begin{lem}\label{lem: towards Wick spinor t-o full vev}
  Using the notation of \cref{thm: Wick spinor t-o full vev}, the following identity holds:
  \begin{align}\label{eqn: towards Wick spinor t-o full vev}
    &(\op{a}_{\vec{p},r})_{\vec{p},r}\,(\op{b}_{\bar{\vec{p}},\bar{r}})_{\bar{\vec{p}},\bar{r}}\,\normalord \Psi^{\epsilon_{1}}(x_{1}) \cdots \nonumber\\
    &\cdots  \Psi^{\epsilon_{n}}(x_{n}) \normalord\,(\op{a}^{\dagger}_{\vec{k},s})_{\vec{k},s}\,(\op{b}^{\dagger}_{\bar{\vec{k}},\bar{s}})_{\bar{\vec{k}},\bar{s}} \nonumber\\
    &= \normalord(\op{a}_{\vec{p},r})_{\vec{p},r}\,(\op{b}_{\bar{\vec{p}},\bar{r}})_{\bar{\vec{p}},\bar{r}}\,\Psi^{\epsilon_{1}}(x_{1}) \cdots \nonumber\\
    &\hphantom{={}}\cdots  \Psi^{\epsilon_{n}}(x_{n}) \,(\op{a}^{\dagger}_{\vec{k},s})_{\vec{k},s}\,(\op{b}^{\dagger}_{\bar{\vec{k}},\bar{s}})_{\bar{\vec{k}},\bar{s}} \normalord\nonumber\\
    &\hphantom{={}}+ \sum_{\mathclap{\substack{\text{single}\\\text{contractions}}}} \normalord
    \contraction{a_{(\vec{p},r)_{1}} \cdots}{\vphantom{\Psi}a}{_{(\vec{p},r)_{i}} \cdots }{\Psi}%
    a_{(\vec{p},r)_{1}} \cdots a_{(\vec{p},r)_{i}} \cdots \Psi^{\epsilon_{j}}(x_{j}) \cdots b^{\dagger}_{(\bar{\vec{k}},\bar{s})_{\bar{N}}}\normalord\nonumber\\
    &\hphantom{={}}+ \sum_{\mathclap{\substack{\text{double}\\\text{contractions}}}} \cdots~.
  \end{align}
  where contractions are only allowed between two ladder operators, and between a ladder operator and a field (not between two fields).
\end{lem}

The proof is very similar to the proof of \cref{lem: towards Wick complex t-o full vev}, but with the additional complexity that one has to keep track of signs introduced by anticommuting the involved operators.

\begin{proof}[Proof of \Cref{lem: towards Wick spinor t-o full vev}]
  To achieve normal-ordering, the annihilation operators have to be moved past the field operators and the creation operators to the right, and the creation operators subsequently past the fields to the left. Recall the anticommutation relations between the ladder operators \cref{eqn: acr ladder op spinor field}, with all remaining anticommutators vanishing. For the first step, we need to know how the annihilation operators anticommute with the fields. From the above relations and the separation of the spinor field according to \cref{eqn: spinor field parts,eqn: field separation spinor} one obtains:
  \begin{align}\label{eqn: acr w field towards Wick spinor t-o full vev}
    \{a_{\vec{k},s},\Psi(x)\} &= 0~,  & \{a_{\vec{k},s},\conj\Psi(x)\} &= \id\, \conj{\psi}_{\vec{k},s,+}(x)~,\nonumber\\
    \{b_{\vec{k},s},\conj\Psi(x)\} &= 0 ~, & \{b_{\vec{k},s},\Psi(x)\} &=\id\,  \psi_{\vec{k},s,-}(x)~.
  \end{align}
  
  Let us begin with one annihilation operator on the left: to move it to the right, it has to be anticommuted with any field operator $\Psi^{\epsilon_{i}}(x_{i})$ at some point, picking up an anticommutator in the process. Each field operator decomposes into two parts, one containing the creation operators, one the annihilation operators
  \begin{multline}
  a_{\vec{k},s}  \,\normalord \Psi^{\epsilon_{1}}(x_{1}) \cdots \Psi^{\epsilon_{i}}(x_{i})  \cdots \Psi^{\epsilon_{n}}(x_{n}) \normalord \\
  = a_{\vec{k},s}  \,\normalord \Psi^{\epsilon_{1}}(x_{1}) \cdots  [\Psi^{\epsilon_{i}}_{a}(x_{i}) + \Psi_{b}^{\epsilon_{i}}(x_{i})]\cdots \Psi^{\epsilon_{n}}(x_{n}) \normalord~.
  \end{multline}
  The same can be done for the other fields, and we essentially end up with sums of products of ladder operators (with sums in front, and multiplied mode functions), where each term contains one ladder operator per field. Let us pick out one of these products of ladder operators from the field $\Psi^{\epsilon_{1}}(x_{1}),\dots,\Psi^{\epsilon_{n}}(x_{n})$ without $\Psi^{\epsilon_{i}}(x_{i})$ and denote it $S$ for brevity. Both parts of $\Psi^{\epsilon_{i}}(x_{i})$ are multiplied with $S$, but since the entire product of field operators is then normal-ordered, the position where they are inserted into $S$ will be different. Through normal-ordering, $\Psi^{\epsilon_{i}}_{a}(x_{i})$ will be shifted from position $i$ to some position $i+m$, picking up a sign $(-1)^{m}$, and $\Psi^{\epsilon_{i}}_{b}(x_{i})$ will be shifted from position $i$ to some other position $i+l$, picking up a sign $(-1)^{l}$. Therefore,
  \begin{multline}\label{eqn: towards Wick spinor t-o full vev intermediate}
   a_{\vec{k},s}  \,\normalord \Psi^{\epsilon_{1}}(x_{1}) \cdots \Psi^{\epsilon_{i}}(x_{i})  \cdots \Psi^{\epsilon_{n}}(x_{n}) \normalord \\
   = \cdots + a_{\vec{k},s}  (-1)^{m}S_{1} \Psi^{\epsilon_{i}}_{a}(x_{i})  S_{2} \\
   + a_{\vec{k},s} (-1)^{l}S^{*}_{1} \Psi^{\epsilon_{i}}_{b}(x_{i})  S^{*}_{2} +\dots~,
  \end{multline}
  where $S_{1}S_{2} = S = S^{*}_{1}S^{*}_{2}$. Consider the first term, containing $\Psi^{\epsilon_{i}}_{a}(x_{i})$: the ladder operator $a_{\vec{k},s}$ is successively anticommuted through the operators of $S_{1}$, creating an additional term containing the anticommutator, and picking up a sign for $a_{\vec{k},s}$ in each step. By the time $a_{\vec{k},s}$ reaches $\Psi^{\epsilon_{i}}_{a}(x_{i})$, it already had to be anticommuted through the $(i+m-1)$ normal-ordered ladder operators in $S_{1}$ and has therefore picked up a sign $(-1)^{i+m-1}$, which leaves us with an overall sign $(-1)^{i+2m-1} = (-1)^{i-1}$. Anticommuting $a_{\vec{k},s}$ and $\Psi^{\epsilon_{i}}_{a}(x_{i})$, we get two terms:
  \begin{multline}
    (-1)^{i-1}S_{1}a_{\vec{k},s}\Psi^{\epsilon_{i}}_{a}(x_{i})S_{2}  \\
    = (-1)^{i-1}\{a_{\vec{k},s},\Psi^{\epsilon_{i}}_{a}(x_{i})\}S - \\
    (-1)^{i-1} S_{1}\Psi^{\epsilon_{i}}_{a}(x_{i}) a_{\vec{k},s}S_{2}~,
  \end{multline}
  where we have used the fact that the anticommutator is proportional to the identity operator in any case, and can therefore be pulled in front of the product $S_{1}S_{2} = S$. The same argumentation works on the second term,  $\Psi^{\epsilon_{i}}_{b}(x_{i})$, and we obtain
  \begin{multline}
    (-1)^{i-1}S^{*}_{1}a_{\vec{k},s}\Psi^{\epsilon_{i}}_{b}(x_{i})S^{*}_{2}  \\
    = (-1)^{i-1}\{a_{\vec{k},s},\Psi^{\epsilon_{i}}_{b}(x_{i})\}S \\
    - (-1)^{i-1} S^{*}_{1}\Psi^{\epsilon_{i}}_{b}(x_{i}) a_{\vec{k},s}S^{*}_{2}~.
  \end{multline}
  The two anticommutators can be combined to
  \begin{equation}
    \{a_{\vec{k},s},\Psi^{\epsilon_{i}}_{a}(x_{i})\} + \{a_{\vec{k},s},\Psi^{\epsilon_{i}}_{b}(x_{i})\} = \{a_{\vec{k},s},\Psi^{\epsilon_{i}}(x_{i})\}
  \end{equation}
  such that \cref{eqn: towards Wick spinor t-o full vev intermediate} becomes
  \begin{multline}
   a_{\vec{k},s}  \,\normalord \Psi^{\epsilon_{1}}(x_{1}) \cdots \Psi^{\epsilon_{i}}(x_{i})  \cdots \Psi^{\epsilon_{n}}(x_{n}) \normalord \\=
   \cdots + (-1)^{i}S_{1} \Psi^{\epsilon_{i}}_{a}(x_{i})  a_{\vec{k},s}  S_{2} + (-1)^{i}S^{*}_{1} \Psi^{\epsilon_{i}}_{b}(x_{i}) a_{\vec{k},s}  S^{*}_{2} \\
   + (-1)^{i-1} \{a_{\vec{k},s},\Psi^{\epsilon_{i}}(x_{i})\}S + \dots~.
  \end{multline}
 We have not written out the anticommutators $a_{\vec{k},s}$ has picked up with the operators in $S_{1}$ and $S_{1}^{*}$. Since this needs to be done for any of the field operators, we end up with
  \begin{multline}
    a_{\vec{k},s}  \,\normalord \Psi^{\epsilon_{1}}(x_{1})  \cdots \Psi^{\epsilon_{n}}(x_{n}) \normalord\\
     = (-1)^{n} \normalord \Psi^{\epsilon_{1}}(x_{1})  \cdots \Psi^{\epsilon_{n}}(x_{n})  \normalord a_{\vec{k},s}\\
    + \sum_{i=1}^{n} (-1)^{i-1} \{a_{\vec{k},s},\Psi^{\epsilon_{i}}(x_{i})\} \normalord \Psi^{\epsilon_{1}}(x_{1}) \cdots \\
    \cdots \widehat{\Psi^{\epsilon_{i}}(x_{i})}  \cdots \Psi^{\epsilon_{n}}(x_{n}) \normalord~.
  \end{multline}

 The first term on the right-hand side is normal-ordered, and the sign can be absorbed into anticommuting the annihilation operator back to the beginning of the product. The second term can be rewritten by taking into account the anticommutation relations \cref{eqn: acr w field towards Wick spinor t-o full vev} and defining contractions like in \cref{eqn: spinor field full vev contractions 3}
  \begin{multline}\label{eqn: intermediate result 1 towards Wick spinor t-o full vev 2}
    a_{\vec{k},s}  \,\normalord \Psi^{\epsilon_{1}}(x_{1})  \cdots \Psi^{\epsilon_{n}}(x_{n}) \normalord 
    = \normalord a_{\vec{k},s} \Psi^{\epsilon_{1}}(x_{1})  \cdots \Psi^{\epsilon_{n}}(x_{n})  \normalord \\
    +\sum_{i=1}^{n}%
    \contraction{\normalord}{\vphantom{\Psi}a}{_{\vec{k},s}\Psi^{\epsilon_{1}}(x_{1}) \cdots }{\Psi}%
    \normalord a_{\vec{k},s}\Psi^{\epsilon_{1}}(x_{1}) \cdots \Psi^{\epsilon_{i}}(x_{i}) \cdots \Psi^{\epsilon_{n}}(x_{n}) \normalord~.
  \end{multline}
  Note in particular that the sign is absorbed into the definition of the contraction. Since we have nowhere used that $\op{a}$ is a particle annihilation operator, the same identity holds for antiparticle annihilation operators $\op{b}$.
  
  Now consider what happens if a creation operator is added at the right. Using the last two equations, we can immediately write
  \begin{multline}\label{eqn: towards Wick spinor t-o full vev intermediate 2}
    a_{\vec{k},s}  \,\normalord \Psi^{\epsilon_{1}}(x_{1})  \cdots \Psi^{\epsilon_{n}}(x_{n}) \normalord \,a^{\dagger}_{\vec{p},r} \\ =
    (-1)^{n} \normalord \Psi^{\epsilon_{1}}(x_{1})  \cdots \Psi^{\epsilon_{n}}(x_{n})  \normalord\, a_{\vec{k},s}a^{\dagger}_{\vec{p},r} \\
    + \sum_{i=1}^{n}%
    \contraction{\normalord}{\vphantom{\Psi}a}{_{\vec{k},s}\Psi^{\epsilon_{1}}(x_{1}) \cdots }{\Psi}%
    \normalord a_{\vec{k},s}\Psi^{\epsilon_{1}}(x_{1}) \cdots \Psi^{\epsilon_{i}}(x_{i}) a^{\dagger}_{\vec{p},r} \Psi^{\epsilon_{n}}(x_{n}) \normalord\, a^{\dagger}_{\vec{p},r}~.
  \end{multline}
  Taking into account the anticommutation relations between the ladder operators and defining contractions \cref{eqn: spinor field full vev contractions 2}, the first term on the right-hand side can be expressed as
  \begin{multline}\label{eqn: towards Wick spinor t-o full vev intermediate 3}
    \contraction{\normalord\, }{a}{_{\vec{k},s}\Psi^{\epsilon_{1}}(x_{1})  \cdots \Psi^{\epsilon_{n}}(x_{n})}{a}%
    \normalord\, a_{\vec{k},s}\Psi^{\epsilon_{1}}(x_{1})  \cdots \Psi^{\epsilon_{n}}(x_{n})  a^{\dagger}_{\vec{p},r}\normalord \\
    + (-1)^{n+1}\normalord\Psi^{\epsilon_{1}}(x_{1})  + \cdots \Psi^{\epsilon_{n}}(x_{n}) \normalord \,a^{\dagger}_{\vec{p},r} a_{\vec{k},s}~.
  \end{multline}
  To handle the remaining terms, we need to evaluate expressions of the form
  \begin{equation}
    \normalord \Psi^{\epsilon_{1}}(x_{1}) \cdots \Psi^{\epsilon_{i}}(x_{i})  \cdots \Psi^{\epsilon_{n}}(x_{n}) \normalord\, a^{\dagger}_{\vec{p},r} ~.
  \end{equation}
  The relevant anticommutation relations are
  \begin{equation}\begin{aligned}\label{eqn: acr w field towards Wick spinor t-o full vev 2}
    \{\conj\Psi(x),a^{\dagger}_{\vec{k},s}\} &= 0~, & \{\Psi(x),a^{\dagger}_{\vec{k},s}\} &= \id\, \psi_{\vec{k},s,+}(x)~,  \\
    \{\Psi(x),b^{\dagger}_{\vec{k},s}\} &= 0~, & \{\conj\Psi(x),b^{\dagger}_{\vec{k},s}\} &= \id\, \conj{\psi}_{\vec{k},s,-}(x)~,
  \end{aligned}\end{equation}
  and a reasoning analogous to the one that led to \cref{eqn: intermediate result 1 towards Wick spinor t-o full vev 2} eventually yields
  \begin{multline}\label{eqn: intermediate result 2 towards Wick spinor t-o full vev 2}
    \,\normalord \Psi^{\epsilon_{1}}(x_{1})  \cdots \Psi^{\epsilon_{n}}(x_{n})   \normalord\, a^{\dagger}_{\vec{p},r}\\
    = \normalord \Psi^{\epsilon_{1}}(x_{1})  \cdots \Psi^{\epsilon_{n}}(x_{n}) a^{\dagger}_{\vec{p},r} \normalord \\
    +\sum_{i=1}^{n}%
    \contraction{\normalord \Psi^{\epsilon_{1}}(x_{1}) \cdots }{\Psi}{^{\epsilon_{i}}(x_{i}) \cdots  \Psi^{\epsilon_{n}}(x_{n})}{a}%
    \normalord \Psi^{\epsilon_{1}}(x_{1}) \cdots \Psi^{\epsilon_{i}}(x_{i})  \Psi^{\epsilon_{n}}(x_{n}) \cdots a^{\dagger}_{\vec{p},r}\normalord ~.
  \end{multline}
  Using this in \cref{eqn: towards Wick spinor t-o full vev intermediate 2,eqn: towards Wick spinor t-o full vev intermediate 3}, one obtains
  \begin{widetext}
    \begin{multline}
      a_{\vec{k},s}  \,\normalord \Psi^{\epsilon_{1}}(x_{1}) \cdots  \Psi^{\epsilon_{n}}(x_{n}) \normalord \,a^{\dagger}_{\vec{p},r} =
      \normalord a_{\vec{k},s}  \Psi^{\epsilon_{1}}(x_{1}) \cdots  \Psi^{\epsilon_{n}}(x_{n})  a^{\dagger}_{\vec{p},r} \normalord 
      +\sum_{i=1}^{n}%
      \contraction{\normalord}{\vphantom{\Psi}a}{_{\vec{k},s}\Psi^{\epsilon_{1}}(x_{1}) \cdots }{\Psi}%
      \normalord a_{\vec{k},s}\Psi^{\epsilon_{1}}(x_{1}) \cdots \Psi^{\epsilon_{i}}(x_{i}) \cdots \Psi^{\epsilon_{n}}(x_{n}) a^{\dagger}_{\vec{p},r}\normalord\\
      + \contraction{\normalord\, }{a}{_{\vec{k},s}\Psi^{\epsilon_{1}}(x_{1})  \cdots \Psi^{\epsilon_{n}}(x_{n})}{a}%
      \normalord\, a_{\vec{k},s}\Psi^{\epsilon_{1}}(x_{1})  \cdots \Psi^{\epsilon_{n}}(x_{n})  a^{\dagger}_{\vec{p},r}\normalord
      +\sum_{i=1}^{n}%
      \contraction{\normalord a_{\vec{k},s}\Psi^{\epsilon_{1}}(x_{1}) \cdots }{\Psi}{^{\epsilon_{i}}(x_{i}) \cdots  \Psi^{\epsilon_{n}}(x_{n})}{a}%
      \normalord a_{\vec{k},s}\Psi^{\epsilon_{1}}(x_{1}) \cdots \Psi^{\epsilon_{i}}(x_{i})  \Psi^{\epsilon_{n}}(x_{n}) \cdots a^{\dagger}_{\vec{p},r}\normalord \\
      +\sum_{i=1}^{n}\sum_{\substack{j=1\\j\neq i}}^{n}%
      \contraction{\normalord}{\vphantom{\Psi}a}{_{\vec{k},s}\Psi^{\epsilon_{1}}(x_{1}) \cdots \Psi^{\epsilon_{j}}(x_{j}) \cdots}{\Psi}%
      \bcontraction{\normalord a_{\vec{k},s}\Psi^{\epsilon_{1}}(x_{1}) \cdots}{\Psi}{^{\epsilon_{j}}(x_{j}) \cdots \Psi^{\epsilon_{i}}(x_{i}) \cdots \Psi^{\epsilon_{n}}(x_{n})}{a}
      \normalord a_{\vec{k},s}\Psi^{\epsilon_{1}}(x_{1}) \cdots \Psi^{\epsilon_{j}}(x_{j}) \cdots \Psi^{\epsilon_{i}}(x_{i}) \cdots \Psi^{\epsilon_{n}}(x_{n}) a^{\dagger}_{\vec{p},r}\normalord~.
    \end{multline}
  \end{widetext}
  
  Iterative application of \cref{eqn: intermediate result 1 towards Wick spinor t-o full vev 2,eqn: intermediate result 2 towards Wick spinor t-o full vev 2} then yields \cref{eqn: towards Wick spinor t-o full vev} in the general case.
\end{proof}

\begin{proof}[Proof of \Cref{thm: Wick spinor t-o full vev}]
  Similar to the proof of \cref{thm: Wick complex t-o full vev}, combining \cref{thm: Wick spinor t-o} with \cref{lem: towards Wick spinor t-o full vev} yields \cref{thm: Wick spinor t-o full vev}: at first, one obtains a completely normal ordered term plus all normal-ordered terms with single contractions (between two fields, between two ladder operators and between ladder operator and field) plus all normal-ordered term with two contractions and so on. Taking the vacuum expectation value, all terms vanish due to the normal-ordering, with the exception of those terms that are fully contracted and therefore proportional to the identity $\id$. Full contraction is only possible for an even number of operators, such that products containing an odd number of operators vanish.
\end{proof}

\subsection{UDW-type detector}\label{sec: proofs wick detector}

\Cref{thm: Wick detector t-o full vev} is essentially a simplification of the \cref{thm: Wick spinor t-o full vev} proven in the previous section. Again, first consider

\begin{lem}\label{lem: towards Wick detector t-o full vev}
  Using the notation of \cref{thm: Wick detector t-o full vev}, the following identity holds:
  \begin{multline}
    [\sigma^{-}]^{a}\normalord \op{\mu}(t_{1})\cdots\op{\mu}(t_{n}) \normalord [\sigma^{+}]^{b} \\= \normalord [\sigma^{-}]^{a} \op{\mu}(t_{1})\cdots\op{\mu}(t_{n})  [\sigma^{+}]^{b} \normalord \\
    + \sum_{\mathclap{\substack{\text{single}\\{\text{contractions}}}}}%
    \contraction{\normalord}{\sigma}{^{-}\cdots \op{\mu}(t_{1})\cdots}{\mu}
     \normalord \sigma^{-} \cdots \op{\mu}(t_{1})\cdots\op{\mu}(t_{i})\cdots\op{\mu}(t_{n})  \cdots \sigma^{+} \normalord \\
    + \sum_{\mathclap{\substack{\text{double}\\{\text{contractions}}}}}%
    \contraction{\normalord}{\sigma}{^{-}\cdots \op{\mu}(t_{1})\cdots\op{\mu}(t_{j})\cdots}{\mu}
    \bcontraction{\normalord \sigma^{-} \cdots \op{\mu}(t_{1})\cdots}{\op{\mu}}{(t_{j})\cdots\op{\mu}(t_{i})\cdots\op{\mu}(t_{n})  \cdots }{\sigma}%
     \normalord \sigma^{-} \cdots \op{\mu}(t_{1})\cdots\op{\mu}(t_{j})\cdots\op{\mu}(t_{i})\cdots\op{\mu}(t_{n})  \cdots \sigma^{+} \normalord~.
  \end{multline}
\end{lem}

\begin{proof}[Proof of \Cref{lem: towards Wick detector t-o full vev}]
  The proof runs analogous to the proof of \cref{lem: towards Wick spinor t-o full vev}, but is somewhat simpler since there is only one type of ladder operator ($\sigma$) instead of two ($a$, $b$), and the monopole operator is self-adjoint, unlike the spinor field. This time, the relevant anticommutation relations are \cref{eqn: acr ladd op monopole}, and
  \begin{align}
    \{\sigma^{-},\mu(t)\} &= \id\,e^{+\ii  \Omega t}~, & \{\mu(t),\sigma^{+}\} &= \id\,e^{-\ii  \Omega t}
  \end{align}
  instead of \cref{eqn: acr w field towards Wick spinor t-o full vev} and \cref{eqn: acr w field towards Wick spinor t-o full vev 2}. This gives rise to the contractions defined in \cref{eqn: detector full vev contractions 2,eqn: detector full vev contractions 3}. Beyond that, the proof is follows the exact same steps as before.
\end{proof}

\begin{proof}[Proof of \Cref{thm: Wick detector t-o full vev}]
  The theorem follows directly from \cref{thm: Wick detector t-o} and \cref{lem: towards Wick detector t-o full vev}, analogous to the proof of \cref{thm: Wick spinor t-o full vev}.
\end{proof}

\bibliography{fermion_detector}

\begin{thebibliography}{53}%
\makeatletter
\providecommand \@ifxundefined [1]{%
 \@ifx{#1\undefined}
}%
\providecommand \@ifnum [1]{%
 \ifnum #1\expandafter \@firstoftwo
 \else \expandafter \@secondoftwo
 \fi
}%
\providecommand \@ifx [1]{%
 \ifx #1\expandafter \@firstoftwo
 \else \expandafter \@secondoftwo
 \fi
}%
\providecommand \natexlab [1]{#1}%
\providecommand \enquote  [1]{``#1''}%
\providecommand \bibnamefont  [1]{#1}%
\providecommand \bibfnamefont [1]{#1}%
\providecommand \citenamefont [1]{#1}%
\providecommand \href@noop [0]{\@secondoftwo}%
\providecommand \href [0]{\begingroup \@sanitize@url \@href}%
\providecommand \@href[1]{\@@startlink{#1}\@@href}%
\providecommand \@@href[1]{\endgroup#1\@@endlink}%
\providecommand \@sanitize@url [0]{\catcode `\\12\catcode `\$12\catcode
  `\&12\catcode `\#12\catcode `\^12\catcode `\_12\catcode `\%12\relax}%
\providecommand \@@startlink[1]{}%
\providecommand \@@endlink[0]{}%
\providecommand \url  [0]{\begingroup\@sanitize@url \@url }%
\providecommand \@url [1]{\endgroup\@href {#1}{\urlprefix }}%
\providecommand \urlprefix  [0]{URL }%
\providecommand \Eprint [0]{\href }%
\providecommand \doibase [0]{http://dx.doi.org/}%
\providecommand \selectlanguage [0]{\@gobble}%
\providecommand \bibinfo  [0]{\@secondoftwo}%
\providecommand \bibfield  [0]{\@secondoftwo}%
\providecommand \translation [1]{[#1]}%
\providecommand \BibitemOpen [0]{}%
\providecommand \bibitemStop [0]{}%
\providecommand \bibitemNoStop [0]{.\EOS\space}%
\providecommand \EOS [0]{\spacefactor3000\relax}%
\providecommand \BibitemShut  [1]{\csname bibitem#1\endcsname}%
\let\auto@bib@innerbib\@empty
\bibitem [{\citenamefont {Terno}(2014)}]{DannyLoc}%
  \BibitemOpen
  \bibfield  {author} {\bibinfo {author} {\bibfnamefont {D.~R.}\ \bibnamefont
  {Terno}},\ }\href {\doibase 10.1103/PhysRevA.89.042111} {\bibfield  {journal}
  {\bibinfo  {journal} {Phys. Rev. A}\ }\textbf {\bibinfo {volume} {89}},\
  \bibinfo {pages} {042111} (\bibinfo {year} {2014})}\BibitemShut {NoStop}%
\bibitem [{\citenamefont {Unruh}(1976)}]{unruh_notes_1976}%
  \BibitemOpen
  \bibfield  {author} {\bibinfo {author} {\bibfnamefont {W.~G.}\ \bibnamefont
  {Unruh}},\ }\href {\doibase 10.1103/PhysRevD.14.870} {\bibfield  {journal}
  {\bibinfo  {journal} {Phys. Rev. D}\ }\textbf {\bibinfo {volume} {14}},\
  \bibinfo {pages} {870} (\bibinfo {year} {1976})}\BibitemShut {NoStop}%
\bibitem [{\citenamefont
  {Seligman~DeWitt}(1979)}]{seligman_dewitt_quantum_1979}%
  \BibitemOpen
  \bibfield  {author} {\bibinfo {author} {\bibfnamefont {B.}~\bibnamefont
  {Seligman~DeWitt}},\ }in\ \href@noop {} {\emph {\bibinfo {booktitle} {General
  Relativity: An Einstein Centenary Survey}}},\ \bibinfo {editor} {edited by\
  \bibinfo {editor} {\bibfnamefont {S.~W.}\ \bibnamefont {Hawking}}\ and\
  \bibinfo {editor} {\bibfnamefont {W.}~\bibnamefont {Israel}}}\ (\bibinfo
  {publisher} {Cambridge University Press},\ \bibinfo {address} {Cambridge},\
  \bibinfo {year} {1979})\ p.\ \bibinfo {pages} {680}\BibitemShut {NoStop}%
\bibitem [{\citenamefont {Jonsson}\ \emph {et~al.}(2015)\citenamefont
  {Jonsson}, \citenamefont {Mart\'{i}n-Mart\'{i}nez},\ and\ \citenamefont
  {Kempf}}]{Jonsson2015}%
  \BibitemOpen
  \bibfield  {author} {\bibinfo {author} {\bibfnamefont {R.~H.}\ \bibnamefont
  {Jonsson}}, \bibinfo {author} {\bibfnamefont {E.}~\bibnamefont
  {Mart\'{i}n-Mart\'{i}nez}}, \ and\ \bibinfo {author} {\bibfnamefont
  {A.}~\bibnamefont {Kempf}},\ }\href {\doibase 10.1103/PhysRevLett.114.110505}
  {\bibfield  {journal} {\bibinfo  {journal} {Phys. Rev. Lett.}\ }\textbf
  {\bibinfo {volume} {114}},\ \bibinfo {pages} {110505} (\bibinfo {year}
  {2015})}\BibitemShut {NoStop}%
\bibitem [{\citenamefont {Jonsson}\ \emph {et~al.}(2014)\citenamefont
  {Jonsson}, \citenamefont {Mart\'{i}n-Mart\'{i}nez},\ and\ \citenamefont
  {Kempf}}]{jonsson_quantum_2014}%
  \BibitemOpen
  \bibfield  {author} {\bibinfo {author} {\bibfnamefont {R.~H.}\ \bibnamefont
  {Jonsson}}, \bibinfo {author} {\bibfnamefont {E.}~\bibnamefont
  {Mart\'{i}n-Mart\'{i}nez}}, \ and\ \bibinfo {author} {\bibfnamefont
  {A.}~\bibnamefont {Kempf}},\ }\href {\doibase 10.1103/PhysRevA.89.022330}
  {\bibfield  {journal} {\bibinfo  {journal} {Phys. Rev. A}\ }\textbf {\bibinfo
  {volume} {89}},\ \bibinfo {pages} {022330} (\bibinfo {year}
  {2014})}\BibitemShut {NoStop}%
\bibitem [{\citenamefont {Summers}\ and\ \citenamefont
  {Werner}(1985)}]{Algebra1}%
  \BibitemOpen
  \bibfield  {author} {\bibinfo {author} {\bibfnamefont {S.~J.}\ \bibnamefont
  {Summers}}\ and\ \bibinfo {author} {\bibfnamefont {R.~F.}\ \bibnamefont
  {Werner}},\ }\href {\doibase 10.1016/0375-9601(85)90093-3} {\bibfield
  {journal} {\bibinfo  {journal} {Phys. Lett. A}\ }\textbf {\bibinfo {volume}
  {110}},\ \bibinfo {pages} {257} (\bibinfo {year} {1985})}\BibitemShut
  {NoStop}%
\bibitem [{\citenamefont {Valentini}(1991)}]{Valentini1991}%
  \BibitemOpen
  \bibfield  {author} {\bibinfo {author} {\bibfnamefont {A.}~\bibnamefont
  {Valentini}},\ }\href {\doibase 10.1016/0375-9601(91)90952-5} {\bibfield
  {journal} {\bibinfo  {journal} {Phys. Lett. A}\ }\textbf {\bibinfo {volume}
  {153}},\ \bibinfo {pages} {321} (\bibinfo {year} {1991})}\BibitemShut
  {NoStop}%
\bibitem [{\citenamefont {Reznik}(2003)}]{Reznik2003}%
  \BibitemOpen
  \bibfield  {author} {\bibinfo {author} {\bibfnamefont {B.}~\bibnamefont
  {Reznik}},\ }\href {\doibase 10.1023/A:1022875910744} {\bibfield  {journal}
  {\bibinfo  {journal} {Found. Phys.}\ }\textbf {\bibinfo {volume} {33}},\
  \bibinfo {pages} {167} (\bibinfo {year} {2003})}\BibitemShut {NoStop}%
\bibitem [{\citenamefont {Steeg}\ and\ \citenamefont
  {Menicucci}(2009)}]{VerSteeg2009}%
  \BibitemOpen
  \bibfield  {author} {\bibinfo {author} {\bibfnamefont {G.~V.}\ \bibnamefont
  {Steeg}}\ and\ \bibinfo {author} {\bibfnamefont {N.~C.}\ \bibnamefont
  {Menicucci}},\ }\href {\doibase 10.1103/PhysRevD.79.044027} {\bibfield
  {journal} {\bibinfo  {journal} {Phys. Rev. D}\ }\textbf {\bibinfo {volume}
  {79}},\ \bibinfo {pages} {044027} (\bibinfo {year} {2009})}\BibitemShut
  {NoStop}%
\bibitem [{\citenamefont {Mart\'{i}n-Mart\'{i}nez}\ \emph
  {et~al.}(2013)\citenamefont {Mart\'{i}n-Mart\'{i}nez}, \citenamefont {Brown},
  \citenamefont {Donnelly},\ and\ \citenamefont {Kempf}}]{Farming}%
  \BibitemOpen
  \bibfield  {author} {\bibinfo {author} {\bibfnamefont {E.}~\bibnamefont
  {Mart\'{i}n-Mart\'{i}nez}}, \bibinfo {author} {\bibfnamefont {E.~G.}\
  \bibnamefont {Brown}}, \bibinfo {author} {\bibfnamefont {W.}~\bibnamefont
  {Donnelly}}, \ and\ \bibinfo {author} {\bibfnamefont {A.}~\bibnamefont
  {Kempf}},\ }\href {\doibase 10.1103/PhysRevA.88.052310} {\bibfield  {journal}
  {\bibinfo  {journal} {Phys. Rev. A}\ }\textbf {\bibinfo {volume} {88}},\
  \bibinfo {pages} {052310} (\bibinfo {year} {2013})}\BibitemShut {NoStop}%
\bibitem [{\citenamefont {Fulling}(1973)}]{fulling_nonuniqueness_1973}%
  \BibitemOpen
  \bibfield  {author} {\bibinfo {author} {\bibfnamefont {S.~A.}\ \bibnamefont
  {Fulling}},\ }\href {\doibase 10.1103/PhysRevD.7.2850} {\bibfield  {journal}
  {\bibinfo  {journal} {Phys. Rev. D}\ }\textbf {\bibinfo {volume} {7}},\
  \bibinfo {pages} {2850} (\bibinfo {year} {1973})}\BibitemShut {NoStop}%
\bibitem [{\citenamefont {Davies}(1975)}]{davies_scalar_1975}%
  \BibitemOpen
  \bibfield  {author} {\bibinfo {author} {\bibfnamefont {P.~C.~W.}\
  \bibnamefont {Davies}},\ }\href {\doibase 10.1088/0305-4470/8/4/022}
  {\bibfield  {journal} {\bibinfo  {journal} {J. Phys. A}\ }\textbf {\bibinfo
  {volume} {8}},\ \bibinfo {pages} {609} (\bibinfo {year} {1975})}\BibitemShut
  {NoStop}%
\bibitem [{\citenamefont {Grove}\ and\ \citenamefont
  {Ottewill}(1983)}]{grove_notes_1983}%
  \BibitemOpen
  \bibfield  {author} {\bibinfo {author} {\bibfnamefont {P.~G.}\ \bibnamefont
  {Grove}}\ and\ \bibinfo {author} {\bibfnamefont {A.~C.}\ \bibnamefont
  {Ottewill}},\ }\href {\doibase 10.1088/0305-4470/16/16/029} {\bibfield
  {journal} {\bibinfo  {journal} {J. Phys. A}\ }\textbf {\bibinfo {volume}
  {16}},\ \bibinfo {pages} {3905} (\bibinfo {year} {1983})}\BibitemShut
  {NoStop}%
\bibitem [{\citenamefont {Birrell}\ and\ \citenamefont
  {Davies}(1984)}]{birrell_quantum_1984}%
  \BibitemOpen
  \bibfield  {author} {\bibinfo {author} {\bibfnamefont {N.~D.}\ \bibnamefont
  {Birrell}}\ and\ \bibinfo {author} {\bibfnamefont {P.~C.~W.}\ \bibnamefont
  {Davies}},\ }\href@noop {} {\emph {\bibinfo {title} {Quantum Fields in Curved
  Space}}}\ (\bibinfo  {publisher} {Cambridge University Press},\ \bibinfo
  {address} {Cambridge},\ \bibinfo {year} {1984})\BibitemShut {NoStop}%
\bibitem [{\citenamefont {Crispino}\ \emph {et~al.}(2008)\citenamefont
  {Crispino}, \citenamefont {Higuchi},\ and\ \citenamefont
  {Matsas}}]{crispino_unruh_2008}%
  \BibitemOpen
  \bibfield  {author} {\bibinfo {author} {\bibfnamefont {L.~C.~B.}\
  \bibnamefont {Crispino}}, \bibinfo {author} {\bibfnamefont {A.}~\bibnamefont
  {Higuchi}}, \ and\ \bibinfo {author} {\bibfnamefont {G.~E.~A.}\ \bibnamefont
  {Matsas}},\ }\href {\doibase 10.1103/RevModPhys.80.787} {\bibfield  {journal}
  {\bibinfo  {journal} {Rev. Mod. Phys.}\ }\textbf {\bibinfo {volume} {80}},\
  \bibinfo {pages} {787} (\bibinfo {year} {2008})}\BibitemShut {NoStop}%
\bibitem [{\citenamefont {Hinton}(1983)}]{hinton_particle_1983}%
  \BibitemOpen
  \bibfield  {author} {\bibinfo {author} {\bibfnamefont {K.~J.}\ \bibnamefont
  {Hinton}},\ }\href {\doibase 10.1088/0305-4470/16/9/018} {\bibfield
  {journal} {\bibinfo  {journal} {J. Phys. A}\ }\textbf {\bibinfo {volume}
  {16}},\ \bibinfo {pages} {1937} (\bibinfo {year} {1983})}\BibitemShut
  {NoStop}%
\bibitem [{\citenamefont {Hinton}(1984)}]{hinton_particle_1984}%
  \BibitemOpen
  \bibfield  {author} {\bibinfo {author} {\bibfnamefont {K.~J.}\ \bibnamefont
  {Hinton}},\ }\href {\doibase 10.1088/0264-9381/1/1/006} {\bibfield  {journal}
  {\bibinfo  {journal} {Class. Quantum Grav.}\ }\textbf {\bibinfo {volume}
  {1}},\ \bibinfo {pages} {27} (\bibinfo {year} {1984})}\BibitemShut {NoStop}%
\bibitem [{\citenamefont {Sriramkumar}(2001)}]{sriramkumar_response_2001}%
  \BibitemOpen
  \bibfield  {author} {\bibinfo {author} {\bibfnamefont {L.}~\bibnamefont
  {Sriramkumar}},\ }\href@noop {} {\enquote {\bibinfo {title} {On the response
  of non-linearly coupled, accelerated detectors in odd-dimensional flat
  spacetimes},}\ } (\bibinfo {year} {2001}),\ \Eprint
  {http://arxiv.org/abs/gr-qc/0106054} {arXiv:gr-qc/0106054} \BibitemShut
  {NoStop}%
\bibitem [{\citenamefont {Candelas}\ and\ \citenamefont
  {Sciama}(1977)}]{candelas_irreversible_1977}%
  \BibitemOpen
  \bibfield  {author} {\bibinfo {author} {\bibfnamefont {P.}~\bibnamefont
  {Candelas}}\ and\ \bibinfo {author} {\bibfnamefont {D.~W.}\ \bibnamefont
  {Sciama}},\ }\href {\doibase 10.1103/PhysRevLett.38.1372} {\bibfield
  {journal} {\bibinfo  {journal} {Phys. Rev. Lett.}\ }\textbf {\bibinfo
  {volume} {38}},\ \bibinfo {pages} {1372} (\bibinfo {year}
  {1977})}\BibitemShut {NoStop}%
\bibitem [{\citenamefont {Takagi}(1986)}]{takagi_vacuum_1986}%
  \BibitemOpen
  \bibfield  {author} {\bibinfo {author} {\bibfnamefont {S.}~\bibnamefont
  {Takagi}},\ }\href {\doibase 10.1143/PTPS.88.1} {\bibfield  {journal}
  {\bibinfo  {journal} {Prog. Theor. Phys. Suppl.}\ }\textbf {\bibinfo {volume}
  {88}},\ \bibinfo {pages} {1} (\bibinfo {year} {1986})}\BibitemShut {NoStop}%
\bibitem [{\citenamefont {Reznik}\ \emph {et~al.}(2005)\citenamefont {Reznik},
  \citenamefont {Retzker},\ and\ \citenamefont
  {Silman}}]{reznik_violating_2005}%
  \BibitemOpen
  \bibfield  {author} {\bibinfo {author} {\bibfnamefont {B.}~\bibnamefont
  {Reznik}}, \bibinfo {author} {\bibfnamefont {A.}~\bibnamefont {Retzker}}, \
  and\ \bibinfo {author} {\bibfnamefont {J.}~\bibnamefont {Silman}},\ }\href
  {\doibase 10.1103/PhysRevA.71.042104} {\bibfield  {journal} {\bibinfo
  {journal} {Phys. Rev. A}\ }\textbf {\bibinfo {volume} {71}},\ \bibinfo
  {pages} {042104} (\bibinfo {year} {2005})}\BibitemShut {NoStop}%
\bibitem [{\citenamefont {Cliche}\ and\ \citenamefont
  {Kempf}(2011)}]{cliche_vacuum_2011}%
  \BibitemOpen
  \bibfield  {author} {\bibinfo {author} {\bibfnamefont {M.}~\bibnamefont
  {Cliche}}\ and\ \bibinfo {author} {\bibfnamefont {A.}~\bibnamefont {Kempf}},\
  }\href {\doibase 10.1103/PhysRevD.83.045019} {\bibfield  {journal} {\bibinfo
  {journal} {Phys. Rev. D}\ }\textbf {\bibinfo {volume} {83}},\ \bibinfo
  {pages} {045019} (\bibinfo {year} {2011})}\BibitemShut {NoStop}%
\bibitem [{\citenamefont {Mart{\'\i}n-Mart{\'\i}nez}\ \emph
  {et~al.}(2013{\natexlab{a}})\citenamefont {Mart{\'\i}n-Mart{\'\i}nez},
  \citenamefont {Brown}, \citenamefont {Donnelly},\ and\ \citenamefont
  {Kempf}}]{martin-martinez_sustainable_2013}%
  \BibitemOpen
  \bibfield  {author} {\bibinfo {author} {\bibfnamefont {E.}~\bibnamefont
  {Mart{\'\i}n-Mart{\'\i}nez}}, \bibinfo {author} {\bibfnamefont {E.~G.}\
  \bibnamefont {Brown}}, \bibinfo {author} {\bibfnamefont {W.}~\bibnamefont
  {Donnelly}}, \ and\ \bibinfo {author} {\bibfnamefont {A.}~\bibnamefont
  {Kempf}},\ }\href {\doibase 10.1103/PhysRevA.88.052310} {\bibfield  {journal}
  {\bibinfo  {journal} {Phys. Rev. A}\ }\textbf {\bibinfo {volume} {88}},\
  \bibinfo {pages} {052310} (\bibinfo {year} {2013}{\natexlab{a}})}\BibitemShut
  {NoStop}%
\bibitem [{\citenamefont {Lin}\ and\ \citenamefont
  {Hu}(2010)}]{lin_entanglement_2010}%
  \BibitemOpen
  \bibfield  {author} {\bibinfo {author} {\bibfnamefont {S.-Y.}\ \bibnamefont
  {Lin}}\ and\ \bibinfo {author} {\bibfnamefont {B.~L.}\ \bibnamefont {Hu}},\
  }\href {\doibase 10.1103/PhysRevD.81.045019} {\bibfield  {journal} {\bibinfo
  {journal} {Phys. Rev. D}\ }\textbf {\bibinfo {volume} {81}},\ \bibinfo
  {pages} {045019} (\bibinfo {year} {2010})}\BibitemShut {NoStop}%
\bibitem [{\citenamefont {Olson}\ and\ \citenamefont
  {Ralph}(2011)}]{olson_entanglement_2011}%
  \BibitemOpen
  \bibfield  {author} {\bibinfo {author} {\bibfnamefont {S.~J.}\ \bibnamefont
  {Olson}}\ and\ \bibinfo {author} {\bibfnamefont {T.~C.}\ \bibnamefont
  {Ralph}},\ }\href {\doibase 10.1103/PhysRevLett.106.110404} {\bibfield
  {journal} {\bibinfo  {journal} {Phys. Rev. Lett.}\ }\textbf {\bibinfo
  {volume} {106}},\ \bibinfo {pages} {110404} (\bibinfo {year}
  {2011})}\BibitemShut {NoStop}%
\bibitem [{\citenamefont {Cliche}\ and\ \citenamefont
  {Kempf}(2010)}]{cliche_relativistic_2010}%
  \BibitemOpen
  \bibfield  {author} {\bibinfo {author} {\bibfnamefont {M.}~\bibnamefont
  {Cliche}}\ and\ \bibinfo {author} {\bibfnamefont {A.}~\bibnamefont {Kempf}},\
  }\href {\doibase 10.1103/PhysRevA.81.012330} {\bibfield  {journal} {\bibinfo
  {journal} {Phys. Rev. A}\ }\textbf {\bibinfo {volume} {81}},\ \bibinfo
  {pages} {012330} (\bibinfo {year} {2010})}\BibitemShut {NoStop}%
\bibitem [{\citenamefont {Mart{\'\i}n-Mart{\'\i}nez}\ \emph
  {et~al.}(2013{\natexlab{b}})\citenamefont {Mart{\'\i}n-Mart{\'\i}nez},
  \citenamefont {Montero},\ and\ \citenamefont {del
  Rey}}]{martin-martinez_wavepacket_2013}%
  \BibitemOpen
  \bibfield  {author} {\bibinfo {author} {\bibfnamefont {E.}~\bibnamefont
  {Mart{\'\i}n-Mart{\'\i}nez}}, \bibinfo {author} {\bibfnamefont
  {M.}~\bibnamefont {Montero}}, \ and\ \bibinfo {author} {\bibfnamefont
  {M.}~\bibnamefont {del Rey}},\ }\href {\doibase 10.1103/PhysRevD.87.064038}
  {\bibfield  {journal} {\bibinfo  {journal} {Phys. Rev. D}\ }\textbf {\bibinfo
  {volume} {87}},\ \bibinfo {pages} {064038} (\bibinfo {year}
  {2013}{\natexlab{b}})}\BibitemShut {NoStop}%
\bibitem [{\citenamefont {Alhambra}\ \emph {et~al.}(2014)\citenamefont
  {Alhambra}, \citenamefont {Kempf},\ and\ \citenamefont
  {Mart\'in-Mart\'inez}}]{Alhambra2013}%
  \BibitemOpen
  \bibfield  {author} {\bibinfo {author} {\bibfnamefont {A.~M.}\ \bibnamefont
  {Alhambra}}, \bibinfo {author} {\bibfnamefont {A.}~\bibnamefont {Kempf}}, \
  and\ \bibinfo {author} {\bibfnamefont {E.}~\bibnamefont
  {Mart\'in-Mart\'inez}},\ }\href {\doibase 10.1103/PhysRevA.89.033835}
  {\bibfield  {journal} {\bibinfo  {journal} {Phys. Rev. A}\ }\textbf {\bibinfo
  {volume} {89}},\ \bibinfo {pages} {033835} (\bibinfo {year}
  {2014})}\BibitemShut {NoStop}%
\bibitem [{\citenamefont {Soffel}\ \emph {et~al.}(1980)\citenamefont {Soffel},
  \citenamefont {M{\"u}ller},\ and\ \citenamefont
  {Greiner}}]{soffel_dirac_1980}%
  \BibitemOpen
  \bibfield  {author} {\bibinfo {author} {\bibfnamefont {M.}~\bibnamefont
  {Soffel}}, \bibinfo {author} {\bibfnamefont {B.}~\bibnamefont {M{\"u}ller}},
  \ and\ \bibinfo {author} {\bibfnamefont {W.}~\bibnamefont {Greiner}},\ }\href
  {\doibase 10.1103/PhysRevD.22.1935} {\bibfield  {journal} {\bibinfo
  {journal} {Phys. Rev. D}\ }\textbf {\bibinfo {volume} {22}},\ \bibinfo
  {pages} {1935} (\bibinfo {year} {1980})}\BibitemShut {NoStop}%
\bibitem [{\citenamefont {Hawking}(1975)}]{hawking_particle_1975}%
  \BibitemOpen
  \bibfield  {author} {\bibinfo {author} {\bibfnamefont {S.~W.}\ \bibnamefont
  {Hawking}},\ }\href {\doibase 10.1007/BF02345020} {\bibfield  {journal}
  {\bibinfo  {journal} {Commun. Math. Phys.}\ }\textbf {\bibinfo {volume}
  {43}},\ \bibinfo {pages} {199} (\bibinfo {year} {1975})}\BibitemShut
  {NoStop}%
\bibitem [{\citenamefont {Montero}\ and\ \citenamefont
  {Mart{\'\i}n-Mart{\'\i}nez}(2011)}]{montero_fermionic_2011}%
  \BibitemOpen
  \bibfield  {author} {\bibinfo {author} {\bibfnamefont {M.}~\bibnamefont
  {Montero}}\ and\ \bibinfo {author} {\bibfnamefont {E.}~\bibnamefont
  {Mart{\'\i}n-Mart{\'\i}nez}},\ }\href {\doibase 10.1103/PhysRevA.83.062323}
  {\bibfield  {journal} {\bibinfo  {journal} {Phys. Rev. A}\ }\textbf {\bibinfo
  {volume} {83}},\ \bibinfo {pages} {062323} (\bibinfo {year}
  {2011})}\BibitemShut {NoStop}%
\bibitem [{\citenamefont {Iyer}\ and\ \citenamefont
  {Kumar}(1980)}]{iyer_detection_1980}%
  \BibitemOpen
  \bibfield  {author} {\bibinfo {author} {\bibfnamefont {B.~R.}\ \bibnamefont
  {Iyer}}\ and\ \bibinfo {author} {\bibfnamefont {A.}~\bibnamefont {Kumar}},\
  }\href {\doibase 10.1088/0305-4470/13/2/015} {\bibfield  {journal} {\bibinfo
  {journal} {J. Phys. A}\ }\textbf {\bibinfo {volume} {13}},\ \bibinfo {pages}
  {469} (\bibinfo {year} {1980})}\BibitemShut {NoStop}%
\bibitem [{\citenamefont {Takagi}(1985)}]{takagi_response_1985}%
  \BibitemOpen
  \bibfield  {author} {\bibinfo {author} {\bibfnamefont {S.}~\bibnamefont
  {Takagi}},\ }\href {\doibase 10.1143/PTP.74.501} {\bibfield  {journal}
  {\bibinfo  {journal} {Prog. Theor. Phys.}\ }\textbf {\bibinfo {volume}
  {74}},\ \bibinfo {pages} {501} (\bibinfo {year} {1985})}\BibitemShut
  {NoStop}%
\bibitem [{\citenamefont {Sriramkumar}\ and\ \citenamefont
  {Padmanabhan}(1996)}]{sriramkumar_finite-time_1996}%
  \BibitemOpen
  \bibfield  {author} {\bibinfo {author} {\bibfnamefont {L.}~\bibnamefont
  {Sriramkumar}}\ and\ \bibinfo {author} {\bibfnamefont {T.}~\bibnamefont
  {Padmanabhan}},\ }\href {\doibase 10.1088/0264-9381/13/8/005} {\bibfield
  {journal} {\bibinfo  {journal} {Class. Quantum Grav.}\ }\textbf {\bibinfo
  {volume} {13}},\ \bibinfo {pages} {2061} (\bibinfo {year}
  {1996})}\BibitemShut {NoStop}%
\bibitem [{\citenamefont {Louko}\ and\ \citenamefont
  {Satz}(2008)}]{louko_transition_2008}%
  \BibitemOpen
  \bibfield  {author} {\bibinfo {author} {\bibfnamefont {J.}~\bibnamefont
  {Louko}}\ and\ \bibinfo {author} {\bibfnamefont {A.}~\bibnamefont {Satz}},\
  }\href {\doibase 10.1088/0264-9381/25/5/055012} {\bibfield  {journal}
  {\bibinfo  {journal} {Class. Quantum Grav.}\ }\textbf {\bibinfo {volume}
  {25}},\ \bibinfo {pages} {055012} (\bibinfo {year} {2008})}\BibitemShut
  {NoStop}%
\bibitem [{\citenamefont {Diaz}\ and\ \citenamefont
  {Stephany}(2003)}]{diaz_radiative_2002}%
  \BibitemOpen
  \bibfield  {author} {\bibinfo {author} {\bibfnamefont {D.~E.}\ \bibnamefont
  {Diaz}}\ and\ \bibinfo {author} {\bibfnamefont {J.}~\bibnamefont
  {Stephany}},\ }\href {http://arxiv.org/abs/gr-qc/0201096} {\bibfield
  {journal} {\bibinfo  {journal} {Rev. Mex. Fis.}\ }\textbf {\bibinfo {volume}
  {49S3}},\ \bibinfo {pages} {120} (\bibinfo {year} {2003})}\BibitemShut
  {NoStop}%
\bibitem [{\citenamefont {Langlois}(2006)}]{langlois_causal_2006}%
  \BibitemOpen
  \bibfield  {author} {\bibinfo {author} {\bibfnamefont {P.}~\bibnamefont
  {Langlois}},\ }\href {\doibase 10.1016/j.aop.2006.01.013} {\bibfield
  {journal} {\bibinfo  {journal} {Ann. Phys. (Amsterdam)}\ }\textbf {\bibinfo
  {volume} {321}},\ \bibinfo {pages} {2027} (\bibinfo {year}
  {2006})}\BibitemShut {NoStop}%
\bibitem [{\citenamefont {B{\'e}ssa}\ \emph {et~al.}(2012)\citenamefont
  {B{\'e}ssa}, \citenamefont {Due{\~n}as},\ and\ \citenamefont
  {Svaiter}}]{bessa_accelerated_2012}%
  \BibitemOpen
  \bibfield  {author} {\bibinfo {author} {\bibfnamefont {C.~H.~G.}\
  \bibnamefont {B{\'e}ssa}}, \bibinfo {author} {\bibfnamefont {J.~G.}\
  \bibnamefont {Due{\~n}as}}, \ and\ \bibinfo {author} {\bibfnamefont {N.~F.}\
  \bibnamefont {Svaiter}},\ }\href {\doibase 10.1088/0264-9381/29/21/215011}
  {\bibfield  {journal} {\bibinfo  {journal} {Class. Quantum Grav.}\ }\textbf
  {\bibinfo {volume} {29}},\ \bibinfo {pages} {215011} (\bibinfo {year}
  {2012})}\BibitemShut {NoStop}%
\bibitem [{\citenamefont {Harikumar}\ and\ \citenamefont
  {Verma}(2013)}]{harikumar_uniformly_2013}%
  \BibitemOpen
  \bibfield  {author} {\bibinfo {author} {\bibfnamefont {E.}~\bibnamefont
  {Harikumar}}\ and\ \bibinfo {author} {\bibfnamefont {R.}~\bibnamefont
  {Verma}},\ }\href {\doibase 10.1142/S0217732313500636} {\bibfield  {journal}
  {\bibinfo  {journal} {Mod. Phys. Lett. A}\ }\textbf {\bibinfo {volume}
  {28}},\ \bibinfo {pages} {1350063} (\bibinfo {year} {2013})}\BibitemShut
  {NoStop}%
\bibitem [{\citenamefont {Mart{\'\i}n-Mart{\'\i}nez}\ and\ \citenamefont
  {Louko}(2014)}]{martin-martinez_particle_2014}%
  \BibitemOpen
  \bibfield  {author} {\bibinfo {author} {\bibfnamefont {E.}~\bibnamefont
  {Mart{\'\i}n-Mart{\'\i}nez}}\ and\ \bibinfo {author} {\bibfnamefont
  {J.}~\bibnamefont {Louko}},\ }\href {\doibase 10.1103/PhysRevD.90.024015}
  {\bibfield  {journal} {\bibinfo  {journal} {Phys. Rev. D}\ }\textbf {\bibinfo
  {volume} {90}},\ \bibinfo {pages} {024015} (\bibinfo {year}
  {2014})}\BibitemShut {NoStop}%
\bibitem [{\citenamefont {Alonso}\ \emph {et~al.}(1997)\citenamefont {Alonso},
  \citenamefont {Vincenzo},\ and\ \citenamefont
  {Mondino}}]{alonso_boundary_1997}%
  \BibitemOpen
  \bibfield  {author} {\bibinfo {author} {\bibfnamefont {V.}~\bibnamefont
  {Alonso}}, \bibinfo {author} {\bibfnamefont {S.~D.}\ \bibnamefont
  {Vincenzo}}, \ and\ \bibinfo {author} {\bibfnamefont {L.}~\bibnamefont
  {Mondino}},\ }\href {\doibase 10.1088/0143-0807/18/5/001} {\bibfield
  {journal} {\bibinfo  {journal} {Eur. J. Phys.}\ }\textbf {\bibinfo {volume}
  {18}},\ \bibinfo {pages} {315} (\bibinfo {year} {1997})}\BibitemShut
  {NoStop}%
\bibitem [{\citenamefont {Sakurai}\ and\ \citenamefont
  {Napolitano}(2011)}]{sakurai_modern_2011}%
  \BibitemOpen
  \bibfield  {author} {\bibinfo {author} {\bibfnamefont {J.~J.}\ \bibnamefont
  {Sakurai}}\ and\ \bibinfo {author} {\bibfnamefont {J.}~\bibnamefont
  {Napolitano}},\ }\href@noop {} {\emph {\bibinfo {title} {Modern Quantum
  Mechanics}}},\ \bibinfo {edition} {2nd}\ ed.\ (\bibinfo  {publisher}
  {Addison-Wesley},\ \bibinfo {address} {New York},\ \bibinfo {year}
  {2011})\BibitemShut {NoStop}%
\bibitem [{\citenamefont {Bjorken}\ and\ \citenamefont
  {Drell}(1965)}]{bjorken_relativistic_1965}%
  \BibitemOpen
  \bibfield  {author} {\bibinfo {author} {\bibfnamefont {J.~D.}\ \bibnamefont
  {Bjorken}}\ and\ \bibinfo {author} {\bibfnamefont {S.~D.}\ \bibnamefont
  {Drell}},\ }\href@noop {} {\emph {\bibinfo {title} {Relativistic Quantum
  Fields}}}\ (\bibinfo  {publisher} {McGraw-Hill},\ \bibinfo {address} {New
  York},\ \bibinfo {year} {1965})\BibitemShut {NoStop}%
\bibitem [{\citenamefont {Svaiter}\ and\ \citenamefont
  {Svaiter}(1992)}]{svaiter_inertial_1992}%
  \BibitemOpen
  \bibfield  {author} {\bibinfo {author} {\bibfnamefont {B.~F.}\ \bibnamefont
  {Svaiter}}\ and\ \bibinfo {author} {\bibfnamefont {N.~F.}\ \bibnamefont
  {Svaiter}},\ }\href {\doibase 10.1103/PhysRevD.46.5267} {\bibfield  {journal}
  {\bibinfo  {journal} {Phys. Rev. D}\ }\textbf {\bibinfo {volume} {46}},\
  \bibinfo {pages} {5267} (\bibinfo {year} {1992})}\BibitemShut {NoStop}%
\bibitem [{\citenamefont {Greiner}\ and\ \citenamefont
  {Reinhardt}(2008)}]{greiner_quantum_2008}%
  \BibitemOpen
  \bibfield  {author} {\bibinfo {author} {\bibfnamefont {W.}~\bibnamefont
  {Greiner}}\ and\ \bibinfo {author} {\bibfnamefont {J.}~\bibnamefont
  {Reinhardt}},\ }\href@noop {} {\emph {\bibinfo {title} {Quantum
  Electrodynamics}}}\ (\bibinfo  {publisher} {Springer},\ \bibinfo {address}
  {New York},\ \bibinfo {year} {2008})\BibitemShut {NoStop}%
\bibitem [{\citenamefont {Brenna}\ \emph {et~al.}(2013)\citenamefont {Brenna},
  \citenamefont {Brown}, \citenamefont {Mann},\ and\ \citenamefont
  {Mart\'{i}n-Mart\'{i}nez}}]{martin-martinez_Unruh2013}%
  \BibitemOpen
  \bibfield  {author} {\bibinfo {author} {\bibfnamefont {W.~G.}\ \bibnamefont
  {Brenna}}, \bibinfo {author} {\bibfnamefont {E.~G.}\ \bibnamefont {Brown}},
  \bibinfo {author} {\bibfnamefont {R.~B.}\ \bibnamefont {Mann}}, \ and\
  \bibinfo {author} {\bibfnamefont {E.}~\bibnamefont
  {Mart\'{i}n-Mart\'{i}nez}},\ }\href {\doibase 10.1103/PhysRevD.88.064031}
  {\bibfield  {journal} {\bibinfo  {journal} {Phys. Rev. D}\ }\textbf {\bibinfo
  {volume} {88}},\ \bibinfo {pages} {064031} (\bibinfo {year}
  {2013})}\BibitemShut {NoStop}%
\bibitem [{\citenamefont {Peskin}\ and\ \citenamefont
  {Schroeder}(1995)}]{peskin_introduction_1995}%
  \BibitemOpen
  \bibfield  {author} {\bibinfo {author} {\bibfnamefont {M.~E.}\ \bibnamefont
  {Peskin}}\ and\ \bibinfo {author} {\bibfnamefont {D.~V.}\ \bibnamefont
  {Schroeder}},\ }\href@noop {} {\emph {\bibinfo {title} {An Introduction To
  Quantum Field Theory}}}\ (\bibinfo  {publisher} {Addison-Wesley},\ \bibinfo
  {address} {New York},\ \bibinfo {year} {1995})\BibitemShut {NoStop}%
\bibitem [{\citenamefont {Br\'adler}\ and\ \citenamefont
  {J\'auregui}(2012)}]{Bradler2012}%
  \BibitemOpen
  \bibfield  {author} {\bibinfo {author} {\bibfnamefont {K.}~\bibnamefont
  {Br\'adler}}\ and\ \bibinfo {author} {\bibfnamefont {R.}~\bibnamefont
  {J\'auregui}},\ }\href {\doibase 10.1103/PhysRevA.85.016301} {\bibfield
  {journal} {\bibinfo  {journal} {Phys. Rev. A}\ }\textbf {\bibinfo {volume}
  {85}},\ \bibinfo {pages} {016301} (\bibinfo {year} {2012})}\BibitemShut
  {NoStop}%
\bibitem [{\citenamefont {Montero}\ and\ \citenamefont
  {Mart\'{i}n-Mart\'{i}nez}(2012)}]{Montero2012b}%
  \BibitemOpen
  \bibfield  {author} {\bibinfo {author} {\bibfnamefont {M.}~\bibnamefont
  {Montero}}\ and\ \bibinfo {author} {\bibfnamefont {E.}~\bibnamefont
  {Mart\'{i}n-Mart\'{i}nez}},\ }\href {\doibase 10.1103/PhysRevA.85.016302}
  {\bibfield  {journal} {\bibinfo  {journal} {Phys. Rev. A}\ }\textbf {\bibinfo
  {volume} {85}},\ \bibinfo {pages} {016302} (\bibinfo {year}
  {2012})}\BibitemShut {NoStop}%
\bibitem [{\citenamefont {Hodgkinson}\ \emph {et~al.}(2014)\citenamefont
  {Hodgkinson}, \citenamefont {Louko},\ and\ \citenamefont
  {Ottewill}}]{Jormacircular}%
  \BibitemOpen
  \bibfield  {author} {\bibinfo {author} {\bibfnamefont {L.}~\bibnamefont
  {Hodgkinson}}, \bibinfo {author} {\bibfnamefont {J.}~\bibnamefont {Louko}}, \
  and\ \bibinfo {author} {\bibfnamefont {A.~C.}\ \bibnamefont {Ottewill}},\
  }\href {\doibase 10.1103/PhysRevD.89.104002} {\bibfield  {journal} {\bibinfo
  {journal} {Phys. Rev. D}\ }\textbf {\bibinfo {volume} {89}},\ \bibinfo
  {pages} {104002} (\bibinfo {year} {2014})}\BibitemShut {NoStop}%
\bibitem [{\citenamefont {Ng}\ \emph {et~al.}(2014)\citenamefont {Ng},
  \citenamefont {Hodgkinson}, \citenamefont {Louko}, \citenamefont {Mann},\
  and\ \citenamefont {Mart\'{i}n-Mart\'{i}nez}}]{Keith2014}%
  \BibitemOpen
  \bibfield  {author} {\bibinfo {author} {\bibfnamefont {K.~K.}\ \bibnamefont
  {Ng}}, \bibinfo {author} {\bibfnamefont {L.}~\bibnamefont {Hodgkinson}},
  \bibinfo {author} {\bibfnamefont {J.}~\bibnamefont {Louko}}, \bibinfo
  {author} {\bibfnamefont {R.~B.}\ \bibnamefont {Mann}}, \ and\ \bibinfo
  {author} {\bibfnamefont {E.}~\bibnamefont {Mart\'{i}n-Mart\'{i}nez}},\ }\href
  {\doibase 10.1103/PhysRevD.90.064003} {\bibfield  {journal} {\bibinfo
  {journal} {Phys. Rev. D}\ }\textbf {\bibinfo {volume} {90}},\ \bibinfo
  {pages} {064003} (\bibinfo {year} {2014})}\BibitemShut {NoStop}%
\bibitem [{\citenamefont {Ju\'arez-Aubry}\ and\ \citenamefont
  {Louko}(2014)}]{JormaJuarez}%
  \BibitemOpen
  \bibfield  {author} {\bibinfo {author} {\bibfnamefont {B.~A.}\ \bibnamefont
  {Ju\'arez-Aubry}}\ and\ \bibinfo {author} {\bibfnamefont {J.}~\bibnamefont
  {Louko}},\ }\href {http://stacks.iop.org/0264-9381/31/i=24/a=245007}
  {\bibfield  {journal} {\bibinfo  {journal} {Class. Quantum Grav.}\ }\textbf
  {\bibinfo {volume} {31}},\ \bibinfo {pages} {245007} (\bibinfo {year}
  {2014})}\BibitemShut {NoStop}%
\bibitem [{\citenamefont {Brenna}\ \emph {et~al.}()\citenamefont {Brenna},
  \citenamefont {Mann},\ and\ \citenamefont
  {Mart\'{i}n-Mart\'{i}nez}}]{Anti_Unruh2013}%
  \BibitemOpen
  \bibfield  {author} {\bibinfo {author} {\bibfnamefont {W.~G.}\ \bibnamefont
  {Brenna}}, \bibinfo {author} {\bibfnamefont {R.~B.}\ \bibnamefont {Mann}}, \
  and\ \bibinfo {author} {\bibfnamefont {E.}~\bibnamefont
  {Mart\'{i}n-Mart\'{i}nez}},\ }\href@noop {} {\enquote {\bibinfo {title}
  {Anti-{U}nruh {P}henomena.}}\ }\Eprint {http://arxiv.org/abs/1504.02468}
  {arXiv:1504.02468 [quant-ph]} \BibitemShut {NoStop}%
\end{thebibliography}%

\end{document}